%% file: main.tex
\documentclass[10pt, aps,pra,onecolumn,superscriptaddress,nofootinbib]{revtex4-2}

\input{preamble.tex}
\input{figures.tex}

\usepackage{etoolbox}

\AtBeginDocument{%
  \setlength{\abovedisplayskip}{5pt}%
  \setlength{\belowdisplayskip}{5pt}%
  \setlength{\abovedisplayshortskip}{5pt}%
  \setlength{\belowdisplayshortskip}{5pt}%
}

\newcommand{\manualtocentry}[2]{%
  \noindent
  \begingroup
  \hypersetup{linkcolor=mycolr!50!black}%
  \hyperref[#1]{\ref*{#1}.~#2}%
  \nobreak{\dotfill}\nobreak%
  \pageref{#1}%
  \\[0.1cm]
  \endgroup
}

\newcommand{\manualtocsubentry}[2]{%
  \phantom{~}\hspace{0.75em}%
  \hyperref[#1]{\ref*{#1}.~#2}%
  \nobreak{\dotfill}\nobreak%
  \pageref{#1}%
  \\[0.1cm]
}

\newcommand{\altmanualtocentry}[2]{%
  \noindent
  \hyperref[#1]{#2}%
  \nobreak{\dotfill}\nobreak%
  \pageref{#1}%
  \\[0.1cm]
}
\newcommand{\manualtocentryapp}[2]{%
  \noindent%
  \hyperref[#1]{\ref*{#1}.~#2}%
  \nobreak{\dotfill}\nobreak%
  \pageref{#1}%
  \\[0.1cm]
}

\begin{document}
\title{On the Computational Complexity of Guided Berry Phase Estimation}
\author{Gabriel Waite}\email{gabe.waite@uts.edu.au}
\affiliation{CQSI, School of Computer Science, University of Technology Sydney, NSW 2007, Australia}
\begin{abstract}
    We prove that deciding the Berry phase for parameterised $2$-local qubit Hamiltonians is \clw{BQP}{complete} when presented with a classical description of a guiding state, promised to overlap with the ground state of the system.
    Our results extend to systems with weighted Heisenberg interactions and when restricted to a $2$D square or triangular lattice geometry.
    The techniques we develop leverage the Schrieffer--Wolff transformation, typically used in the construction of perturbative gadget reductions for local Hamiltonian problems, extending it to parameterised families of Hamiltonians.
    We demonstrate that there exists a choice of parameterised simulator Hamiltonians whose Berry phase well-approximates that of a parameterised target family.
    Using the perturbative gadget reduction framework of Oliveira and Terhal and of Schuch and Verstraete, we adapt the arguments to parameterised interactions and demonstrate the error bounds in the resulting simulation can be controlled.
    Additionally, we provide an explicit proof that families of $1$-local Hamiltonians have a Berry phase that can be efficiently computed to inverse-polynomial precision.
    This establishes a complexity transition between $1$-local and $2$-local Hamiltonian families.
\end{abstract}

\maketitle
\vspace{-2em}
\section*{Table of Contents}
{
\twocolumngrid
\begin{spacing}{1}
% \begin{center}
\raggedright
\setlength{\parindent}{0pt}
    \manualtocentry{sec:introduction}{Introduction}
      \manualtocsubentry{sec:prior-work}{Prior Work}
      \manualtocsubentry{sec:summary-of-techniques-and-results}{Summary of Techniques and Results}
      \manualtocsubentry{sec:organisation}{Organisation}
    \manualtocentry{sec:preliminaries}{Preliminaries}
      \manualtocsubentry{sec:liouville-superoperator}{Liouville Superoperator}
      \manualtocsubentry{sec:families-of-parameterised-hamiltonians}{Families of Parameterised Hamiltonians}
      \manualtocsubentry{sec:berry-phase-definition}{Berry Phase Definition}
      \manualtocsubentry{sec:computational-complexity}{Computational Complexity}
      \manualtocsubentry{sec:problem-statement}{Problem Statement}
      \manualtocsubentry{sec:analytic-preliminaries}{Analytic Preliminaries}
    \manualtocentry{sec:uniform-simulation}{Uniform Simulation of Parameterised Hamiltonian Families}
      \manualtocsubentry{sec:definitions-and-basic-properties}{Definitions and Basic Properties}
      \manualtocsubentry{sec:composition-of-simulations}{Composition of Simulations}
      \manualtocsubentry{sec:generalised-composition}{Generalised Composition}
    \manualtocentry{sec:sw-transformation}{Schrieffer-Wolff Transformation for Parameterised Hamiltonian Families}
      \manualtocsubentry{sec:simulator-hamiltonians}{Simulator Hamiltonians}
      \manualtocsubentry{sec:exact-sw-transformation}{Exact Schrieffer--Wolff Transformation}
      \manualtocsubentry{sec:perturbative-coefficients-and-truncation}{Perturbative Coefficients and Truncation}
      \manualtocsubentry{sec:uniform-control-exact-and-truncated-components}{Uniform Control of the Exact and Truncated Components}
      \manualtocsubentry{sec:tensor-product-penalties}{Tensor Product Penalties and Local Perturbations}
    \manualtocentry{sec:berry-phase-simulation}{Berry Phase Simulation}
      \manualtocsubentry{sec:setup-and-assumptions}{Setup and Assumptions}
      \manualtocsubentry{sec:comparing-ground-space-projectors}{Comparing the Ground-Space Projectors}
      \manualtocsubentry{sec:berry-phase-holonomy}{The Berry Phase as a Holonomy}
      \manualtocsubentry{sec:berry-phase-simulation-theorem}{The Berry Phase Simulation Theorem}
    \manualtocentry{sec:perturbative-gadget-reductions}{Perturbative Gadget Reductions}
      \manualtocsubentry{sec:target-simulator-hamiltonian-family-preliminaries}{Target and Simulator Preliminaries}
      \manualtocsubentry{sec:poly-control-criterion}{A Polynomial-Control Criterion}
      \manualtocsubentry{sec:first-order-reductions}{First-Order Reductions}
      \manualtocsubentry{sec:second-order-reductions}{Second-Order Reductions}
      \manualtocsubentry{sec:third-order-reductions}{Third-Order Reductions}
      \manualtocsubentry{sec:gadgetised-two-local-driving-terms}{Gadgetised Sums of Two-Local Driving Terms}
    \manualtocentry{sec:extended-hardness-results}{Extended Hardness Results}
      \manualtocsubentry{sec:preserving-parameters}{Preserving Problem Parameters}
      \manualtocsubentry{sec:improved-hardness-to-2-local}{Improved Hardness to 2-Local Hamiltonians}
    \manualtocentry{sec:berry-phase-estimation-parameterised-1-local}{Berry Phase Estimation for Parameterised Families of 1-Local Hamiltonians}
      \manualtocsubentry{sec:main-result-1l}{Main Result}
    \manualtocentry{sec:conclusion}{Conclusion}
      \manualtocsubentry{sec:open-questions}{Open Questions}
    \altmanualtocentry{sec:acknowledgments}{Acknowledgments}
    \altmanualtocentry{sec:bibliography}{References}
    \altmanualtocentry{app:toc}{Appendices}
% \end{center}
\end{spacing}
\onecolumngrid
}
\newpage
\input{sections/introduction.tex}
\input{sections/preliminaries.tex}
\input{sections/uniform-simulation.tex}
\input{sections/sw-transformation.tex}
\input{sections/berry-phase-simulation.tex}
\input{sections/perturbative-gadget-reductions.tex}
\begin{figure}[!ht]
    \centering
    \begin{tikzpicture}
        \pic{promiseinterval};
    \end{tikzpicture}
    \caption{A schematic diagram illustrating the preservation of the promise intervals under a perturbation of size $\upsilon$.
    The black intervals represent the original promise intervals $[a+\g,b-\g]_{2\pi}$ and $[b+\g,a-\g]_{2\pi}$, while the coloured intervals represent the new promise intervals $[a+\g',b-\g']_{2\pi}$ and $[b+\g',a-\g']_{2\pi}$.
    The arrows indicate the shifts in the endpoints of the intervals due to the perturbation $\upsilon$ and the promise gap $\g$.}
    \label{fig:interval-preservation}
\end{figure}
\input{sections/extended-hardness-results.tex}

\begin{figure}[!ht]
\centering
\begin{subfigure}{0.3\textwidth}
    \begin{tikzpicture}
        \pic{localisingvertex};
    \end{tikzpicture}
    \caption{Localising a vertex.}
    \label{fig:localising-vertex}
\end{subfigure}
\hfill
\begin{subfigure}{0.3\textwidth}
    \begin{tikzpicture}
        \pic{subdividedcrossing};
    \end{tikzpicture}
    \caption{A Subdivided crossing edge.}
    \label{fig:subdivided-crossing-edge}
\end{subfigure}
\hfill
\begin{subfigure}{0.3\textwidth}
    \begin{tikzpicture}
        \pic{localisecrossing};
    \end{tikzpicture}
    \caption{Localising a crossing.}
    \label{fig:localising-crossing}
\end{subfigure}
\caption{
    Gadget transformations used to localise vertices, subdivide crossing edges, and remove crossings in the geometric reduction.
    Revised from Ref.~\cite{oliveira2008complexity}.
}
\label{fig:gadgets-for-graph-elements}
\end{figure}

\input{sections/berry-phase-estimation-parameterised-1-local.tex}
\input{sections/conclusion.tex}

\section*{acknowledgments}\label{sec:acknowledgments}
We thank the following individuals for their contributions.
Yuval Sanders for helpful discussions and insights regarding the formalisation of our extended Schrieffer--Wolff transformation.
Jannis Ruh for explaining the polarisation identity and for proposing the connection between it and perturbation gadget decompositions.
Sam Elman and Michael Bremner for providing valuable feedback and discussions on the manuscript and suggestions for improvements on the presentation of the material.
GW is supported by a scholarship from the Sydney Quantum Academy. 
This work was supported by the ARC Centre of Excellence for Quantum Computation and Communication Technology (CQC2T), project number CE170100012.

\textbf{Use of Generative AI Tools.} In light of the rise in declarations regarding the use of generative AI tools in academic writing, we declare the following.
\emph{Authenticity of ideas:}---The ideas presented in this work are original (of human origin) and have been developed independently by the authors (who are human) based on their expertise and understanding of the subject matter.
\emph{Specific tools used:}---Generative AI tools \textbf{[}\textsl{GitHub Copilot, Claude Sonnet 4.6, GPT-5.6 Luna}\textbf{]} were used to assist with language and formatting, this included: grammar correction, sentence/statement restructuring, ensuring consistency in terminology throughout the manuscript, and assistance with summarising mathematical content.
\emph{Effect on the work:}---Suggested changes were used selectively and served predominantly as guidance.
\emph{Quantification of generative AI contribution:}---Relative to the human contribution, the use of generative AI tools for development of proofs and mathematical content was minimal and primarily focused on language and formatting assistance.

\bibliographystyle{apsrev4-2}
\bibliography{ref}\label{sec:bibliography}

\appendix

\section*{Table of Contents}\label{app:toc}

\begin{center}
\begin{minipage}{\textwidth}
\manualtocentryapp{app:preliminaries-proofs}{Proofs of Main Results in Section~\ref{sec:preliminaries} (Preliminaries)}
\manualtocentryapp{app:uniform-simulation-proofs}{Proofs of Main Results in Section~\ref{sec:uniform-simulation} (Uniform Simulation of Parameterised Hamiltonian Families)}
\manualtocentryapp{app:sw-transformation-proofs}{Proofs of Main Results in Section~\ref{sec:sw-transformation} (Schrieffer-Wolff Transformation for Parameterised Hamiltonian Families)}
\manualtocentryapp{app:berry-phase-simulation-proofs}{Proofs of Main Results in Section~\ref{sec:berry-phase-simulation} (Berry Phase Simulation)}
\manualtocentryapp{app:perturbative-gadget-reductions-proofs}{Proofs of Main Results in Section~\ref{sec:perturbative-gadget-reductions} (Perturbative Gadget Reductions)}
\manualtocentryapp{app:extended-hardness-results-proofs}{Proofs of Main Results in Section~\ref{sec:extended-hardness-results} (Extended Hardness Results)}
\manualtocentryapp{app:berry-phase-estimation-parameterised-1-local-proofs}{Proofs of Main Results in Section~\ref{sec:berry-phase-estimation-parameterised-1-local} (Berry Phase Estimation for Parameterised Families of 1-Local Hamiltonians)}
\manualtocentryapp{app:perturbative-gadget-extensions}{Perturbative Gadget Extensions}
\end{minipage}
\end{center}

\input{appendices/section-proofs/preliminaries-proofs.tex}
\input{appendices/section-proofs/uniform-simulation-proofs.tex}
\input{appendices/section-proofs/sw-transformation-proofs.tex}
\input{appendices/section-proofs/berry-phase-simulation-proofs.tex}
\input{appendices/section-proofs/perturbative-gadget-reductions-proofs.tex}
\input{appendices/section-proofs/extended-hardness-results-proofs.tex}
\input{appendices/section-proofs/berry-phase-estimation-parameterised-1-local-proofs.tex}
\input{appendices/perturbative-gadget-extensions.tex}
\end{document}

%% file: preamble.tex
\usepackage[margin=0.75in]{geometry}
\usepackage[T1]{fontenc}
\usepackage{newtxtext}
\usepackage{slantsc}
\usepackage{fix-cm}
\usepackage[verbose=false]{microtype}
\usepackage{setspace}
\usepackage{amsmath}
\usepackage{amssymb}
\usepackage{mathtools}
\usepackage{physics}
\usepackage{cancel}
\usepackage[frak=euler]{mathalpha}

\usepackage{bm}
\numberwithin{equation}{section}

\usepackage{graphicx}
\usepackage[table]{xcolor}
\usepackage{subcaption}
\usepackage{tikz}
\usepackage{pgfplots}
\usepackage[outline]{contour}
\usepackage{pifont}

\usepackage{paralist}
\usepackage{multirow}
\usepackage{tabularx}
\usepackage{booktabs}
\usepackage[shortlabels]{enumitem}

\usepackage{etoolbox}
\usepackage{hyperref}
\usepackage[capitalise]{cleveref}
\usepackage{amsthm}
\usepackage{thmtools}
\usepackage{thm-restate}

\definecolor{mycolr}{RGB}{71,115,117}
\definecolor{mycolrtwo}{RGB}{205,205,205}

\hypersetup{
    colorlinks=true,
    linkcolor=mycolr,
    citecolor=mycolr,
    urlcolor=mycolr
}

\makeatletter
\def\l@subsubsection#1#2{}
\makeatother

\usetikzlibrary{
    arrows.meta,
    angles,
    graphs,
    graphs.standard,
    fit,
    decorations.markings,
    decorations.pathmorphing,
    decorations.pathreplacing,
    decorations.shapes,
    decorations.text,
    patterns,
    patterns.meta,
    shadows,
    shapes.geometric,
    positioning,
    quotes,
    calc
}
\pgfplotsset{compat=1.17}
\contourlength{1.2pt}

\tikzset{bigarrow/.style={
    postaction={
        decorate,
        decoration={
            markings,
            mark=at position \pgfdecoratedpathlength-0.5pt with {\arrow[mycolrtwo,line width=#1] {>}; },
            mark=between positions 0 and \pgfdecoratedpathlength-8pt step 0.5pt with {
                \pgfmathsetmacro\myval{multiply(divide(
                    \pgfkeysvalueof{/pgf/decoration/mark info/distance from start}, \pgfdecoratedpathlength),100)};
                \pgfsetfillcolor{mycolrtwo!\myval!mycolr};
                \pgfpathcircle{\pgfpointorigin}{#1};
                \pgfusepath{fill};}
}}}}

\makeatletter
\input{t1ntxtlf.fd}
\DeclareFontShape{T1}{ntxtlf}{m}{sl}
    {<-> \ntx@scaled ptmro8t}{}
\DeclareFontShape{T1}{ntxtlf}{b}{sl}
    {<-> \ntx@scaled ptmbo8t}{}
\DeclareFontShape{T1}{ntxtlf}{bx}{sl}
    {<-> ssub * ntxtlf/b/sl}{}
\makeatother

\DeclareRobustCommand{\sbseries}{\fontseries{sb}\selectfont}
\DeclareTextFontCommand{\textsb}{\sbseries}

\DeclareMathAlphabet{\mathdcal}{U}{dutchcal}{m}{n}
\DeclareMathAlphabet{\mathdbcal}{U}{dutchcal}{b}{n}

\makeatletter
\DeclareFontEncoding{LS1}{}{}
\DeclareFontEncoding{LS2}{}{\noaccents@}
\DeclareFontSubstitution{LS1}{stix}{m}{n}
\DeclareFontSubstitution{LS2}{stix}{m}{n}
\makeatother
\DeclareMathAlphabet{\mathscr}{LS1}{stixscr}{m}{n}

\DeclareSymbolFont{rsfso}{U}{rsfso}{m}{n}
\DeclareSymbolFontAlphabet{\mathrsf}{rsfso}

\newtheorem{corollary}{Corollary}[section]
\newtheorem{lemma}{Lemma}[section]
\newtheorem{proposition}{Proposition}[section]
\newtheorem{fact}{Fact}[section]

\theoremstyle{definition}
\newtheorem{defn}{Definition}[section]

\newtheorem{assumption}{Assumption}[section]

\theoremstyle{remark}
\newtheorem{remark}{Remark}[section]
\AtEndEnvironment{remark}{\null\hfill\ensuremath{\diamond}}

\newenvironment{definition}[1]{%
    \pushQED{\hfill$\diamond$}%
    \begin{defn}#1%
}{%
    \popQED
    \end{defn}%
}

\newcommand{\tn}[1]{\textnormal{#1}}

\newcommand{\cl}[1]{\textnormal{\bfseries #1}}
\newcommand{\clw}[2]{\textnormal{\bfseries #1}\textnormal{-}\textnormal{#2}}
\renewcommand{\sc}[1]{\textnormal{\textsc{#1}}}

\newcommand{\Gate}[1]{{\fontfamily{cmr}\selectfont\textsc{#1}}}

\newcommand{\poly}[1]{\mathrm{poly}\mathopen{}\left(#1\right)}

\newcommand{\e}{\mathrm{e}}
\renewcommand{\i}{\mathrm{i}}
\newcommand{\g}{\textsl{g}}
\newcommand{\B}{\{0,1\}}
\newcommand{\ind}[1]{\textbf{1}[#1]}

\newcommand{\V}{\mathcal{V}}

\newcommand{\U}{\mathcal{U}}

\renewcommand{\H}{\mathcal{H}}
\newcommand{\Hil}{\mathcal{H}}
\newcommand{\D}{{\rm D}}

\newcommand{\w}{\mathdcal{W}}

\newcommand{\bmw}{\mathdbcal{W}}

\renewcommand{\L}{\mathrsf{L}}
\newcommand{\pattce}{\mathbin{\text{\ding{66}}}}

\renewcommand{\ket}[1]{\lvert #1\rangle}
\renewcommand{\bra}[1]{\langle #1\rvert}
\renewcommand{\braket}[2]{\langle #1\vert #2\rangle}
\renewcommand{\mel}[3]{\langle #1\rvert #2\lvert #3\rangle}

\renewcommand{\norm}[1]{\lVert #1\rVert}
\newcommand{\snorm}[1]{\lVert #1\rVert_{\sup}}
\newcommand{\hw}[1]{\norm{#1}_1}
\newcommand{\supp}[1]{{\mathsf S}\mathopen{}(#1)}
\newcommand{\unit}[1]{\bm{[}\,#1\,\bm{]}}

\newcommand{\wt}[1]{\widetilde{#1}}
\newcommand{\bs}[1]{\boldsymbol{#1}}
\newcommand{\ps}{\partial_s}

\newcommand{\problemdef}[3]{%
    \begin{center}
        \begin{tabularx}{\linewidth}{@{}l X@{}}
            \multicolumn{2}{@{}l}{#1:}
                \\ \addlinespace[2pt]
            \textit{Instance:}
                & {#2}
                \\ \addlinespace[1.5pt]
            \textit{Problem:}
                & {#3}
                \\
        \end{tabularx}
    \end{center}
}

%% file: figures.tex
\tikzset{
    localitytransition/.pic = {
        \path[bigarrow=1.5pt] (-1,-1) .. controls (5,-1) and (7,-1) .. (9,-1) node[right] {locality};
        \draw[thick, mycolr!50!black, fill=mycolr!50!black] (-0.5,-1) circle (0.1) node[below, yshift=-10pt, xshift=10pt, rotate=-30] {$1$-local} node[above, draw=black, yshift=8pt] {\small This work};
        \draw[thick, mycolr!50!black, fill=mycolr!50!black] (1.5,-1) circle (0.1) node[below, yshift=-10pt, xshift=10pt, rotate=-30] {$2$-local} node[above, draw=black, yshift=8pt] {\small This work};
        \draw[thick, black!20, fill=black!20] (3.5,-1) circle (0.1) node[below, yshift=-10pt, xshift=10pt, rotate=-30, black!20] {$3$-local};;
        \draw[thick, black!20, fill=black!20] (5.5,-1) circle (0.1) node[below, yshift=-10pt, xshift=10pt, rotate=-30, black!20] {$4$-local};
        \draw[thick, black, fill=black] (7.5,-1) circle (0.1) node[below, yshift=-10pt, xshift=10pt, rotate=-30] {$5$-local} node[above, draw=black, yshift=8pt] {\small Ref.~\cite{hayakawa2025computational}};

        \draw[dashed, thick] (0.5, -2.5) -- (0.5, 0);

        \draw[thick, decoration={brace,mirror,raise=0.5cm}, decorate] (0.75,-1.75) -- (8.75,-1.75) node [pos=0.5,anchor=north,yshift=-0.55cm] {\clw{BQP}{complete}}; 
        \draw[thick, decoration={brace,mirror,raise=0.5cm}, decorate] (-0.75,-1.75) -- (0.25,-1.75) node [pos=0.5,anchor=north,yshift=-0.55cm] {\cl{P}}; 
    },
    logicalflow/.pic = {
        \useasboundingbox (0.5, 3) rectangle (14.5, -7.5);
        \node[fill=mycolr!20!white, align=center, text width=4.5cm] (a) at (3,2) {Section~\ref{sec:uniform-simulation} \\[0.2cm] \footnotesize{Uniform Simulation of Parameterised Hamiltonian Families}};
        \node[fill=mycolr!20!white, align=center, text width=4.5cm] (b) at (12,0) {Section~\ref{sec:sw-transformation} \\[0.2cm] \footnotesize{Schrieffer-Wolff Transformation for Parameterised Hamiltonian Families}};
        \node[fill=mycolr!20!white, align=center, text width=4.5cm] (c) at (3,-3) {Section~\ref{sec:berry-phase-simulation} \\[0.2cm] \footnotesize{Berry Phase Simulation}};
        \node[fill=mycolr!20!white, align=center, text width=4.5cm] (d) at (12,-4) {Section~\ref{sec:perturbative-gadget-reductions} \\[0.2cm] \footnotesize{Perturbative Gadget Reductions}};
        \node[fill=mycolr!20!white, align=center, text width=4.5cm] (e) at (7.5,-6.5) {Section~\ref{sec:extended-hardness-results} \\[0.2cm] \footnotesize{Extended Hardness Results}};
        \path[very thick, -latex] (a.south) edge[out=-80, in=90, looseness=1.2] (b.north);
        \draw[very thick, -latex] (a.south) to[bend right=15] (c.north);
        \path[very thick, -latex] (a.south) edge[out=-80, in=120, looseness=1.2] (d.north);
        \draw[very thick, -latex] (b.south) to[bend left=15] (d.north);
        \path[very thick, -latex] (d.south) edge[out=-90, in=90, looseness=1.2] (e.north);
        \path[very thick, -latex] (c.south) edge[out=-90, in=90, looseness=1.2] (e.north);

        \draw[fill=white,thick] (3.15,-0.25) ellipse (1.35cm and 0.5cm) node[align=center] {\footnotesize Definition of \\[-0.15cm]\footnotesize uniform simulation};

        \draw[fill=white,thick] (12.35,-2) ellipse (1.15cm and 0.5cm) node[align=center] {\footnotesize Existence of \\[-0.15cm]\footnotesize simulator family};

        \draw[fill=white,thick] (4.5,-4.5) ellipse (1.15cm and 0.5cm) node[align=center] {\footnotesize Closeness of \\[-0.15cm]\footnotesize Berry phase};

        \draw[fill=white,thick] (10.5,-5.25) ellipse (1.5cm and 0.35cm) node[align=center] {\footnotesize \textsl{subdivision}, \textsl{fork}, \textsl{cross}};
    },
    berryloop/.pic = {
        \coordinate (A) at (0,0);\coordinate (B) at (1,0);\coordinate (C) at (2,1);\coordinate (D) at (3.75,0.5);\coordinate (E) at (2.5,1.5);\coordinate (F) at (1,2); 

        \draw[line width=1mm, mycolr!65] plot[smooth cycle, tension=1] coordinates {(A) (B) (C) (D) (E) (F)} node[xshift=1.5cm,above] {$C$};
        \draw[line width=0.35mm, white] plot[smooth cycle, tension=1] coordinates {(A) (B) (C) (D) (E) (F)};

        \draw[line width=0.35mm, -stealth, black!80] (A) -- +(45:0.3);
        \draw[line width=0.35mm, -stealth, black!75] (B) -- +(60:0.3);
        \draw[line width=0.35mm, -stealth, black!65] (C) -- +(85:0.3);
        \draw[line width=0.35mm, -stealth, black!60] (D) -- +(123:0.3);
        \draw[line width=0.35mm, -stealth, black!55] (E) -- +(277:0.3);
        \draw[line width=0.35mm, -stealth, black!50] (F) -- +(67:0.3); \draw[line width=0.35mm, -stealth, black!45] (A) -- +(105:0.3);

        \draw[black, opacity=0.6] (A) circle(0.3);
        \draw[black, opacity=0.6] (B) circle(0.3);
        \draw[black, opacity=0.6] (C) circle(0.3);
        \draw[black, opacity=0.6] (D) circle(0.3);
        \draw[black, opacity=0.6] (E) circle(0.3);
        \draw[black, opacity=0.6] (F) circle(0.3);

       \draw[fill=black,draw=black,thick] (A) circle(0.04); 
        \draw[fill=white,draw=black,thick] (B) circle(0.04);
        \draw[fill=white,draw=black,thick] (C) circle(0.04);
        \draw[fill=white,draw=black,thick] (D) circle(0.04);
        \draw[fill=white,draw=black,thick] (E) circle(0.04);
        \draw[fill=white,draw=black,thick] (F) circle(0.04);

        \draw[thick, ->] (1,1)++(-100:0.3) arc (-100:170:0.3) node[above, pos=0.5] {$s$};

        \coordinate (z) at (-0.25,1.85);
        \draw[thick, black, opacity=0.6] (z) circle(0.5);
        \draw[thick, mycolr] (z)++(45:0.25) arc (45:105:0.25) node[above, pos=0.5, yshift=-1pt, mycolr] {$\vartheta$};
        \draw[line width=0.35mm, -stealth, black!80] (z) -- +(45:0.5);
        \draw[line width=0.35mm, -stealth, black!45] (z) -- +(105:0.5);
        \draw[fill=black, draw=black,thick] (z) circle(0.045);

        \path (A) -- +(20:0.3) coordinate (a1);
        \path (A) -- +(170:0.3) coordinate (a2);
        \path (A) -- +(170:0.3) coordinate (a2);
        \path (z) -- +(0:0.5) coordinate (z1);
        \path (z) -- +(190:0.5) coordinate (z2);

        \draw[dashed, black, opacity=0.6] (a1) -- (z1);
        \draw[dashed, black, opacity=0.6] (a2) -- (z2);
    },
    gadgetrealisation/.pic = {
        \draw[thick] (0,0) -- (1,-1) node[above, pos=0.5, rotate=-45] {$V^{a,1}$};
        \draw[thick] (0,-2) -- (1,-1) node[above, pos=0.5, rotate=45] {$V^{a,2}$};
        \draw[thick] (2.5,-1) -- (1,-1) node[below, pos=0.5] {$V^{a,3}$};
        \draw[thick, fill=black] (0,0) circle (0.1) node[above, yshift=5pt] {$u$};
        \draw[thick, fill=black] (0,-2) circle (0.1) node[below, yshift=-5pt] {$v$};
        \draw[thick, fill=black] (2.5,-1) circle (0.1) node[below, yshift=-5pt] {$w$};
        \draw[thick, fill=mycolr, draw=mycolr] (1,-1) circle (0.1) node[below, yshift=-5pt, mycolr] {\small $m_a$};

        \node at (1.25,-3) {Gadget};
        \node at (4.5,-1) {\LARGE $\overset{\Omega^a}{\rightsquigarrow}$};
        \begin{scope}[xshift=6.5cm]
            \draw[thick, fill=black] (0,0) circle (0.1) node[above, yshift=5pt] {$u$};
            \draw[thick, fill=black] (0,-2) circle (0.1) node[below, yshift=-5pt] {$v$};
            \draw[thick, fill=black] (2.5,-1) circle (0.1) node[below, yshift=-5pt] {$w$};

            \draw[very thick] (0,0) -- (2.5,-1);
            \draw[very thick] (0,-2) -- (2.5,-1);
            \draw[thick, dashed] (0,0) -- (0,-2);
            \draw[thick, dashed] (0,0)  to[in=40,out=140,loop,distance=2cm] (0,0);
            \draw[thick, dashed] (0,-2)  to[in=-40,out=-140,loop,distance=2cm] (0,-2);
            \draw[thick, dashed] (2.5,-1)  to[in=-50,out=50,loop,distance=2cm] (2.5,-1);
            \node at (1.25,-3) {Realisation};
        \end{scope}
    },
    subdivision/.pic = {
        \draw[thick, mycolr!80] (-0.5,3.75) rectangle ++(4,-4.5);
        \begin{scope}[yshift=1.5cm]
            \draw[thick](0,0)--(3,0) node[midway, above, yshift=2.55pt] {$A_uB_v$}; 
            \draw[fill=black, thick] (0,0)circle(0.1) node[below, yshift=-5pt] {$u$};
            \draw[fill=black, thick] (3,0)circle(0.1) node[below, yshift=-5pt] {$v$};
        \end{scope}
        \begin{scope}[yshift=-5cm]
        \draw[thick, mycolr!80] (-0.5,3.75) rectangle ++(4,-4.5);
            \begin{scope}[yshift=1.5cm]
                \draw[thick](0,0)--(1.5,0) node[midway, above, yshift=2.55pt] {$A_uX_{m_a}$}; \draw[thick](1.5,0)--(3,0) node[midway, above, yshift=2.55pt] {$X_{m_a}B_v$}; 
                \draw[fill=black, thick] (0,0)circle(0.1) node[below, yshift=-5pt] {$u$};
                \draw[fill=mycolr, thick, draw=mycolr] (1.5,0)circle(0.1) node[mycolr, below, yshift=-5pt] {\small $m_a$};
                \draw[fill=black, thick] (3,0)circle(0.1) node[below, yshift=-5pt] {$v$};
            \end{scope}
        \end{scope}
    },
    cross/.pic = {
        \draw[thick, mycolr!80] (-0.5,3.75) rectangle ++(4,-4.5);
        \draw[thick](0,0)--(3,3) node[midway, above, rotate=45, yshift=7.5pt, xshift=-40pt, fill=white] {$C_wD_s$}; \draw[fill=white, draw=none] (1.5,1.5)circle(0.12);\draw[thick](0,3)--(3,0) node[midway, above, rotate=-45, yshift=7.5pt, xshift=40pt, fill=white] {$A_uB_v$}; 
        \draw[fill=black, thick] (0,0)circle(0.1) node[below, yshift=-5pt] {$w$};
        \draw[fill=black, thick] (3,0)circle(0.1) node[below, yshift=-5pt] {$v$};
        \draw[fill=black, thick] (0,3)circle(0.1) node[above, yshift=5pt] {$u$};
        \draw[fill=black, thick] (3,3)circle(0.1) node[above, yshift=5pt] {$s$};
        \begin{scope}[yshift=-5cm]
            \draw[thick, mycolr!80] (-0.5,3.75) rectangle ++(4,-4.5);
            \draw[thick](0,0)--(1.5,1.5) node[midway, above, rotate=45, yshift=2.5pt] {$C_wX_{m_a}$}; \draw[thick](1.5,1.5)--(3,3) node[midway, below, rotate=45, yshift=-2.5pt] {$X_{m_a}D_s$}; \draw[thick](0,3)--(1.5,1.5) node[midway, above, rotate=-45, yshift=2.5pt] {$A_uX_{m_a}$}; \draw[thick](1.5,1.5)--(3,0) node[midway, below, rotate=-45, yshift=-2.5pt] {$X_{m_a}B_v$}; 
            \draw[fill=black, thick] (0,0)circle(0.1) node[below, yshift=-5pt] {$w$};
            \draw[fill=black, thick] (3,0)circle(0.1) node[below, yshift=-5pt] {$v$};
            \draw[fill=black, thick] (0,3)circle(0.1) node[above, yshift=5pt] {$u$};
            \draw[fill=black, thick] (3,3)circle(0.1) node[above, yshift=5pt] {$s$};
            \draw[fill=mycolr, thick, draw=mycolr] (1.5,1.5)circle(0.1) node[mycolr, right, xshift=5pt] {\small $m_a$};
        \end{scope}
    },
    fork/.pic = {
        \draw[thick, mycolr!80] (-0.5,3.75) rectangle ++(4,-4.5);
        \begin{scope}[yshift=2.799038cm]    
            \draw[thick] (0,0)--(1.5,-{3*sqrt(3)}/2) node[midway, below, rotate=-60, yshift=-2.5pt] {$A_uB_v$};
            \draw[thick] (1.5,-{3*sqrt(3)}/2)--(3,0) node[midway, below, rotate=60, yshift=-2.5pt] {$B_vC_w$}; 
            \draw[fill=black, thick] (0,0)circle(0.1) node[above, yshift=5pt] {$u$};
            \draw[fill=black, thick] (3,0)circle(0.1) node[above, yshift=5pt] {$w$};
            \draw[fill=black, thick] (1.5,-{3*sqrt(3)}/2)circle(0.1) node[below, yshift=-5pt] {$v$};
        \end{scope}
        \begin{scope}[yshift=-5cm]
        \draw[thick, mycolr!80] (-0.5,3.75) rectangle ++(4,-4.5);
            \begin{scope}[yshift=2.799038cm]
                \draw[thick] (0,0)--(1.5,-{3*sqrt(3)}/2 + 1.5) node[midway, below, rotate=-40, yshift=-2.5pt] {$A_uX_{m_a}$};
                \draw[thick] (1.5,-{3*sqrt(3)}/2 + 1.5)--(3,0) node[midway, above, rotate=40, yshift=2.5pt] {$X_{m_a}B_v$}; 
                \draw[thick] (1.5,-{3*sqrt(3)}/2 + 1.5)--(1.5,-{3*sqrt(3)}/2) node[midway, above, rotate=-90, yshift=2.5pt] {$X_{m_a}B_v$}; 
                \draw[fill=black, thick] (0,0)circle(0.1) node[above, yshift=5pt] {$u$};
                \draw[fill=black, thick] (3,0)circle(0.1) node[above, yshift=5pt] {$w$};
                \draw[fill=black, thick] (1.5,-{3*sqrt(3)}/2)circle(0.1) node[below, yshift=-5pt] {$v$};
                \draw[fill=mycolr, thick, draw=mycolr] (1.5,-{3*sqrt(3)}/2 + 1.5)circle(0.1) node[mycolr, above, yshift=5pt] {\small $m_a$};
            \end{scope}
        \end{scope}
    },
    promiseinterval/.pic = {
        \path (0,0) -- +(17.5:2) coordinate (a); \path (0,0) -- +(30:2) coordinate (a1); \path (0,0) -- +(5:2) coordinate (a2);
        \path (0,0) -- +(132.5:2) coordinate (b); \path (0,0) -- +(120:2) coordinate (b1); \path (0,0) -- +(145:2) coordinate (b2);

        \path (0,0) -- +(25:2) coordinate (x1); \path (0,0) -- +(10:2) coordinate (x2);
        \path (0,0) -- +(125:2) coordinate (y1); \path (0,0) -- +(140:2) coordinate (y2);

        \draw[thick] (0,0) circle (2);
        \draw[thick, dashed, black!30] (-2,0) -- (2,0);

        \node[rotate around={120:(a1)}] at (a1) {\Large [}; \node[rotate around={95:(a2)}] at (a2) {\Large ]};
        \node[rotate around={30:(b1)}] at (b1) {\Large [}; \node[rotate around={55:(b2)}] at (b2) {\Large ]};
        \draw[thick, fill=black] (a) circle (0.02) node[right] {$a$}; \draw[thick, fill=black] (a1) circle (0.02); \draw[thick, fill=black] (a2) circle (0.02);
        \draw[thick, fill=black] (b) circle (0.02) node[above] {$b$}; \draw[thick, fill=black] (b1) circle (0.02); \draw[thick, fill=black] (b2) circle (0.02);

        \draw[thick, mycolr,fill=mycolr] (x1) circle (0.02); \draw[thick, mycolr,fill=mycolr] (x2) circle (0.02);
        \node[mycolr,rotate around={115:(x1)}] at (x1) {\Huge [}; \node[mycolr,rotate around={100:(x2)}] at (x2) {\Huge ]};
        \draw[thick, mycolr,fill=mycolr] (y1) circle (0.02); \draw[thick, mycolr,fill=mycolr] (y2) circle (0.02);
        \node[mycolr,rotate around={35:(y1)}] at (y1) {\Huge [}; \node[mycolr,rotate around={50:(y2)}] at (y2) {\Huge ]};

        \coordinate (z) at (0,0);
        \draw[thick, black!30, fill=black!30] (z) circle(0.02);
        \draw[thick,line width=3.5mm, opacity=0.2] (z)++(30:2) arc (30:120:2); \draw[thick,line width=3.5mm, opacity=0.2] (z)++(5:2) arc (5:-215:2);
        \draw[thick,line width=6.5mm, mycolr, opacity=0.3] (z)++(25:2) arc (25:125:2); \draw[thick,line width=6.5mm, mycolr, opacity=0.3] (z)++(10:2) arc (10:-220:2);

        \draw[thick, fill=black] (2,0) circle (0.02) node[below right] {$0$}; \draw[thick, fill=black] (-2,0) circle (0.02) node[below left] {$\pi$};

        \draw[opacity=0.6,thick] (b) circle(0.6);
        \path (z) -- +(115:2) coordinate (f); \path (z) -- +(150:2) coordinate (g);
        \draw[opacity=0.6,thick, dashed] (f) -- +(115:2); \draw[opacity=0.6,thick, dashed] (g) -- +(150:2);
        \draw[mycolr, opacity=0.3] (z) -- +(125:4); \draw[mycolr, opacity=0.3] (z) -- +(140:4);
        \draw[opacity=0.3] (z) -- +(120:4); \draw[opacity=0.3] (z) -- +(145:4);

        \path (0,0) -- +(132.5:4) coordinate (c);
        \draw[opacity=0.6,thick] (c) circle(1.2); 
        \path (z) -- +(115:4) coordinate (p); \path (z) -- +(150:4) coordinate (q);
        \draw[line width=0.6mm] (p) arc (115:150:4);
        \draw[thick, fill=black] (c) circle (0.05) node[above, yshift=0.2cm,xshift=-0.1cm] {\large $b$};

        \path (z) -- +(125:4) coordinate (u); 
        \path (z) -- +(140:4) coordinate (v); 
        \path (z) -- +(120:4) coordinate (m); 
        \path (z) -- +(145:4) coordinate (n); 
        
        \draw[-stealth, mycolr] (m) to [bend left=55] ( $ (u)!.02!(0,0) $ ) node[yshift=-0.1cm, xshift=0.3cm] {$\upsilon$}; \draw[-stealth, mycolr] (n) to [bend right=55] ( $ (v)!.02!(0,0) $ ) node[yshift=-0.3cm, xshift=0.05cm] {$\upsilon$};
        \draw[-stealth] (c) to [bend left=55] ( $ (0,0)!1.02!(m) $ ) node[above left, xshift=-0.1cm, yshift=-0.1cm] {$\g$}; \draw[-stealth] (c) to [bend right=55] ( $ (0,0)!1.02!(n) $ ) node[above right, xshift=0cm, yshift=0.4cm] {$\g$};
        \draw[thick, mycolr,fill=mycolr] (u) circle (0.05); \draw[thick, mycolr,fill=mycolr] (v) circle (0.05);
        \draw[thick, fill=black] (m) circle (0.05); \draw[thick, fill=black] (n) circle (0.05);
        \draw[thick, fill=black] (c) circle(0.02);
    },
    latticepaths/.pic = {
        \def\width{10}
        \def\step{0.5}
        \def\term{\width}
        \pgfmathsetmacro\term{\term-\step}

        \foreach \x/\y in {1/1.5, 2/8, 3/4, 7/7, 8.5/2, 9.5/9.5}{
            \draw[draw=none, fill=gray!45, draw=none] (\x-0.5,\y-0.5) rectangle (\x+0.5,\y+0.5);
        };
        
        \foreach \i in {0.5, 1, 1.5, ..., 9.5}{
            \draw[mycolr] (\i,0)--(\i,\width);
            \draw[mycolr!55!black] (0,\i)--(\width,\i);
        };
    
        \foreach \j in {0.5, 1, 1.5, ..., 10}{
            \draw[mycolr!55] (10-\j,0)--(10,\j);
            \draw[mycolr!55] (0,\j)--(10-\j,10);
        };
    
        \draw[gray, dashed] (1,1.5)--(8.5,2)--(9.5,9.5)--(7,7)--(3,4)--(2,8);
        \draw[gray, dashed] (1,1.5)--(3,4);
        \draw[gray, dashed] (7,7)--(8.5,2);
    
        \foreach \x/\y in {1/1.5, 2/8, 3/4, 7/7, 8.5/2, 9.5/9.5}{
            \draw[fill = black] (\x,\y) circle (2pt);
        };

        \draw[very thick] (1,1.5)--(4.5,1.5)--(5,2)--(8,2)--(8.5,2);
        \draw[very thick] (3,4)--(2,3)--(2,2.5)--(1,1.5);
        \draw[very thick] (3,4)--(3,5)--(2.5,5)--(2.5,7)--(2,7)--(2,8);
        \draw[very thick] (3,4)--(4,5)--(4.5,5)--(6,6.5)--(6.5,6.5)--(7,7);
        \draw[very thick] (8.5,2)--(8.5,2.5)--(8,2.5)--(8,4.5)--(7.5,4.5)--(7.5,6.5)--(7,6.5)--(7,7);
        \draw[very thick] (7,7)--(9.5,9.5);
        \draw[very thick] (8.5,2)--(9,2.5)--(9,7.5)--(9.5,8)--(9.5,9.5);
    
        \draw[thick, black] (0,0) rectangle (\width,\width);
    },
    localisingvertex/.pic = {
        \draw[thick, mycolr!80] (-0.5,3.75) rectangle ++(4,-4.5);
        \begin{scope}[yshift=1.5cm]
            \draw[fill=black, thick] (1.5,0)circle(0.1) node[below, yshift=-5pt] {$u$};
            \draw[thick](1.5,0)--(3,1); \draw[fill=black, thick] (3,1)circle(0.1);
            \draw[thick](1.5,0)--(0,-2); \draw[fill=black, thick] (0,-2)circle(0.1);
            \draw[thick](1.5,0)--(0,1.3); \draw[fill=black, thick] (0,1.3)circle(0.1);
            \draw[thick](1.5,0)--(1.2,1.8); \draw[fill=black, thick] (1.2,1.8)circle(0.1);
            \draw[thick](1.5,0)--(3,-1.2); \draw[fill=black, thick] (3,-1.2)circle(0.1);
        \end{scope}
        \begin{scope}[yshift=-5cm]
        \draw[thick, mycolr!80] (-0.5,3.75) rectangle ++(4,-4.5);
            \begin{scope}[yshift=1.5cm]
                \draw[fill=black, thick] (1.5,0)circle(0.1) node[below, yshift=-5pt] {$u$};
                \draw[thick](1.5,0)--(3,1); \draw[fill=black, thick] (3,1)circle(0.1); \draw[fill=mycolr, draw=mycolr, thick] (2.25,0.5)circle(0.1);
                \draw[thick](1.5,0)--(0,-2); \draw[fill=black, thick] (0,-2)circle(0.1); \draw[fill=mycolr, draw=mycolr, thick] (0.75,-1)circle(0.1);
                \draw[thick](1.5,0)--(0,1.3); \draw[fill=black, thick] (0,1.3)circle(0.1); \draw[fill=mycolr, draw=mycolr, thick] (0.75,0.65)circle(0.1);
                \draw[thick](1.5,0)--(1.2,1.8); \draw[fill=black, thick] (1.2,1.8)circle(0.1); \draw[fill=mycolr, draw=mycolr, thick] (1.35,0.9)circle(0.1);
                \draw[thick](1.5,0)--(3,-1.2); \draw[fill=black, thick] (3,-1.2)circle(0.1); \draw[fill=mycolr, draw=mycolr, thick] (2.25,-0.6)circle(0.1);
            \end{scope}
        \end{scope}
    },
    subdividedcrossing/.pic = {
        \draw[thick, mycolr!80] (-0.5,3.75) rectangle ++(4,-4.5);
        \draw[thick] (0.75,3)--(0.75,0); \draw[fill=white, draw=none] (0.75,1.5)circle(0.1); \draw[fill=black, thick] (0.75,0)circle(0.1); \draw[fill=black, thick] (0.75,3)circle(0.1);
        \draw[thick] (1.5,3)--(1.5,0); \draw[fill=white, draw=none] (1.5,1.5)circle(0.1); \draw[fill=black, thick] (1.5,0)circle(0.1); \draw[fill=black, thick] (1.5,3)circle(0.1);
        \draw[thick] (2.25,3)--(2.25,0); \draw[fill=white, draw=none] (2.25,1.5)circle(0.1); \draw[fill=black, thick] (2.25,0)circle(0.1); \draw[fill=black, thick] (2.25,3)circle(0.1);
        
        \draw[thick] (0,1.5)--(3,1.5);
        \draw[fill=black, thick] (0,1.5)circle(0.1); \draw[fill=black, thick] (3,1.5)circle(0.1);
        \begin{scope}[yshift=-5cm]
            \draw[thick, mycolr!80] (-0.5,3.75) rectangle ++(4,-4.5);
            \draw[thick] (0.75,3)--(0.75,0); \draw[fill=white, draw=none] (0.75,1.5)circle(0.1); \draw[fill=black, thick] (0.75,0)circle(0.1); \draw[fill=black, thick] (0.75,3)circle(0.1);
            \draw[thick] (1.5,3)--(1.5,0); \draw[fill=white, draw=none] (1.5,1.5)circle(0.1); \draw[fill=black, thick] (1.5,0)circle(0.1); \draw[fill=black, thick] (1.5,3)circle(0.1);
            \draw[thick] (2.25,3)--(2.25,0); \draw[fill=white, draw=none] (2.25,1.5)circle(0.1); \draw[fill=black, thick] (2.25,0)circle(0.1); \draw[fill=black, thick] (2.25,3)circle(0.1);
            
            \draw[thick] (0,1.5)--(3,1.5);
            \draw[fill=black, thick] (0,1.5)circle(0.1); \draw[fill=black, thick] (3,1.5)circle(0.1);
            \draw[fill=mycolr, draw=mycolr] (1.125,1.5)circle(0.1);
            \draw[fill=mycolr, draw=mycolr] (1.875,1.5)circle(0.1);
        \end{scope}
    },
    localisecrossing/.pic = {
        \draw[thick, mycolr!80] (-0.5,3.75) rectangle ++(4,-4.5);
        \draw[thick](0,0)--(3,3); \draw[fill=white, draw=none] (1.5,1.5)circle(0.12);\draw[thick](0,3)--(3,0); 
        \draw[fill=black, thick] (0,0)circle(0.1);
        \draw[fill=black, thick] (3,0)circle(0.1);
        \draw[fill=black, thick] (0,3)circle(0.1);
        \draw[fill=black, thick] (3,3)circle(0.1);
        \begin{scope}[yshift=-5cm]
            \draw[thick, mycolr!80] (-0.5,3.75) rectangle ++(4,-4.5);
            \draw[thick](0,0)--(3,3); \draw[fill=white, draw=none] (1.5,1.5)circle(0.12);\draw[thick](0,3)--(3,0); 
            \draw[fill=black, thick] (0,0)circle(0.1);
            \draw[fill=black, thick] (3,0)circle(0.1);
            \draw[fill=black, thick] (0,3)circle(0.1);
            \draw[fill=black, thick] (3,3)circle(0.1);

            \draw[fill=mycolr, draw=mycolr] (0.75,0.75)circle(0.1);
            \draw[fill=mycolr, draw=mycolr] (2.25,2.25)circle(0.1);
            \draw[fill=mycolr, draw=mycolr] (2.25,0.75)circle(0.1);
            \draw[fill=mycolr, draw=mycolr] (0.75,2.25)circle(0.1);
        \end{scope}
    },
}

%% file: sections/introduction.tex
\section{Introduction}
\label{sec:introduction}
The study and classification of topological phases of matter is a key area of modern condensed matter physics.
Topological invariants --- quantities preserved under continuous deformations --- play an important role in understanding phenomena such as the quantum Hall effect~\cite{thouless1982quantized} and topological insulators~\cite{hasan2010colloquium}.
A fundamental example is the Berry phase~\cite{pancharatnam1956generalized,berry1984quantal}: the geometric phase accumulated when a quantum state evolves adiabatically along a closed path in parameter space.
Such an evolution can be expressed as
\begin{equation}
    \ket{\psi(1)} = \e^{-\i \alpha} \e^{\i \vartheta} \ket{\psi(0)},
\end{equation}
where $\alpha$ is the dynamical phase and $\vartheta$ the Berry phase.
This geometric phase governs a wide range of physical effects emerging in the study of topological quantum matter~\cite{sprinkart2024tutorial}, electronic properties~\cite{xiao2010berry}, and polarisation phenomena~\cite{watanabe2018inequivalent}.

The robustness of geometric phases to local perturbations inspired holonomic quantum computation~\cite{zanardi1999holonomic, jones2000geometric}, where gates exploit adiabatic evolution in degenerate ground state manifolds.
This topological character parallels other fault-tolerant schemes that have been proposed~\cite{kitaev2003fault,freedman2002modular} and motivates a broader question: \emph{to what extent do topological properties of physical systems inherently encode complex computational problems?}

In this work, we prove that estimating the Berry phase to inverse-polynomial precision is \clw{BQP}{complete} for ${\rm C}^2$ families of $2$-local qubit Hamiltonians on square and triangular lattices.\footnote{A precise statement of this result is given in Section~\ref{sec:summary-of-techniques-and-results}, and the computational problem is defined in Section~\ref{sec:problem-statement}.}
The input includes a classical description of a guiding state promised to have non-negligible overlap with the ground state.
We also prove that the corresponding problem for $1$-local families lies in \cl{P}, even without a guiding state.

\subsection{Prior Work}
\label{sec:prior-work}
A central challenge in Hamiltonian complexity is to determine how physical constraints such as locality, dimensionality, and geometric structure affect the boundary between quantum and classical computation~\cite{kempe2006complexity,oliveira2008complexity,cubitt2018universal}.
Perturbative gadgets are a principal tool for constructing complexity-theoretic reductions between families of Hamiltonians, typically preserving computational hardness.
In this approach, a target interaction is reproduced in the low-energy sector of a simulator Hamiltonian by coupling the system to ancillary mediator qubits with a large energy penalty~\cite{kempe2006complexity,oliveira2008complexity}.
The Schrieffer--Wolff transformation provides a systematic expansion of the corresponding effective Hamiltonian~\cite{bravyi2011schrieffer}, while the simulation framework of Bravyi and Hastings formalises the resulting approximation and its composition across successive gadget reductions~\cite{bravyi2016complexity}.
These methods underpin locality reductions and lattice embeddings for many versions and variants of the \sc{Local Hamiltonian} problem~\cite{oliveira2008complexity,piddock2017complexity}.

Hayakawa, Sakamoto and Kiumi~\cite{hayakawa2025computational} considered three promise problem variants of Berry phase estimation distinguished by the auxiliary information supplied with the Hamiltonian path.
For the variant in which the input includes a guiding state with non-negligible ground-state overlap, \clw{BQP}{completeness} for parameterised families of $5$-local Hamiltonians was proven.
The \clw{BQP}{hardness} reduction utilised an adapted version of the Feynman--Kitaev circuit-to-Hamiltonian construction to design a closed (local Hamiltonian) family of the form $H_\zeta(s) = H + \zeta\w(s)$, where the static bulk Hamiltonian $H$ encodes the computation and $\w(s)$ drives the geometric phase.
In particular, the driving term may be chosen as
\begin{equation}\label{eq:driving-pauli-decomposition-intro}
    \w(s) = \frac{\cos(2\pi s)}{2}X_q + \frac{\sin(2\pi s)}{2}Y_q - \frac{\cos(2\pi s)}{2}X_qZ_c - \frac{\sin(2\pi s)}{2}Y_qZ_c,
\end{equation}
where $q$ is the output qubit and $c$ is the final clock qubit.
Although \cref{eq:driving-pauli-decomposition-intro} is already $2$-local, the bulk Hamiltonian is $5$-local, leaving open whether the computational hardness persists when the full family is $2$-local and geometrically constrained.

At first sight, this locality reduction appears to be an immediate application of existing gadgets.
Standard simulation guarantees can be applied pointwise: at each fixed $s$, they control the low-energy spectrum and eigenvectors of the simulator in terms of the errors $\eta$ and $\varepsilon$~\cite{bravyi2016complexity}.
Such control is sufficient for many static reductions and for tracking instantaneous low-energy information along a smooth path, but it does not by itself provide a way to compare the Berry phases of the two families.
A uniformly small error operator $E(s)$ can vary rapidly with $s$, though the Berry phase is sensitive to this variation even when the target and simulator ground states are close at every point.
Moreover, the isometry identifying the target space with the exact low-energy band can itself contribute a geometric phase, including when the pointwise Hamiltonian error vanishes.
Controlling the Berry phase therefore requires bounds on $\snorm{\ps E}$ and on the motion of the witnessing isometry, together with a periodic choice of that isometry at the endpoints.

A second obstruction appears when the gadget transformations themselves depend on $s$.
Applying a static Schrieffer--Wolff result separately at each parameter value does not supply a coherent, differentiable family of low-energy isometries or uniform derivative bounds for the exact and truncated series.
In a parallel gadget construction, the relevant low- and high-energy projectors are global, so a perturbation that is off-diagonal for one mediator need not be globally off-diagonal when another mediator is excited.
Consequently, the exact global Schrieffer--Wolff generator is not generally the sum of isolated single-gadget generators, and cross terms must be excluded or bounded before the familiar effective interactions can be used.

We therefore retain the algebraic gadget constructions of prior work, but augment them with a regular simulation theory for closed Hamiltonian families, a Berry phase simulation theorem, uniform parameter-dependent Schrieffer--Wolff estimates, and parameter-dependent perturbative gadget extensions.
These additions identify the precise hypotheses under which the standard reductions preserve the computational promise rather than assuming that pointwise low-energy simulation also preserves the Berry phase.

\subsection{Summary of Techniques and Results}
\label{sec:summary-of-techniques-and-results}
We partition our results into several contributions, which we summarise below.
For reference, we provide an informal definition of the \sc{Guided-State Berry Phase Estimation} problem (a formal definition can be found in \cref{sec:problem-statement}).

\problemdef{(Informal) \sc{Guided-State Berry Phase Estimation}}
{A family $\Upsilon = \{H(s)\}_{s \in [0,1]}$ of parameterised local Hamiltonians, a parameter $\textsl{g} \geq 1/\poly{n}$, reals $a < b$ with $b - a \geq 2\textsl{g}$, a parameter $\delta_0 \in [1/\poly{n}, 1 - 1/\poly{n}]$, and a classical description of a guiding state $\ket{\xi}$ which is promised to satisfy $\abs{\braket{\phi_0(0)}{\xi}}^2 \geq \delta_0$ (where $\ket{\phi_0(0)}$ is the ground state of $H(0)$).}
{Decide whether the Berry phase $\vartheta$ of $\Upsilon$ lies within $I_{\textnormal{in}} = [a + \textsl{g}, b - \textsl{g}]$ or $I_{\textnormal{out}} = [b + \textsl{g}, a - \textsl{g}]$.}

\cref{fig:locality-transition} illustrates the results of this work, for the problem we study, with respect to the locality of the Hamiltonian and previously known complexity results.

\subparagraph{Regular Uniform Simulation and Composition.}
We first extend the static simulation framework of Ref.~\cite{bravyi2016complexity} from a single Hamiltonian to a family $\Upsilon = \{H(s)\}_{s\in[0,1]}$.
A uniform simulation consists of a simulator family $\wt\Upsilon = \{\wt H(s)\}_{s\in[0,1]}$, a fixed encoding isometry $\V$, and a ${\rm C}^2$ witness family $\U(s)$ whose image is the exact low-energy band of $\wt H(s)$.
Uniformly in $s$, the transformed Hamiltonian $\U^\dagger(s)\wt H(s)\U(s)$ approximates $H(s)$ within $\varepsilon$, while $\U(s)$ remains within $\eta$ of the fixed encoding $\V$; see Definition~\ref{def:uniform-simulation}.
The distinction between the two isometries is essential for computational purposes: $\V$ is the efficiently described encoding used to transform the guiding state, whereas $\U(s)$ tracks the exact low-energy subspace used in the analysis.
A regular simulation additionally requires $\U(1) = \U(0)$ and $\snorm{\ps\U}\leq\kappa$; see Definition~\ref{def:regular-simulation}.

The usual eigenvalue and ground-state guarantees then hold uniformly along the path, with errors scaling as $O(\poly{\eta^{-1},\varepsilon^{-1}})$.
Composition requires additional care beyond a naive concatenation of simulations, introducing a parameter-dependent unitary rotation to align specific subspaces.
We show that a chain of compositions can have error parameters handled systematically by appropriate renormalisations of the individual simulation errors.

\subparagraph{Berry-Phase Stability.}
Using the newfound definition of regular uniform simulation, we can begin to analyse the difference in Berry phases between the target and simulator families.
For a regular simulation, set $\widehat H(s)=\U^\dagger(s)\wt H(s)\U(s)$ and $E(s) = H(s)-\widehat H(s)$.
We compare the target and transformed ground-state projectors and their Kato parallel transports, which avoids choosing eigenvector phases in the perturbation argument.
The resulting estimate in \cref{eq:main-result-bp-error} separates the Berry-phase error into four contributions.
The two static terms are $J_1=\gamma_\star^{-2}\max\{\snorm{\ps H},\snorm{\ps\widehat H}\}\varepsilon$ and $J_3=\gamma_\star^{-1}\varepsilon$, with the dynamic terms given by $J_2=\gamma_\star^{-1}\snorm{\ps E}$ and $J_4=\kappa$.
Consequently, we require small $\eta$, $\varepsilon$ and $\kappa$ to ensure accurate Berry-phase simulation.
Our main Berry phase simulation result is stated formally in \cref{thm:berry-phase-simulation}.

\subparagraph{Parameterised Schrieffer--Wolff Analysis and Gadgets.}
We next analyse simulator families of the form $\wt H(s)=\Delta H_0+V(s)$ with fixed penalty Hamiltonian $H_0$ and closed ${\rm C}^2$ perturbation $V(s)$.
The exact canonical Schrieffer--Wolff generator $S(s)$ defines the witness $\U(s)=\e^{-S(s)}\V$; it follows that the regularity and derivative characterisation of $S(s)$ directly determine $\kappa$.
Writing $v_0=\snorm{V}$ and $v_1=\snorm{\ps V}$, \cref{lemma:uniform-control-exact-truncated} gives uniform bounds $\snorm{S}=O(\Delta^{-1}v_0)$ and $\snorm{\ps S}=O(\Delta^{-1}v_1)$, together with derivative bounds for the truncated-series remainders.
When the leading block-off-diagonal perturbation is independent of $s$, \cref{cor:static-leading-sw-generator} improves the derivative scale to $\kappa=O(\Delta^{-2}v_0v_1)$.

For parallel gadgets, we perform the transformation with respect to the global penalty splitting and derive the corresponding generator and effective-Hamiltonian coefficients in \cref{cor:sw-generators-parallel,cor:heff-coefficients-parallel}.
The excitation-sector analysis in \cref{sec:parallel-effective-hamiltonian-decomposition} shows that, under standard mediator-support conditions, the required second- and third-order effective terms decompose into independent gadget contributions even though the exact generator is global.
We then treat the two scenarios needed to conclude the hardness reduction: 
\begin{inparaenum}[(1)]
    \item the driving terms may be carried unchanged through a gadget layer,
    \item a sum of $2$-local driving terms may itself be realised by a second-order parameter-dependent gadget.
\end{inparaenum}
In each case, the analysis tracks not only the effective-Hamiltonian error but also $\snorm{\ps E}$ and $\kappa$.
\cref{lem:poly-control-general} then presents a criterion for controlling the error contributions in order to achieve an inverse-polynomial error on the Berry-phase simulation.
That is, a sufficiently large polynomial scaling of the penalty $\Delta$ suffices to ensure all four contributions $J_1,\ldots,J_4$ are inverse-polynomially small.

\subparagraph{Two-Local Hardness.}
We apply these tools first to a spatially sparse version of the $5$-local construction of Ref.~\cite{hayakawa2025computational}, obtaining \cref{theorem:spatially-sparse-5-local}.
Starting from this family, \textsl{subdivision} and \textsl{3-to-2-local} gadgets reduce the bulk locality, while \textsl{fork}, \textsl{cross}, and further \textsl{subdivision} gadgets reduce the Pauli degree, remove non-planar crossings, and embed the interaction graph into a square or triangular lattice~\cite{oliveira2008complexity,piddock2017complexity}.
Only a constant number of gadget layers is required, so the regular-simulation composition bounds preserve inverse-polynomial accuracy with polynomial interaction strengths.

The reduction must also preserve the promise of the computational problem, not merely the low-energy spectrum.
If the Berry phase changes by at most $\upsilon$, we reduce the arc margin from $\g$ to $\g'=\g-\upsilon$; if the ground state changes by $\mu=O(\eta+\gamma_\star^{-1}\varepsilon)$, the guiding-state overlap remains at least $\delta_0'=(\sqrt{\delta_0}-\mu)^2$.
The local encoding appends mediator qubits in fixed product states, so the semi-classical encoded subset state remains efficiently preparable from its classical description in the sense of Ref.~\cite{waite2026physically}.
Combining these facts with the \cl{BQP} containment algorithm of Ref.~\cite{hayakawa2025computational} proves that $(2,2,\delta_0)$-\sc{GSBPE} is \clw{BQP}{complete} on square and triangular lattices for every $\delta_0\in[1/\poly{n},1-1/\poly{n}]$; see \cref{thm:main-theorem}.
We then further improve the hardness result by extending it to systems with weighted Heisenberg interactions, while maintaining the lattice embedding constraints, by employing the reduction techniques of Schuch and Verstraete~\cite{schuch2009computational}, adapted to parameterised interactions.

\subparagraph{The One-Local Algorithm.}
Finally, for a $1$-local family we decompose the ground-state projector as $P(s) =\bigotimes_{j=1}^nP_j(s)$ and approximate the Berry phase by the cyclic Bargmann product of the one-qubit projectors in \cref{eq:one-local-bargmann-invariant,eq:one-local-discrete-phase}.
Because the estimator is written entirely in terms of projectors, it is invariant under independent rephasings of the sampled eigenvectors.
By discretising the cyclic path into $N$ segments, we obtain an approximation $\vartheta_N$ to the Berry phase $\vartheta$ with an error (roughly) bounded by $1/N$.
Choosing $N = \poly{n}$ and evaluating the $2\times2$ projectors to logarithmic precision yields a polynomial-time algorithm, which does not require the use of the guiding state or its classical description.
Together with the $2$-local hardness result, \cref{thm:one-local-berry-phase-in-p} establishes the claimed locality transition for \sc{Guided-State Berry Phase Estimation}.
The corresponding locality thresholds for the other Berry phase estimation variants considered in Ref.~\cite{hayakawa2025computational} remain open.

\begin{figure}[!ht]
    \centering
    \begin{tikzpicture}
        \pic[scale=1.2]{localitytransition};
    \end{tikzpicture}
    \caption{Locality transition for the \sc{Guided-State Berry Phase Estimation} problem. The dashed line indicates the threshold between polynomial-time solvable and \clw{BQP}{complete} instances.}
    \label{fig:locality-transition}
\end{figure}

\subsection{Organisation}
\label{sec:organisation}
We collect the proofs of the main results, presented in the following sections, in the appendices.
Each section has its own appendix where the detailed proofs and supporting results are given.
Where appropriate, we provide a proof sketch in the main text to convey the key ideas, though, the majority of the proof details follow from the surrounding context.

The main body of this work is structured as follows.
Section~\ref{sec:preliminaries} provides the necessary background and notation.
Section~\ref{sec:uniform-simulation} introduces the concept of uniform simulation for parameterised Hamiltonian families and defines what it means to compose such simulations.
Section~\ref{sec:sw-transformation} presents the Schrieffer--Wolff transformation and its application to parameterised Hamiltonians, extending the analysis to the parallel setting.
Section~\ref{sec:berry-phase-simulation} discusses the error analysis for simulating the Berry phase using parameterised Hamiltonian families.
Section~\ref{sec:perturbative-gadget-reductions} introduces perturbative gadget reductions and their adaptation to two different families of parameterised Hamiltonians.
Section~\ref{sec:extended-hardness-results} presents extended hardness results for the \sc{Guided-State Berry Phase Estimation} problem.
Section~\ref{sec:berry-phase-estimation-parameterised-1-local} provides the analysis for the Berry phase estimation problem restricted to $1$-local Hamiltonian families.
Section~\ref{sec:conclusion} concludes with a summary of the results and potential directions for future work.

Figure~\ref{fig:logical-flow} gives a schematic overview of the logical flow of this work (excluding Sections~\ref{sec:preliminaries}, \ref{sec:berry-phase-estimation-parameterised-1-local}, and \ref{sec:conclusion}).
Arrows indicate (loosely) the dependencies and connections between sections.

\begin{figure}[!h]
    \centering
    \begin{tikzpicture}
        \pic[xscale=1, yscale=0.875]{logicalflow};
    \end{tikzpicture}
\caption{Logical flow of the sections and key concepts in this work. Arrows indicate dependencies between sections, and ellipses present a short summary of the connection.}
\label{fig:logical-flow}
\end{figure}

%% file: sections/preliminaries.tex
\section{Preliminaries}
\label{sec:preliminaries}
In this section, we fix the operator conventions, Hamiltonian-family assumptions, Berry-phase definition, computational problem, and analytic tools used throughout the paper.

For a finite-dimensional operator $X$, we write $\norm{X}$ for the spectral norm.
For an operator family $X(s)$ with $s \in [0,1]$, we write $\snorm{X} = \sup_{s \in [0,1]} \norm{X(s)}$ for its uniform spectral norm.
The standard norm axioms hold for $\snorm{\cdot}$:
\begin{inparaenum}[(i)]
    \item \emph{non-negativity:} $\snorm{X} \ge 0$ and $\snorm{X} = 0$ if and only if $X(s) = 0$ for all $s \in [0,1]$;
    \item \emph{absolute homogeneity:} $\snorm{\alpha X} = \abs{\alpha} \snorm{X}$ for any scalar $\alpha \in \mathbb{C}$;
    \item \emph{subadditivity:} $\snorm{X + Y} \le \snorm{X} + \snorm{Y}$ for any operators $X$ and $Y$.\footnote{A short proof of this fact is as follows.
    $\snorm{X + Y} = \sup_s \norm{X(s) + Y(s)} \le \sup_s (\norm{X(s)} + \norm{Y(s)}) \leq \snorm{X} + \snorm{Y}$.}
\end{inparaenum}
Note that $\snorm{X} \leq \epsilon > 0$ implies that $\norm{X(s)} \leq \epsilon$ for all $s \in [0,1]$.
We abuse derivative notation and write $\ps X$ for the derivative of $X(s)$ with respect to $s$, and $\ps^k X$ for the $k$-th derivative of $X(s)$ with respect to $s$.
For two operators $X$ and $Y$, we write $X \preceq Y$ if $Y - X$ is positive semidefinite, and $X \prec Y$ if $Y - X$ is positive definite.

For an operator $X$ acting on finitely many qubits, let $\supp{X}$ denote the set of qubits on which $X$ acts non-trivially.
For a collection $\{X_a\}_{a \in [A]}$, let $G$ denote the \emph{overlap graph} with vertex set $[A]$ and an edge between $a$ and $b$ if and only if $\supp{X_a} \cap \supp{X_b} \neq \emptyset$.
A tuple of indices is a \emph{connected cluster} if the corresponding vertices induce a connected subgraph of $G$.

We write $\unit{x}$ for the physical dimension of $x$, with $e$ the unit of energy and $1$ a dimensionless unit.
For any operator $X$, $\unit{\norm{X}} = \unit{X}$.

For a variable $x$, we define $\poly{x} = O(x^k)$ for some fixed integer constant $k \ge 0$; for a set of variables $\{x_i\}_{i \in I}$, we $\poly{\{x_i\}_{i \in I}}$ typically refers to $O\big(\sum_{i \in I} x_i^{k_i}\big)$ for some set of fixed integer constants $\{k_i\}_{i \in I}$.
Taking the maximum over the set of variables, we have $\poly{\{x_i\}_{i \in I}} = O\big(\max_{i \in I} x_i^{k_i}\big)$, i.e., the polynomial scaling is dominated by the largest single term among the variables.

\subsection{Liouville Superoperator}
\label{sec:liouville-superoperator}
We now define the Liouville superoperator and its Moore--Penrose pseudoinverse as well as collect useful properties for later use.

Let $X$ be a Hermitian operator and let $Y$ be a bounded operator on a finite-dimensional Hilbert space $\mathcal{H}$.
We define the Liouville superoperator $\L_X$ associated with $X$ by
\begin{equation}\label{eq:liouville_superoperator}
    \L_X(Y) = \comm{X}{Y} = XY - YX.
\end{equation}
Note that ${\rm ker}(\L_X) = \{Y \in \mathscr{B}(\mathcal{H}) : \comm{X}{Y} = 0\}$, where $\mathscr{B}(\mathcal{H})$ denotes the set of bounded operators on $\mathcal{H}$.
Observe that the identity operator $I$ and $X$ itself both belong to ${\rm ker}(\L_X)$, since $\comm{X}{I} = 0$ and $\comm{X}{X} = 0$.
Moreover, if $O \in \mathscr{B}(\mathcal{H})$ satisfies $\comm{X}{O} = Y$, then so does $O + K$ for any $K \in {\rm ker}(\L_X)$.

Suppose that $X = \sum_{\lambda \in \sigma(X)} \lambda P_\lambda$ is the spectral decomposition of $X$, where $\sigma(X)$ denotes the set of eigenvalues of $X$ and $P_\lambda$ is the projector onto the eigenspace corresponding to $\lambda$.
We denote the Moore--Penrose pseudoinverse~\cite{benisrael2003generalized} of $\L_X$ by $\L_X^{\pattce}$ and define it as
\begin{equation}\label{eq:liouville_pseudoinverse_def}
    \L_X^{\pattce}(Y) = \sum_{\substack{\lambda, \mu \in \sigma(X) \\ \lambda \neq \mu}} \frac{P_\lambda Y P_\mu}{\lambda - \mu}.
\end{equation}
Consequently, ${\rm Im}(\L_X) = \{Y : P_\lambda Y P_\lambda = 0 \text{ for all } \lambda \in \sigma(X)\}$ and $\L_X \L_X^{\pattce}(Y) = Y$ for all $Y \in {\rm Im}(\L_X)$.
We use the Liouville superoperator and its pseudoinverse throughout the Schrieffer--Wolff analysis of Section~\ref{sec:sw-transformation}.
Whenever $\L_X$ is restricted to a subspace on which it is invertible, the restriction of $\L_X^{\pattce}$ coincides with the ordinary inverse on that subspace.

We present two useful lemmas consolidating some of the properties of the Liouville superoperator and its pseudoinverse.

\begin{restatable}{lemma}{liouvilleProperties}\label{lem:liouville-properties}
    Let $\mathcal{H}$ be a finite-dimensional Hilbert space and $A$ be a finite index set.
    Suppose that $\{X_a, Y_a, Z_a\}_{a \in A}$ are sets of bounded operators on $\mathcal{H}$ for each $a \in A$.
    Define $\bs{O} = \sum_{a \in A} O_a$ for $O \in \{X, Y, Z\}$.
    For any $P, Q \in \mathscr{B}(\mathcal{H})$, let $\L_{P}(Q) = \comm{P}{Q}$.
    Then the following statements hold.
    \begin{enumerate}
        \item $\L_{\bs{X}}(\bs{Y}) = \sum_{a,b \in A} \L_{X_a}(Y_b)$;
        \item If $\supp{X_a} \cap \supp{Y_b} = \emptyset$ for $a \neq b$, then $\L_{\bs{X}}(\bs{Y}) = \sum_{a \in A} \L_{X_a}(Y_a)$;
        \item Suppose that $\L_{X_a}(Z_a) = Y_a$ for each $a \in A$ and that $\comm{X_a}{Z_b} = 0$ for all $a \neq b$, then $\L_{\bs{X}}(\bs{Z}) = \bs{Y}$;
        \item Let $X$ be fixed and let $\L_X^{\pattce}$ be the Moore-Penrose pseudoinverse of $\L_X$, then if $Y_a \in {\rm Im}(\L_X)$ for each $a \in A$, we have that $\L_X^{\pattce}\left(\sum_{a \in A} Y_a\right) = \sum_{a \in A} \L_X^{\pattce}(Y_a)$.
    \end{enumerate}
\end{restatable}

\begin{restatable}{lemma}{liouvillePseudoinverseBound}\label{lem:liouville-pseudoinverse-bound}
    Let $A$ be a finite index set and, for each $a\in A$, let $p_a$ and $q_a$ be complementary single-qubit
    projectors. 
    Define $X_a \coloneqq q_a$ for each $a \in A$ and $\bs{X} = \sum_{a\in A}X_a$.
    Let $\bs{P}=\bigotimes_{a\in A}p_a$, $\bs{Q}=I-\bs{P}$, and suppose that $O$ is block-off-diagonal with respect to $\bs{P}\oplus\bs{Q}$, so that $O = \bs{P}O\bs{Q}+\bs{Q}O\bs{P}$.
    Then, for every $\Delta>0$,
    \begin{equation}\label{eq:liouville-pseudoinverse-bound}
        \norm{\L_{\Delta\bs{X}}^{\pattce}(O)} \leq \Delta^{-1}\norm{O}.
    \end{equation}
\end{restatable}

\begin{restatable}{proposition}{moorePenrosePseudoinverseBlockDiagonalOperator}\label{prop:moore-penrose-pseudoinverse-bdop}
    Let $X$ be a Hermitian operator defined on a Hilbert space $\mathcal{H} = \mathcal{P} \oplus \mathcal{Q}$.
    Suppose that $X$ is block-diagonal with respect to this decomposition, so that $X = PXP + QXQ$ with $P$ and $Q$ the projectors onto $\mathcal{P}$ and $\mathcal{Q}$, respectively. 
    Assume further that $PXP = 0$ and $QXQ \succeq Q$.
    Define the operator 
    \begin{equation}
        \Sigma = Q \big( X\bigr|_{{\rm Im}(Q)} \big)^{-1} Q.
    \end{equation}
    Then $\Sigma$ is the Moore--Penrose pseudoinverse of $X$ on $\mathcal{H}$.
\end{restatable}

Proofs of these results are given in Appendix~\ref{app:preliminaries-proofs}.

\subsection{Families of Parameterised Hamiltonians}
\label{sec:families-of-parameterised-hamiltonians}
We now fix the parameter domain and the regularity, closure, and gap assumptions imposed on every Hamiltonian family.
We identify the control manifold with $S^1 \cong \mathbb{R}/\mathbb{Z}$ and parameterise one traversal of the circle by $s\in[0,1]$, where the endpoints represent the same point. 
Accordingly, every Hamiltonian family considered below is \emph{closed}, with $H(0) = H(1)$.

Let $\Upsilon \coloneqq \{H(s)\}_{s\in [0,1]}$ be a family of $k$-local Hamiltonians on $n$ qubits, where $H\in{\rm C}^2\bigl([0,1];{\rm Herm}(\mathcal{H})\bigr)$ with $\mathcal{H}=(\mathbb{C}^2)^{\otimes n}$.
For each $s\in [0,1]$, suppose that
\begin{equation}
    H(s) = \sum_{j=1}^{m} h_j(s),
\end{equation}
where $m = \poly{n}$ and each $h_j(s)$ acts non-trivially on at most $k$ qubits.

For each $s\in [0,1]$, let $\{\lambda_j(s),\ket{\phi_j(s)}\}_{j=0}^{2^n-1}$ be an eigensystem of $H(s)$, ordered such that $\lambda_0(s)<\lambda_1(s)\leq\cdots\leq\lambda_{2^n-1}(s)$.
We assume that the ground state is non-degenerate for every $s$ and define the spectral gap by $\gamma(s)\coloneqq\lambda_1(s)-\lambda_0(s)$ with $\gamma_\star\coloneqq\inf_{s\in [0,1]}\gamma(s)$.

We call a family \emph{poly-gapped} if $\gamma_\star\geq1/\poly{n}$ and throughout we assume all families are uniformly poly-gapped, unless stated otherwise.
We further assume the uniform bounds $\snorm{\partial_s^r H}\leq\poly{n}$ for $r=0,1,2$.

Since $H$ is ${\rm C}^2$, the ground state is non-degenerate, and the gap remains open along the path, the ground-state projector $P(s) \coloneqq \ketbra{\phi_0(s)}$ is a ${\rm C}^2$ family. 
Moreover, $H(0) = H(1)$ and uniqueness of the ground state imply that $P(0) = P(1)$. 
A normalised ground-state eigenvector may be chosen ${\rm C}^2$ along $[0,1]$, but its phase need not return to its initial value. 
In general, there exists some $\alpha\in\mathbb{R}$ such that $\ket{\phi_0(1)}=\e^{\i\alpha}\ket{\phi_0(0)}$.

We denote a perturbed family by $\wt{\Upsilon} \coloneqq \{\wt{H}(s)\}_{s\in [0,1]}$ acting on a Hilbert space $\wt{\mathcal{H}}$, with eigensystem $\{\wt{\lambda}_j(s),\ket{\wt{\phi}_j(s)}\}_j$.
Unless stated otherwise, $\wt{\Upsilon}$ satisfies the analogous regularity, closure and spectral assumptions.

\subsection{Berry Phase Definition}
\label{sec:berry-phase-definition}
The Berry phase $\vartheta$ is the geometric phase acquired by $\ket{\phi_0(s)}$ as $s$ traverses a closed loop $C$ adiabatically~\cite{pancharatnam1956generalized, berry1984quantal}; see \cref{fig:berry-phase} for an illustration.
Choose a normalised ${\rm C}^1$ ground-state section satisfying the periodic-gauge condition $\ket{\phi_0(1)}=\ket{\phi_0(0)}$.
Such a section always exists for a non-degenerate closed path.
Then
\begin{equation}\label{eq:berry-phase}
    \vartheta = \int_0^1 \dd{s}\, \i\mel{\phi_0(s)}{\partial_s}{\phi_0(s)} \pmod{2\pi} \in \mathbb{R}/2\pi\mathbb{Z},
\end{equation}
defines the Berry phase $\vartheta$.
The integrand of \cref{eq:berry-phase} defines the \emph{Berry connection} $A(s)$; we take the representative $\vartheta \in [0, 2\pi)$.
For a non-periodic section one must add the endpoint term $\arg\braket{\phi_0(0)}{\phi_0(1)}$ to the integral.
Equivalently, the phase is the holonomy of Kato parallel transport, as proved in \cref{prop:berry-holonomy}.
The periodic-gauge expression \cref{eq:berry-phase} is invariant modulo $2\pi$ under every periodic gauge change (\cref{prop:berry-phase-gauge-invariance} in Appendix~\ref{app:preliminaries-proofs}).

The computation of $\vartheta$ rests on the adiabatic theorem, which we briefly recall.
A system initialised in the ground state of $H(0)$ remains close to the instantaneous ground state of a slowly varying $H(s)$, provided the gap $\gamma(s)$ is bounded away from zero~\cite{jansen2007bounds}.
Writing $U(T)$ for the evolution generated by $\dv{t}U(t) = -\i H(t/T)U(t)$ over total time $T$, the adiabatic error is $\norm{U(T)\ket{\phi_0(0)} - \ket{\phi_0(1)}} \le \varepsilon$.
Ref.~\cite{ambainis2006elementary} shows that if the first two derivatives of $H(s)$ and the inverse minimum gap are $\poly{n}$, then some $T = \poly{n, 1/\varepsilon}$ achieves adiabatic error at most $\varepsilon$.

\begin{figure}[!ht]
    \centering
    \begin{tikzpicture}
        \pic[scale=1.75]{berryloop};
    \end{tikzpicture}
    \caption{
        A visualisation of the Berry phase $\vartheta$ acquired by the ground state $\ket{\phi_0(s)}$ as $s$ is varied adiabatically from $0$ to $1$ along a closed loop $C$. 
        The arrows represent a parallel-transported ground-state vector, with the dynamical phase removed, at various points along the path.
        Since the projector path is closed, the initial and final transported vectors differ by the holonomy angle $\vartheta$.
    }
    \label{fig:berry-phase}    
\end{figure}

\subsection{Computational Complexity}
\label{sec:computational-complexity}
We consider only \emph{promise-problem} variants throughout and drop the ``promise'' prefix; a notion of hardness or completeness for the classes we consider is therefore understood in the promise sense.
We identify a circuit with the partial Boolean function it induces: for a family $\{A_n\}_{n\in\mathbb{N}}$ and input $x\in\B^*$, $A_{\abs{x}}(x)$ is the random variable obtained by running $A_{\abs{x}}$ on the standard initialisation determined by $x$ and measuring the output qubit, and $\Pr[A_{\abs{x}}(x) = \mathtt{1}]$ its probability of returning $\mathtt{1}$.
We set $n = \abs{x}$ and suppress the ancilla initialisation unless needed.
A family is \emph{polynomial-time generated} (uniform) if a deterministic polynomial-time Turing machine outputs a description of $A_n$ on input $1^n$.

\begin{definition}[\cl{BQP}]\label{def:bqp}
    Let $L = (L_{\sc{yes}}, L_{\sc{no}})$ be a promise problem and $a, b : \mathbb{N}\to[0,1]$.
    Then $L \in \cl{BQP}(a,b)$ if there is a polynomial-time generated quantum circuit family $Q = \{Q_n\}$ on $n + \poly{n}$ input qubits with one output qubit such that $\Pr[Q_n(x) = \mathtt{1}] \ge a(n)$ for $x\in L_{\sc{yes}}$ and $\Pr[Q_n(x) = \mathtt{1}] \le b(n)$ for $x\in L_{\sc{no}}$.
\end{definition}

We set $\cl{BQP} \coloneqq \cl{BQP}(2/3, 1/3)$; by repetition and majority voting, $\cl{BQP} = \cl{BQP}(1 - 2^{-q}, 2^{-q})$ for any polynomially bounded $q\ge 2$.
A problem is \clw{BQP}{hard} if every \cl{BQP} problem reduces to it under a polynomial-time (Karp) reduction, and \clw{BQP}{complete} if additionally in \cl{BQP}.
We take all \cl{BQP} circuits to be built from $2$-local gates in a finite universal set, e.g.\ $\{I, \Gate{Had}, T, \Gate{Cnot}\}$; the identity $I$ accommodates the pre-idle sequences of the history-state construction, and where it is omitted these are implemented by trivial-action sequences such as $\Gate{Had}^2$ or $T^8$~\cite{waite2026physically}.

\subsection{Problem Statement}
\label{sec:problem-statement}
We introduce the notion of a \emph{presentation} for a single local Hamiltonian which provides a classical description of the Hamiltonian.

\begin{definition}[Presentation of a Local Hamiltonian]
    Let $H = \sum_{j \in [m]} h_j$ be a $k$-local Hamiltonian on $n$ qubits, where each $h_j$ acts non-trivially on at most $k$ qubits.
    A \emph{presentation} of $H$ is a classical description of the Hamiltonian comprising a list of the $m$ local terms $h_j$ with their supports ${\sf S}_j \coloneqq \supp{h_j}$ and a specification of a qubit ordering on $\mathcal{H}$.
    Each list element is a pair $({\sf S}_j, \langle h_j \rangle)$, where ${\sf S}_j \subseteq [n]$ is a list of the qubits on which $h_j$ acts non-trivially, and the description $\langle h_j \rangle$ of $h_j$ constitutes a classical description of the $2^k \times 2^k$ matrix representing $h_j$ in the computational basis.
    Each component of the description $\langle h_j \rangle$ is represented using $\poly{n}$ classical bits in a known encoding.
\end{definition}

Given this description and indices $i,j \in \{0,1\}^n$, a polynomial-time classical algorithm can output a $\poly{n}$-bit representation of the matrix element $\mel{i}{H}{j}$~\cite{waite2026thesis}.
A $k$-local Hamiltonian is also $O(n^k)$-sparse in the computational basis~\cite{waite2026thesis} and is therefore row-sparse and row-computable, and hence simulable in the sense of Ref.~\cite{aharonov2003adiabatic}.
We extend the notion of a presentation to a family of parameterised Hamiltonians via access to a query model.
Specifically, on input $s\in[0,1]$ specified to $\poly{n}$ bits of precision, a presentation of $\Upsilon = \{H(s)\}_{s\in[0,1]}$ outputs a presentation of $H(s)$.
It also records the promised closure, gappedness, and ${\rm C}^r$ regularity of the family and supplies polynomial bounds on the first $r$ derivatives with respect to $s$.
We denote the presentation by $\mathcal{P}$.

We now formally define the main computational problem that we study in this work.

\problemdef
{\sc{$(r, k, \delta_0)$-Guided-State Berry Phase Estimation}} % name
{
A presentation $\mathcal{P}$ of a closed, poly-gapped and ${\rm C}^r$ family of $k$-local Hamiltonians $\Upsilon = \{H(s)\}_{s\in[0,1]}$ on $n$ qubits.
A separation parameter $\g\in(0,\pi)$ with $\g=1/\poly{n}$ specified in binary using $\poly{n}$ bits.
Two reals $a,b\in[0,2\pi)$, specified using $\poly{n}$ bits, that define the \emph{inner} arc $I_{\tn{in}}=[a+\g,\,b-\g]_{2\pi}$ and the \emph{outer} arc $I_{\tn{out}}=[b+\g,\,a-\g]_{2\pi}$, which are disjoint and separated by $2\g$, both specified using $\poly{n}$ bits.
An overlap parameter $\delta_0$ with $1/\poly{n}\le\delta_0<1$ specified in binary using $\poly{n}$ bits, and a polynomial-size classical guiding-state description $C_\xi$.
} % instance
{
Decide whether $\Upsilon$ has a Berry phase $\vartheta$ within the inner arc $I_{\tn{in}}$ (\sc{yes} case), or whether $\Upsilon$ has a Berry phase $\vartheta$ within the outer arc $I_{\tn{out}}$ (\sc{no} case), promised that one of these is the case and that the guiding state satisfies $|\braket{\phi_0(0)}{\xi}|^2\geq\delta_0$.
} % problem

We abbreviate the problem as $(r,k,\delta_0)$-\sc{GSBPE}.

We follow the requirements for guiding-state descriptions given in Ref.~\cite{waite2026physically} and briefly summarise the relevant points.
For the \cl{BQP} containment result of Ref.~\cite{hayakawa2025computational}, the classical description $C_\xi$ must support:
\begin{inparaenum}[(a)]
    \item uniform, polynomial-time preparation of the associated quantum state $\ket{\xi}$, and \label{item:uniform-preparation}
    \item polynomial-time classical compilation of the quantum circuit that prepares $\ket{\xi}$ from the classical description $C_\xi$.\footnote{Note that polynomial-time classical compilation implies a polynomial-space complexity as well.}\label{item:efficient-compilation}
\end{inparaenum}
The \clw{BQP}{hardness} construction of Ref.~\cite{hayakawa2025computational} can be constructed to use the so-called \emph{semi-classical subset states}~\cite{gharibian2023dequantizing} as guiding states.
These states are special cases of a related family known as semi-classical \emph{encoded} subset states~\cite{cade2023improved,waite2025guided}, which admit classical descriptions satisfying both Item~\ref{item:uniform-preparation} and Item~\ref{item:efficient-compilation}~\cite{waite2026physically}.
We present a formal definition of semi-classical encoded subset states below, which we will use in our hardness results.

\begin{definition}[Semi-Classical Encoded Subset States~\cite{cade2023improved,waite2025guided}]
    \label{def:semi-classical-encoded-subset-state}
    For any nonempty subset $S \subset \B^n$ such that $\abs{S} = O(\poly{n})$, let $\V = (V_j)_{j\in[n]}$ be an ordered set of isometries where, for each $j$ we have $V_j : \mathbb{C}^2 \to (\mathbb{C}^2)^{\otimes m_j}$ with $m_j = O(1)$.
    The semi-classical encoded subset state $\ket{S_{\V}}$ over $(S,\V)$ is defined as
    \begin{equation}
        \ket{S_{\V}} \coloneqq \frac{1}{\sqrt{\abs{S}}} \sum_{x \in S} \bigotimes_{j\in[n]}\, V_j\ket{x_j}.
    \end{equation}
    The components $\ket{x} = \bigotimes_{j\in[n]}\, \ket{x_j}$ are the standard basis states.
\end{definition}

A semi-classical subset state is a special case of a semi-classical encoded subset state where each $V_j$ is the identity operator.
A sufficient classical description of a semi-classical encoded subset state is a classical description of the subset $S$ and the isometries $\V$.
Given a secondary set of isometries $\mathcal{W}$ satisfying the same conditions as $\V$, we can update the semi-classical encoded subset state in classical polynomial time by composing the appropriate isometries; see Refs.~\cite{cade2023improved,waite2025guided,waite2026physically} for further details.
Throughout this work, when we reference a guiding state being a semi-classical encoded subset state, we implicitly assume that the classical description of the guiding state includes both the subset $S$ and the isometries $\V$, satisfying Items~\ref{item:uniform-preparation} and~\ref{item:efficient-compilation}.

We also note that the problem is monotone in all three parameters.
For $r'>r$, ${\rm C}^{r'}\subseteq {\rm C}^r$ gives \sc{$(r',k,\delta_0)$-GSBPE} $\subseteq$ \sc{$(r,k,\delta_0)$-GSBPE}; for $k'\leq k$, \sc{$(r,k',\delta_0)$-GSBPE} $\subseteq$ \sc{$(r,k,\delta_0)$-GSBPE}; and for $\delta_0'\geq\delta_0$, \sc{$(r,k,\delta_0')$-GSBPE} $\subseteq$ \sc{$(r,k,\delta_0)$-GSBPE}.
Ref.~\cite{hayakawa2025computational} gives a \cl{BQP} verification algorithm for \sc{GSBPE} assuming only $k$-locality and a guiding state satisfying Items~\ref{item:uniform-preparation} and~\ref{item:efficient-compilation}.
Containment in \cl{BQP} therefore holds for every family considered here.
Details of this algorithm can be found in Ref.~\cite{hayakawa2025computational}; see also Ref.~\cite{waite2026physically} for details on state preparation from a classical description of a semi-classical encoded subset state.
In \cref{sec:berry-phase-estimation-parameterised-1-local} we give an explicit ``guiding-state-free'' algorithm for one-local families, placing $(2,1,\delta_0)$-\sc{GSBPE} in \cl{P} for all $\delta_0$.
For formality and completeness, we recall the hardness result of Ref.~\cite{hayakawa2025computational}.

\begin{restatable}[Ref.~\cite{hayakawa2025computational}]{theorem}{bqphardness}\label{thm:hayakawa-bqp-hardness}
    The $(2, 5, \delta_0)$-\sc{GSBPE} problem is \clw{BQP}{hard} for any $\delta_0 \in [1/\poly{n}, 1 - 1/\poly{n}]$ and when the guiding state is a semi-classical encoded subset state.
\end{restatable}

\subsection{Analytic Preliminaries}
\label{sec:analytic-preliminaries}
We collect the few notions from complex analysis used in the spectral-projector estimates. 
Throughout, a contour is a piecewise-${\rm C}^1$ closed curve in $\mathbb{C}$, oriented positively (counterclockwise), and $\abs{\mathscr{C}}$ denotes its arc length.

\subsubsection{Holomorphic functions and Cauchy's theorem}
A function $f:\mathcal{O}\to\mathbb{C}$ on an open set $\mathcal{O}\subseteq\mathbb{C}$ is \emph{holomorphic} (equivalently \emph{analytic}) if it is complex-differentiable at every point of $\mathcal{O}$. 
\emph{Cauchy's theorem} states that if $f$ is holomorphic on and inside a contour $\mathscr{C}$, then $\oint_\mathscr{C} f(z)\,\dd z = 0$.
A point $z_0$ where $f$ fails to be holomorphic may be a \emph{simple pole}, where $f(z)=\frac{c}{z-z_0}+(\text{holomorphic})$; its \emph{residue} is $c$.
The residue theorem gives $\frac{1}{2\pi\i}\oint_\mathscr{C} f(z)\,\dd z = \sum_{z_0 \text{ inside } \mathscr{C}} c$, where the sum ranges over the poles enclosed by $\mathscr{C}$.
In particular,
\begin{equation}
    \frac{1}{2\pi\i}\oint_\mathscr{C} \frac{\dd{z}}{z-\lambda} = 
    \begin{cases}
        1 & \lambda \text{ inside } \mathscr{C},\\ 
        0 & \lambda \text{ outside } \mathscr{C},
    \end{cases}
\end{equation}
which is the only instance we use: a pole inside $\mathscr{C}$ integrates to $1$, while an integrand holomorphic inside $\mathscr{C}$ (no enclosed pole) integrates to $0$.

\subsubsection{Resolvents and the contour estimate}
For a Hermitian operator $X$ with spectrum $\sigma(X)$, the \emph{resolvent} $R_X(z)=(z-X)^{-1}$ is defined for $z\notin\sigma(X)$ and satisfies $\|R_X(z)\| = {\rm dist}(z,\sigma(X))^{-1}$, where ${\rm dist}(z,S)=\inf_{\lambda\in S}\abs{z-\lambda}$. 
We repeatedly use the \emph{contour (ML) estimate}, that is, for any operator-valued $F$ continuous on $\mathscr{C}$,
\begin{equation}
    \left\|\frac{1}{2\pi\i}\oint_\mathscr{C} F(z)\,\dd z\right\| \leq \frac{\abs{\mathscr{C}}}{2\pi}\,\sup_{z\in\mathscr{C}}\|F(z)\|.
\end{equation}

\subsubsection{Continuity of the spectrum}\label{sec:continuity-of-spectrum}
The eigenvalues of a Hermitian operator depend continuously on the operator (in operator norm). 
Hence, if $s\mapsto H(s)$ is continuous and a contour $\mathscr{C}$ encloses exactly the eigenvalue $\lambda_0(s_0)$ of $H(s_0)$ and no other, then $\mathscr{C}$ continues to enclose exactly $\lambda_0(s)$ for all $s$ in a neighbourhood of $s_0$; we call such a contour \emph{admissible}. 
Since the Riesz projector $\frac{1}{2\pi\i}\oint_\mathscr{C} \dd{z}\, R(z,s)$ is independent of the admissible contour, we may hold $\mathscr{C}$ fixed on a neighbourhood of $s_0$.
We may then differentiate under the integral sign without a boundary contribution: $\partial_s\oint_\mathscr{C} \dd{z}\, R(z,s) = \oint_\mathscr{C} \dd{z}\, \ps R(z,s)$.
Differentiating $R (z - H) = I$ gives $\ps R(z,s) = R(z,s)\,(\ps H(s))\,R(z,s)$.

\subsubsection{Unimodular functions and continuous phase lifts}
A complex number is \emph{unimodular} if it has modulus $1$. 
A ${\rm C}^1$, nowhere-vanishing function $c:[0,1]\to\mathbb{C}$ admits a ${\rm C}^1$ \emph{phase lift}, i.e., a real-valued ${\rm C}^1$ function $\beta$ with $c(s)=\abs{c(s)}\e^{\i\beta(s)}$, unique up to an additive constant in $2\pi\mathbb{Z}$.
When $\abs{c}\equiv 1$ this reduces to $c(s)=\e^{\i\beta(s)}$.
 
\subsubsection{Lipschitz continuity}
A function $f:\mathcal{O}\to\mathbb{C}$ is \emph{Lipschitz continuous} if there exists a constant $L$ such that $\abs{f(z)-f(z')}\leq L\abs{z-z'}$ for all $z,z'\in\mathcal{O}$; the smallest such $L$ is the \emph{Lipschitz constant} of $f$. A Lipschitz continuous function is uniformly continuous, and in particular bounded on compact sets. 
If $f$ is holomorphic on a convex set $\mathcal{O}$ and $\sup_{z\in\mathcal{O}}\abs{f'(z)}\leq M$, then $f$ is Lipschitz continuous with Lipschitz constant at most $M$.
If the local terms of $\Upsilon$ have bounded derivative norms, then $H(s)$ is Lipschitz continuous with constant at most the sum of those bounds.
These analytic conventions will be used in the spectral-projector and Berry-phase estimates of Section~\ref{sec:berry-phase-simulation}.

\subsubsection{Weierstrass M-test}
The Weierstrass M-test provides a criterion for uniform convergence of series of functions. 
Suppose $\{f_j\}_{j\geq 1}$ is a sequence of functions defined on a set $\mathcal{O}$ and there exists a sequence of non-negative numbers $\{M_j\}_{j\geq 1}$ such that $\abs{f_j(z)} \leq M_j$ for all $z \in \mathcal{O}$ and all $j \geq 1$. 
If $\sum_{j\geq 1} M_j$ converges, then the series $\sum_{j\geq 1} f_j(z)$ converges uniformly on $\mathcal{O}$.

%% file: sections/uniform-simulation.tex
\section{Uniform Simulation of Parameterised Hamiltonian Families}
\label{sec:uniform-simulation}
In this section, we extend the static simulation framework of Bravyi and Hastings~\cite{bravyi2016complexity} to parameterised Hamiltonian families and prove that the resulting simulations compose.
The definition uses a fixed encoding isometry $\V$ for the full path and a parameter-dependent witness family of isometries $\{\U(s)\}_{s \in [0,1]}$.
Controlling the derivative of this witness family will later control the Berry-phase simulation error.

\subsection{Definitions and Basic Properties}
\label{sec:definitions-and-basic-properties}
Let $\Upsilon = \{H(s)\}_{s\in[0,1]}$ be a ${\rm C}^2$, poly-gapped, closed target family on a Hilbert space $\Hil$ of dimension $2^n$, and let $\wt{\Upsilon} = \{\wt{H}(s)\}_{s\in[0,1]}$ be a simulator family on $\wt{\Hil}$ of dimension $2^m$, where $m \ge n$.
Fix an isometry $\V : \Hil \to \wt{\Hil}$ embedding $\Hil$ into a $2^n$-dimensional subspace of $\wt{\Hil}$.
For each $s$, suppose the spectrum of $\wt{H}(s)$ splits into a low-energy band of the first $2^n$ eigenvalues and a high-energy band, with $\mathcal{E}_{2^n}(\wt{H}(s)) \subseteq \wt{\Hil}$ the low-energy subspace and $\Delta(s)$ the separating gap.
We call $\Delta_\star \coloneqq \inf_{s\in[0,1]}\Delta(s)$ the \emph{band gap} of the simulator path and assume $\Delta_\star > 0$.

We now define the notion of a uniform simulation of a parameterised family of Hamiltonians we use throughout this work.

\begin{definition}[Uniform simulation]\label{def:uniform-simulation}
    We say $(\wt{\Upsilon}, \V)$ \emph{uniformly simulates} $\Upsilon$ with error $(\eta, \varepsilon)$ if there is a family of ${\rm C}^2$ isometries $\Phi = \{\U(s)\}_{s\in[0,1]}$ such that, for every $s\in[0,1]$,
    \begin{enumerate}[leftmargin=1.25cm, label=\textup{(C\arabic*)}]
        \item the image of $\U(s)$ equals $\mathcal{E}_{2^n}(\wt{H}(s))$; \label{item:c1}
        \item $\norm{H(s) - \U^\dagger(s) \wt{H}(s) \U(s)} \le \varepsilon$; \label{item:c2}
        \item $\norm{\U(s) - \V} \le \eta$. \label{item:c3} \qedhere
    \end{enumerate}
\end{definition}

We write $E(s) \coloneqq H(s) - \U^\dagger(s)\wt{H}(s)\U(s)$ for the simulation error operator, and therefore $\snorm{E} \le \varepsilon$.
Also, note that $\unit{\eta} = 1$ and $\unit{\varepsilon} = e$.

The term \emph{uniform} emphasises that $\V$ and $(\eta, \varepsilon)$ are independent of $s$. 
The isometries $\U(s)$ are taken to be ${\rm C}^2$ for technical reasons related to the control of the Berry phase.
We say the simulation is \emph{witnessed by} $\Phi$ to emphasise that existence of a suitable collection $\{\U(s)\}$ is part of the definition, and we stress that the properties of the witness are not inherited by the fixed encoding $\V$.\footnote{Here, `witness' denotes the family $\{\U(s)\}$ satisfying the simulation conditions, not a witness in the sense of Merlin--Arthur complexity theory.}

The witness $\Phi$ need not be unique because Conditions~\ref{item:c1}--\ref{item:c3} leave gauge freedom within the low-energy band.
The Berry-phase analysis of Section~\ref{sec:berry-phase-simulation} requires a closed witness whose derivative is controlled along the path.
For $\Phi = \{\U(s)\}_{s\in[0,1]}$, define its \emph{derivative norm} by $\kappa(\Phi) = \snorm{\ps\U}$; since $s$ is dimensionless, $\unit{\kappa} = 1$.
This motivates the following definition which describes what we mean by a \emph{regular} simulation.

\begin{definition}[Regular simulation]\label{def:regular-simulation}
    A uniform simulation $(\wt{\Upsilon}, \V)$ of $\Upsilon$ is \emph{regular} with derivative norm $\kappa$ if there is a witness $\Phi = \{\U(s)\}_{s\in[0,1]}$ satisfying Definition~\ref{def:uniform-simulation} and, additionally,
    \begin{enumerate}[leftmargin=1.25cm, label=\textup{(C\arabic*)}]
        \setcounter{enumi}{3}
        \item $\U(1) = \U(0)$; \label{item:c4}
        \item $\snorm{\ps\U} \le \kappa$. \label{item:c5}
    \end{enumerate}
    We call such a $\Phi$ a \emph{regular witness family} and write that $(\wt{\Upsilon}, \V)$ is an $(\eta, \varepsilon, \kappa)$-simulation of $\Upsilon$.
\end{definition}

\begin{remark}\label{rem:kappa-witness-dependent}
    Different witnesses for the same pair $(\wt{\Upsilon}, \V)$ can yield different $\kappa$. 
    Hence, $\kappa$ is a property of the witness, not of the pair.
    From a reduction standpoint we never write $\U(s)$ explicitly; we need only that a witness with small $\kappa$ exists.
    In every construction below we exhibit the existence of a witness satisfying \cref{item:c1,item:c2,item:c3,item:c4,item:c5}.
\end{remark}

The eigenvalue- and ground-state-simulation lemmas of Ref.~\cite{bravyi2016complexity} extend to parameterised families by applying the static results pointwise in $s$.
We impose $2\varepsilon < \gamma_\star$ when ground-state non-degeneracy must be preserved in the low-energy band (i.e., to prevent a level crossing).

\begin{restatable}[Eigenvalue simulation]{lemma}{eigenvaluesimulation}\label{lem:eigenvalue-simulation}
    If $(\wt{\Upsilon}, \V)$ is an $(\eta, \varepsilon)$-simulation of $\Upsilon$, then for every $s\in[0,1]$ and $j\in\{0, \dots, 2^n-1\}$,
    \begin{equation}\label{eq:eigenvalue-simulation}
        |\lambda_j(s) - \wt{\lambda}_j(s)| \le \varepsilon.
    \end{equation}
\end{restatable}

A proof is given in Appendix~\ref{app:uniform-simulation-proofs}.

\begin{restatable}[Ground state simulation]{lemma}{groundstatesimulation}\label{lem:ground-state-simulation}
    Assume each $H(s) \in \Upsilon$ has a non-degenerate ground state $\ket{\phi_0(s)}$ with $\gamma(s) \ge \gamma_\star > 0$, and let $(\wt{\Upsilon}, \V)$ be an $(\eta, \varepsilon)$-simulation with $2\varepsilon < \gamma_\star$.
    Then for every $s\in[0,1]$, $\wt{H}(s)$ has a non-degenerate ground state and the phases of $\ket{\phi_0(s)}$ and $\ket{\wt{\phi}_0(s)}$ may be chosen so that
    \begin{equation}\label{eq:ground-state-simulation}
        \norm{\ket{\wt{\phi}_0(s)} - \V\ket{\phi_0(s)}} = O(\eta + \gamma_\star^{-1}\varepsilon).
    \end{equation}
\end{restatable}

A proof is given in Appendix~\ref{app:uniform-simulation-proofs}.

\subsection{Composition of Simulations}
\label{sec:composition-of-simulations}
We now compose a simulation of the target family $\Upsilon$ by an intermediate family $\wt{\Upsilon}_1$ with a simulation of that intermediate family by a final family $\wt{\Upsilon}_2$.
Unfortunately, the naive approach of composing two simulations by simply concatenating the isometries may not always yield a valid simulation, as the low-energy subspace of the final simulator may not align with that of the intermediate one~\cite{bravyi2016complexity}.
We therefore insert a unitary rotation that aligns these subspaces and control both its norm and its derivative.

\begin{restatable}[Parameterised unitary rotation]{lemma}{parameterisedunitaryrotation}\label{lem:parameterised-unitary-rotation}
    Let $A(s), B(s)$ be closed ${\rm C}^2$ families of equal-rank orthogonal projectors with $\norm{A(s)-B(s)}\le\delta$ for all $s$, and suppose $4\delta\le1$.
    Then there is a closed ${\rm C}^2$ family of unitaries $U(s)$ and a universal constant $C > 0$ such that, for all $s$,
    \begin{enumerate}[label=\arabic*)]
        \item $U(s)A(s)U^\dagger(s) = B(s)$;
        \item $\norm{U(s) - I} \le C\delta$;
        \item $\norm{\ps U(s)} \le C(\norm{\ps A(s)} + \norm{\ps B(s)})$.
    \end{enumerate}
\end{restatable}

A proof is given in Appendix~\ref{app:uniform-simulation-proofs}.

\begin{restatable}[Composition of two simulations]{lemma}{compositionofsimulations}\label{lem:composition-of-simulations}
    Let $(\wt\Upsilon_1, \V_1)$ be an $(\eta_1, \varepsilon_1, \kappa_1)$-simulation of $\Upsilon$ witnessed by $\Phi_1 = \{\U_1(s)\}$, and $(\wt\Upsilon_2, \V_2)$ an $(\eta_2, \varepsilon_2, \kappa_2)$-simulation of $\wt\Upsilon_1$ witnessed by $\Phi_2 = \{\U_2(s)\}$, with band gaps $\Delta_1, \Delta_2$.
    Suppose $8\varepsilon_2 \le \Delta_1$ and $\varepsilon_1, \varepsilon_2 < \snorm{H}$.
    Then there is a ${\rm C}^2$ family of unitary rotations $T(s)$ such that $\Phi=\{T(s)\U_2(s)\circ\U_1(s)\}$ witnesses an $(\eta,\varepsilon,\kappa)$-simulation of $\Upsilon$ with
    \begin{align}
        \eta &= \eta_1 + \eta_2 + O(\Delta_1^{-1}\varepsilon_2), &
        \varepsilon &= \varepsilon_1 + \varepsilon_2 + O(\Delta_1^{-1}\varepsilon_2\snorm{H}),
    \end{align}
    and derivative norm
    \begin{equation}
        \kappa = O(\kappa_1 + \kappa_2 + \Delta_1^{-1}\max\{\snorm{\ps\wt{H}_1},\snorm{\ps\wt{H}_2}\}).
    \end{equation}
    If both simulations are closed, so is the composition.
\end{restatable}

A proof is given in Appendix~\ref{app:uniform-simulation-proofs}.
Note that the fixed encoding of the composed simulation is $\V = \V_2 \circ \V_1$.
Now we can iterate this composition to a chain of simulations, as described in the next subsection.

\subsection{Generalised Composition}
\label{sec:generalised-composition}
We iterate the composition to a chain of $M$ simulations, each simulating the previous.
For notational convenience write $\wt{\Upsilon}_0 = \Upsilon$ and $\wt{H}_0 = H$.
For $m = 1, \dots, M$, let $\wt{\Upsilon}_m$ simulate $\wt{\Upsilon}_{m-1}$ and write $\wt{H}_m$ for its Hamiltonian.

\begin{corollary}\label{cor:composition-of-simulations}
    Let $(\wt\Upsilon_i, \V_i)$ be an $(\eta_i, \varepsilon_i, \kappa_i)$-simulation of $\wt{\Upsilon}_{i-1}$ witnessed by $\Phi_i = \{\U_i(s)\}$, with band gap $\Delta_i$, for $i = 1, \dots, M$.
    Suppose $8\varepsilon_i \le \Delta_{i-1}$ for $i = 2, \dots, M$ and $\varepsilon_i < \snorm{\wt{H}_0}$ for all $i$.
    Then there are ${\rm C}^2$ unitary rotations $\{T_i(s)\}_{i=2}^M$ such that $\Phi = \{T_M(s)\,\U_M \circ \cdots \circ T_2(s)\,\U_2 \circ \U_1(s)\}$ witnesses an $(\eta, \varepsilon, \kappa)$-simulation of $\Upsilon$ with
    \begin{align}
        \eta &= \sum_{i=1}^{M}\eta_i + O\Big(\sum_{i=2}^{M}\Delta_{i-1}^{-1}\varepsilon_i\Big), &
        \varepsilon &= \sum_{i=1}^{M}\varepsilon_i + O\Big(\sum_{i=2}^{M}\Delta_{i-1}^{-1}\varepsilon_i\snorm{\wt{H}_0}\Big),
    \end{align}
    and derivative norm
    \begin{equation}
        \kappa = O\Big(\sum_{i=1}^{M}\kappa_i + \sum_{i=2}^{M}\Delta_{i-1}^{-1}\max\{\snorm{\ps\wt{H}_{i-1}}, \snorm{\ps\wt{H}_i}\}\Big).
    \end{equation}
    If all simulations are closed, so is the composition.
    Taking the worst stage, with $\chi_* = \max_i \chi_i$ for $\chi\in\{\eta, \varepsilon, \kappa\}$, $\Delta_* = \min_i\Delta_i$, and $\snorm{\ps\wt{H}_*} = \max_i\snorm{\ps\wt{H}_i}$,
    \begin{align}
        \eta &= O\big(M(\eta_* + \Delta_*^{-1}\varepsilon_*)\big), &
        \varepsilon &= O\big(M(\varepsilon_* + \Delta_*^{-1}\varepsilon_*\snorm{H})\big), &
        \kappa &= O\big(M(\kappa_* + \Delta_*^{-1}\snorm{\ps\wt{H}_*})\big).
    \end{align}
\end{corollary}

Note that the fixed encoding of the composed simulation is $\V = \V_M \circ \cdots \circ \V_2 \circ \V_1$.
Since $\unit{\Delta_i}$, $\unit{\snorm{H}}$, $\unit{\snorm{\ps\wt{H}_i}}$, $\unit{\varepsilon} = e$ and $\unit{\eta}$, $\unit{\kappa}$, $\unit{M} = 1$, each expression is dimensionally consistent.
We always take $\Delta_* \gg M\max\{\eta_*^{-1}\varepsilon_*,\snorm{H},\kappa_*^{-1}\snorm{\ps\wt{H}_*}\}$, so that $\chi \approx M\chi_*$ for $\chi\in\{\eta, \varepsilon, \kappa\}$.
In the sequel, the number of composition stages is constant, $M = O(1)$, reflecting the idea that we only need a fixed number of simulation layers in the reductions.
Equivalently, this corresponds to performing a constant number of perturbative gadget constructions in \emph{series}.
The accumulated simulation errors and derivative bounds therefore remain within constant factors of the corresponding worst-stage quantities, which is the form used in the reductions of Section~\ref{sec:extended-hardness-results}.

%% file: sections/sw-transformation.tex
\section{Schrieffer--Wolff Transformation for Parameterised Hamiltonian Families}
\label{sec:sw-transformation}
In this section, we develop the Schrieffer--Wolff transformation for parameterised families of Hamiltonians and its application to perturbative gadget reductions that will be used later in this work to extend the \clw{BQP}{hardness} results for the \sc{GSBPE} problem.
Our treatment follows the standard framework of Refs.~\cite{datta1996low,bravyi2008quantum,schuch2009computational,bravyi2011schrieffer,bravyi2016complexity,piddock2017complexity}, adopting the corresponding setting and conventions for the relevant components of the simulator Hamiltonian.

Passing from a static Hamiltonian to a parameterised family introduces two new requirements. 
First, the Schrieffer--Wolff transformation and the resulting effective Hamiltonian must be defined coherently as functions of the parameter. 
Second, the estimates required for the simulation must hold uniformly over the parameter domain. 
At each fixed parameter value, the perturbative coefficients should reproduce the algebraic structure of the static construction. 
We verify this structure explicitly and then establish bounds that hold uniformly along the path.
We verify this explicitly and consequently find that our parameterised construction retains the familiar static simulator Hamiltonian arguments.

Throughout this section, we develop the Schrieffer--Wolff framework from the perspective of its intended application to parameterised perturbative gadget reductions, rather than as a purely standalone theoretical construct.
We first treat the single-target parameterised setting and then extend the analysis to what is known as the \emph{parallel} setting. 
In the latter case, we prove that the relevant contributions decompose into a sum of corresponding single-target contributions, thus mirroring both the structure found in the single-target setting and the familiar parallel static case.

The underlying strategy is to impose a large energy penalty parameter so that the low-energy behaviour of a carefully constructed simulator Hamiltonian approximates that of a chosen target system. 
Given a target family $\Upsilon = \{H(s)\}_{s\in[0,1]}$ in a class $\mathcal{F}$, one constructs a simulator family $\wt{\Upsilon} = \{\wt{H}(s)\}_{s\in[0,1]}$ in a more restricted class $\mathcal{F}' \subseteq \mathcal{F}$, such as a class of lower-locality Hamiltonians. 
The Schrieffer--Wolff transformation block-diagonalises the simulator relative to its unperturbed low- and high-energy subspaces, thereby defining an exact low-energy effective Hamiltonian and providing the isometry $\U$ that witnesses the simulation. 
Truncating the associated perturbative expansion then produces an approximate effective Hamiltonian that can be compared directly with the target.

For the complexity-theoretic reductions considered later in this work, identifying the formal structure of the effective Hamiltonian is not by itself sufficient. 
The penalty scale, perturbation strengths, and truncation order must be chosen so that all relevant simulation errors remain uniformly controlled while the resulting construction retains the required efficiency and locality.

We begin by specifying the parameterised perturbed Hamiltonians used in our later gadget constructions and recording the basic properties of their low- and high-energy block structure. 
We then introduce the Schrieffer--Wolff transformation and the corresponding exact effective Hamiltonian, before deriving its perturbative expansion and truncated approximation. 
Next, we establish uniform control of the transformation and of the exact and truncated effective Hamiltonians over the parameter domain. 
Finally, we treat parallel perturbative reductions and prove the decomposition into their constituent single-target contributions.

\subsection{Simulator Hamiltonians}
\label{sec:simulator-hamiltonians}
Let $\wt{\H} = \H_- \oplus \H_+$ and define $\Pi_-$, $\Pi_+$ to be the orthogonal projectors onto $\H_-$ and $\H_+$, respectively.
We refer to $\H_-$ as the low-energy subspace and $\H_+$ as the high-energy subspace.
For an operator $O$ on $\wt{\H}$, write $O_{\alpha\beta} = \Pi_\alpha O \Pi_\beta$ for $\alpha, \beta \in \{-, +\}$ and define 
\begin{align}
    O_\tn{d} &\coloneqq O_{--} + O_{++}, &
    O_\tn{od} &\coloneqq O_{-+} + O_{+-},
\end{align}
as the \emph{block-diagonal} and \emph{block-off-diagonal} parts of $O$, respectively.
We define $\D_\tn{d} = \{O : O = O_\tn{d}\}$ as the set of block-diagonal operators on $\wt{\H}$ and $\D_\tn{od} = \{O : O = O_\tn{od}\}$ as the set of block-off-diagonal operators on $\wt{\H}$.

We recall some simple properties of block-diagonal and block-off-diagonal operators.

\begin{lemma}[\cite{waite2026thesis}]
    \label{lem:block-commutators}
    If $A \in \D_\tn{d}$ and $B \in \D_\tn{od}$, then $\comm{A}{B}$ is block-off-diagonal.
    If $B_1, B_2 \in \D_\tn{od}$, then $\comm{B_1}{B_2}$ is block-diagonal.
\end{lemma}

\begin{corollary}\label{cor:L-preserves-block-structure}
    Let $A \in \D_{\tn{d}}$, then $\L_{A}$ preserves $\D_{\tn{d}}$ and $\D_{\tn{od}}$.
    That is, $\L_A(\D_{\tn{d}}) \subseteq \D_{\tn{d}}$ and $\L_A(\D_{\tn{od}}) \subseteq \D_{\tn{od}}$.
\end{corollary}

The families of simulator Hamiltonians we construct are of the standard form found in the literature~\cite{kempe2006complexity,oliveira2008complexity,biamonte2008realizable,schuch2009computational,cubitt2016complexity,piddock2017complexity,bravyi2016complexity,cubitt2018universal,waite2025complexitya}, extended to the parameter-dependent setting.
Specifically, we consider Hamiltonians structured as
\begin{equation}
    \wt{H}(s) = \Delta H_0 + V(s),
\end{equation}
where $H_0$ is a block-diagonal Hamiltonian with unit gap and $\Delta H_0$ has band gap $\Delta$, while $V(s)$ is a perturbation with block-diagonal and block-off-diagonal components carrying the parameter dependence.
We focus on this form because it is the one used by the perturbative gadgets developed later.
Section~\ref{sec:perturbative-gadget-reductions} further specialises $V(s)$ for first-, second-, and third-order reductions.

In this work we always choose $H_0$ to be independent of the parameter $s$ and 
\begin{align}\label{eq:H0-block-structure}
    (H_0)_{--} &= 0, &
    (H_0)_{++} &\succeq \Pi_+.
\end{align}
Notice that $H_0$ has no ordinary inverse on the full Hilbert space, since $(H_0)_{--} = 0$. 
We define the quantity 
\begin{equation}\label{eq:Sigma-definition}
    \Sigma \coloneqq \Pi_+ \left( H_0 \bigr|_{{\rm Im}(\Pi_+)} \right)^{-1} \Pi_+,
\end{equation}
which is precisely the Moore--Penrose pseudoinverse $H_0^{\pattce}$.
From this definition, it follows that
\begin{equation}\label{eq:Sigma-inverse-identity}
    (H_0)_{++} \Sigma = \Sigma (H_0)_{++} = \Pi_+.
\end{equation}
We also assume throughout that $2\snorm{V} < \Delta$ to avoid closing the band gap of $H_0$ (and for analytic reasons we explore momentarily).

In the case where $H_0 = \Pi_+$, we characterise useful facts about the map $\L_{\Delta H_0}$ that will be applied in subsequent analyses.

\begin{restatable}{lemma}{LInverse}\label{lem:L-inverse}
    Let $\Delta > 0$ and let $H_0 = \Pi_+$ and therefore satisfy \cref{eq:H0-block-structure}.
    Then $\L_{\Delta H_0}$ restricted to $\D_\tn{od}$ is a linear bijection of $\D_\tn{od}$ onto itself, whose inverse is the restriction of the pseudoinverse \cref{eq:liouville_pseudoinverse_def} to $\D_\tn{od}$:
    \begin{equation}\label{eq:L-inverse}
        \L_{\Delta H_0}^{\pattce}(X) = \Delta^{-1} \big( \Sigma X_{+-} - X_{-+} \Sigma \big),
        \qquad X \in \D_\tn{od}.
    \end{equation}
    Consequently, for every $X \in \D_\tn{od}$ the equation $\L_{\Delta H_0}(S) = X$ has a unique solution $S \in \D_\tn{od}$, namely $S = \L_{\Delta H_0}^{\pattce}(X)$.
\end{restatable}

A proof is given in Appendix~\ref{app:sw-transformation-proofs}.

Having specified the simulator families, we now define the exact Schrieffer--Wolff transformation and its effective Hamiltonian.

\subsection{Exact Schrieffer--Wolff Transformation}
\label{sec:exact-sw-transformation}
Take $\wt{\Pi}_-(s)$ to be the projector onto the low-energy subspace of the perturbed Hamiltonian $\wt{H}(s)$.
When the perturbation $V(s)$ is sufficiently small compared to the band gap $\Delta$, there exists a unique anti-Hermitian and block-off-diagonal operator $S(s)$, with $\snorm{S} < \pi/2$~\cite{bravyi2011schrieffer}, such that
\begin{equation}
    \wt{\Pi}_-(s) = \e^{-S(s)} \Pi_- \e^{S(s)}.
\end{equation}
The operator $S(s)$ is the exact Schrieffer--Wolff generator that block-diagonalises the perturbed Hamiltonian $\e^{S(s)} \wt{H}(s) \e^{-S(s)}$.
Restricting to the low-energy subspace defines the \emph{effective} Hamiltonian
\begin{equation}\label{eq:exact-effective-hamiltonian}
    H_{\tn{eff}}(s) = \Pi_- \e^{S(s)} \wt{H}(s) \e^{-S(s)} \Pi_-.
\end{equation}
Therefore, $H_{\tn{eff}}(s)$ is exactly unitarily equivalent to the low-energy block of the perturbed Hamiltonian $\wt{H}(s)$.

Let $H(s)$ be the target Hamiltonian and let $\V : \Hil \to \wt{\Hil}$ be a fixed isometry with $\V\V^\dagger = \Pi_-$.
Write
\begin{equation}\label{eq:encoded-target-hamiltonian}
    \overline{H}(s) = \V H(s) \V^\dagger
\end{equation}
as the \emph{encoded} target Hamiltonian in the unperturbed low-energy subspace ${\rm Im}(\Pi_-)$.
The exact Schrieffer--Wolff isometry is taken to be
\begin{equation}
    \U(s) = \e^{-S(s)} \V.
\end{equation}
It follows that $\U(s)\U(s)^\dagger = \wt{\Pi}_-(s)$ and so $\U(s)$ has exactly the image required of a simulation witness (cf. Section~\ref{sec:uniform-simulation}).
Moreover, 
\begin{equation}
    \U(s)^\dagger \wt{H}(s) \U(s) = \V^\dagger H_{\tn{eff}}(s) \V.
\end{equation}

These observations reduce the simulation proof to showing that the exact effective Hamiltonian in \cref{eq:exact-effective-hamiltonian} is close in operator norm to the encoded target Hamiltonian in \cref{eq:encoded-target-hamiltonian}.
In the next subsection, we will develop perturbative expansions for the exact generator $S(s)$ and the exact effective Hamiltonian, which we subsequently truncate to obtain approximate expressions; the exact objects are generally not available in closed form.
Our applications of these ideas, in the computational reduction for the \sc{GSBPE} problem, will require proving that the Berry phase of the simulator system is close to the Berry phase of the target Hamiltonian.

Before proceeding to the perturbative expansions, we establish some differentiability and closure properties of the exact canonical Schrieffer--Wolff transformation, assuming that the perturbation $V(s)$ is itself closed and ${\rm C}^2$.

\begin{restatable}{proposition}{regularityperiodicitycanonicalsw}\label{prop:regularity-periodicity-canonical-sw}
    Let $H_0$ be independent of $s$ and let $V\in{\rm C}^2$ be closed.
    Suppose that $2\snorm{V}<\Delta$, so that the low-energy band of $\widetilde{H}(s) = \Delta H_0+V(s)$ remains uniformly separated from the high-energy band and the canonical Schrieffer--Wolff transformation is well-defined for every $s\in[0,1]$.
    Then the exact canonical generator $S(s)$ and the unitary $\e^{-S(s)}$ are ${\rm C}^2$ families satisfying
    \begin{align}
        S(0) &= S(1), & 
        \e^{-S(0)} &= \e^{-S(1)}.
    \end{align}
    If $\V$ is independent of $s$, then the isometry $\U(s)=\e^{-S(s)}\V$ is also ${\rm C}^2$ and closed.
\end{restatable}

A proof is given in Appendix~\ref{app:sw-transformation-proofs}.

\subsection{Perturbative Coefficients and Truncation}
\label{sec:perturbative-coefficients-and-truncation}
Fix $s$ temporarily, introduce a formal coupling parameter $t$ and set $\wt{H}(t) = \Delta H_0 + t V$.
The exact Schrieffer--Wolff generator and exact effective Hamiltonian can be expanded as power series in $t$; these series are analytic in $t$ on the disc $\abs{t}<\Delta/(2\snorm{V})$, while \cite{bravyi2011schrieffer} guarantees that the series for the exact Schrieffer--Wolff generator and effective Hamiltonian converge absolutely on the disc $\abs{t}<\Delta/(16\snorm{V})$.
Specifically, we write
\begin{align}
    S(t) &= \sum_{j \geq 1} t^j S_j, &
    H_{\tn{eff}}(t) &= \sum_{j \geq 0} t^j K_j.
\end{align}
For brevity we have set $K_j = H_{\tn{eff},j}$, where $H_{\tn{eff},j}$ are the coefficients in the power series expansion of the exact effective Hamiltonian.
Note that since $(H_0)_{--} = 0$, we have $K_0 = 0$ and so the effective Hamiltonian series starts at order $t$.

At $t=1$, define the truncated generator and truncated effective Hamiltonian as
\begin{align}\label{eq:truncated-generator-and-effective-hamiltonian}
    S^{[p]} &\coloneqq \sum_{j=1}^{p} S_j, &
    H_{\tn{eff}}^{[p]} &\coloneqq \sum_{j=0}^{p} K_j,
\end{align}
and define their respective remainders as
\begin{align}\label{eq:remainder-terms-generator-and-effective-hamiltonian}
    R^S_{p+1} &\coloneqq S - S^{[p]} = \sum_{j \geq p+1} S_j, &
    R^{H_{\tn{eff}}}_{p+1} &\coloneqq H_{\tn{eff}} - H_{\tn{eff}}^{[p]} = \sum_{j \geq p+1} K_j.
\end{align}
The operator $H_{\tn{eff}}^{[p]}$ is obtained by truncating the Taylor-BCH expansion, \cref{eq:bch-expansion}, of the exact effective Hamiltonian at order $p$.
It is not, in general, the full untruncated expression $(\e^{S^{[p]}} \wt{H} \e^{-S^{[p]}})_{--}$ (\cref{eq:exact-effective-hamiltonian}).

The generator coefficients $S_j$ are obtained by cancelling the off-diagonal block of the BCH expansion, order-by-order~\cite{bravyi2011schrieffer}.
That is, by inserting $S(t) = \sum_{j\geq 1} t^j S_j$ into 
\begin{equation}\label{eq:bch-expansion}
    \e^{S(t)} (\Delta H_0 + t V) \e^{-S(t)} = \sum_{k \geq 0} \frac{1}{k!} \L^k_{S(t)}(\Delta H_0 + t V),
\end{equation}
and setting the block-off-diagonal part to zero at each order of $t$.

\subsubsection{Schrieffer--Wolff generator coefficients}
We collect the first three generator coefficients explicitly in the following proposition.
Note that we re-introduce the parameter $s$ in this result for consistency, but continue to temporarily suppress it in the subsequent discussions.

\begin{restatable}{proposition}{swcoefficients}\label{prop:sw-coefficients}
    The first three Schrieffer--Wolff generator coefficients are given by
    \begin{align}
        S_1(s) &= \L_{\Delta H_0}^{\pattce} \big( \big( V(s) \big)_{\tn{od}} \big), \label{eq:sw1} \\
        S_2(s) &= \L_{\Delta H_0}^{\pattce} \big( \comm{S_1(s)}{\big( V(s) \big)_{\tn{d}}} \big), \label{eq:sw2} \\
        S_3(s) &= \L_{\Delta H_0}^{\pattce}\big( \comm{S_2(s)}{\big( V(s) \big)_{\tn{d}}} \big) - \frac{1}{3} \L_{\Delta H_0}^{\pattce}\big( \comm{S_1(s)}{\comm{S_1(s)}{\comm{S_1(s)}{\Delta H_0}}} \big). \label{eq:sw3}
    \end{align}
\end{restatable}

A proof is given in Appendix~\ref{app:sw-transformation-proofs}.

\subsubsection{Effective Hamiltonian coefficients}
Now that we have explicit expressions for the first three generator coefficients $S_1, S_2, S_3$, we can use them to construct the corresponding effective Hamiltonian coefficients $K_1, K_2, K_3$.
We again re-introduce the parameter $s$ in the effective Hamiltonian coefficients for the following proposition.

\begin{restatable}{proposition}{heffcoefficients}\label{prop:heff-coefficients}
    Let $\Sigma = H_0^{\pattce}$.
    The first three nonzero coefficients of the effective Hamiltonian are given by
    \begin{align}
        K_1(s) &= \big( V(s) \big)_{--}, \label{eq:K1}\\
        K_2(s) &= - \Delta^{-1} \big( V(s) \big)_{-+} \Sigma \big( V(s) \big)_{+-}, \label{eq:K2}\\
        K_3(s) &= \Delta^{-2} \left( \big( V(s) \big)_{-+} \Sigma \big( V(s) \big)_{++} \Sigma \big( V(s) \big)_{+-} - \frac{1}{2} \acomm{\big( V(s) \big)_{--}}{\big( V(s) \big)_{-+} \Sigma^2 \big( V(s) \big)_{+-}} \right). \label{eq:K3}
    \end{align}
\end{restatable}

A proof is given in Appendix~\ref{app:sw-transformation-proofs}.

From \cref{prop:heff-coefficients}, it can be shown that $\norm{K_j} = O(\Delta^{1-j} \norm{V}^j)$ for $j = 1, 2, 3$.
As we will later establish, the choice of perturbation $V$ will contain terms that scale with $\Delta$ (to some fractional power).
In what follows, an ``order-$p$'' reduction uses the perturbative coefficients through $K_p$ to generate the target interaction.
The resulting error is determined after the penalty scale $\Delta$ is fixed.

\begin{remark}\label{rmk:ignoring-third-order-terms}
    In third-order gadget design, the contribution $V_{-+}\Sigma V_{++}\Sigma V_{+-}$ is used to generate the target interaction.
    The anticommutator contribution in \cref{eq:K3} is instead absorbed into the correction and error terms identified in the later reduction.
    We therefore analyse the decomposition of the first contribution separately below.
\end{remark}
    
\subsection{Uniform Control of the Exact and Truncated Components}
\label{sec:uniform-control-exact-and-truncated-components}
The coefficients derived above in \cref{prop:sw-coefficients,prop:heff-coefficients} will be used to identify the interactions produced by a perturbative gadget reduction.
To turn such reductions into simulation results we must control the difference between the truncated and exact objects, uniformly along the entire parameter space of interest.
For a fixed order $p$, we truncate the generator and effective-Hamiltonian series and bound the resulting remainders uniformly in $s$.
The following lemma supplies the required bounds on the exact and truncated components.

\begin{restatable}{lemma}{uniformcontrolexacttruncated}\label{lemma:uniform-control-exact-truncated}
    Let $H_0$ be independent of $s$ and satisfy \cref{eq:H0-block-structure}.
    Let $V \in {\rm C}^2$ be closed and suppose 
    \begin{align}
        \snorm{V} &\leq v_0, &
        \snorm{\ps V} &\leq v_1.
    \end{align}
    Whenever $32 v_0 < \Delta$, the exact low-energy band of $\Delta H_0 + V(s)$ is obtained from $\Pi_-$ by a closed ${\rm C}^2$ Schrieffer--Wolff unitary $\e^{-S(s)}$.
    For every fixed $p \geq 1$, we have:
    \begin{align}
        \snorm{S} &= O(\Delta^{-1} v_0), & \snorm{\ps S} &= O(\Delta^{-1} v_1); \label{eq:uniform-control-s}\\
        \snorm{R^S_{p+1}} &= O(\Delta^{-(p+1)} v_0^{p+1}), & \snorm{\ps R^S_{p+1}} &= O(\Delta^{-(p+1)} v_1 v_0^{p}); \label{eq:uniform-control-remainder-s}\\
        \snorm{R^{H_\tn{eff}}_{p+1}} &= O(\Delta^{-p} v_0^{p+1}), & \snorm{\ps R^{H_\tn{eff}}_{p+1}} &= O(\Delta^{-p} v_1 v_0^{p}). \label{eq:uniform-control-remainder-heff}
    \end{align}
    The constants implicit in the asymptotic notation may depend on the fixed order $p$, but are independent of $s$, $\Delta$, $v_0$, and $v_1$.
    Moreover, if $\V$ is independent of $s$ and $\U(s) = \e^{-S(s)}\V$, then $\U$ is closed and ${\rm C}^2$, with 
    \begin{equation}\label{eq:kappa-bound}
        \kappa = \snorm{\ps \U} \leq \snorm{\ps S} = O(\Delta^{-1} v_1).
    \end{equation}
\end{restatable}

A proof is given in Appendix~\ref{app:sw-transformation-proofs}.

The later reductions use \cref{lemma:uniform-control-exact-truncated} to construct exact simulation bounds from truncated conditions.
If $\snorm{\overline{H} - H_\tn{eff}^{[p]}} \leq \epsilon$, then
$\snorm{\overline{H} - H_\tn{eff}} \leq \epsilon + O(\Delta^{-p} v_0^{p+1})$.
The same lemma gives the encoding bound $\snorm{\U - \V} = O(\Delta^{-1} v_0)$.

In a number of cases we consider, the first Schrieffer--Wolff coefficient $S_1$ is static, i.e., independent of $s$.
The following corollary addresses this case.

\begin{restatable}[Static leading Schrieffer--Wolff generator]{corollary}{staticleadingswgenerator}\label{cor:static-leading-sw-generator}
    Under the hypotheses of
    \cref{lemma:uniform-control-exact-truncated}, suppose additionally that the block-off-diagonal part of the perturbation is independent of $s$, i.e., $\ps \big( V(s) \big)_{\tn{od}} = 0$.
    Then the first Schrieffer--Wolff coefficient $S_1$ is independent of $s$. 
    Consequently,
    \begin{align}
        \snorm{S} &= O(\Delta^{-1}v_0), &
        \snorm{\ps S} &= O(\Delta^{-2}v_0v_1).
    \end{align}
    If $\U(s) = \e^{-S(s)}\V$, where $\V$ is independent of $s$, then its derivative norm satisfies
    \begin{equation}
        \kappa = \snorm{\ps \U} = O(\Delta^{-2}v_0v_1).
    \end{equation}
\end{restatable}

A proof is provided in Appendix~\ref{app:sw-transformation-proofs}.

\subsection{Tensor-Product Penalties and Local Perturbations}
\label{sec:tensor-product-penalties}
We now extend the preceding analysis to a sum of single-qubit penalties coupled to a sum of local perturbations.
This setting models the parallel application of perturbative gadgets and requires exact \emph{global} blocks defined by the joint low- and high-energy splitting.
We first derive the global coefficients and then state the additional mediator assumptions under which the effective-Hamiltonian contributions decompose into single-gadget terms.

\subsubsection{The parallel penalty Hamiltonian}
We consider a penalty Hamiltonian that is a sum of pairwise-disjoint single-qubit penalty terms
\begin{equation}\label{eq:parallel-penalty-hamiltonian}
    \bs{H}_0 = \sum_{a\in[A]} H_0^a.
\end{equation}
For each $a\in[A]$, let $p_a$ be the ground-space projector of $H_0^a$ and let $q_a = I-p_a$. 
We assume for all $a\in[A]$ that the following conditions hold:
\begin{align}\label{eq:single-qubit-penalty-conditions}
    p_a H_0^a p_a &= 0, & 
    q_a H_0^a q_a &\succeq q_a, & 
    {\sf S}(H_0^a)\cap{\sf S}(H_0^{a'}) &=\emptyset \quad \text{for } a\neq a' .
\end{align}
Each $H_0^a$ acts on a single penalty qubit $m_a$ and there is some $e_a \geq 1$ such that $H_0^a = e_a q_a$. 
Note that we use the terminology ``mediator qubits'' interchangeably with ``penalty qubits'' throughout this work.
For simplicity, we take the unit penalty $H_0^a = q_a$.
We also record that the inverse of each penalty term restricted to its excited subspace is $\Sigma^a = (H_0^a)^{\pattce} = q_a$.

The full low- and high-energy projectors of $\bs{H}_0$ are
\begin{align}
    \bs{\Pi}_- &= \bigotimes_{a\in[A]}p_a, & 
    \bs{\Pi}_+ &= I-\bs{\Pi}_-.
\end{align}
We refer to $\bs{\Pi}_{\pm}$ as the \emph{global} projectors.
Here and throughout, we suppress identities on the unpenalised system qubits whenever they are clear from context.
The spectral decomposition of $\bs{H}_0$ is 
\begin{equation}\label{eq:H0_spectral_decomposition}
    \bs{H}_0 = \sum_{x \in \{0,1\}^{A} \setminus 0^A} \hw{x} \ketbra{x},
\end{equation}
where $\hw{x}$ denotes the Hamming weight of the bit string $x \in \{0,1\}^{A}$ and $0^A$ is the all-zero string.

There exists a simple correspondence between subsets of $[A]$ and bit strings in $\{0,1\}^{A}$ given by $X \subseteq [A] \mapsto x \in \{0,1\}^{A}$ with $x_a = 1$ if and only if $a \in X$.
Observe that $\{ \{a\} : a \in [A] \} \cong \{ x : \hw{x} = 1 \}$.\footnote{We use the notation $\cong$ to denote the correspondence between subsets of $[A]$ and bit strings in $\{0,1\}^{A}$.}
As a short example, take $A = 2$, then, $\emptyset \cong 00$, $\{1\} \cong 10$, $\{2\} \cong 01$, and $\{1,2\} \cong 11$.

For $X\subseteq[A]$, define the excitation-sector projector and its excitation energy by
\begin{align}
    Q_X &\coloneqq \bigotimes_{a\in X}q_a\bigotimes_{a\notin X}p_a, & 
    E_X &\coloneqq |X| \cong \hw{x}.
\end{align}
The projector $Q_X$ selects the subspace in which precisely the penalty qubits indexed by $X$ are excited. 
In particular, 
\begin{align}
    \bs{\Pi}_- &\eqqcolon Q_\emptyset, &
    \bs{\Pi}_+ &\eqqcolon \sum_{\emptyset\neq X\subseteq[A]}Q_X.
\end{align}

We now prove some basic properties of the excitation-sector decomposition and its relation to the penalty Hamiltonian.

\begin{restatable}{proposition}{excitationsectordecomposition}
\label{clm:excitation-sector-decomposition}
    For every $X\subseteq[A]$,
    \begin{equation}\label{eq:H0-sector-eigenvalue}
        \bs{H}_0 Q_X = \abs{X}\, Q_X .
    \end{equation}
    Consequently, the spectral gap of $\bs{H}_0$ above ${\rm Im}(\bs{\Pi}_-)$ is one, and the
    reduced inverse
    \begin{equation}
        \bs{\Sigma} = \bs{\Pi}_+\left(\bs{H}_0\big|_{{\rm Im}(\bs{\Pi}_+)}\right)^{-1}\bs{\Pi}_+
                    = \bs{H}_0^{\pattce}
    \end{equation}
    has the excitation-sector decomposition
    \begin{equation}
        \bs{\Sigma} = \sum_{\emptyset\neq X\subseteq[A]}\bs{\Sigma}_X ,
    \end{equation}
    where $\bs{\Sigma}_X = \abs{X}^{-1}Q_X$.
    Moreover, $\norm{\bs{\Sigma}_X} = \abs{X}^{-1}$ and $\norm{\bs{\Sigma}} = 1$.
\end{restatable}

A proof is given in Appendix~\ref{app:sw-transformation-proofs}.

Notice when $X=\{a\}$ we are in the singleton sector in which only the penalty qubit $m_a$ is excited, then
\begin{equation}\label{eq:singleton-sigma}
    \bs{\Sigma}_{\{a\}} = Q_{\{a\}} = Q_{\{a\}}\,\Sigma^a\, Q_{\{a\}} = Q_{\{a\}}\,(H_0^a)^{\pattce}\,Q_{\{a\}}.
\end{equation}

\subsubsection{The parallel simulator Hamiltonian family}
The perturbed Hamiltonian family is now defined to have the form
\begin{equation}
    \wt{\bs{H}}(s) = \Delta\bs{H}_0 + \bs{V}(s),
\end{equation}
where $\bs{V}(s) = \sum_{a\in[A]}V^a(s)$ and each $V^a(s)$ is a sum of local perturbations acting on a constant-size subset of system qubits together with the penalty qubit $m_a$.
Perturbative constructions may also contain system-only terms, but the decomposition below concerns only terms involving a single penalty qubit.

The terms $V^a(s)$ are not assumed to have disjoint supports globally, but restricted to the mediator qubits they are disjoint.
More concretely, let ${\sf S}\bigr|_{{\rm med}}(X) \coloneqq {\sf S}(X) \cap {\rm med}$ denote the restriction of the support of $X$ to the mediator qubits, then 
\begin{equation}\label{eq:med-support-disjoint}
    {\sf S}\bigr|_{{\rm med}} \big( V^a(s) \big) \cap {\sf S}\bigr|_{{\rm med}} \big( V^{b}(s) \big) = \emptyset
\end{equation}
for $a \neq b$.
We remark that each $V^a$ can be of the form
\begin{equation}
    V^a(s) = \sum_{\ell=1}^{{\sf d}_a} V^{a,\ell}(s),
\end{equation}
where each $V^{a,\ell}(s)$ acts on a constant-size subset of the system qubits $R_{a,\ell}$ and the penalty qubit $m_a$.

The supports of the terms $V^a$ define an interaction hypergraph. 
This support-overlap structure will determine which mixed perturbative contributions can occur at each order in the effective Hamiltonian coefficients.

\subsubsection{The global Schrieffer--Wolff transformation}\label{sec:global-sw-transformation}
We now apply the Schrieffer--Wolff construction outlined above with respect to the global splitting $\bs{\Pi}_-\oplus\bs{\Pi}_+$.
For any operator $\bs{O}$, define $\bs{O}_{\alpha\beta} = \bs{\Pi}_\alpha\bs{O}\bs{\Pi}_\beta$ for $\alpha,\beta\in\{-,+\}$.
The global block-off-diagonal and block-diagonal components are
\begin{align}
    \bs{O}_{\tn{od}} &= \bs{O}_{-+}+\bs{O}_{+-}, & 
    \bs{O}_{\tn{d}} &= \bs{O}_{--}+\bs{O}_{++},
\end{align}
and we define the corresponding operator spaces by 
\begin{align}
    \bs{\D}_{\tn{od}} &= \{\bs{O} : \bs{O}_{--}=\bs{O}_{++}=0\}, &
    \bs{\D}_{\tn{d}} &= \{\bs{O} : \bs{O}_{-+}=\bs{O}_{+-}=0\}.
\end{align}

We can extend \cref{lem:L-inverse} to the case of the parallel penalty Hamiltonian via the following corollary.

\begin{restatable}{corollary}{linverseparallel}\label{cor:L-inverse-parallel}
    Let $\Delta>0$ and let $\bs{H}_0$ be the parallel penalty Hamiltonian of \cref{eq:parallel-penalty-hamiltonian}.
    Then $\L_{\Delta\bs{H}_0}$ restricted to $\bs{\D}_{\tn{od}}$ is a linear bijection of $\bs{\D}_{\tn{od}}$ onto itself, whose inverse is the restriction of the pseudoinverse
    \cref{eq:liouville_pseudoinverse_def} to $\bs{\D}_{\tn{od}}$:
    \begin{align}\label{eq:L-inverse-parallel}
        \L_{\Delta\bs{H}_0}^{\pattce}(\bs{X}) &= \Delta^{-1}\big( \bs{\Sigma}\,\bs{X}_{+-} - \bs{X}_{-+}\,\bs{\Sigma} \big), &
        \bs{X} &\in \bs{\D}_{\tn{od}}.
    \end{align}
    Consequently, for every $\bs{X}\in \bs{\D}_{\tn{od}}$ the equation $\L_{\Delta\bs{H}_0}(\bs{S})=\bs{X}$ has a unique solution $\bs{S}\in \bs{\D}_{\tn{od}}$, namely $\bs{S}=\L_{\Delta\bs{H}_0}^{\pattce}(\bs{X})$.
\end{restatable}

A proof is provided in Appendix~\ref{app:sw-transformation-proofs}.

By \cref{lem:liouville-pseudoinverse-bound}, $\norm{\L_{\Delta\bs{H}_0}^{\pattce}(\bs{X})}\leq\Delta^{-1}\norm{\bs{X}}$ for every $\bs{X}\in \bs{\D}_{\tn{od}}$.

All blocks in this subsection are defined with respect to the global splitting.
A local term may be block-off-diagonal relative to one splitting $p_a\oplus q_a$ while having a nonzero global $(++)$ block. 
For example, on a state in which another penalty qubit is already excited, the same term may map one vector in ${\rm Im}(\bs{\Pi}_+)$ to another vector in ${\rm Im}(\bs{\Pi}_+)$.
Consequently, the global Schrieffer--Wolff generator is not generally obtained by assigning an isolated generator to each local perturbation and summing the results.

Assume that $2\snorm{\bs{V}}<\Delta$, and let $\bs{S}(s)$ be the exact canonical Schrieffer--Wolff generator of $\wt{\bs{H}}(s)$.
If $\bs{\V}$ is the fixed isometry with image ${\rm Im}(\bs{\Pi}_-)$, then
\begin{equation}
    \bs{\U}(s) = \e^{-\bs{S}(s)}\bs{\V}
\end{equation}
has image equal to the exact low-energy band of $\wt{\bs{H}}(s)$. 
This exact global isometry is the witnessing family used later in the Hamiltonian and Berry-phase simulation results.

The corresponding exact effective Hamiltonian is
\begin{equation}
    \bs{H}_{\tn{eff}}(s) = \bs{\Pi}_-\e^{\bs{S}(s)}\wt{\bs{H}}(s)\e^{-\bs{S}(s)}\bs{\Pi}_- .
\end{equation}
Both $\bs{S}(s)$ and $\bs{H}_{\tn{eff}}(s)$ are therefore defined globally with respect to the joint splitting $\bs{\Pi}_-\oplus\bs{\Pi}_+$.

\subsubsection{Coefficients of the exact global components}
To derive the perturbative coefficients of the exact global generator and effective Hamiltonian, we again introduce the formal parameter $t$ and follow \cref{sec:perturbative-coefficients-and-truncation}.
The derivation uses \cref{cor:L-inverse-parallel} in place of \cref{lem:L-inverse}.

We first record the global generator coefficients.

\begin{restatable}{corollary}{swgeneratorsparallel}\label{cor:sw-generators-parallel}
    The first three Schrieffer--Wolff generator coefficients of $\wt{\bs{H}}(s)$ with respect to the
    global splitting $\bs{\Pi}_-\oplus\bs{\Pi}_+$ are
    \begin{align}
        \bs{S}_1(s) &= \L_{\Delta\bs{H}_0}^{\pattce}\big( \big(\bs{V}(s)\big)_{\tn{od}} \big), \label{eq:S1-parallel}\\
        \bs{S}_2(s) &= \L_{\Delta\bs{H}_0}^{\pattce}\big( \comm{\bs{S}_1(s)}{\big(\bs{V}(s)\big)_{\tn{d}}} \big), \label{eq:S2-parallel}\\
        \bs{S}_3(s) &= \L_{\Delta\bs{H}_0}^{\pattce}\Big( \comm{\bs{S}_2(s)}{\big(\bs{V}(s)\big)_{\tn{d}}} - \tfrac{1}{3}\comm{\bs{S}_1(s)}{\comm{\bs{S}_1(s)}{\comm{\bs{S}_1(s)}{\Delta\bs{H}_0}}} \Big). \label{eq:S3-parallel}
    \end{align}
\end{restatable}

A proof is provided in Appendix~\ref{app:sw-transformation-proofs}.

Next, we turn to the perturbative coefficients of the exact global effective Hamiltonian.

\begin{restatable}{corollary}{heffcoefficientsparallel}\label{cor:heff-coefficients-parallel}
    Let $\bs{\Sigma} = \bs{H}_0^{\pattce}$ be as in \cref{clm:excitation-sector-decomposition}.
    The first three nonzero coefficients of the effective Hamiltonian $\bs{H}_{\tn{eff}}(s)$,
    taken with respect to the global splitting $\bs{\Pi}_-\oplus\bs{\Pi}_+$, are
    \begin{align}
        \bs{K}_1(s) &= \big( \bs{V}(s) \big)_{--}, \label{eq:K1-parallel}\\
        \bs{K}_2(s) &= -\Delta^{-1}\,\big( \bs{V}(s) \big)_{-+}\bs{\Sigma}\big( \bs{V}(s) \big)_{+-}, \label{eq:K2-parallel}\\
        \bs{K}_3(s) &= \Delta^{-2}\left( \big( \bs{V}(s) \big)_{-+}\bs{\Sigma}\big( \bs{V}(s) \big)_{++}\bs{\Sigma}\big( \bs{V}(s) \big)_{+-} - \tfrac{1}{2}\acomm{\big( \bs{V}(s) \big)_{--}}{\big( \bs{V}(s) \big)_{-+}\bs{\Sigma}^2\big( \bs{V}(s) \big)_{+-}} \right). \label{eq:K3-parallel}
    \end{align}
\end{restatable}

A proof is provided in Appendix~\ref{app:sw-transformation-proofs}.
Note that we refer to the first group of terms in $\bs{K}_3$ as ``$\bs{K}_3^{(1)}$'' and additionally follow the comments in \cref{rmk:ignoring-third-order-terms}.

Although the parallel formulas mirror the single-block formulas, the global splitting can introduce mixed contributions between local perturbations.
The decomposition below identifies conditions on the mediator couplings intended to remove such cross-gadget terms from the effective-Hamiltonian coefficients.
Section~\ref{sec:perturbative-gadget-reductions} then uses this structure in the parallel gadget constructions.

\subsubsection{Decompositions of the parallel effective Hamiltonian coefficients}
\label{sec:parallel-effective-hamiltonian-decomposition}
We now discuss how the parallel effective Hamiltonian coefficients can be decomposed into contributions from individual sets of penalty qubits, recovering the standard form readily assumed in the literature~\cite{piddock2017complexity}.
Recall that we are only concerned with $\bs{K}_1(s)$ (of \cref{eq:K1-parallel}), $\bs{K}_2(s)$ (of \cref{eq:K2-parallel}) and the first group of terms in $\bs{K}_3(s)$, denoted $\bs{K}_3^{(1)}(s)$ (of \cref{eq:K3-parallel}) (cf. \cref{rmk:ignoring-third-order-terms}).

By linearity of the projection onto $\bs{\Pi}_-$, the first-order coefficient decomposes as
\begin{equation}
    \bs{K}_1(s) = \sum_{a\in[A]} \bs{\Pi}_- V^a(s) \bs{\Pi}_- .
\end{equation}
The following lemma establishes analogous decompositions for the second- and third-order terms (with an additional assumption on each local perturbation $V^a(s)$).

\begin{restatable}{lemma}{ParallelEffectiveHamiltonianDecomposition}\label{lem:parallel-effective-hamiltonian-decomposition}
    Let $\bs{K}_2(s)$ and $\bs{K}_3^{(1)}(s)$ be the second- and (first group of) third-order parallel effective Hamiltonian coefficients of \cref{eq:K2-parallel,eq:K3-parallel}, respectively. 
    Suppose that the perturbation $\bs{V}(s)$ decomposes as $\bs{V}(s) = \sum_{a\in[A]} V^a(s)$ and for each $a\in[A]$, there exists a constant ${\sf d}_a$ such that
    \begin{equation}
        V^a(s) \coloneqq \sum_{\ell \in [{\sf d}_a]} V^{a,\ell}(s),
    \end{equation}
    where each $V^{a,\ell}(s)$ acts on penalty qubit $m_a$ and a constant number of system qubits.
    Take $\big( V^a(s) \big)_{\alpha\beta} \coloneqq \bs{\Pi}_\alpha V^a(s) \bs{\Pi}_\beta$ for $\alpha, \beta \in \{-,+\}$ and define the operators
    \begin{align}
        \begin{split}\label{eq:Omega-a}
            \Omega^a(s) &\coloneqq \big( V^a(s) \big)_{-+} (H_0^a)^{\pattce} \big( V^a(s) \big)_{+-} \\
            &= \sum_{\ell, \ell' \in [{\sf d}_a]} \big( V^{a,\ell}(s) \big)_{-+} (H_0^a)^{\pattce} \big( V^{a,\ell'}(s) \big)_{+-}, 
        \end{split} \\[0.3em]
        \begin{split}\label{eq:Psi-a}
            \Psi^a(s) &\coloneqq \big( V^a(s) \big)_{-+} (H_0^a)^{\pattce} \big( V^a(s) \big)_{++} (H_0^a)^{\pattce} \big( V^a(s) \big)_{+-} \\
            &= \sum_{\ell, \ell', \ell'' \in [{\sf d}_a]} \big( V^{a,\ell}(s) \big)_{-+} (H_0^a)^{\pattce} \big( V^{a,\ell'}(s) \big)_{++} (H_0^a)^{\pattce} \big( V^{a,\ell''}(s) \big)_{+-}.
        \end{split}
    \end{align}
    Then, 
    \begin{equation}\label{eq:K2-decomposition}
        \bs{K}_2(s) = -\Delta^{-1}\sum_{a \in [A]} \Omega^a(s).
    \end{equation}
    If, additionally, $\bs{\Pi}_-V^a(s)\bs{\Pi}_-=0$ for every $a\in[A]$ and $s\in[0,1]$, then
    \begin{equation}    \label{eq:K3-decomposition}
        \bs{K}_3^{(1)}(s) = \Delta^{-2}\sum_{a \in [A]} \Psi^a(s). 
    \end{equation}
    In particular, there are no cross-terms between contributions acting on different penalty qubits.
\end{restatable}

A proof is given in Appendix~\ref{app:sw-transformation-proofs}.

The operator $\Omega^a$ collects the second-order contribution to $\bs{K}_2$ associated with penalty qubit $m_a$.
It still contains mixed terms between local components $\ell \neq \ell'$ of the same gadget; \cref{fig:omega-example} illustrates the interactions these terms may generate.
Similarly, $\Psi^a$ collects the corresponding third-order contribution to $\bs{K}_3^{(1)}$ under the same assumption that $\bs{\Pi}_-V^a(s)\bs{\Pi}_- = 0$ for every $a\in[A]$ and $s\in[0,1]$.
Note that this assumption is reflected in the constructed third-order simulator Hamiltonians in Section~\ref{sec:third-order-reductions}.

For simplicity, assume the locality of each $V^{a,\ell}$ is $2$, and that the degree of $m_a$ is bounded by some constant $c$.
The largest possible local term that can emerge in $\Omega^a$ is of locality $c+1$ (projected onto the low-energy subspace this is locality $c$).
However, in all applications of perturbative gadgets we consider, these cross-terms are either intentional or the unwanted components can be trivially cancelled.
That is, we choose the perturbations to couple to specific mediator qubits to create the desired effective interactions rather than there being a case where there is an accumulation of unwanted, higher-locality terms.
This is outlined in Ref.~\cite{oliveira2008complexity} whose gadget constructions we extend in this work.
In Section~\ref{sec:perturbative-gadget-reductions} we discuss in more detail how the components at a given order in the perturbative expansion behave.

\begin{figure}[!ht]
    \centering
    \begin{tikzpicture}
        \pic{gadgetrealisation};
    \end{tikzpicture}
    \caption{
        A schematic illustration of a gadget and the interactions generated through $\Omega^a$, using the \textsl{fork} gadget as an example.
        The thick lines on the right denote desired interactions, the dashed lines denote unwanted interactions, and the loops denote self-interactions.
    }
    \label{fig:omega-example}
\end{figure}

%% file: sections/berry-phase-simulation.tex
\section{Berry Phase Simulation}
\label{sec:berry-phase-simulation}
In this section, we establish a Berry-phase analogue of the eigenvalue- and ground-state-simulation results extended to parameterised families in Section~\ref{sec:uniform-simulation}; see \cref{lem:eigenvalue-simulation,lem:ground-state-simulation} and Ref.~\cite{bravyi2016complexity}.
Those results compare spectral data pointwise in the parameter.
The Berry phase, in contrast, depends on the full path traced by the ground state.

More precisely, if two Hamiltonians are close at each parameter value and the relevant spectral gaps remain (sufficiently) large, then their eigenvalues and ground spaces are close pointwise. 
Pointwise control alone does not imply Berry-phase control as two families may remain uniformly close in operator norm while accumulating different geometric phases if their difference varies sufficiently rapidly along the path. 
We therefore require control of both the simulation error and its variation with the parameter (see Definitions~\ref{def:uniform-simulation} and \ref{def:regular-simulation}). 
This is reflected in the ${\rm C}^2$ regularity conditions of Assumption~\ref{assump:berry} and in the appearance of both $\varepsilon$ and $\snorm{\ps E}$ in the final bound. 
These assumptions are incorporated into the promises of valid instances of \sc{GSBPE}; accordingly, our result concerns the promised class of regular, uniformly gapped families rather than arbitrary Hamiltonian paths.

The proof proceeds by expressing the Berry phase as a holonomy and comparing the resulting ground-state projector paths. 
We first bound the spectral projectors of Hamiltonians that are close in operator norm and then use Kato's parallel-transport construction to identify the Berry phase with the holonomy of the ground-state path. 
The comparison is carried out through a family $\widehat{\Upsilon}$ we introduce in \cref{eq:transformed-family}, which acts on the target space and can therefore be compared directly with $\Upsilon$.
This separates the argument into two steps: 
\begin{inparaenum}[(i)]
    \item the comparison between $\vartheta$ and $\widehat{\vartheta}$, controlled by the simulation error and its variation, and 
    \item the comparison between $\widehat{\vartheta}$ and $\wt{\vartheta}$, controlled by the parameter dependence of the witnessing isometry.
\end{inparaenum}
These estimates are combined to yield the final Berry phase simulation result.

\subsection{Setup and Assumptions}
\label{sec:setup-and-assumptions}
Throughout this section, $\Upsilon = \{H(s)\}_{s\in[0,1]}$ is a target family and $(\wt{\Upsilon}, \V)$ is an $(\eta, \varepsilon, \kappa)$-regular simulation of $\Upsilon$ (Definitions~\ref{def:uniform-simulation} and \ref{def:regular-simulation}), witnessed by the isometry family $\Phi = \{\U(s)\}_{s\in[0,1]}$ with $\snorm{\ps\U}\leq\kappa$.

We pull the simulator family back onto the target space and define the \emph{transformed family} $\widehat{\Upsilon} = \{\widehat{H}(s)\}_{s\in[0,1]}$ by
\begin{equation}\label{eq:transformed-family}
    \widehat{H}(s) \coloneqq \U(s)^\dagger\, \wt{H}(s)\, \U(s).
\end{equation}

Recall that the \emph{simulation error operator} $E(s) \coloneqq H(s) - \widehat{H}(s)$ satisfies $\norm{E(s)} \le \varepsilon$ by Condition~\ref{item:c2}.
Since the image of $\U(s)$ is $\mathcal{E}_{2^n}(\wt{H}(s))$, the operator $\widehat{H}(s)$ is unitarily equivalent to the restriction of $\wt{H}(s)$ to its low-energy band; in particular, writing $\ket{\wt{\phi}_0(s)}$ for the ground state of $\wt{H}(s)$, the state
\begin{equation}\label{eq:transformed-ground-state}
    \ket{\psi_0(s)} \coloneqq \U(s)^\dagger\ket{\wt{\phi}_0(s)}
\end{equation}
is the ground state of $\widehat{H}(s)$, with ground energy $\mu_0(s)$ equal to that of $\wt{H}(s)$.
We write $P(s) = \ketbra{\phi_0(s)}$ and $\widehat{P}(s) = \ketbra{\psi_0(s)}$ for the ground-space projectors of $H(s)$ and $\widehat{H}(s)$.

Before proceeding with the analysis, we state some assumptions on the regularity, accuracy, and closure of the families involved, motivated by the valid instances of the \sc{GSBPE} problem.

\begin{assumption}\label{assump:berry}
    The following assumptions are made throughout this section.
    \begin{enumerate}[label=\textup{(A\arabic*)}]
           \item \emph{Regularity:} $H,\wt{H}\in{\rm C}^2$ and $\U$ is chosen ${\rm C}^2$; hence $\widehat H,E,P,\widehat P\in{\rm C}^2$.
            We additionally assume that $\snorm{\ps H}$, $\snorm{\wt H}$, $\snorm{\ps\wt H}$, $\snorm{\ps\widehat H}$ are bounded above by polynomials in $n$ and that $\snorm{\ps E}$ and $\snorm{\ps\U}$ are bounded above by polynomials in $1/n$. \label{item:a1}
        \item \emph{Accuracy:} $\varepsilon < \gamma_\star/4$. \label{item:a2}
        \item \emph{Closure:} $H(1) = H(0)$, $\wt{H}(1) = \wt{H}(0)$, $\U(1) = \U(0)$; hence $P, \widehat{P}$ are closed. \label{item:a3}
    \end{enumerate}
\end{assumption}

We also remark that Assumption~\ref{item:a2} and Weyl's inequality imply that the gap of $\widehat{H}(s)$ is at least $\gamma_\star - 2\varepsilon > \gamma_\star/2$, ensuring that the ground-space projectors $P$ and $\widehat{P}$ are well-defined throughout.

All quantities are dimensionless except energies:
$\unit{\varepsilon}, \unit{\gamma_\star}, \unit{\snorm{H}}, \unit{\snorm{\ps H}} = e$ and 
$\unit{s}, \unit{\eta}, \unit{\kappa}, \unit{P} = 1$.
The bounds below are dimensionally consistent ratios of these quantities.
We omit the routine dimensional checks (though it can be instructive to verify them).

\subsection{Comparing the Ground-Space Projectors}
\label{sec:comparing-ground-space-projectors}

We now study how close the spectral projectors $P(s)$ and $\widehat{P}(s)$ are to each other along the path.
Importantly, we must obtain bounds for both the difference of the projectors themselves and the difference of their derivatives with respect to the path parameter $s$.
To do so, we relate the projectors to the resolvents of the Hamiltonians via the Riesz formula.

For each $s\in[0,1]$, let $\mathscr{C}(s)$ be the positively oriented circle of radius $\gamma_\star/2$ centred at $\lambda_0(s)$, so that $\abs{\mathscr{C}(s)} = \pi\gamma_\star$, and write
\begin{align}\label{eq:resolvents}
    R(z,s) &\coloneqq \big(z - H(s)\big)^{-1}, &
    \widehat{R}(z,s) &\coloneqq \big(z - \widehat{H}(s)\big)^{-1},
\end{align}
for the resolvents of $H(s)$ and $\widehat{H}(s)$, respectively.
The choice of radius is a consequence of Assumption~\ref{item:a2}: with $\varepsilon < \gamma_\star/4$, the same contour encircles the ground energy of both $H(s)$ and $\widehat{H}(s)$ and no other eigenvalue of either, uniformly in $s$.

The projectors defined in Section~\ref{sec:setup-and-assumptions} are then recovered as
\begin{align}\label{eq:riesz}
    P(s) &= \frac{1}{2\pi\i}\oint_{\mathscr{C}(s)} \dd{z}\, R(z,s), &
    \widehat{P}(s) &= \frac{1}{2\pi\i}\oint_{\mathscr{C}(s)} \dd{z}\, \widehat{R}(z,s).
\end{align}
See \cref{rem:riesz-equivalence} for the equivalence of this definition with the standard spectral projector definition.

We now consider the distance between the spectral projectors $P(s)$ and $\widehat{P}(s)$ across the entire path $s\in[0,1]$.

\begin{restatable}[Projector difference]{proposition}{spectralprojectorbound}\label{prop:spectral-projector-bound}
    For all $s\in[0,1]$, $\ \|P(s) - \widehat{P}(s)\| = O(\gamma_\star^{-1}\varepsilon)$.
\end{restatable}

A proof is given in Appendix~\ref{app:berry-phase-simulation-proofs}.

In addition to the difference of the projectors themselves, we also consider the difference of their derivatives with respect to the path parameter $s$.

\begin{restatable}[Difference of projector derivatives]{proposition}{derivativeProjectorDifference}\label{prop:derivative-projector-difference}
    For all $s\in[0,1]$,
    \begin{equation}\label{eq:derivative-projector-difference}
        \|\partial_s P(s) - \partial_s \widehat{P}(s)\| = O\big(\gamma_\star^{-2}\max\{\norm{\ps H(s)}, \norm{\ps\widehat{H}(s)}\}\,\varepsilon + \gamma_\star^{-1} \norm{\ps E(s)}\big).
    \end{equation}
\end{restatable}

A proof is given in Appendix~\ref{app:berry-phase-simulation-proofs}.

The later holonomy comparison uses both projector bounds above.
In particular, \cref{eq:derivative-projector-difference} depends on $\norm{\ps E(s)}$, which need not be small just because $\varepsilon$ is small.
Section~\ref{sec:perturbative-gadget-reductions} therefore establishes separate control of $\ps E(s)$ for the gadget families used in the hardness reduction.

\subsection{The Berry Phase as a Holonomy}
\label{sec:berry-phase-holonomy}
We recall Kato's construction of adiabatic parallel transport~\cite{kato1950adiabatic}.
Define the \emph{transport generator} ${\rm G}(s) \coloneqq \comm{\partial_s P(s)}{P(s)}$.
Since $P(s)$ and $\partial_s P(s)$ are Hermitian, ${\rm G}(s)$ is anti-Hermitian; it depends only on the projector family and is therefore gauge invariant.
By Assumption~\ref{item:a1}, ${\rm G}(s)$ is continuous, so the initial value problem
\begin{equation}\label{eq:kato-ivp}
    \ps T(s) = {\rm G}(s)\, T(s), \qquad T(0) = I,
\end{equation}
has a unique solution, and anti-Hermiticity of ${\rm G}(s)$ makes $T(s)$ unitary for every $s$.

\begin{restatable}[Intertwining property~\cite{kato1950adiabatic}]{proposition}{katoIntertwining}\label{prop:kato-intertwining}
    For all $s\in[0,1]$, $\ P(s) = T(s)\, P(0)\, T(s)^\dagger$.
\end{restatable}

By Assumption~\ref{item:a3}, the path is closed, so \cref{prop:kato-intertwining} gives $T(1)P(0)T(1)^\dagger = P(1) = P(0)$.
Thus, $T(1)$ commutes with the rank-one projector $P(0)$ and acts unitarily on the line $\mathbb{C}\ket{\phi_0(0)}$ by a phase:
\begin{align}\label{eq:holonomy-phase}
    T(1)\ket{\phi_0(0)} &= \e^{\i\alpha}\ket{\phi_0(0)}, & \Tr[P(0)\, T(1)] &= \mel{\phi_0(0)}{T(1)}{\phi_0(0)} = \e^{\i\alpha}.
\end{align}

Define the transported state $\ket{\xi(s)} \coloneqq T(s)\ket{\phi_0(0)}$, a ${\rm C}^1$ unit-norm ground eigenvector of $H(s)$ along the path with $\ket{\xi(0)} = \ket{\phi_0(0)}$.
It follows that for all $s\in[0,1]$, \cref{prop:kato-intertwining} gives
\begin{equation}\label{eq:xi-in-image}
    P(s)\ket{\xi(s)} = P(s)T(s)\ket{\phi_0(0)} = T(s)P(0)\ket{\phi_0(0)} = T(s)\ket{\phi_0(0)} = \ket{\xi(s)}.
\end{equation}
Thus, $\ket{\xi(s)}$ lies in the image of $P(s)$.
Moreover, for all $s\in[0,1]$, using \cref{eq:kato-ivp,eq:xi-in-image}, we have
\begin{equation}\label{eq:parallel-transport}
    \mel{\xi}{{\rm G}}{\xi} = \mel{\xi}{(\partial_s P)P}{\xi} - \mel{\xi}{P(\partial_s P)}{\xi} = \mel{\xi}{\partial_s P}{\xi} - \mel{\xi}{\partial_s P}{\xi} = 0,
\end{equation}
and therefore $\ket{\xi(s)}$ satisfies the parallel transport condition along the path.

A \emph{${\rm C}^1$ gauge} for $\Upsilon$ is a ${\rm C}^1$ family $\ket{\phi_0(s)}$ of normalised ground eigenvectors of $H(s)$; one exists by Assumption~\ref{item:a1} and $\gamma_\star > 0$ (for instance $\ket{\xi(s)}$ above).
The \emph{Berry connection} of the gauge is
\begin{equation}\label{eq:berry-connection}
    A(s) \coloneqq \i\mel{\phi_0(s)}{\partial_s}{\phi_0(s)},
\end{equation}
and we set $\varTheta \coloneqq \int_0^1 A(s)\,\dd{s} \in \mathbb{R}$.
Differentiating $\braket{\phi_0}{\phi_0} = 1$ shows $\mel{\phi_0}{\partial_s}{\phi_0}$ is purely imaginary, so $A(s)$ and $\varTheta$ are real.
The connection integral depends on a general gauge; the gauge-invariant phase is
$\varTheta+\arg\braket{\phi_0(0)}{\phi_0(1)}$ modulo $2\pi$.
In a periodic gauge the endpoint term vanishes and $\vartheta\equiv\varTheta\pmod{2\pi}$.
The following proposition identifies this gauge-independent quantity with $\Tr[P(0)T(1)]$.

\begin{restatable}[Berry phase as holonomy]{lemma}{berryphaseholonomy}\label{prop:berry-holonomy}
    Let $P(s) = \ketbra{\phi_0(s)}$ be a ${\rm C}^1$ family of rank-one spectral projectors over a closed path $P(1) = P(0)$, let ${\rm G}(s), T(s)$ be as in \cref{eq:kato-ivp}, and let $\ket{\phi_0(s)}$ be any ${\rm C}^1$ gauge.
    Then
    \begin{equation}\label{eq:berry-holonomy-identity}
        \Tr[P(0)\, T(1)] = \exp\Big(\i \int_0^1 \dd{s}\, \i\mel{\phi_0(s)}{\partial_s}{\phi_0(s)}\Big) \braket{\phi_0(0)}{\phi_0(1)} = \e^{\i \varTheta}\, \braket{\phi_0(0)}{\phi_0(1)}.
    \end{equation}
\end{restatable}

In a \emph{periodic gauge} $\ket{\phi_0(1)}=\ket{\phi_0(0)}$, so \cref{eq:berry-holonomy-identity} becomes
\begin{align}
    \Tr[P(0)T(1)] &= \e^{\i\varTheta}, &
    \arg\Tr[P(0)T(1)] &\equiv \varTheta \equiv \vartheta\pmod{2\pi}.
\end{align}
A periodic gauge always exists along a closed path.
For the remainder of the paper, we therefore adopt the identification $\vartheta = \arg\Tr[P(0)T(1)]$.
The same convention applies to $\widehat{\Upsilon}$ and $\wt{\Upsilon}$, which are closed by Assumption~\ref{item:a3}.

Since Berry phases are defined modulo $2\pi$, we compare them in the circular metric $d_{2\pi}(a, b) \coloneqq \min_{k\in\mathbb{Z}}\abs{a - b - 2\pi k}$.

\subsection{The Berry Phase Simulation Theorem}
\label{sec:berry-phase-simulation-theorem}
We now collect the above ideas to prove that the Berry phase of the simulator family can be made arbitrarily close to that of the target family under suitable simulation conditions.
Let ${\rm G} = \comm{\partial_s P}{P}$, $\widehat{{\rm G}} = \comm{\partial_s\widehat{P}}{\widehat{P}}$ be the transport generators of $\Upsilon, \widehat{\Upsilon}$, with transport unitaries $T, \widehat{T}$ solving $\partial_s T = {\rm G}T$, $\partial_s\widehat{T} = \widehat{{\rm G}}\widehat{T}$, $T(0) = \widehat{T}(0) = I$.
Since both generators are anti-Hermitian, the transport unitaries they generate differ by at most the integral of the generator difference; via \cref{prop:duhamel}, $\norm{T(s) - \widehat{T}(s)}$ inherits the bound of \cref{lem:generator-comparison}.

We conclude the following lemma bounding the difference between the Berry phases of the target family $\vartheta$ and the transformed family $\widehat{\vartheta}$.

\begin{restatable}{lemma}{berryPhaseComparison}\label{lem:berry-phase-comparison}
    With $\vartheta, \widehat{\vartheta}$ the Berry phases of $\Upsilon, \widehat{\Upsilon}$,
    \begin{equation}
        d_{2\pi}(\vartheta, \widehat{\vartheta}) = O\big( \gamma_\star^{-2}\max\{\snorm{\ps H}, \snorm{\ps\widehat{H}}\}\,\varepsilon + \gamma_\star^{-1} \snorm{\ps E} + \gamma_\star^{-1}\varepsilon \big).
    \end{equation}
\end{restatable}

It remains to compare $\widehat{\Upsilon}$ with $\wt{\Upsilon}$.
Fix a ${\rm C}^1$ periodic gauge $\ket{\wt{\phi}_0(s)}$ for $\wt{H}(s)$; by Assumption~\ref{item:a3} the induced gauge $\ket{\psi_0(s)} = \U(s)^\dagger\ket{\wt{\phi}_0(s)}$ for $\widehat{\Upsilon}$ is ${\rm C}^1$ and periodic, with connections $\wt{A}(s) = \i\mel{\wt{\phi}_0(s)}{\partial_s}{\wt{\phi}_0(s)}$ and $\widehat{A}(s) = \i\mel{\psi_0(s)}{\partial_s}{\psi_0(s)}$ and phases $\wt{\vartheta}, \widehat{\vartheta}$.

\begin{restatable}{lemma}{pullbackVsSimulator}\label{lem:pullback-vs-simulator}
    With the conventions and assumptions detailed in Sections~\ref{sec:setup-and-assumptions} and \ref{sec:berry-phase-simulation-theorem},
    \begin{equation}\label{eq:pullback-vs-simulator}
        d_{2\pi}(\widehat\vartheta,\wt\vartheta) \le \snorm{\ps\U}.
    \end{equation}
\end{restatable}

This contribution depends only on the derivative norm $\kappa$ of the witness family, not on other simulation parameters $\varepsilon$ or $\eta$.
It may therefore remain large even when the pointwise simulation errors vanish.
The gadget analysis must consequently exhibit a witness family $\U(s)$ with controlled derivative norm; see Section~\ref{sec:uniform-simulation}.

We can now state the main result, which bounds the Berry phase of the target family in terms of the simulation parameters.

\begin{restatable}[Berry phase simulation]{theorem}{berryPhaseSimulation}\label{thm:berry-phase-simulation}
    Let $(\wt{\Upsilon}, \V)$ be an $(\eta, \varepsilon, \kappa)$-regular simulation of $\Upsilon$ witnessed by $\Phi = \{\U(s)\}_{s\in[0,1]}$ with $\kappa = \snorm{\ps\U}$, satisfying Assumptions~\ref{item:a1}, \ref{item:a2}, and \ref{item:a3}.
    Let $\vartheta$ and $\wt{\vartheta}$ be the Berry phases of $\Upsilon$ and $\wt{\Upsilon}$, respectively.
    Then
    \begin{equation}\label{eq:main-result-bp-error}
        d_{2\pi}(\vartheta, \wt{\vartheta}) 
        = O(
            \gamma_\star^{-2}\max\{\snorm{\ps H}, \snorm{\ps\widehat{H}}\}\,\varepsilon
            + \gamma_\star^{-1} \snorm{\ps E}
            + \gamma_\star^{-1}\varepsilon + \kappa 
        ).
    \end{equation}
\end{restatable}

For convenience, we introduce the following notation to refer to the individual contributions to the Berry-phase error.

\begin{definition}[Berry-phase error contributions]\label{def:J-contributions}
    We denote the four terms of \cref{eq:main-result-bp-error} by
    \begin{align}\label{eq:J-contributions}
        J_1 &= \gamma_\star^{-2}\max\{\snorm{\ps H}, \snorm{\ps\widehat{H}}\}\,\varepsilon, &
        J_2 &= \gamma_\star^{-1} \snorm{\ps E}, &
        J_3 &= \gamma_\star^{-1}\varepsilon, &
        J_4 &= \kappa.
    \end{align}
\end{definition}

The error contributions fall into two groups:
\begin{inparaenum}
    \item \emph{static:} $J_1$ and $J_3$ depend on the simulation error $\varepsilon$ at fixed target gap $\gamma_\star$;
    \item \emph{dynamic:} $J_2$ depends on the derivative of the error operator $E(s)$, while $J_4$ depends on the derivative norm $\kappa$ of the witness family.
\end{inparaenum}
To control the Berry-phase error it suffices to bound $J_1,\dots,J_4$, which depend on $\varepsilon$, $\snorm{\partial_sE}$, and $\kappa$.
For the hardness reduction, we additionally require small $\eta$ and $\varepsilon$ to maintain a sufficient guiding-state overlap.
This is carried out for the perturbative gadget constructions in Section~\ref{sec:perturbative-gadget-reductions}.

%% file: sections/perturbative-gadget-reductions.tex
\section{Perturbative Gadget Reductions}
\label{sec:perturbative-gadget-reductions}
In this section, we specialise the preceding analysis to perturbative gadget reductions between parameterised target and simulator families.
Such a reduction adjoins mediator qubits with a large energy penalty so that the low-energy subspace of $\wt H(s)$ reproduces $H(s)$ up to an error controlled by the perturbative order.

We fix the isometry
\begin{equation}\label{eq:mediator-isometry}
    \bs{\V} = I_{\rm sys} \otimes \bigotimes_{a\in[A]} \ket{0}_{m_a},
\end{equation}
which appends every mediator qubit in a fixed reference state.
The isometry $\bs{\V}$ is independent of $s$, as required of the embedding in Definition~\ref{def:uniform-simulation}, and ${\rm sys}$ is disjoint from the mediator qubits ${\rm med} = \{m_a\}_{a\in[A]}$.
We define the witness family $\Phi = \{\bs{\U}(s)\}_{s\in[0,1]}$ by $\bs{\U}(s) = \e^{-\bs{S}(s)}\bs{\V}$, where $\bs{S}(s)$ is the exact canonical Schrieffer--Wolff generator.
Our goal is to obtain a polynomially small derivative norm $\kappa = \snorm{\ps\bs{\U}}$ and thereby control the Berry-phase simulation error.
Note that the isometry in \cref{eq:mediator-isometry} is sufficient for the gadget reductions of Ref.~\cite{oliveira2008complexity}, but is not applicable to those of Ref.~\cite{schuch2009computational}; compatability with this latter group follows from a trivial modification (see Appendix~\ref{app:perturbative-gadget-extensions}).

For simplicity, we take the number of parallel gadget terms to equal the number of target terms $A$.
Note that in some cases one takes $O(\poly{A})$ simulator terms, which only inflates the polynomial factors in the final bounds and does not affect the inverse-polynomial asymptotics.
For ease of reference, we refer to Hamiltonian terms that carry an $s$-dependence as \emph{driving terms} (denoting the number of driving terms by $B$), and we refer to Hamiltonian terms that are not simulated by the gadgets as \emph{carried terms}.

The content of this section is organised into two main categories: 
\begin{inparaenum}[(1)]
    \item first-, second-, and third-order perturbative gadget reductions where we do not simulate any driving terms,
    \item second-order perturbative gadget reductions where we simulate $2$-local driving terms.
\end{inparaenum}
In each case we construct a family of Schrieffer--Wolff generators $\bs{S}(s)$ and bound the derivative scale $\kappa = \snorm{\ps\bs{\U}}$ of the resulting encoding isometry family $\bs{\U}(s) = \e^{-\bs{S}(s)}\bs{\V}$.
The bounds on $\kappa$ are then combined with the other contributions $J_1, J_2, J_3$ (cf. Definition~\ref{def:J-contributions} or \cref{eq:J-contributions}) to give the final Berry-phase simulation error. 
For the extended hardness results in \cref{sec:extended-hardness-results}, we do not require third-order gadget reductions for driving terms, hence their exclusion from the second category.

Our treatment is tailored to the parameter-dependent gadgets required for the hardness reduction rather than to every possible parameter-dependent construction.
For clarity, we also make simplifying assumptions about terms that are not gadgetised.
When these terms are polynomially numerous, the same arguments apply after retaining the corresponding polynomial factors in the size parameters.

\subsection{Target and Simulator Preliminaries}
\label{sec:target-simulator-hamiltonian-family-preliminaries}
We first analyse the case in which all driving terms are retained directly in the block-diagonal sector.
We treat this as a warm-up case for the more general scenario where driving terms are simulated by the gadgets, detailed in \cref{sec:gadgetised-two-local-driving-terms}.

We use the class of parameterised $5$-local target families arising in the \clw{BQP}{hardness} construction of Ref.~\cite{hayakawa2025computational}.
This restricted form suffices for our applications and supplies the notation used for both the carried- and gadgetised-driving analyses below.

\begin{definition}\label{def:target-family}
    Let $\Upsilon = \{\bs{H}_\zeta(s)\}_{s\in[0,1]}$ be the target family of local Hamiltonians.
    Take
    \begin{equation}
        \bs{H}_\zeta(s) = \bs{H} + \zeta\, \bmw(s),
    \end{equation}
    where $\bs{H} = \sum_{a\in[A]} H^{a}$ is a sum of $A$ local terms, each with $\norm{H^{a}} \le 1$, and
    \begin{align}\label{eq:carried-driving-scales}
        \bmw(s) &= \sum_{b\in[B]}\w^b(s), &
        w_r &\coloneqq \sum_{b\in[B]}\snorm{\ps^r \w^b}\quad (r=0,1,2).
    \end{align}
    We assume every $\w^b$ is $2$-local, ${\rm C}^2$, and closed, and that $\snorm{\ps^r\w^b} = O(1)$ for each $b \in [B]$ and each $r=0,1,2$; hence $w_r = O(B)$.
    We further assume that $A + B=\poly{n}$, $A\ge1$, and $0<\zeta<1$ is a parameter fixed for each instance.
\end{definition}

We now define the corresponding simulator family of Hamiltonians relevant to the first part of this section.

\begin{definition}\label{def:simulator-family}
    Let $\wt{\Upsilon} = \{\wt{\bs{H}}_\zeta(s)\}_{s\in[0,1]}$ be the simulator family of Hamiltonians.
    Take 
    \begin{equation}
        \wt{\bs{H}}_\zeta(s) = \wt{\bs{H}} + \zeta\, \bmw(s)\otimes I_{{\rm med}},
    \end{equation}
    where $\wt{\bs{H}} = \sum_{a\in[A]} (\Delta H_0^a + V^a)$ is a sum of $A$ local gadget Hamiltonians with $\norm{H_0^a} \le 1$, and the full perturbation satisfies
    \begin{equation}
        2\snorm{\sum_a V^a + \zeta\,\bmw\otimes I_{{\rm med}}} < \Delta.
    \end{equation}
\end{definition}

Adopting the summation-to-bold notation, we define 
\begin{align}
    \bs{H}_0 &\coloneqq \sum_{a\in[A]} H_0^a, &
    \bs{V} &\coloneqq \sum_{a\in[A]} V^a,
\end{align}
thus, $\wt{\bs{H}} = \Delta \bs{H}_0 + \bs{V}$.
We additionally use the notation as outlined in \cref{sec:global-sw-transformation}, and therefore denote the projectors onto the low-energy and high-energy subspaces of $\bs{H}_0$ by $\bs{\Pi}_-$ and $\bs{\Pi}_+$, respectively.

Direct calculations using the triangle inequality and the bounds in \cref{eq:carried-driving-scales} give
\begin{align}\label{eq:simulator-hamiltonian-norms}
    \snorm{\bs{H}_\zeta} &= O(A+\zeta w_0), &
    \snorm{\wt{\bs{H}}_\zeta} &= O(A\Delta+\zeta w_0), &
    \snorm{\ps\bs{H}_\zeta} &\le \zeta w_1, &
    \snorm{\ps\wt{\bs{H}}_\zeta} &\le \zeta w_1,
\end{align}
for the main Hamiltonians of interest.

In the hardness family of \cref{eq:driving-pauli-decomposition-intro}, $B=4$ and $w_r=O(1)$.
For polynomially many carried terms the factors $w_r = O(B)$ are retained explicitly; choosing $\Delta = \Omega((A+B)^2)$ makes $A\Delta$ dominate $\zeta w_0$.
For convenience, the perturbative sums below omit carried $s$-independent terms.
Such terms may instead be gadgetised, increasing $A$ by at most a polynomial factor, or included among the carried terms as constant functions of $s$.
Under the latter convention, if $\zeta w_0\le A$ does not hold, replace $A$ throughout the asymptotic bounds by $A'\coloneqq A+\zeta w_0 = O(\poly{n})$.
Both conventions therefore preserve the required polynomial scaling.

We pull the simulator back onto the target space and define the transformed family $\widehat{\Upsilon} = \{\widehat{\bs{H}}_\zeta(s)\}_{s\in[0,1]}$ by
\begin{equation}
    \widehat{\bs{H}}_\zeta(s)=\bs{\U}(s)^\dagger\wt{\bs{H}}_\zeta(s)\bs{\U}(s).
\end{equation}
Differentiating with respect to $s$ gives
\begin{equation}
    \ps\widehat{\bs{H}}_\zeta(s) =
        \big[\ps\bs{\U}(s)^\dagger\big]\wt{\bs{H}}_\zeta(s)\bs{\U}(s)
        +\bs{\U}(s)^\dagger\big[\ps\wt{\bs{H}}_\zeta(s)\big]\bs{\U}(s)
        +\bs{\U}(s)^\dagger\wt{\bs{H}}_\zeta(s)\big[\ps\bs{\U}(s)\big].
\end{equation}

Since $\bs{\U}$ is an isometry,
\begin{align}
    \snorm{\widehat{\bs{H}}_\zeta}&\le\snorm{\wt{\bs{H}}_\zeta} = O(A\Delta),\\
    \snorm{\ps\widehat{\bs{H}}_\zeta}&\le\snorm{\ps\bs{H}_\zeta} + \snorm{\ps\bs{E}_\zeta}
    \le \zeta w_1 + \snorm{\ps\bs{E}_\zeta}.
\end{align}
Writing $\Lambda(s) = \e^{-\bs{S}(s)}\big[\ps\e^{\bs{S}(s)}\big]$ we have $\ps\bs{\U}(s) = -\Lambda(s)\bs{\U}(s)$, and since $\bs{S}(s)$ is anti-Hermitian, $\kappa = \snorm{\ps\bs{\U}} \le \snorm{\ps\bs{S}}$ (matching \cref{eq:kappa-bound}).

Each order in the perturbative gadgets we study below requires bounding $\snorm{\bs{S}_\zeta}$ and $\kappa$ (among the other standard quantities).
The following lemma bounds the derivative of the simulation-error operator in a form that is independent of the perturbative order.

\begin{restatable}{lemma}{carriedDrivingSimulationError}\label{lem:carried-driving-simulation-error}
    Let $\Upsilon$ and $\wt\Upsilon$ be as in Definitions~\ref{def:target-family} and \ref{def:simulator-family}, so that
    \begin{align}\label{eq:carried-total-perturbation}
        \wt{\bs{H}}_\zeta(s) &= \Delta\bs{H}_0 + \bs{V}_\zeta(s), &
        \bs{V}_\zeta(s) &\coloneqq \bs{V} + \zeta\,\bmw(s)\otimes I_{{\rm med}},
    \end{align}
    with $\bs{V} = \sum_{a\in[A]} V^a$ independent of $s$.
    Fix $p \in \{1,2,3\}$ and let $v_0 > 0$ satisfy $\snorm{\bs{V}_\zeta} \le v_0 < \Delta/16$.
    Let $\bs\V$ be the fixed encoding isometry with image ${\rm Im}(\bs{\Pi}_-)$, put $\bs\U(s) = \e^{-\bs{S}(s)}\bs\V$, and set
    \begin{equation}
        \bs{E}_\zeta(s) \coloneqq \bs{H}_\zeta(s) - \bs\U(s)^\dagger \wt{\bs{H}}_\zeta(s) \bs\U(s).
    \end{equation}
    Defining the order-$p$ operator
    \begin{equation}\label{eq:carried-matching-defect}
        \bs{M}_p(s) \coloneqq \bs{H}_\zeta(s) - \bs\V^\dagger \bs{H}^{[p]}_{\tn{eff}}(s)\bs\V ,
    \end{equation}
    we have
    \begin{equation}\label{eq:carried-error-decomposition}
        \bs{E}_\zeta(s) = \bs{M}_p(s) - \bs\V^\dagger\,R^{\bs{H}_{\tn{eff}}}_{p+1}(s)\,\bs\V .
    \end{equation}
    Assuming that $(\bs{V})_{--} = 0$, then, for $p \in \{1,2,3\}$, we have
    \begin{equation}\label{eq:carried-defect-derivative}
        \ps\bs{M}_p(s) =
        \begin{cases}
            0,                                                 & p \in \{1,2\}, \\[0.3em]
            - \bs\V^\dagger \big( \ps\bs{K}_3(s) \big) \bs\V , & p = 3,
        \end{cases}
        \qquad\text{with}\quad
        \snorm{\ps\bs{K}_3} = O\big(\Delta^{-2} v_0^{2} \zeta w_1\big),
    \end{equation}
    where $\bs{K}_3$ is the third-order coefficient of \cref{eq:K3-parallel}.
    Consequently,
    \begin{align}\label{eq:carried-final-bounds}
        \snorm{\bs{E}_\zeta} &\le \snorm{\bs{M}_p} + O\big(\Delta^{-p} v_0^{\,p+1}\big), &
        \snorm{\ps\bs{E}_\zeta} &= O\big(\Delta^{-\min\{p,2\}} v_0^{\min\{p,2\}} \zeta w_1\big).
    \end{align}
\end{restatable}

A proof is given in Appendix~\ref{app:perturbative-gadget-reductions-proofs}.

Recalling that $J_2$ of \cref{eq:J-contributions} involves the derivative of the simulation error operator, we can now establish a result that bounds the Berry phase error using \cref{eq:simulator-hamiltonian-norms} and appropriate assumptions on $\varepsilon$ and $\kappa$.

\subsection{A Polynomial-Control Criterion}
\label{sec:poly-control-criterion}
We now isolate the parameter tuning that makes the Berry-phase error $1/\poly{n}$, in a form characterised by decay exponents $q$ and $q'$.
We focus on the carried driving terms scenario in the following analysis (see \cref{app:perturbative-gadget-reductions-proofs} for the gadgetised-driving case).

Recall the four contributions of \cref{thm:berry-phase-simulation},
\begin{align}
    J_1 &= \gamma_\star^{-2}\max\{\snorm{\ps \bs{H}_\zeta}, \snorm{\ps\widehat{\bs{H}}_\zeta}\}\,\varepsilon, &
    J_2 &= \gamma_\star^{-1}\snorm{\ps\bs{E}_\zeta}, &
    J_3 &= \gamma_\star^{-1}\varepsilon, &
    J_4 &= \kappa,
\end{align}
where $\gamma_\star$ is the minimum gap of the target family; the transformed-family gap is bounded below using the simulation accuracy.
It suffices to make $J_1+\cdots+J_4\leq\delta$ for an inverse-polynomial target $\delta$ chosen strictly below the available arc margin.
The resulting promise preservation is then obtained from \cref{prop:gsbpe-simulator-hardness}.

The following proposition provides bounds on the contributions $J_1$ through $J_4$ under a choice of scaling hypotheses for $\kappa$ and $\snorm{\ps\bs{E}_\zeta}$ relative to the penalty $\Delta$.
Note that this choice is consistent with the scaling of the perturbative gadget reductions we study in \cref{sec:first-order-reductions,sec:second-order-reductions,sec:third-order-reductions}, and should not be considered the most general possible scenario.

\begin{restatable}{proposition}{AssembledBoundsGeneral}\label{prop:assembled-bounds-general}
    Set $d_1\coloneqq\max\{1,\zeta w_1\}$.
    Suppose a parallel gadget reduction with penalty $\Delta$ satisfies, for some $q,q' \in (0,2]$,
    \begin{equation}\label{eq:scaling-hypotheses}
        \kappa = O\big(A d_1\,\Delta^{-q}\big), \qquad \snorm{\ps\bs{E}_\zeta} = O\big(\Delta^{-q'} A^2 d_1\big).
    \end{equation}
    Then
    \begin{align}
        J_1 &= O\big(\gamma_\star^{-2}d_1(1 + A^2\Delta^{-q'})\varepsilon\big), &
        J_2 &= O\big(\gamma_\star^{-1}d_1 A^2\Delta^{-q'}\big), \\
        J_3 &= \gamma_\star^{-1}\varepsilon, &
        J_4 &= O\big(d_1 A\Delta^{-q}\big).
    \end{align}
\end{restatable}

A proof is given in Appendix~\ref{app:perturbative-gadget-reductions-proofs}.

By choosing a fixed target accuracy $\delta = 1/\poly{n}$, we can ensure that the contributions $J_1$ through $J_4$ are all bounded by $\delta$ for appropriately chosen $\varepsilon$ and $\Delta$.

\begin{restatable}{lemma}{PolyControlGeneral}\label{lem:poly-control-general}
    Fix a target $\delta = 1/\poly{n}$ and suppose $A,B, \gamma_\star^{-1}, \zeta^{-1}, \eta^{-1} = \poly{n}$ with $\zeta$ fixed by the hardness reduction.
    Suppose a parallel gadget reduction satisfies the hypotheses of \cref{prop:assembled-bounds-general} for some $q, q' \in (0,2]$, and is a valid $(\eta,\varepsilon,\kappa)$-simulation whenever $\Delta \ge \Delta_{*}(\varepsilon,\eta,\kappa) = \poly{\varepsilon^{-1},\eta^{-1},\kappa^{-1}}$.
    Then there exist $\varepsilon = 1/\poly{n}$ and $\Delta = \poly{n}$ such that the reduction is a valid $(\eta,\varepsilon,\kappa)$-simulation and $J_1 + \cdots + J_4 \le \delta$.
\end{restatable}

A proof is given in Appendix~\ref{app:perturbative-gadget-reductions-proofs}; we sketch the parameter choices here.
First choose $\varepsilon = 1/\poly{n}$ so that $J_3 \leq \delta/4$ and the $\Delta$-independent part of $J_1$ is at most $\delta/8$.
Next choose $\Delta$ large enough that $A^2\Delta^{-q'} \leq 1$, that $J_2,J_4\leq\delta/4$, and that the gadget is a valid $(\eta,\varepsilon,\kappa)$-simulation.
For fixed positive $q$ and $q'$, the maximum of these lower bounds remains polynomial in the problem parameters.
Thus $\Delta = \poly{n}$ suffices to make $J_1+\cdots+J_4\leq\delta$.

We proceed to outline the form of the simulator Hamiltonians for first-, second-, and third-order reductions when the driving terms are carried.
At each order we prove there exists a series of choices for the parameters such that the Berry phase error can be controlled to be $1/\poly{n}$.
This satisfies our previous requirements on the existence of a ``good'' family of witness isometries $\bs{\U}(s)$ with polynomially small derivative norm $\kappa = \snorm{\ps\bs{\U}}$.
Our proposals for the simulator Hamiltonians are based on the standard constructions of Bravyi and Hastings~\cite{bravyi2016complexity}.

We analyse the gadgets in parallel because the Schrieffer--Wolff generator $\bs{S}_\zeta$ may contain cross-gadget terms.
Even when the driving terms are carried, these cross terms can contribute to the derivative norm $\kappa$.
Lemma~\ref{lem:parallel-effective-hamiltonian-decomposition} separately identifies when the effective Hamiltonian decomposes into single-gadget contributions.
Recall from \cref{sec:global-sw-transformation} that $(\bs{X})_{\alpha\beta} = \bs{\Pi}_\alpha \bs{X} \bs{\Pi}_\beta$ for $\alpha,\beta \in \{+,-\}$, where $\bs{\Pi}_-$ is the projector onto the low-energy subspace of $\bs{H}_0$ and $\bs{\Pi}_+ = I - \bs{\Pi}_-$.

\subsection{First-Order Reductions}
\label{sec:first-order-reductions}
We now analyse the first-order reduction when the driving terms are carried.
First, we establish the simulation lemma for the single-gadget case, then extend it to the parallel gadget construction.

\begin{restatable}[First-order reduction]{lemma}{FirstOrderReduction}
\label{lem:first-order-reduction}
    Let $\Upsilon=\{H(s)\}_{s\in[0,1]}$ be a closed ${\rm C}^2$ target family with encoded target $\overline{H}(s)$.
    Let $H_0$ be independent of $s$ and satisfy \cref{eq:H0-block-structure}. 
    Suppose there is an $s$-independent Hermitian operator $V_1$ and a block-diagonal Hermitian family $V_2(s)$ such that
    \begin{equation}\label{eq:first-order-uniform-condition}
        \sup_{s\in[0,1]} \norm{ \overline H(s) - \big( V_1 + V_2(s) \big)_{--} } \leq \frac{\varepsilon}{2}.
    \end{equation}
    Assume that $V(s) = V_1+V_2(s)$ is ${\rm C}^2$ and closed.
    Define $L_{0} = \max\left\{1,\norm{V_1},\snorm{V_2}\right\}$ and $L_{1} = \snorm{\ps V_2}$.
    Then the family $\wt{\Upsilon} = \{\wt{H}(s) = \Delta H_0 + V(s)\}_{s\in[0,1]}$ regularly simulates $\Upsilon$ with error $(\eta,\varepsilon,\kappa)$ provided
    \begin{equation}\label{eq:first-order-regular-threshold}
        \Delta \geq \Delta_{*}(\varepsilon,\eta,\kappa) 
        = \Omega\left( \varepsilon^{-1}L_{0}^2 + \eta^{-1}L_{0} + \kappa^{-1/2}(L_{0}L_{1})^{1/2} \right).
    \end{equation}
\end{restatable}

A proof is given in Appendix~\ref{app:perturbative-gadget-reductions-proofs}; we sketch the main ideas here.
The proof follows the first-order Schrieffer--Wolff reduction of \cite[Lemma~4]{bravyi2016complexity} and verifies Conditions~\ref{item:c1}--\ref{item:c5} using \cref{lemma:uniform-control-exact-truncated,cor:static-leading-sw-generator}.
Conditions~\ref{item:c4} and \ref{item:c5} have no static analogue, which accounts for the additional term in \cref{eq:first-order-regular-threshold}.
Thus, \cref{lem:first-order-reduction} gives a uniform family extension of the static reduction.

\subsubsection{Parallel first-order reductions}
We now introduce the parallel extension of the first-order reduction.
Let 
\begin{align}
    \bs{V}_1 &= \sum_{a \in [A]} V_1^a, &
    \bs{V}_2(s) &= \bmw(s) \otimes I_{{\rm med}},
\end{align}
and define the parallel simulator Hamiltonian as
\begin{equation}\label{eq:parallel-simulator-hamiltonian-first-order}
    \wt{\bs{H}}_\zeta(s) = \Delta \bs{H}_0 + \bs{V}_1 + \zeta \bs{V}_2(s),
\end{equation}
where $\bs{H}_0 = \sum_{a \in [A]} H_0^a$.
For simplicty, we assume $\zeta w_0 \leq A$, hence $L_0 = O(A)$.

By \cref{eq:K1-parallel}, we obtain
\begin{equation}
    \bs{H}_{\tn{eff}}^{[1]}(s) = \big( \bs{V}_1 + \zeta \bs{V}_2(s) \big)_{--}.
\end{equation}
Consequently, if $(V^a_1)_{--} = \overline{H}^a$ for every $a$, then 
\begin{equation}
    \bs{H}_{\tn{eff}}^{[1]}(s) = \sum_{a \in [A]} \overline{H}^a + \zeta \big( \bs{V}_2(s) \big)_{--} = \overline{\bs{H}}_\zeta(s),
\end{equation}
since, under the fixed encoding, $\big( \bs{V}_2(s) \big)_{--} = \overline{\bmw}(s)$.
There are no cross terms to analyse at this order.
From \cref{lem:first-order-reduction} a sufficient parallel threshold is
\begin{equation}\label{eq:parallel-threshold-first-order}
    \Delta = \Omega(\varepsilon^{-1} A^2 + \eta^{-1} A + \kappa^{-1/2} (A d_1)^{1/2}).
\end{equation}

The global perturbation satisfies $v_0 = O(A)$, under the assumption $\zeta w_0 \leq A$, and $v_1 \leq \zeta w_1$.
Because the carried driving term is block diagonal, the first-order Schrieffer--Wolff coefficient $\bs{S}_1$ is independent of $s$ and therefore, $\ps \bs{S}_\zeta(s) = \ps R_2^{\bs{S}}(s)$.
Applying \cref{cor:static-leading-sw-generator} then gives
\begin{align}\label{eq:parallel-first-order-bounds-pt1}
    \snorm{\bs{S}_\zeta} &= O(\Delta^{-1} A), &
    \kappa &= O(\Delta^{-2} d_1 A).
\end{align}
Additionally, \cref{lem:carried-driving-simulation-error} with $p=1$ gives
\begin{equation}\label{eq:parallel-first-order-bounds-pt2}
    \snorm{\ps \bs{E}_\zeta} = O(\Delta^{-1} v_0 v_1) = O(\Delta^{-1} A d_1).
\end{equation}
Thus, the hypotheses of \cref{prop:assembled-bounds-general} hold with $q=2$, $q'=1$, and size parameter $O(A)$.
    
\begin{restatable}{corollary}{ParallelFirstOrderSimulationBerryPhase}
\label{cor:parallel-first-order-simulation-berry-phase}
    Let $\Upsilon$ be the target family of Definition~\ref{def:target-family} and let $\wt{\Upsilon}$ be the first-order parallel simulator family of the form in \cref{eq:parallel-simulator-hamiltonian-first-order}, with Berry phases $\vartheta$ and $\wt{\vartheta}$ respectively.
    Suppose $A + \zeta w_0$, $B$, $\gamma_\star^{-1}$, $\zeta^{-1}$ and $\eta^{-1}$ are polynomially bounded in $n$, and fix a target Berry-phase accuracy $\delta = 1/\poly{n}$.
    Then there exist simulation targets $\varepsilon, \kappa = 1/\poly{n}$, and a penalty $\Delta = \poly{n}$, such that $\wt{\Upsilon}$ is a regular $(\eta,\varepsilon,\kappa)$-simulation of $\Upsilon$ satisfying
    \begin{equation}
        d_{2\pi}(\vartheta,\wt{\vartheta}) \le \delta .
    \end{equation}
\end{restatable}

A proof is provided in Appendix~\ref{app:perturbative-gadget-reductions-proofs}.

\subsection{Second-Order Reductions}
\label{sec:second-order-reductions}
We now analyse the second-order reduction when the driving terms are carried.
First, we establish the simulation lemma for the single-gadget case, then extend it to the parallel gadget construction.

\begin{restatable}[Second-order reduction]{lemma}{SecondOrderReduction}\label{lem:second-order-reduction}
    Let $\Upsilon=\{H(s)\}_{s\in[0,1]}$ be a closed ${\rm C}^2$ target family, with encoded target $\overline{H}(s)$.
    Let $H_0$ be independent of $s$ and satisfy \cref{eq:H0-block-structure}. 
    Suppose there is an $s$-independent $V_1$ with $(V_1)_{--} = 0$, and block-diagonal $V_2, V_3(s)$ such that
    \begin{equation}\label{eq:second-order-uniform-condition}
        \sup_{s\in[0,1]} \norm{ \overline{H}(s) - \big( V_2 + V_3(s) \big)_{--} + (V_1)_{-+}H_0^{\pattce}(V_1)_{+-} } \le \varepsilon/2.
    \end{equation}
    Assume that $V(s) = \Delta^{1/2} V_1 + V_2 + V_3(s)$ is ${\rm C}^2$ and closed.
    Define $L_0 = \max\{\norm{V_1}, \norm{V_2}, \snorm{V_3}\}$ and $L_1 = \snorm{\ps V_3}$.
    Then the family $\wt{\Upsilon} = \{\wt{H}(s) = \Delta H_0 + V(s)\}_{s\in[0,1]}$ regularly simulates $\Upsilon$ with error $(\eta,\varepsilon,\kappa)$ provided
       \begin{equation}\label{eq:second-order-regular-threshold}
        \Delta \geq \Delta_{*}(\varepsilon,\eta,\kappa) 
        = \Omega\left( \varepsilon^{-2}L_{0}^6 + \eta^{-2}L_{0}^2 + \kappa^{-2/3}(L_{0}L_{1})^{2/3} \right).
    \end{equation}
\end{restatable}

A proof is given in Appendix~\ref{app:perturbative-gadget-reductions-proofs}; note that the structure follows that of \cref{lem:first-order-reduction}, so we do not provide a sketch here.

\subsubsection{Parallel second-order reductions}
We now introduce the parallel extension of the second-order reduction.
Let 
\begin{align}
    \bs{V}_1 &= \sum_a V_1^a, &
    \bs{V}_2 &= \sum_a V_2^a, &
    \bs{V}_3(s) &= \bmw(s) \otimes I_{{\rm med}},
\end{align}
and define the parallel simulator Hamiltonian as 
\begin{equation}\label{eq:parallel-simulator-hamiltonian-second-order}
    \wt{\bs{H}}_\zeta(s) = \Delta\bs{H}_0 + \Delta^{1/2}\bs{V}_1 + \bs{V}_2 + \zeta\,\bs{V}_3(s).
\end{equation}
For simplicty, we assume $\zeta w_0 \leq A$.

By \cref{eq:K1-parallel,eq:K2-parallel} we obtain 
\begin{equation}\label{eq:parallel-second-order-effective-hamiltonian}
    \bs{H}_{\tn{eff}}^{[2]}(s) = \big( \bs{V}_2 + \zeta\,\bs{V}_3(s) \big)_{--} - (\bs{V}_1)_{-+} \bs{\Sigma} (\bs{V}_1)_{+-}.
\end{equation}
Since we have chosen each $V^a$ to act only non-trivially on its corresponding mediator qubit, \cref{eq:K2-decomposition} holds and thus
\begin{equation}\label{eq:second-order-no-cross-gadget}
    (\bs{V}_1)_{-+} \bs{\Sigma} (\bs{V}_1)_{+-} = \sum_a (V_1^a)_{-+} (H_0^a)^{\pattce} (V_1^a)_{+-},
\end{equation}
implying there are no cross-gadget contributions at this order.

For each $a \in [A]$, suppose that $V_1^a$ and $V_2^a$ are chosen so that
\begin{equation}\label{eq:second-order-local-error}
    \norm{\overline{H}^a - (V_2^a)_{--} + (V_1^a)_{-+} (H_0^a)^{\pattce} (V_1^a)_{+-}} \leq \varepsilon_a \leq \frac{\varepsilon}{2A}.
\end{equation}
From \cref{eq:second-order-no-cross-gadget} and the fact that $\big( \bs{V}_3(s) \big)_{--} = \overline{\bmw}(s)$, we have that 
\begin{equation}
    \overline{\bs{H}}_\zeta(s) - \bs{H}_{\tn{eff}}^{[2]}(s) = \sum_a \left[ \overline{H}^a - (V_2^a)_{--} + (V_1^a)_{-+} (H_0^a)^{\pattce} (V_1^a)_{+-} \right].
\end{equation}
The triangle inequality and \cref{eq:second-order-local-error} then give
\begin{equation}
    \sup_{s \in [0,1]} \norm{\overline{\bs{H}}_\zeta(s) - \bs{H}_{\tn{eff}}^{[2]}(s)} \leq \frac{\varepsilon}{2}.
\end{equation}

Assuming $\norm{V_1^a}, \norm{V_2^a} = O(1)$ for all $a$ and $\zeta w_0 \leq A$, we have $L_0 = O(A)$.
Applying \cref{lem:second-order-reduction} therefore gives the sufficient threshold 
\begin{equation}
    \Delta = \Omega(\varepsilon^{-2} A^6 + \eta^{-2} A^2 + \kappa^{-2/3} (d_1 A)^{2/3}).
\end{equation}

The global perturbation satisfies $v_0 = O(\Delta^{1/2} A)$, under the assumption $\zeta w_0 \leq A$ and $\norm{V_2^a} = O(1)$, and $v_1 \leq \zeta w_1$.
Again, the first-order Schrieffer--Wolff coefficient is independent of $s$, therefore, $\ps \bs{S}_\zeta(s) = \ps \bs{S}_{2,\zeta}(s) + \ps R_3^{\bs{S}}(s)$.
Applying \cref{cor:static-leading-sw-generator} then gives
\begin{align}\label{eq:parallel-second-order-bounds-pt1}
    \snorm{\bs{S}_\zeta} &= O(\Delta^{-1/2} A), &
    \kappa &= O(\Delta^{-3/2} d_1 A).
\end{align}
Additionally, \cref{lem:carried-driving-simulation-error} with $p=2$ gives
\begin{equation}\label{eq:parallel-second-order-bounds-pt2}
    \snorm{\ps \bs{E}_\zeta} = O(\Delta^{-2} v_0^2 v_1) = O(\Delta^{-1} A^2 d_1).
\end{equation}
Thus, the hypotheses of \cref{prop:assembled-bounds-general} hold with $q=3/2$, $q'=1$, and size parameter $O(A)$.

\begin{restatable}{corollary}{ParallelSecondOrderSimulationBerryPhase}
\label{cor:parallel-second-order-simulation-berry-phase}
    Let $\Upsilon$ be the target family of Definition~\ref{def:target-family} and let $\wt{\Upsilon}$ be the second-order parallel simulator family of the form in \cref{eq:parallel-simulator-hamiltonian-second-order}, with Berry phases $\vartheta$ and $\wt{\vartheta}$ respectively.
    Suppose $A + \zeta w_0$, $B$, $\gamma_\star^{-1}$, $\zeta^{-1}$ and $\eta^{-1}$ are polynomially bounded in $n$, and fix a target Berry-phase accuracy $\delta = 1/\poly{n}$.
    Then there exist simulation targets $\varepsilon, \kappa = 1/\poly{n}$, and a penalty $\Delta = \poly{n}$, such that $\wt{\Upsilon}$ is a regular $(\eta,\varepsilon,\kappa)$-simulation of $\Upsilon$ satisfying
    \begin{equation}
        d_{2\pi}(\vartheta,\wt{\vartheta}) \le \delta .
    \end{equation}
\end{restatable}

The proof follows analogously from \cref{cor:parallel-first-order-simulation-berry-phase} and the bounds established for the second-order case.

\subsection{Third-Order Reductions}
\label{sec:third-order-reductions}
We now analyse the third-order reduction when the driving terms are carried.
First, we establish the simulation lemma for the single-gadget case, then extend it to the parallel gadget construction.

\begin{restatable}[Third-order reduction]{lemma}{ThirdOrderReduction}\label{lem:third-order-reduction}
    Let $\Upsilon=\{H(s)\}_{s\in[0,1]}$ be a closed ${\rm C}^2$ target family with encoded target $\overline{H}(s)$.
    Let $H_0$ be independent of $s$ and satisfy \cref{eq:H0-block-structure}.
    Suppose there is an $s$-independent $V_1$ with $(V_1)_{--}=0$, an $s$-independent block-diagonal $V_2$ with $(V_2)_{--} = (V_1)_{-+}H_0^{\pattce}(V_1)_{+-}$, and block-diagonal $V_3,V_4(s)$ such that
    \begin{equation}\label{eq:third-order-uniform-condition}
        \sup_{s \in [0,1]} \norm{ \overline{H}(s) - \big( V_3 + V_4(s) \big)_{--} - (V_1)_{-+}H_0^{\pattce}(V_1)_{++}H_0^{\pattce}(V_1)_{+-} } \le \varepsilon/2.
    \end{equation}
    Assume that $V(s) = \Delta^{2/3} V_1 + \Delta^{1/3} V_2 + V_3 + V_4(s)$ is ${\rm C}^2$ and closed.
    Define $L_0 = \max\{\norm{V_1}, \norm{V_2}, \norm{V_3}, \snorm{V_4}\}$ and $L_1 = \snorm{\ps V_4}$.
    Then the family $\wt{\Upsilon} = \{\wt{H}(s) = \Delta H_0 + V(s)\}_{s\in[0,1]}$ regularly simulates $\Upsilon$ with error $(\eta,\varepsilon, \kappa)$ provided
    \begin{equation}\label{eq:third-order-regular-threshold}
        \Delta \geq \Delta_{*}(\varepsilon,\eta,\kappa) 
        = \Omega\left( \varepsilon^{-3}L_{0}^{12} + \eta^{-3}L_{0}^{3} + \kappa^{-4/3}(L_{0}L_{1})^{4/3} \right).
    \end{equation}
\end{restatable}

A proof is given in Appendix~\ref{app:perturbative-gadget-reductions-proofs}; note that the structure follows that of \cref{lem:first-order-reduction}, so we do not provide a sketch here.

\subsubsection{Parallel third-order reductions}
We now introduce the parallel extension of the third-order reduction.
Let 
\begin{align}
    \bs{V}_1 &= \sum_a V_1^a, &
    \bs{V}_2 &= \sum_a V_2^a, &
    \bs{V}_3 &= \sum_a V_3^a, &
    \bs{V}_4(s) &= \bmw(s) \otimes I_{{\rm med}},
\end{align}
and define the parallel simulator Hamiltonian as
\begin{equation}\label{eq:parallel-simulator-hamiltonian-third-order}
    \wt{\bs{H}}_\zeta(s) = \Delta\bs{H}_0 + \Delta^{2/3}\bs{V}_1 + \Delta^{1/3}\bs{V}_2 + \bs{V}_3 + \zeta\bs{V}_4(s),
\end{equation}
where $\bs{H}_0 = \sum_a H_0^a$.
For simplicity, we assume $\zeta w_0 \leq A$.

By \cref{eq:K1-parallel,eq:K2-parallel,eq:K3-parallel} we obtain
\begin{align}\label{eq:third-order-effective-hamiltonian}
    \begin{split}
        \bs{H}_{\tn{eff}}^{[3]}(s) &= 
        \big( \Delta^{1/3}\bs{V}_2 + \bs{V}_3 + \zeta\bs{V}_4(s) \big)_{--} 
        - \Delta^{1/3} (\bs{V}_1)_{-+} \bs{\Sigma} (\bs{V}_1)_{+-} 
        + (\bs{V}_1)_{-+} \bs{\Sigma} (\bs{V}_1)_{++} \bs{\Sigma} (\bs{V}_1)_{+-} \\
        &\qquad+ O(\Delta^{-2/3} \max \{\Delta^{1/3}\norm{\bs{V}_2}, \norm{\bs{V}_3}, \snorm{\bs{V}_4} \} \norm{\bs{V}_1}^2).
    \end{split}
\end{align}
We further assume that $\max \{\Delta^{1/3}\norm{\bs{V}_2}, \norm{\bs{V}_3}, \snorm{\bs{V}_4} \} = \Delta^{1/3} \norm{\bs{V}_2}$ giving $O(\Delta^{-1/3} \norm{\bs{V}_2} \norm{\bs{V}_1}^2)$ as the error term in the third-order effective Hamiltonian.
Recall \cref{rmk:ignoring-third-order-terms} discussed earlier that we ``ignore'' the second group of terms in $\bs{K}_3(s)$ (\cref{eq:K3-parallel}); such terms represent the $O(\Delta^{-1/3} \norm{\bs{V}_2} \norm{\bs{V}_1}^2)$ contribution above.

Since we have chosen each $V^a$ to act only non-trivially on its corresponding mediator qubit, and $(V_1^a)_{--} = 0$ for all $a$, \cref{eq:K2-decomposition,eq:K3-decomposition} hold, giving
\begin{align}\label{eq:third-order-no-cross-gadget}
    \begin{split}
        (\bs{V}_1)_{-+} \bs{\Sigma} (\bs{V}_1)_{+-} &= \sum_a (V_1^a)_{-+} (H_0^a)^{\pattce} (V_1^a)_{+-}, \\
        (\bs{V}_1)_{-+} \bs{\Sigma} (\bs{V}_1 )_{++} \bs{\Sigma} (\bs{V}_1)_{+-} &= \sum_a (V_1^a)_{-+} (H_0^a)^{\pattce} (V_1^a)_{++} (H_0^a)^{\pattce} (V_1^a)_{+-},
    \end{split}
\end{align}
implying no cross-gadget contributions in the third-order effective Hamiltonian, modulo contributions from the error term.

For each $a \in [A]$, suppose that $V_1^a$ and $V_3^a$, are chosen so that 
\begin{equation}\label{eq:third-order-local-error}
    \norm{\overline{H}^a - (V_3^a)_{--} - (V_1^a)_{-+} (H_0^a)^{\pattce} (V_1^a)_{++} (H_0^a)^{\pattce} (V_1^a)_{+-}} \leq \varepsilon_a \leq \frac{\varepsilon}{2A},
\end{equation}
with $(V_1^a)_{--} = 0$, and $V_2^a$ is chosen so that $(V_2^a)_{--} = (V_1^a)_{-+} (H_0^a)^{\pattce} (V_1^a)_{+-}$.
From \cref{eq:third-order-no-cross-gadget} and the fact that $\big( \bs{V}_4(s) \big)_{--} = \overline{\bmw}(s)$, we have that 
\begin{equation}
    \overline{\bs{H}}_\zeta(s) - \bs{H}_{\tn{eff}}^{[3]}(s) = \sum_a \left[ \overline{H}^a - (V_3^a)_{--} - (V_1^a)_{-+} (H_0^a)^{\pattce} (V_1^a)_{++} (H_0^a)^{\pattce} (V_1^a)_{+-} \right] + O(\Delta^{-1/3} \norm{\bs{V}_2} \norm{\bs{V}_1}^2).
\end{equation}
The triangle inequality and \cref{eq:third-order-local-error} then give
\begin{equation}
    \norm{\overline{\bs{H}}_\zeta(s) - \bs{H}_{\tn{eff}}^{[3]}(s)} \leq \sum_a \varepsilon_a + O(\Delta^{-1/3} \norm{\bs{V}_2} \norm{\bs{V}_1}^2) \leq \frac{\varepsilon}{2} + O(\Delta^{-1/3} \norm{\bs{V}_2} \norm{\bs{V}_1}^2).
\end{equation}
Assuming $\norm{V_1^a}, \norm{V_2^a}, \norm{V_3^a} = O(1)$ for all $a$ and $\zeta w_0 \leq A$, we have $L_0 = O(A)$ and the error term $O(\Delta^{-1/3} \norm{\bs{V}_2} \norm{\bs{V}_1}^2) = O(\Delta^{-1/3} A^3) \leq \varepsilon/2$ for $\Delta = \Omega(A^9 \varepsilon^{-3})$.
Applying \cref{lem:third-order-reduction} therfore gives the sufficient threshold
\begin{equation}\label{eq:parallel-threshold-third-order}
    \Delta = \Omega(\varepsilon^{-3}A^{12} + \eta^{-3}A^3 + \kappa^{-3/4}(Ad_1)^{3/4}).
\end{equation}

The global perturbation satisfies $v_0 = O(\Delta^{2/3} A)$, under the assumption $\zeta w_0 \leq A$ and $\norm{V_1^a}, \norm{V_2^a}, \norm{V_3^a} = O(1)$, and $v_1 \leq \zeta w_1$.
Again, the first-order Schrieffer--Wolff coefficient is independent of $s$, therefore, $\ps \bs{S}_\zeta(s) = \ps \bs{S}_{2,\zeta}(s) + \ps \bs{S}_{3,\zeta}(s) + \ps R_4^{\bs{S}}(s)$.
Applying \cref{cor:static-leading-sw-generator} then gives 
\begin{align}\label{eq:parallel-third-order-bounds-pt1}
    \snorm{\bs{S}_\zeta} &= O(\Delta^{-1/3} A), &
    \kappa &= O(\Delta^{-4/3} d_1 A).
\end{align}
Additionally, \cref{lem:carried-driving-simulation-error} with $p=3$ gives
\begin{equation}\label{eq:parallel-third-order-bounds-pt2}
    \snorm{\ps \bs{E}_\zeta} = O(\Delta^{-2} v_0^2 v_1) = O(\Delta^{-2/3} A^2 d_1).
\end{equation}
Thus, the hypotheses of \cref{prop:assembled-bounds-general} hold with $q=4/3$, $q'=2/3$, and size parameter $O(A)$.

\begin{restatable}{corollary}{ParallelThirdOrderSimulationBerryPhase}
\label{cor:parallel-third-order-simulation-berry-phase}
    Let $\Upsilon$ be the target family of Definition~\ref{def:target-family} and let $\wt{\Upsilon}$ be the third-order parallel simulator family of the form in \cref{eq:parallel-simulator-hamiltonian-third-order}, with Berry phases $\vartheta$ and $\wt{\vartheta}$ respectively.
    Suppose $A + \zeta w_0$, $B$, $\gamma_\star^{-1}$, $\zeta^{-1}$ and $\eta^{-1}$ are polynomially bounded in $n$, and fix a target Berry-phase accuracy $\delta = 1/\poly{n}$.
    Then there exist simulation targets $\varepsilon, \kappa = 1/\poly{n}$, and a penalty $\Delta = \poly{n}$, such that $\wt{\Upsilon}$ is a regular $(\eta,\varepsilon,\kappa)$-simulation of $\Upsilon$ satisfying
    \begin{equation}
        d_{2\pi}(\vartheta,\wt{\vartheta}) \le \delta .
    \end{equation}
\end{restatable}

The proof follows analogously from \cref{cor:parallel-first-order-simulation-berry-phase} and the bounds established for the third-order case.

\subsection{Gadgetised Sums of Two-Local Driving Terms}
\label{sec:gadgetised-two-local-driving-terms}
We now analyse $2$-local driving terms that are realised perturbatively using mediator qubits.
We first extend \cref{lem:second-order-reduction} to one carried term and one gadgetised interaction, and then formulate the corresponding parallel construction.

We explicitly restrict attention to $2$-local target Hamiltonians that are ${\rm C}^2$ in $s$ and closed, so $H(0) = H(1)$.
The structure of the target Hamiltonians is now given as 
\begin{equation}\label{eq:target-hamiltonian-gadget-sum-form}
    H(s) = \Gamma(s) + \mathdcal{K}(s),
\end{equation}
where 
\begin{equation}
    \mathdcal{K}(s) = \mathdcal{C}(s) - \frac{1}{2} \big(\mathdcal{Q}(s)\big)^2,
\end{equation}
for some $2$-local operator $\mathdcal{C}(s)$ and $1$-local operator $\mathdcal{Q}(s)$.
More specifically,
\begin{align}
    \mathdcal{C}(s) &= \sum_{\ell \in [{\sf d}_{\mathdcal{C}}]} \mathdcal{C}^\ell(s), &
    \mathdcal{Q}(s) &= \sum_{\ell \in [{\sf d}_{\mathdcal{Q}}]} \mathdcal{Q}^\ell(s).
\end{align}
In a moment, we show examples forms of the individual terms $\mathdcal{C}^\ell(s)$ and $\mathdcal{Q}^\ell(s)$.

In the present case, we assume that the second-order effective Hamiltonian can be chosen to agree exactly with the encoded target Hamiltonian.
This assumption holds for the constructions used below; extending the result to a more general setting is straightforward.

\begin{restatable}[Parameter-dependent second-order reduction]{lemma}{SecondOrderParameterisedReduction}\label{lem:second-order-parameterised-reduction}
    Let $\Upsilon=\{H(s)\}_{s\in[0,1]}$ be a closed ${\rm C}^2$ target family, with encoded target $\overline{H}(s)$.
    Let $H_0$ be independent of $s$ and satisfy \cref{eq:H0-block-structure}. 
    Suppose there is $V_1(s)$ with $\big( V_1(s) \big)_{--} = 0$, and block-diagonal $V_2(s)$ such that
    \begin{equation}\label{eq:second-order-parameterised-uniform-condition}
        \big( V_2(s) \big)_{--} - \big( V_1(s) \big)_{-+}H_0^{\pattce}\big( V_1(s) \big)_{+-} = \overline{H}(s).
    \end{equation}
    Assume that $V(s) = \Delta^{1/2} V_1(s) + V_2(s)$ is ${\rm C}^2$ and closed.
    Define $L_0 = \max\{\snorm{V_i}\}_{i\in\{1,2\}}$ and $L_1 = \max\{\snorm{\ps V_i}\}_{i\in\{1,2\}}$.
    Then the family $\wt{\Upsilon} = \{\wt{H}(s) = \Delta H_0 + V(s)\}_{s\in[0,1]}$ regularly simulates $\Upsilon$ with error $(\eta,\varepsilon,\kappa)$ provided
       \begin{equation}\label{eq:second-order-parameterised-regular-threshold}
        \Delta \geq \Delta_{*}(\varepsilon,\eta,\kappa) 
        = \Omega\left( \varepsilon^{-2}L_{0}^6 + \eta^{-2}L_{0}^2 + \kappa^{-2}L_1^2 \right).
    \end{equation}
\end{restatable}

A proof is given in Appendix~\ref{app:perturbative-gadget-reductions-proofs}; note that the structure follows that of \cref{lem:second-order-reduction}, so we do not provide a sketch here.

The target is allowed to contain a sum of such interactions and an independently parameterised background.
As every driving interaction, to which we apply a gadget, is already $2$-local we will only require the use of second-order gadgets.
Moreover, the specific second-order gadgets we will use stem from Oliveira and Terhal's \textsl{subdivision}, \textsl{fork}, and \textsl{cross} gadgets \cite{oliveira2008complexity} which all admit the same structure (see \cref{eq:polarisation-map} below).

We present a pictorial representation of a single application of each of these gadgets, in \cref{fig:example-gadgets}, to demonstrate the connections each mediator qubit has to the perturbation terms.
Notice that the mediator in each case has a constant number of perturbative connections to the system qubits.

\begin{figure}[!hb]
\centering
\begin{subfigure}{0.3\textwidth}
    \begin{tikzpicture}
        \pic{subdivision};
    \end{tikzpicture}
    \caption{The \textsl{subdivision} gadget.}
    \label{fig:subdivision}
\end{subfigure}
\hfill
\begin{subfigure}{0.3\textwidth}
    \begin{tikzpicture}
        \pic{cross};
    \end{tikzpicture}
    \caption{The \textsl{cross} gadget.}
    \label{fig:cross}
\end{subfigure}
\hfill
\begin{subfigure}{0.3\textwidth}
    \begin{tikzpicture} 
        \pic{fork};
    \end{tikzpicture}
    \caption{The \textsl{fork} gadget.}
    \label{fig:fork}
\end{subfigure}
\caption{
    Revised from Ref.~\cite{waite2025complexitya}. 
    These illustrations omit representations of cross-terms that can emerge.
    The top card of each column represents the target interaction, while the bottom card represents the form of the gadget used to generate it.
}
\label{fig:example-gadgets}
\end{figure}

The motivation for introducing parameter-dependent second-order layers is to systematically decompose the driving Hamiltonian in \cref{eq:driving-pauli-decomposition-intro} into components that can be individually addressed by second-order gadgets (\textsl{subdivision}, \textsl{cross}, or \textsl{fork}).
At each layer, we partition $\bmw(s)=\sum_{\alpha} \w^{\alpha}(s)$ into terms retained in the carried background $\bs\Gamma(s)$ and terms $\mathdcal{K}^a(s)$ generated by gadgets.
This partition may change between layers.

\subsubsection{The parallel setting}
We now define one parallel parameter-dependent second-order layer and its simulator.

\begin{definition}[A parameter-dependent second-order layer]\label{def:parameter-dependent-second-order-layer}
    Let $\Upsilon = \{\bs{H}(s)\}_{s\in[0,1]}$ be a family of target Hamiltonians.
    Take
    \begin{equation}\label{eq:general-gadgetised-target}
        \bs{H}(s) = \bs{\Gamma}(s) + \mathdbcal{K}(s),
    \end{equation}
    where $\bs{\Gamma}(s) = \sum_{b\in[B]} \Gamma^b(s)$ is a sum of $B$ $2$-local terms that are ${\rm C}^2$ and closed with $\snorm{\ps^r \Gamma^b(s)} = O(1)$ for all $b \in [B]$ and $r \in \{0,1,2\}$, and
    \begin{equation}
        \mathdbcal{K}(s) = \sum_{a\in[A]} \mathdcal{K}^a(s).
    \end{equation}
    We assume every $\mathdcal{K}^a$ is ${\rm C}^2$, closed and of the form
    \begin{equation}
        \mathdcal{K}^a(s) = \mathdcal{C}^a(s) - \frac{1}{2} \big(\mathdcal{Q}^a(s)\big)^2 
        = \sum_{\ell = 1}^{{\sf d}_{a,\mathdcal{C}}} \mathdcal{C}^{a,\ell}(s) - \frac{1}{2} \bigg( \sum_{\ell =1}^{{\sf d}_{a,\mathdcal{Q}}} \mathdcal{Q}^{a,\ell}(s) \bigg)^2,
    \end{equation}
    where ${\sf d}_{a,\mathdcal{C}}$ and ${\sf d}_{a,\mathdcal{Q}}$ are the respective numbers of components in the decompositions of $\mathdcal{C}^a(s)$ and $\mathdcal{Q}^a(s)$, and all $(a,\ell)$ components are ${\rm C}^2$ and closed with $\snorm{\ps^r \mathdcal{C}^{a,\ell}(s)}, \snorm{\ps^r \mathdcal{Q}^{a,\ell}(s)} = O(1)$ for all $a,\ell$ and $r \in \{0,1,2\}$.
    We further assume that $A + B = \poly{n}$ with $A,B \geq 1$.
\end{definition}

Note that $\sum_a {\sf d}_{a,\mathdcal{C}} = {\sf d}_{\mathdcal{C}}$ and $\sum_a {\sf d}_{a,\mathdcal{Q}} = {\sf d}_{\mathdcal{Q}}$.
We will now define the corresponding simulator family of local Hamiltonians that will align exactly with the, appropriately defined, encoded target family.

\begin{definition}\label{def:general-driving-simulator}
    Let $\wt{\Upsilon} = \{\wt{\bs{H}}(s)\}_{s\in[0,1]}$ be the corresponding simulator family of local Hamiltonians.
    Take 
    \begin{equation}\label{eq:general-driving-simulator-form}
        \wt{\bs{H}}(s) = \Delta \bs{H}_0 + \Delta^{1/2} \bs{V}_1(s) + \bs{V}_2(s),
    \end{equation}
    where 
    \begin{align}\label{eq:general-driving-simulator-pieces}
        \bs{H}_0 &= \sum_{a\in[A]}\ketbra{1}_{m_a}, &
        \bs{V}_1(s) &= \frac{1}{\sqrt{2}}\sum_{a\in[A]}\sum_{\ell=1}^{{\sf d}_{a,\mathdcal{Q}}} \mathdcal{Q}^{a,\ell}(s)X_{m_a}, &
        \bs{V}_2(s) &= \left(\bs\Gamma(s)+\sum_{a\in[A]}\sum_{\ell=1}^{{\sf d}_{a,\mathdcal{C}}} \mathdcal{C}^{a,\ell}(s)\right)\otimes I_{{\rm med}}, 
    \end{align}
    for a set of distinct mediator qubits ${\rm med} = \{m_a\}_{a\in[A]}$.
\end{definition}

We again use the notation as outlined in \cref{sec:global-sw-transformation}, and therefore denote the projectors onto the low-energy and high-energy subspaces of $\bs{H}_0$ by $\bs{\Pi}_-$ and $\bs{\Pi}_+$, respectively.

Two important features of this definition are that the penalty Hamiltonian $\bs{H}_0$ is independent of $s$ and that the perturbative interactions use distinct mediator-excitation sectors.
The latter feature removes second-order cross terms across different gadgets, that is, we can employ the decomposition for \cref{eq:K2-decomposition} and define 
\begin{equation}
    \Omega^a(s) = \frac{1}{2} \bigg( \sum_{\ell,\ell' \in [{\sf d}_{a,\mathdcal{Q}}]} \mathdcal{Q}^{a,\ell}(s) \mathdcal{Q}^{a,\ell'}(s) \bigg) \bs{\Pi}_- = \frac{1}{2} \big( \mathdcal{Q}^a(s) \big)^2 \bs{\Pi}_-.
\end{equation}
Hence, there are no cross-gadget second-order terms, i.e., no contributions involving different gadgets $a$ and $a'$ with $a\ne a'$.
Also, observe that 
\begin{equation}
    \big( \bs{V}_2(s) \big)_{--} = \bigg(\bs\Gamma(s)+\sum_{a\in[A]} \mathdcal{C}^a(s)\bigg) \bs{\Pi}_-,
\end{equation}
and therefore via substitution into \cref{eq:parallel-second-order-effective-hamiltonian}, the effective Hamiltonian to second order is given by
\begin{equation}
    \bs{H}_{\tn{eff}}^{[2]}(s) \equiv \overline{\bs{H}}(s),
\end{equation}
which verifies the claim that we recover the exact encoded target Hamiltonian to second order.

\subsubsection{The choice of interaction terms}
For our purposes, we require the operators appearing in \cref{eq:general-driving-simulator-pieces} to have the desired simulator locality of $2$.
In the static formulation of the \textsl{subdivision}, \textsl{cross} and \textsl{fork} gadgets, every $\mathdcal{Q}^{a,\ell}$ is a sum of single-qubit Pauli operators and every $\mathdcal{C}^{a,\ell}$ is at most $2$-local, so \cref{eq:general-driving-simulator-form} is $2$-local.

For $r=0,1,2$, we define the parameters
\begin{align}
    q_r &\coloneqq \sum_{a\in[A]}\sum_{\ell=1}^{{\sf d}_{a,\mathdcal{Q}}}\snorm{\ps^r \mathdcal{Q}^{a,\ell}}, &
    c_r &\coloneqq \sum_{a\in[A]}\sum_{\ell=1}^{{\sf d}_{a,\mathdcal{C}}}\snorm{\ps^r \mathdcal{C}^{a,\ell}}, &
    \tau_r &\coloneqq \sum_{b\in[B]}\snorm{\ps^r\Gamma^b}.
\end{align}
and the scale function
\begin{equation}
    \label{eq:general-driving-perturbation-scale}
    l_r(\Delta) \coloneqq (\Delta/2)^{1/2}q_r + c_r + \tau_r.
\end{equation}
Definition~\ref{def:parameter-dependent-second-order-layer} gives $q_r,c_r = O(A)$ and $\tau_r = O(B)$.
Consequently, $l_r(\Delta) = O(\Delta^{1/2}A+B)$, which is polynomial because $A+B = \poly{n}$.
Moreover, for simplicity we take $l_r(\Delta) = O(\Delta^{1/2}A)$.

From Ref.~\cite{oliveira2008complexity} (and \cref{fig:example-gadgets}), ${\sf d}_{a, \mathdcal{Q}} \leq 4$ and ${\sf d}_{a, \mathdcal{C}} \leq 16$ (both constant with respect to $n$) for all $a\in[A]$.
Note that $\bs{\Gamma}(s)$ can be partitioned into those terms dependent on $s$ and those that are not.
However, we do not consider this partitioning for a slightly more general treatment.

Each $\mathdcal{K}^a(s)$ will collect a particular set of target interactions we intend to simulate using perturbative gadgets.
The decomposition $\mathdcal{K}^a \mapsto (\mathdcal{C}^a, \mathdcal{Q}^a)$ can be recovered using an analogy with the polarisation identity for real vectors.
Specifically, this choice can be organised using the polarisation map $\Xi$ 
\begin{equation}\label{eq:polarisation-map}
    \Xi(X, Y; R) \coloneqq \left(\frac{1}{2}(X^2+Y^2)-R, X-Y\right).
\end{equation}
If the commuting operators $X$ and $Y$ and the correction term $R$ are chosen so that $\mathdcal{K}^a = XY - R$, then $(\mathdcal{C}^a, \mathdcal{Q}^a) = \Xi(X, Y; R)$ satisfies $\mathdcal{K}^a = \mathdcal{C}^a - \frac{1}{2} (\mathdcal{Q}^a)^2$.
Here $R$ removes any unwanted interactions generated by the product $XY$, and note that when $X$ and $Y$ are Pauli operators the first term in the polarisation map, $\frac{1}{2}(X^2+Y^2)-R$, reduces to $I-R$.

For the static \textsl{subdivision} gadget we can use $(X, Y; R) = (A_u, B_v; -1)$, for the \textsl{fork} gadget we can use $(X, Y; R) = (B_v, A_u + C_w; \tfrac{3}{2} + A_uC_w)$, and for the \textsl{cross} gadget we can use $(X, Y; R) = (A_u+C_w, B_v+D_s; 1 + A _uC_w + B_vD_s)$.
\cref{eq:polarisation-map} then directly gives the required pair $(\mathdcal{C}^a, \mathdcal{Q}^a)$ --- for a \textsl{fork} gadget, we find that $\mathdcal{C} = \frac{3}{2}I + A_u C_w$ and $\mathdcal{Q} = B_v - A_u - C_w$.

In our application of these constructions (given in Section~\ref{sec:extended-hardness-results}) we do not restrict the \emph{type} of allowed interaction.
For example, the results of Biamonte and Love~\cite{biamonte2008realizable} consider only Hamiltonians whose interactions are composed of $2$-local Pauli operators generated by $X$ and $Z$.
The only restriction on interaction type we impose is that all terms are $2$-local (and will be Pauli operators).
In Appendix~\ref{app:perturbative-gadget-extensions} we construct explicit extensions to the \textsl{subdivision}, \textsl{fork}, and \textsl{cross} gadgets that will be applied in our extended hardness results.
Additionally, we extend the four gadgets: \textsl{parameter-fixing}, \textsl{Pauli-to-Ising}, \textsl{Ising-to-XX} and \textsl{XX-to-Heisenberg}, from Ref.~\cite{schuch2009computational} to the parameterised setting.
This secondary family of gadgets also adhere to a decomposition analogous to Definition~\ref{def:general-driving-simulator}, albeit with the mediator qubit having all possible Pauli operator interactions.
However, for the sake of over-cumbering the presentation, we do not explicitly detail this expansion here and defer the technical discussion to Appendix~\ref{app:perturbative-gadget-extensions}.

The carried and gadgetised constructions can be combined at second order by treating each bulk term as a constant function of $s$ and applying \cref{lem:second-order-parameterised-reduction} to it.
We chose to present the two cases separately to keep their scale parameters distinct.

\subsubsection{Control of the Berry phase}
To conclude, we state the following corollary establishing control of the Berry phase in the context of the parameterised perturbative gadget reductions.
First, we see that $v_0, v_1 = O(\Delta^{1/2} A)$ from the assumptions above \cref{lemma:uniform-control-exact-truncated} gives
\begin{align}
    \snorm{\bs{S}} &= O(\Delta^{-1/2} A), &
    \kappa &= O(\Delta^{-1/2} A).
\end{align}
Applying a straightforward corollary of \cref{lem:carried-driving-simulation-error} (\cref{cor:gadgetised-driving-simulation-error}) gives
\begin{equation}
    \snorm{\ps \bs{E}} = O(\Delta^{-2} v_0^2 v_1) = O(\Delta^{-1/2} A^3).
\end{equation}
Thus, analogous hypotheses to those of \cref{prop:assembled-bounds-general} hold for the second-order parameterised perturbative gadget reductions (see \cref{prop:parallel-second-order-exact-assembled-bounds}), with parameter $q = 1/2$.

\begin{corollary}
\label{cor:parallel-second-order-exact-simulation-berry-phase}
    Let $\Upsilon$ be the target family of Definition~\ref{def:parameter-dependent-second-order-layer} and let $\wt{\Upsilon}$ be an exact second-order parallel simulator family of Definition~\ref{def:general-driving-simulator}, with Berry phases $\vartheta$ and $\wt{\vartheta}$ respectively.
    Suppose $A$, $B$, $\gamma_\star^{-1}$, and $\eta^{-1}$ are polynomially bounded in $n$, and fix a target Berry-phase accuracy $\delta = 1/\poly{n}$.
    Then there exist simulation targets $\varepsilon, \kappa = 1/\poly{n}$, and a penalty $\Delta = \poly{n}$, such that $\wt{\Upsilon}$ is a regular $(\eta,\varepsilon,\kappa)$-simulation of $\Upsilon$ satisfying
    \begin{equation}
        d_{2\pi}(\vartheta,\wt{\vartheta}) \le \delta .
    \end{equation}
\end{corollary}

The proof is analogous to that of Corollaries~\ref{cor:parallel-first-order-simulation-berry-phase}--\ref{cor:parallel-third-order-simulation-berry-phase}, adapted using \cref{cor:gadgetised-driving-simulation-error} and \cref{prop:parallel-second-order-exact-assembled-bounds}.

%% file: sections/extended-hardness-results.tex
\section{Extended Hardness Results}
\label{sec:extended-hardness-results}
Our reduction starts from the $5$-local Hamiltonians in the \cl{BQP}-hardness construction of Ref.~\cite{hayakawa2025computational}.
These families have the bulk-driving decomposition in Section~\ref{sec:prior-work}, with the $2$-local driving term given by \cref{eq:driving-pauli-decomposition-intro}.
Here $q$ is the qubit measured in the original \cl{BQP} circuit, and $c$ is the final clock qubit in the Feynman--Kitaev construction.

Before we establish the improved hardness results, we first discuss what parameters a valid reduction must preserve with respect to the \sc{Guided-State Berry Phase Estimation} problem. 

\subsection{Preserving Problem Parameters}
\label{sec:preserving-parameters}
A valid reduction must preserve every parameter of the problem definition.
For \sc{GSBPE} the parameters of interest are:
\begin{inparaenum}[(a)]
    \item the promise arcs,
    \item the promise gap $\g$,
    \item the guiding-state overlap $\delta_0$ and its classical description $C_\xi$, and
    \item the structural conditions (closed, gapped, non-degenerate).
\end{inparaenum}
We first record how the arcs and overlap transform under a bounded perturbation of the Berry phase and ground state, then consider how the remaining parameters are preserved under the reduction.

We map an instance $x = \langle\Upsilon, a, b, \g, C_\xi, \delta_0\rangle$ with Berry phase $\vartheta$ to $\wt{x} = \langle\wt{\Upsilon}, \wt{a}, \wt{b}, \g', C_{\wt{\xi}}, \wt{\delta}_0\rangle$ with Berry phase $\wt{\vartheta}$.
The reduction must map yes-instances to yes-instances and no-instances to no-instances, so an algorithm for $\wt{x}$ decides $x$.
Recall the disjoint arcs $I_{\tn{in}} = [a+\g, b-\g]_{2\pi}$, $I_{\tn{out}} = [b+\g, a-\g]_{2\pi}$ with $\abs{b - a} \ge 2\g$ and $\g = 1/\poly{n}$.
Suppose the reduction moves the phase by at most $\upsilon > 0$, i.e.\ $d_{2\pi}(\wt{\vartheta}, \vartheta) \le \upsilon$.
Keep the arc centres fixed, so $\wt{a} = a$ and $\wt{b} = b$, and set $\g' = \g - \upsilon$.
The resulting arcs are $\wt{I}_{\tn{in}} = [a+\g', b-\g']_{2\pi}$ and $\wt{I}_{\tn{out}} = [b+\g', a-\g']_{2\pi}$.
If $2\upsilon < \g$, these arcs remain disjoint and satisfy $|\wt{b} - \wt{a}| \ge 2\g' \ge 1/\poly{n}$; see \cref{fig:interval-preservation}.

For the guiding-state overlap we use the following elementary tracking bound.

\begin{lemma}[\cite{waite2025guided}]\label{lma:norm-tracking}
    Let $\ket{a}, \ket{b}, \ket{c} \in (\mathbb{C}^2)^{\otimes n}$ be normalised with $\norm{\ket{a} - \ket{b}} \le \epsilon_{ab}$ and $\abs{\braket{b}{c}}^2 \ge \delta_{bc}$.
    Then
    \begin{equation}
        \abs{\braket{a}{c}}^2
        \ge \big(\max\{0,\sqrt{\delta_{bc}}-\epsilon_{ab}\}\big)^2.
    \end{equation}
\end{lemma}

\begin{restatable}{proposition}{gsbpesimulatorhardness}\label{prop:gsbpe-simulator-hardness}
    Let $\Upsilon$ be a target family and $(\wt{\Upsilon}, \V)$ a regular $(\eta, \varepsilon, \kappa)$-simulation of $\Upsilon$ witnessed by $\Phi = \{\U(s)\}$ with $\kappa = \snorm{\ps\U}$, satisfying Assumptions~\ref{item:a1}--\ref{item:a3}.
    Assume that the fixed isometry $\V = \bigotimes_a \V^a$ describes a tensor product of local encoding isometries.
    Let $\upsilon$ be the Berry-phase error of \cref{thm:berry-phase-simulation} (so $d_{2\pi}(\vartheta, \wt{\vartheta}) \le \upsilon$) and $\mu = O(\eta + \gamma_\star^{-1}\varepsilon)$ the ground-state error of \cref{lem:ground-state-simulation}.
    Given a \sc{GSBPE} instance $x=\langle\Upsilon,a,b,\g,C_\xi,\delta_0\rangle$ where $C_\xi$ is a classical description of a semi-classical encoded subset state $\ket{\xi}$ and with $2\upsilon < \g$ and $2\mu \le \sqrt{\delta_0}$, let $C_{\wt\xi}$ append the classical description of $\V$ to $C_\xi$.
    Then the instance
    \begin{align}
        \wt{x} &= \big\langle \wt{\Upsilon}, a , b , \g' , C_{\wt\xi} , \delta_0' \big\rangle,
        & \g' &= \g - \upsilon,
        & \delta_0' &= (\sqrt{\delta_0} - \mu)^2 \ge \delta_0/4,
    \end{align}
    is a valid \sc{GSBPE} instance with $\wt{x}\in\sc{yes}$ whenever $x\in\sc{yes}$ and $\wt{x}\in\sc{no}$ whenever $x\in\sc{no}$.
    Hence, \sc{GSBPE} for $\wt{\Upsilon}$ is at least as hard as for $\Upsilon$ in the sense of polynomial-time reducibility.
\end{restatable}

A proof is given in Appendix~\ref{app:extended-hardness-results-proofs}.

In the reductions we construct that follow Ref.~\cite{oliveira2008complexity}, the encoding isometries are local and correspond to ``attaching'' polynomially many $\ket{0}$ ancilla qubits to the guiding state~\cite{cade2023improved,waite2026physically}.
More concretely, if $\ket{\xi}$ is the guiding state for the target Hamiltonian, then the guiding state for the simulator Hamiltonian is $\ket{\wt{\xi}} = \ket{\xi}\ket{0^m}$, where $m = |{\rm med}|$ is the number of ancilla qubits introduced by the local encoding isometries.
The reductions that follow Ref.~\cite{schuch2009computational} also employ local encoding isometries that attach ancilla qubits to the guiding state, though the specific structure of ancilla qubits may differ depending on the details of the perturbative gadget construction; we outline these states in \cref{eq:ancilla-qubit-states-sv} and also in Appendix~\ref{app:perturbative-gadget-extensions} (with their corresponding isometries).

\begin{figure}[!ht]
    \centering
    \begin{tikzpicture}
        \pic[scale=0.75]{latticepaths};
    \end{tikzpicture}
    \caption{Illustration of the geometric construction used in the analysis showing the arrangement of system qubits and the rerouting of connections between them.
    Dashed lines indicate the original connections between system qubits and bold lines indicate the rerouted connections.}
    \label{fig:geometric-embedding}
\end{figure}

\subsection{Improved Hardness to 2-Local Hamiltonians}
\label{sec:improved-hardness-to-2-local}
We proceed to prove our main result, showing that \sc{GSBPE} is \clw{BQP}{complete} for families of $2$-local Hamiltonians on square and triangular lattices.
The first step is to prove that spatially-sparse $5$-local Hamiltonians are \clw{BQP}{complete}.
The resulting Hamiltonians are of an identical form to those in Ref.~\cite{hayakawa2025computational}, but with a spatially-sparse interaction graph.

\begin{definition}[Spatially Sparse Graph~\cite{oliveira2008complexity}]
    A graph is spatially sparse if
    \begin{inparaenum}[(i)]
        \item every vertex is incident to $O(1)$ edges;
        \item the graph has a straight-line drawing in the plane in which every edge crosses $O(1)$ other edges and has length $O(1)$.
    \end{inparaenum}
\end{definition}

\begin{restatable}{theorem}{spatiallysparsefivelocal}\label{theorem:spatially-sparse-5-local}
    The $(2, 5, \delta_0)$-\sc{GSBPE} problem is \clw{BQP}{complete} for any $\delta_0 \in [1/\poly{n}, 1 - 1/\poly{n}]$, when the guiding state is a semi-classical encoded subset state and the Hamiltonians have a spatially-sparse interaction graph.
\end{restatable}

A short proof is provided in Appendix~\ref{app:extended-hardness-results-proofs}; we sketch the main ideas here.
Apply the spatially sparse reduction of Ref.~\cite{oliveira2008complexity} to the \cl{BQP}-hardness construction of Ref.~\cite{hayakawa2025computational}.
This preserves the form of the Hamiltonian (keeping $\w(s)$ effectively unchanged) while making its interaction graph spatially sparse (the $5$-local bulk interactions are now spatially sparse), and the arguments of Ref.~\cite{waite2026physically} preserve the guiding-state overlap and classical description.

Starting from this result, we follow the geometric reductions of Refs.~\cite{oliveira2008complexity,cubitt2016complexity,piddock2017complexity,waite2025complexitya,waite2026physically} to obtain $2$-local interactions on square and triangular lattices.

\begin{restatable}[Main theorem]{theorem}{maintheorem}\label{thm:main-theorem}
    The $(2,2,\delta_0)$-\sc{GSBPE} problem is \clw{BQP}{complete} for any $\delta_0 \in [1/\poly{n}, 1 - 1/\poly{n}]$, when the guiding state is a semi-classical encoded subset state and the Hamiltonians are restricted to square or triangular lattices.
\end{restatable}

A proof is given in Appendix~\ref{app:extended-hardness-results-proofs}; we sketch the main ideas here.
Starting from the spatially sparse $5$-local Hamiltonians of \cref{theorem:spatially-sparse-5-local}, apply the \textsl{subdivision} gadget to reduce the bulk locality from $5$ to $3$~\cite{kempe2006complexity,oliveira2008complexity}.
Then apply the \textsl{3-to-2-local} gadget to reduce the bulk locality from $3$ to $2$~\cite{oliveira2008complexity}.
These are static gadget templates, while the present reduction carries their $s$-dependent terms using the analysis of Section~\ref{sec:perturbative-gadget-reductions}.

Next apply geometric reductions that control Pauli degree, remove non-planar crossings, and embed the resulting planar interaction graph into the square or triangular lattice.
Recall that the Pauli degree of a qubit is the triple $({\rm deg}_X, {\rm deg}_Y, {\rm deg}_Z)$, where ${\rm deg}_\sigma$ counts the number of $\sigma$-type Pauli operators acting on the qubit.
Figure~\ref{fig:gadgets-for-graph-elements} illustrates the gadgets used for these graph elements, and Appendix~\ref{app:perturbative-gadget-extensions} is intended to give their parameter-dependent constructions.
Subject to those constructions, Section~\ref{sec:perturbative-gadget-reductions} controls the Berry-phase error, while Ref.~\cite{waite2026physically} preserves the guiding-state overlap and classical description.
Proposition~\ref{prop:gsbpe-simulator-hardness} then transfers hardness to the simulator family.

Both \cref{theorem:spatially-sparse-5-local,thm:main-theorem} also hold for semi-classical subset states because the encoding isometries are trivial.
We state the results for semi-classical encoded subset states to accommodate future reductions with non-trivial (but not complicated) encodings, such as those of Ref.~\cite{schuch2009computational}.

Together, these reductions establish \cref{thm:main-theorem}; the parameter-dependent gadget constructions are given in Appendix~\ref{app:perturbative-gadget-extensions}.

We now extend the hardness result to Hamiltonians generated by so-called \emph{weighted Heisenberg interactions} on square or triangular lattices.
Take $\boldsymbol{\alpha} = (\alpha_1, \alpha_2, \alpha_3) \in \mathbb{R}^3$ and define the $2$-local interaction $W(\boldsymbol{\alpha}) = \alpha_1 XX + \alpha_2 YY + \alpha_3 ZZ$ as the weighted Heisenberg interaction.
Suppose that $\boldsymbol{\alpha}(s)$ is a parameter-dependent vector function of $s$, such that $\alpha_i(s)$ are ${\rm C}^2$ and closed; we say $\boldsymbol{\alpha}(s)$ is therefore ${\rm C}^2$ and closed.
Starting from the $2$-local lattice Pauli Hamiltonians obtained from \cref{thm:main-theorem}, we can apply the gadget chain reduction framework of Schuch and Verstraete~\cite{schuch2009computational} to obtain Hamiltonians generated by parameter-dependent weighted Heisenberg interactions.

\begin{restatable}[Main corollary]{corollary}{maincorollary}\label{cor:main-corollary}
    The $(2,2,\delta_0)$-\sc{GSBPE} problem is \clw{BQP}{complete} for any $\delta_0 \in [1/\poly{n}, 1 - 1/\poly{n}]$, when the guiding state is a semi-classical encoded subset state and the Hamiltonian terms are generated by weighted Heisenberg interactions, with local fields, restricted to square or triangular lattices.
\end{restatable}

The proof of this corollary follows by adapting the proof of \cref{thm:main-theorem} to the reduction chain of Ref.~\cite{schuch2009computational} and applying the parameterised extensions of the \textsl{parameter-fixing}, \textsl{Pauli-to-Ising}, \textsl{Ising-to-XX}, and \textsl{XX-to-Heisenberg} gadgets, as detailed in Appendix~\ref{app:perturbative-gadget-extensions}.
The reduction to the square lattice is achieved via arguments in Ref.~\cite{schuch2009computational} and reductions to the triangular lattice follow from Ref.~\cite{waite2026physically}.
Moreover, the preservation of the guiding-state overlap and classical description is ensured by the constructions in Ref.~\cite{waite2026physically}; in this case, the local isomatries still correspond to the additonal of single ancilla qubits in state such as
\begin{align}\label{eq:ancilla-qubit-states-sv}
    &\ket{0}, &
    &\frac{1}{\sqrt{2}}(\ket{0} - \i \ket{1}), &
    &\frac{1}{\sqrt{2}}(\ket{0} + \e^{\i \pi/4} \ket{1}),
\end{align}
which are trivially handled using the arguments in Refs.~\cite{cade2023improved,waite2025guided,waite2026physically} and remain semi-classical encoded subset states.

%% file: sections/berry-phase-estimation-parameterised-1-local.tex
\section{Berry Phase Estimation for Parameterised Families of 1-Local Hamiltonians}
\label{sec:berry-phase-estimation-parameterised-1-local}
In this section, we prove that the Berry phase of a $1$-local Hamiltonian family can be estimated classically to inverse-polynomial precision under the ${\rm C}^2$ assumptions of \cref{sec:preliminaries}.
The algorithm approximates the Berry phase by a gauge-invariant cyclic product of one-qubit ground-state projectors.

For each $s$, collect all terms acting non-trivially on qubit $j$ into a single-qubit Hamiltonian and remove its scalar part.
This gives
\begin{equation}\label{eq:one-local-traceless-decomposition}
    H(s)=c(s)I+\sum_{j=1}^n g_j(s),
\end{equation}
where each $g_j(s)$ is a traceless Hermitian operator acting on qubit $j$.
The scalar term $c(s)I$ does not affect the ground-state projectors or the Berry phase.
Let $P_j(s)$ be the ground-state projector of $g_j(s)$, and let $\gamma_j(s)$ be its spectral gap.
Since the terms in \cref{eq:one-local-traceless-decomposition} act on distinct qubits, the unique ground-state projector and the spectral gap of $H(s)$ are
\begin{align}\label{eq:one-local-product-ground-space}
    P(s)&=\bigotimes_{j=1}^n P_j(s), &
    \gamma(s)&=\min_{j\in[n]}\gamma_j(s).
\end{align}
In particular, $\gamma_j(s)\geq\gamma_\star$ for every $j$ and $s$, where $\gamma_\star$ is the lower bound on the spectral gap of the single-qubit terms.
Note that $\gamma_\star \geq 1/\poly{n}$ can be encoded using $\poly{n}$ bits.

Let $L_r$ be the upper bound on $\snorm{\partial_s^rH}$ for $r\in\{1,2\}$.
The uniqueness of the decomposition into a scalar part and traceless single-qubit parts implies
\begin{equation}\label{eq:one-local-derivative-scales}
    \snorm{\partial_s^r g_j}\leq L_r
\end{equation}
for every $j\in[n]$ and $r\in\{1,2\}$.
The operators $\partial_s^r g_j(s)$ commute, are traceless, and act on distinct qubits, so the spectral width of their sum is twice the sum of their norms.
This spectral width is unchanged by the scalar part of $\ps^r H(s)$ and is at most $2\norm{\partial_s^rH(s)}$.
Closedness of $H(s)$ and uniqueness of the traceless decomposition imply $g_j(1)=g_j(0)$ and hence $P_j(1)=P_j(0)$ for every $j$.

Note that since each $P_j(s)$ has rank one, it determines a line rather than a single vector.
That is, any two normalised vectors spanning ${\rm Im}(P_j(s))$ differ by a phase.
We must therefore choose a specific normalised vector spanning ${\rm Im}(P_j(s))$ at each $s$; we require this choice to be smooth in $s$ and to return to itself at $s=1$.
Such a choice always exists in this case.

Fix any normalised vector $\ket{u_j(0)}$ spanning ${\rm Im}(P_j(0))$ and let $\ket{u_j(s)}$ solve 
\begin{equation}\label{eq:one-local-parallel-transport-ode}
    \ps\ket{u_j(s)} = \comm{\ps P_j(s)}{P_j(s)}\ket{u_j(s)} .
\end{equation}
Such a solution exists and is unique, and it remains a normalised vector spanning ${\rm Im}(P_j(s))$ for all $s$.
Moreover, it is parallel-transported in the sense that $\braket{u_j(s)}{\ps u_j(s)}=0$.
However, the vector $\ket{u_j(s)}$ need not return to itself at $s=1$; it may acquire a phase, giving
\begin{equation}
    \ket{u_j(1)} = \e^{\i\vartheta_j}\ket{u_j(0)},
\end{equation}
where $\vartheta_j\in\mathbb{R}$ is the local Berry phase.
We can ``unwind'' this phase by defining a new vector
\begin{equation}
    \ket{\phi_{0,j}(s)} \coloneqq \e^{-\i\vartheta_j s}\ket{u_j(s)},
\end{equation}
which is smooth in $s$ and returns to itself at $s=1$.
Therefore, the vector $\ket{\phi_{0,j}(s)}$ provides a smooth, closed choice of representative for the line determined by $P_j(s)$.

\begin{remark}
    Note that $\vartheta_j$ is identified with the usual connection integral (cf. \cref{eq:berry-phase}).
    Since $\braket{u_j(s)}{\ps u_j(s)} = 0$, differentiating $\ket{\phi_{0,j}(s)} = \e^{-\i\vartheta_j s}\ket{u_j(s)}$ gives
    \begin{equation}\label{eq:one-local-constant-connection}
        \i\mel{\phi_{0,j}(s)}{\ps}{\phi_{0,j}(s)}
        = \i\big(-\i\vartheta_j + \braket{u_j(s)}{\ps u_j(s)}\big)
        = \vartheta_j ,
    \end{equation}
    so the local Berry connection is constant and its integral over $[0,1]$ equals $\vartheta_j$.
\end{remark}

Recall that $\ket{\phi_0(s)}=\bigotimes_{j=1}^n\ket{\phi_{0,j}(s)}$ is a closed ground-state section for $H(s)$, and its Berry connection satisfies
\begin{equation}\label{eq:one-local-berry-phase-sum}
    \i\mel{\phi_0(s)}{\partial_s}{\phi_0(s)}
    =\sum_{j=1}^n\i\mel{\phi_{0,j}(s)}{\partial_s}{\phi_{0,j}(s)}.
\end{equation}
Consequently, the global Berry phase is the sum of the local Berry phases modulo $2\pi$.

Fix an integer $N\geq2$ and define the mesh points $s_\ell=\ell/N$ for $\ell\in\{0,\ldots,N-1\}$.
For each qubit $j$, define the cyclic Bargmann invariant~\cite{bargmann1964note,rabei1999bargmann}
\begin{equation}\label{eq:one-local-bargmann-invariant}
    \varDelta_{j,N}\coloneqq \Tr\left[P_j(s_0)P_j(s_{N-1})\cdots P_j(s_1)\right],
\end{equation}
and set
\begin{align}\label{eq:one-local-discrete-phase}
    \varDelta_N &\coloneqq \prod_{j=1}^n \varDelta_{j,N}, &
    \vartheta_N &\coloneqq \arg \varDelta_N \pmod{2\pi}.
\end{align}
If $\ket{\phi_{0,j}^{[\ell]}}$ is any normalised vector in the image of $P_j(s_\ell)$ and $\ket{\phi_{0,j}^{[N]}}\coloneqq\ket{\phi_{0,j}^{[0]}}$, then
\begin{equation}\label{eq:one-local-bargmann-overlaps}
    \varDelta_{j,N}
    =\prod_{\ell=0}^{N-1}\braket{\phi_{0,j}^{[\ell+1]}}{\phi_{0,j}^{[\ell]}}.
\end{equation}
The phase of \cref{eq:one-local-bargmann-overlaps} is invariant under an independent rephasing of every sampled eigenvector since all phase factors cancel around the cycle.
The cyclic definition also uses only mesh points in $[0,1)$; closure is enforced by the condition $P_j(1) = P_j(0)$.

We next bound the discretisation error using the upper bounds on the first two derivatives of the Hamiltonian.

\begin{restatable}{lemma}{oneLocalDiscreteBerryPhase}\label{lem:one-local-discrete-bp}
    Define
    \begin{equation}\label{eq:one-local-second-derivative-scale}
        M_2\coloneqq \gamma_\star^{-1} L_2 + 5\gamma_\star^{-2} L_1^2,    
    \end{equation}
    with $\overline{M}_2 \coloneqq \max\{1,M_2\}$.
    If $N \geq 2n \overline{M}_2$, then $\varDelta_N \neq 0$ and
    \begin{align}\label{eq:one-local-discrete-bp-bound}
        d_{2\pi}(\vartheta_N,\vartheta) & \leq \frac{n M_2}{N}, &
        \abs{\varDelta_N} & \geq 1 - \frac{n M_2}{2 N} \geq \frac{3}{4}.
    \end{align}
\end{restatable}

A proof is given in Appendix~\ref{app:berry-phase-estimation-parameterised-1-local-proofs}.  

\begin{remark}
    \Cref{lem:one-local-discrete-bp} is the statement that \cref{eq:one-local-discrete-phase} computes the Berry phase defined in \cref{eq:berry-phase}. 
    Combining \cref{eq:one-local-berry-phase-sum,eq:one-local-constant-connection} gives $\vartheta \equiv \sum_j \vartheta_j \pmod{2\pi}$, while \cref{eq:one-local-discrete-continuous-comparison} gives $\arg\varDelta_N \equiv \sum_j\vartheta_j + \sum_{j,\ell}\arg a_{j,\ell}$. 
    The bound $\abs{\arg a_{j,\ell}}\le M_2/N^2$ then controls the error between the two by $n M_2/N$, which vanishes as the mesh is refined in the limit $N\to\infty$.
\end{remark}

\subsection{Main Result}\label{sec:main-result-1l}

We now state the following theorem which concludes the existence of the polynomial-time algorithm for deciding the $(2,1,\delta_0)$-\sc{GSBPE} problem.

\begin{restatable}{theorem}{oneLocalBerryPhaseInP}\label{thm:one-local-berry-phase-in-p}
    For every $\delta_0 \in [0,1]$ and any choice of guiding state description, the $(2,1,\delta_0)$-\sc{GSBPE} problem is in \cl{P}.
\end{restatable}

A proof is given in Appendix~\ref{app:berry-phase-estimation-parameterised-1-local-proofs}.  

%% file: sections/conclusion.tex
\section{Conclusion}
\label{sec:conclusion}
In this work, we have demonstrated that estimating the Berry phase for parameterised families of $2$-local Hamiltonians restricted to two-dimensional lattice geometries is \clw{BQP}{complete} when provided with a classical description of the guiding state that is promised to overlap with the ground state.
For the case of $1$-local Hamiltonians, estimating the Berry phase to inverse-polynomial precision can be done efficiently on a classical computer, placing the problem in \cl{P}.
To establish our results for the $2$-local case, we extended the Schrieffer--Wolff framework for perturbative gadget constructions to handle parameterised families of Hamiltonians. 
This allows us to construct simulator Hamiltonians whose Berry phase is approximately equal to that of the target Hamiltonian, up to an error that can be made arbitrarily small by choosing sufficiently large, but polynomially-large interactions.

The perspective offered by this work is primarily complexity-theoretic, treating the interest in the Berry phase from the standpoint of computational difficulty rather than by experimental considerations.
Though, we were motivated by physical considerations such as geometrical restrictions to the qubit interactions and the locality thereof, to understand the computational complexity in so-called physically-motivated settings.\footnote{Physically-motivated settings in this context typically refer to $2$-local qubit Hamiltonians, often restricted to specific geometries and interaction types such as antiferromagnetic Heisenberg or $XY$, as a way of modelling realistic physical systems studied in many-body physics.}
Our results extend the prior work of Hayakawa, Sakamoto and Kiumi~\cite{hayakawa2025computational} who first introduced the study of the computational complexity of estimating the Berry phase for local Hamiltonians.
It was proven in Ref.~\cite{hayakawa2025computational} that estimating the Berry phase for $5$-local Hamiltonians is \clw{BQP}{complete}, given a classical description of the guiding state that is promised to overlap with the ground state (among two other computational tasks established but not considered in this work).
In addition to improving results on this new physical problem, our work extends the ongoing literature studying the computational complexity of guided Hamiltonian problems~\cite{richter2007two,gharibian2023dequantizing,cade2022complexity,cade2023improved,waite2025guided,waite2026physically}.

We anticipate that future work will continue to explore the boundaries of classical and quantum computational capabilities in estimating the Berry phase, potentially relaxing some of the assumptions made in this study and extending the results to broader classes of Hamiltonians.
The (static) perturbative gadget literature often assumes the simulation-error parameters $\eta$ and $\varepsilon$ to be sufficiently small (typically $\eta, \varepsilon \leq 1/\poly{n}$) and subsequently excludes a rigorous error analysis.
We believe that given our parameter-dependent analysis, analogous assumptions can be supposed for $\eta$, $\varepsilon$ and $\kappa$ moving forward, i.e., we can assume $\eta, \varepsilon, \kappa \leq 1/\poly{n}$, for a sufficiently large polynomial in the system size.
In addition to complexity-theoretic considerations, we expect the ideas of this work may have use in the context of theoretical physics, particularly in understanding the role of the Berry phase in many-body quantum systems and its implications for quantum simulation and computation.

\subsection{Open Questions}\label{sec:open-questions}
We outline several open questions and potential directions for future research.
\begin{enumerate}
    \item Apply our framework to the other two Berry-phase problems defined in Ref.~\cite{hayakawa2025computational}.
    In particular, we conjecture that the Berry phase with an energy threshold problem is \clw{dUQMA}{complete} for $2$-local systems, but a detailed reduction remains open.
    \item Extend the analysis to the Fermi--Hubbard Hamiltonian using the perturbative gadget framework of Ref.~\cite{schuch2009computational} and arguments in Ref.~\cite{waite2026physically}.
    We conjecture \clw{BQP}{completeness} for a weighted Fermi--Hubbard Hamiltonian with local fields, but a detailed reduction remains open.
    \item Extend the framework to the Heisenberg-type gadgets of Ref.~\cite{piddock2017complexity} and the logical-encoding gadgets of Ref.~\cite{cubitt2016complexity}.
    The principal obstacle is their doubly degenerate ground spaces.
    \item Study whether the complexity landscape changes for exponential-precision estimation of the Berry phase~\cite{fefferman2016quantum,deshpande2022importance}.
    \item Determine whether classical techniques efficiently approximate the Berry phase for restricted parameter choices.
    Guided Hamiltonian problems often admit classical approximation methods~\cite{gharibian2023dequantizing,waite2025guided,waite2026physically}, but their applicability to the Berry phase remains unclear.
    \item Investigate the potential for classical verification of the Berry phase in systems with succinctly describable ground states~\cite{resta1994macroscopic,jiang2025local,waite2025complexityb}.
    Supposing we have amplitude oracles $A_j(x) = \braket{x}{\phi_j}$, then the overlap $O_j = \braket{\phi_j}{\phi_{j+1}} = \mathbb{E}_{|A_j(x)|^2}[A_{j+1}(x)/A_j(x)]$; it then follows that $\mathbb{E}[\abs{A_{j+1}(x)/A_j(x)}^2] = 1$ and therefore ${\rm Var}(A_{j+1}(x)/A_j(x)) \le 1$, so that $O_j$ can be estimated to additive error $\varepsilon$ with $O(1/\varepsilon^2)$ samples.
    \item Investigate the counting complexity of multiplicative approximations to the Berry phase, given the relationship between additive \clw{BQP}{hardness} and multiplicative \clw{\#P}{hardness}~\cite{kuperberg2015hard}.
    \item Develop a more general theory of the parameter-dependent Schrieffer--Wolff transformation used here.
    Such a theory could clarify the relationship among geometric phases, perturbation theory, and quantum many-body physics~\cite{bukov2016schrieffer}.
\end{enumerate}

Resolving these questions would clarify which parts of the proposed $1$-local/$2$-local complexity boundary persist under restricted interactions, weaker guiding information, and more general low-energy encodings.

%% file: appendices/section-proofs/preliminaries-proofs.tex
\section{Proofs of Main Results in Section~\ref{sec:preliminaries}}\label{app:preliminaries-proofs}

\begin{proposition}\label{prop:berry-phase-gauge-invariance}
    Let $\Upsilon = \{H(s)\}_{s\in[0,1]}$ be a ${\rm C}^2$ family of Hamiltonians such that $H(0) = H(1)$ and the ground state is non-degenerate for all $s \in [0,1]$ and the spectral gap satisfies $\gamma(s) \geq \gamma_\star > 0$.
    Let $\ket{\phi_0(s)}$ be a smooth, normalised, periodic choice of ground state for $H(s)$.
    Then
    \begin{equation}
        \vartheta = \int_0^1 \dd{s} \ \i\mel{\phi_0(s)}{\partial_s}{\phi_0(s)}
    \end{equation}
    is invariant modulo $2\pi$ under the gauge transformation $\ket{\phi_0(s)} \mapsto \e^{\i\alpha(s)}\ket{\phi_0(s)}$ for any smooth function $\alpha : [0,1] \to \mathbb{R}$ satisfying $\alpha(1) - \alpha(0) = 2\pi k$ for some integer $k$.
\end{proposition}

\begin{proof}
    Under the gauge transformation $\ket{\phi_0(s)} \mapsto \e^{\i\alpha(s)}\ket{\phi_0(s)}$, we have
    \begin{align}
        \mel{\phi_0(s)}{\partial_s}{\phi_0(s)} &\mapsto \mel{\e^{\i\alpha(s)}\phi_0(s)}{\partial_s}{\e^{\i\alpha(s)}\phi_0(s)} \\
        &= \mel{\phi_0(s)}{\partial_s}{\phi_0(s)} + \i \dot{\alpha}(s).
    \end{align}
    Here, $\dot{\alpha}(s) = \dv{\alpha(s)}{s}$ denotes the derivative of $\alpha(s)$ with respect to $s$.
    Therefore, the Berry connection and phase transform as
    \begin{align}
        \i\mel{\phi_0(s)}{\partial_s}{\phi_0(s)}
        &\mapsto \i\mel{\phi_0(s)}{\partial_s}{\phi_0(s)}-\dot\alpha(s),\\
        \vartheta&\mapsto\vartheta-\alpha(1)+\alpha(0).
    \end{align}
    Since $\alpha(1)-\alpha(0)=2\pi k$, we have $\vartheta\mapsto\vartheta-2\pi k$, which is equivalent to $\vartheta$ modulo $2\pi$.
\end{proof}

\liouvilleProperties*

\begin{proof}
    We prove each statement separately.

    \medskip
    \noindent\emph{Item 1.}
    By the definitions of $\bs{X}$, $\bs{Y}$ and the Liouville superoperator, we have
    \begin{equation}
        \L_{\bs{X}}(\bs{Y}) = \comm{\bs{X}}{\bs{Y}} = \sum_{a,b \in {\rm A}} \comm{X_a}{Y_b} = \sum_{a,b \in {\rm A}} \L_{X_a}(Y_b).
    \end{equation}

    \medskip
    \noindent\emph{Item 2.}
    Suppose that $\supp{X_a}\cap\supp{Y_b}=\emptyset$ whenever $a\neq b$. 
    Operators with disjoint support commute; we may therefore express them as $X_a = \widehat{X}_a \otimes I$ and $Y_b = I \otimes \widehat{Y}_b$ relative to a tensor-product decomposition containing their respective supports.
    It follows that $\comm{X_a}{Y_b}=0$ whenever $a\neq b$.
    Applying the first statement, we obtain
    \begin{equation}
        \L_{\bs{X}}(\bs{Y}) = \sum_{a,b \in {\rm A}} \L_{X_a}(Y_b) = \sum_{a \in {\rm A}} \L_{X_a}(Y_a) + \sum_{\substack{a,b \in {\rm A} \\ a \neq b}} \L_{X_a}(Y_b),
    \end{equation}
    where the terms with $a\neq b$ vanish due to the disjoint support.

    \medskip
    \noindent\emph{Item 3.}
    Apply the first statement to $\bs{X}$ and $\bs{Z}$ to obtain
    \begin{equation}
        \L_{\bs{X}}(\bs{Z}) = \sum_{a \in {\rm A}} \L_{X_a}(Z_a) + \sum_{\substack{a,b \in {\rm A} \\ a \neq b}} \L_{X_a}(Z_b).
    \end{equation}
    The second sum vanishes under the assumption that $\comm{X_a}{Z_b}=0$ whenever $a\neq b$.
    The hypothesis that $\L_{X_a}(Z_a) = Y_a$ for all $a \in {\rm A}$ then implies
    \begin{equation}
        \L_{\bs{X}}(\bs{Z}) = \sum_{a \in {\rm A}} Y_a = \bs{Y}.
    \end{equation}

    \medskip
    \noindent\emph{Item 4.}
    Let $X = \sum_{\lambda \in \sigma(X)} \lambda P_\lambda$ be the spectral decomposition of $X$, where $P_\lambda$ are the corresponding spectral projectors.
    By definition 
    \begin{equation}
        \L^{\pattce}_X(Y) = \sum_{\substack{\lambda, \mu \in \sigma(X) \\ \lambda \neq \mu}} \frac{P_\lambda Y P_\mu}{\lambda - \mu}.
    \end{equation}
    Applying this definition to $\bs{Y} = \sum_{a \in {\rm A}} Y_a$, we obtain
    \begin{equation}
        \L^{\pattce}_{X}(\bs{Y}) = \sum_{\substack{\lambda, \mu \in \sigma(X) \\ \lambda \neq \mu}} \frac{P_\lambda \left(\sum_{a \in {\rm A}} Y_a\right) P_\mu}{\lambda - \mu} = \sum_{a \in {\rm A}} \sum_{\substack{\lambda, \mu \in \sigma(X) \\ \lambda \neq \mu}} \frac{P_\lambda Y_a P_\mu}{\lambda - \mu} = \sum_{a \in {\rm A}} \L^{\pattce}_X(Y_a).
    \end{equation}
\end{proof}

\liouvillePseudoinverseBound*

\begin{proof}
    For each $k \in \{0,\ldots,\abs{A}\}$, let $Q_k$ denote the projector onto the subspace containing exactly $k$ excited qubits.
    Then
    \begin{align}
        \bs{X} &= \sum_{k=0}^{\abs{A}} k Q_k,& 
        Q_0 &= \bs{P},& 
        \bs{Q} &= \sum_{k=1}^{\abs{A}} Q_k.
    \end{align}
    Define
    \begin{equation}
        G \coloneqq \sum_{k=1}^{\abs{A}} \frac{1}{k} Q_k.
    \end{equation}
    Since the nonzero eigenvalues of $\bs{X}$ are at least one, we have $\norm{G}\leq 1$.

    Because $O = \bs{P} O \bs{Q} + \bs{Q} O \bs{P}$, the definition of the Liouville pseudoinverse gives
    \begin{equation}
        \L_{\Delta\bs{X}}^{\pattce}(O) = \Delta^{-1}\big( G O \bs{P} - \bs{P} O G \big).
    \end{equation}
    Relative to the decomposition $\mathcal{H} = {\rm Im}(\bs{P}) \oplus {\rm Im}(\bs{Q})$, these operators have the block forms
    \begin{align}
        O &= 
        \begin{pmatrix}
            0 & \bs{P}O\bs{Q} \\
            \bs{Q}O\bs{P} & 0
        \end{pmatrix}, &
        \L_{\Delta\bs{X}}^{\pattce}(O) &= \Delta^{-1} 
        \begin{pmatrix}
            0 & -\bs{P}OG \\ 
            GO\bs{P} & 0
        \end{pmatrix}.
    \end{align}
    For any block-off-diagonal operator $\left(\begin{smallmatrix} 0&B\\ C&0 \end{smallmatrix}\right)$, its norm is $\max\{\norm{B},\norm{C}\}$.
    Therefore,
    \begin{equation}
        \norm{\L_{\Delta\bs{X}}^{\pattce}(O)}
        = \Delta^{-1}\max\left\{\norm{\bs{P}OG},\norm{GO\bs{P}}\right\}
        \leq \Delta^{-1}\max\left\{\norm{\bs{P}O\bs{Q}},\norm{\bs{Q}O\bs{P}}\right\}
        = \Delta^{-1}\norm{O}.
    \end{equation}
\end{proof}

\moorePenrosePseudoinverseBlockDiagonalOperator*

\begin{proof}
    Recall that the Moore-Penrose pseudoinverse of an operator $X$ satisfies:
    \begin{enumerate}
        \item $X X^{\pattce} X = X$;
        \item $X^{\pattce} X X^{\pattce} = X^{\pattce}$;
        \item $(X X^{\pattce})^\dagger = X X^{\pattce}$;
        \item $(X^{\pattce} X)^\dagger = X^{\pattce} X$.
    \end{enumerate}
    By definition, 
    \begin{equation}
        X = \begin{pmatrix}
            0 & 0 \\
            0 & X_Q
        \end{pmatrix} = 0 \oplus X_Q,
    \end{equation}
    where $X_Q = Q X Q$ is the restriction of $X$ to the subspace $\mathcal{Q}$.
    Similarly, the proposed operator $\Sigma = Q X_Q^{-1} Q$ can be expressed as
    \begin{equation}
        \Sigma = \begin{pmatrix}
            0 & 0 \\
            0 & X_Q^{-1}
        \end{pmatrix} = 0 \oplus X_Q^{-1}.
    \end{equation}
    We now verify that $\Sigma$ satisfies the four Moore-Penrose conditions for $X$ in turn.

    \medskip
    \noindent\textsl{Condition 1.}
    $X \Sigma X = (0 \oplus X_Q)(0 \oplus X_Q^{-1})(0 \oplus X_Q) = 0 \oplus X_Q = X$.

    \medskip
    \noindent\textsl{Condition 2.}
    $\Sigma X \Sigma = (0 \oplus X_Q^{-1})(0 \oplus X_Q)(0 \oplus X_Q^{-1}) = 0 \oplus X_Q^{-1} = \Sigma$.

    \medskip
    \noindent\textsl{Condition 3.}
    $(X \Sigma)^\dagger = ((0 \oplus X_Q)(0 \oplus X_Q^{-1}))^\dagger = (0 \oplus I)^\dagger = 0 \oplus I = X \Sigma$.

    \medskip
    \noindent\textsl{Condition 4.}
    $(\Sigma X)^\dagger = ((0 \oplus X_Q^{-1})(0 \oplus X_Q))^\dagger = (0 \oplus I)^\dagger = 0 \oplus I = \Sigma X$.

    The operator $\Sigma$ thus satisfies all four Moore-Penrose conditions for $X$, and is therefore the Moore-Penrose pseudoinverse of $X$.
\end{proof}

%% file: appendices/section-proofs/uniform-simulation-proofs.tex
\section{Proofs of Main Results in Section~\ref{sec:uniform-simulation}}\label{app:uniform-simulation-proofs}
In this appendix, we include some additional supporting proofs for the main results presented in Section~\ref{sec:uniform-simulation}.
We first outline these supporting results then proceed with the proofs of the main results.

\subsection{Supporting Results}

\begin{restatable}{claim}{rankimagepointwise}\label{claim:rank-image-pointwise}
    With $A(s) = \U_2(s)P_{N,1}(s)\U_2^\dagger(s)$ and $B(s) = P_{N,2}(s)$ on $\wt\Hil_2$:
    \begin{enumerate}
        \item[(i)] $A(s), B(s)$ are orthogonal projectors of equal rank $N$, with ${\rm Im}(A(s)) = \U_2(s)\mathcal{E}_N(\wt{H}_1(s))$;
        \item[(ii)] if $T(s)$ is unitary with $T(s)A(s)T^\dagger(s) = B(s)$, then $\mathcal{E}_N(\wt{H}_2(s)) = T(s)\U_2(s)\mathcal{E}_N(\wt{H}_1(s))$.
    \end{enumerate}
\end{restatable}

\begin{proof}
    We suppress $s$ and prove each item separately.

    \medskip
    \noindent\textsl{(i).}
    $A^\dagger = A$ since $P_{N,1}$ is Hermitian, and $A^2 = \U_2 P_{N,1}\U_2^\dagger\U_2 P_{N,1}\U_2^\dagger = \U_2 P_{N,1}^2\U_2^\dagger = A$ using $\U_2^\dagger\U_2 = I$, so $A$ is an orthogonal projector.
    As $\U_2^\dagger : \wt\Hil_2 \to \wt\Hil_1$ is surjective, $w = \U_2^\dagger v$ ranges over $\wt\Hil_1$ as $v$ ranges over $\wt\Hil_2$, so ${\rm Im}(A) = \U_2(P_{N,1}\wt\Hil_1) = \U_2\mathcal{E}_N(\wt{H}_1)$, of dimension $N$ since $\U_2$ is an isometry.

    \medskip
    \noindent\textsl{(ii).}
    For unitary $T$ with $TAT^\dagger = B$, ${\rm Im}(B) = \{TAw : w \in \wt\Hil_2\} = T\,{\rm Im}(A)$, so $\mathcal{E}_N(\wt{H}_2) = {\rm Im}(B) = T\U_2\mathcal{E}_N(\wt{H}_1)$.
\end{proof}

\subsection{Main Results}
\eigenvaluesimulation*

\begin{proof}
    We suppress the argument $s$.
    By \ref{item:c1}--\ref{item:c2}, $\U^\dagger\wt{H}\U$ is unitarily equivalent to the restriction of $\wt{H}$ to its low band, so its eigenvalues are $\wt{\lambda}_0, \dots, \wt{\lambda}_{2^n-1}$; by \ref{item:c2} it is within $\varepsilon$ of $H$ in norm, and Weyl's inequality gives the claim.
\end{proof}

\groundstatesimulation*

\begin{proof}
    We suppress the argument $s$.
    By \cref{lem:eigenvalue-simulation} and $2\varepsilon < \gamma_\star$, the low band has a non-degenerate minimum.
    Writing $\ket{\psi_0} = \U^\dagger\ket{\wt{\phi}_0}$ (the ground state of $\U^\dagger\wt{H}\U$), the Davis--Kahan bound and an appropriate relative phase give $\norm{\ket{\psi_0} - \ket{\phi_0}} = O(\gamma_\star^{-1}\varepsilon)$, and then
    \begin{equation}
        \norm{\ket{\wt{\phi}_0} - \V\ket{\phi_0}} \le \norm{(\U - \V)\ket{\phi_0}} + \norm{\U(\ket{\psi_0} - \ket{\phi_0})} \le \eta + O(\gamma_\star^{-1}\varepsilon). \qedhere
    \end{equation}
\end{proof}

\parameterisedunitaryrotation*

\begin{proof}
    We suppress the argument $s$ and prove each item in turn.

    \medskip
    \noindent\textsl{0) Construction.}
    Set $X = BA + (I - B)(I - A) = I + (2B - I)(A - B)$.
    Then $X - I = (2B - I)(A - B)$ is a product of a unitary and a Hermitian operator, so $\norm{X - I} = \norm{A - B} \le \delta < 1$; by Banach's lemma ($\norm{P} < 1 \Rightarrow I - P$ invertible with $\norm{(I - P)^{-1}} \le (1 - \norm{P})^{-1}$), $X$ is invertible.
    Let $X = U\abs{X}$ be its polar decomposition, with $U$ unitary and $\abs{X} = (X^\dagger X)^{1/2}$.

    \medskip
    \noindent\textsl{1) Intertwining.}
    Using $A^2 = A$, $B^2 = B$,
    \begin{align}
        XA &= (BA + (I - B)(I - A))A = BA, & 
        BX &= B(BA + (I - B)(I - A)) = BA,
    \end{align}
    so $XA = BX$.
    Taking adjoints gives $AX^\dagger = X^\dagger B$, and multiplying on the right by $X$ gives $AX^\dagger X = X^\dagger X A$; thus $A$ commutes with $X^\dagger X$, hence with $\abs{X}$ and $\abs{X}^{-1}$.
    Therefore
    \begin{equation}
        UAU^\dagger = X\abs{X}^{-1}A\abs{X}^{-1}X^\dagger = XA\abs{X}^{-2}X^\dagger = BX\abs{X}^{-2}X^\dagger = B,
    \end{equation}
    the last step using $X\abs{X}^{-2}X^\dagger = U\abs{X}\abs{X}^{-2}\abs{X}U^\dagger = I$.

    \medskip
    \noindent\textsl{2) Norm bound.}
    Writing $X^\dagger X = I + F$ with $F = (X - I) + (X - I)^\dagger + (X - I)^\dagger(X - I)$,
    \begin{equation}
        \norm{F} \le 2\norm{X - I} + \norm{X - I}^2 \le 2\delta + \delta^2 \le 3\delta,
    \end{equation}
    Hence $X^\dagger X\succeq(1 - 3\delta)I$ is positive definite and
    \begin{align}
        \norm{ \abs{X}^{-1} } &\le (1 - 3\delta)^{-1/2}, & 
        \norm{ (I + F)^{-1/2} - I } &\le \frac{3\delta}{1 - 3\delta}.
    \end{align}
    It follows that
    \begin{equation}
        \norm{U - I} = \norm{(X - I)\abs{X}^{-1} + (\abs{X}^{-1} - I)} \le \frac{\delta}{(1 - 3\delta)^{1/2}} + \frac{3\delta}{1 - 3\delta},
    \end{equation}
    and for $4\delta \le 1$ this is at most $C\delta$ with $C = 16$, for example.

    \medskip
    \noindent\textsl{3) Derivative norm bound.}
    Differentiating $X$ and using $\norm{2A - I} = \norm{2B - I} = 1$,
    \begin{align}
        \ps X &= (\ps B)(2A - I) + (2B - I)(\ps A), & 
        \norm{\ps X} &\le \norm{\ps A} + \norm{\ps B}.
    \end{align}
    From $\ps U = (\ps X)\abs{X}^{-1} + X\,\ps[\abs{X}^{-1}]$, the first term is at most $(\norm{\ps A} + \norm{\ps B})(1 - 3\delta)^{-1/2}$.
    For the second, set $M = X^\dagger X$, so $\norm{\ps M} \le 2 (1 + \delta) (\norm{\ps A} + \norm{\ps B})$.
    For $t \geq 0$, let $R_t = (M + tI)^{-1}$.
    From Section~\ref{sec:continuity-of-spectrum}, we know that 
    \begin{align}\label{eq:resolvent-derivatives}
        \ps R_t &= -R_t (\ps M)R_t, &
        \ps^2R_t &= 2R_t (\ps M) R_t (\ps M) R_t - R_t (\ps^2 M) R_t.
    \end{align}
    Recall the inverse square root integral representation
    \begin{equation}\label{eq:inverse-square-root-integral}
        M^{-1/2} = \frac{1}{\pi} \int_0^\infty \dd{t}\,t^{-1/2} R_t.
    \end{equation}
    Differentiation under the integral gives
    \begin{equation}
        \ps[M^{-1/2}] = -\frac{1}{\pi} \int_0^\infty \dd{t}\,t^{-1/2} R_t (\ps M) R_t.
    \end{equation}
    If $\lambda_{\min}(M) \ge 1 - 3\delta$, the scalar integral evaluates to
    \begin{equation}
        \norm{\ps[M^{-1/2}]}
        \le \frac{1}{2} (1 - 3\delta)^{-3/2} \norm{\ps M}.
    \end{equation}
    Consequently, $\norm{X\,\ps [\abs {X}^{-1}]}\le C_1 (\norm{\ps A} + \norm{\ps B})$
    and $\norm{\ps U}\le C(\norm{\ps A}+\norm{\ps B})$, 
    with universal constants because $\delta\le 1/4$.

    \medskip
    \noindent\textsl{4) Closure and differentiability.}
    Since $A(0) = A(1)$ and $B(0) = B(1)$, we have $X(0) = X(1)$, hence $\abs{X(0)} = \abs{X(1)}$ and $U(0) = U(1)$.
    Furthermore, $A(s)$ and $B(s)$ are ${\rm C}^2$, the families $X(s)$ and $M(s)=X^\dagger(s)X(s)$ are also ${\rm C}^2$.
    Moreover, the bound established above gives $M(s) \succeq (1 - 3\delta)I \succeq \frac{1}{4} I$ uniformly in $s$.

    The resolvent $R_t(s)=(M(s) + tI)^{-1}$ is ${\rm C}^2$ and satisfies the derivative formulas from \cref{eq:resolvent-derivatives}.
    Since $\norm{R_t(s)} \leq (t + \frac{1}{4})^{-1}$ uniformly in $s$, \cref{eq:inverse-square-root-integral} may be differentiated twice under the integral.
    Consequently, $s\mapsto M(s)^{-1/2}$ is ${\rm C}^2$.
    Since
    \begin{equation}
        U(s) = X(s)M(s)^{-1/2},
    \end{equation}
    it follows that $U(s)$ is ${\rm C}^2$.
\end{proof}

\compositionofsimulations*

\begin{proof}
    We follow the static composition proof of Ref.~\cite{bravyi2016complexity}, adding a derivative bound.
    Let $H, \wt{H}_1, \wt{H}_2$ act on $\Hil, \wt\Hil_1, \wt\Hil_2$ of dimensions $N, N_1, N_2$.
    By \cref{lem:eigenvalue-simulation} and the hypothesis $8\varepsilon_2\leq\Delta_1$, the separation in $\wt{H}_2(s)$ corresponding to the band gap $\Delta_1$ of $\wt H_1(s)$ is at least $\Delta_1 - 2\varepsilon_2\geq 3\Delta_1/4>0$.
    Thus, $\mathcal{E}_N(\wt{H}_2(s))$ is well-defined, and $\U_2(s)$ maps $\wt\Hil_1$ into $\mathcal{E}_{N_1}(\wt{H}_2(s))$.

    Let $P_{N,i}(s)$ project onto $\mathcal{E}_N(\wt{H}_i(s))$.
    With $V(s) = \U_2^\dagger(s)\wt{H}_2(s)\U_2(s) - \wt{H}_1(s)$ (of norm $\le \varepsilon_2$) we have $\mathcal{E}_N(\wt{H}_1(s) + V(s)) = \U_2^\dagger(s)\mathcal{E}_N(\wt{H}_2(s))$, so by the projector-perturbation bound \cite[Lemma~3.1]{bravyi2011schrieffer},
    \begin{equation}
        \norm{P_{N,1}(s) - \U_2^\dagger(s)P_{N,2}(s)\U_2(s)} \le 2\Delta_1^{-1}\varepsilon_2,
    \end{equation}
    Since $\mathcal{E}_{N_1}(\wt{H}_2(s)) = {\rm Im}(\U_2(s))$, conjugation by $\U_2(s)$ gives the same bound on the simulator space.
    Set
    \begin{align}
        A(s) &= \U_2(s)P_{N,1}(s)\U_2^\dagger(s), &
        B(s) &= P_{N,2}(s).
    \end{align}
    Then $\norm{A-B}\le2\Delta_1^{-1}\varepsilon_2\le1/4$.
    By Claim~\ref{claim:rank-image-pointwise}, $A,B$ have equal rank, so \cref{lem:parameterised-unitary-rotation} yields $T(s)$ with
    \begin{align}
        \norm{T(s) - I} &\le 2 C \Delta_1^{-1} \varepsilon_2, &
        \mathcal{E}_N(\wt H_2(s)) &= T(s)\U_2(s)\mathcal{E}_N(\wt H_1(s)).
    \end{align}

    Let $\U(s) = T(s)\U_2(s)\U_1(s)$, an isometry.
    Condition \ref{item:c1} holds as $\U(s)$ maps $\Hil$ onto $\mathcal{E}_N(\wt{H}_2(s))$.
    For \ref{item:c3}, $\norm{\V_2\V_1 - \U} \le \eta_1 + \eta_2 + O(\Delta_1^{-1}\varepsilon_2)$.
    For \ref{item:c2}, writing $\wt{H}_i^{(N)} = \wt{H}_i P_{N,i}$ and using $\wt{H}_1\U_1 = \wt{H}_1^{(N)}\U_1$, $\wt{H}_2 T\U_2\U_1 = \wt{H}_2^{(N)} T\U_2\U_1$,
    \begin{equation}
        \norm{H - \U^\dagger\wt{H}_2\U} \le \varepsilon_1 + \varepsilon_2 + O(\Delta_1^{-1}\varepsilon_2\snorm{H}).
    \end{equation}
    Finally, $\ps\U = (\ps T)\U_2\U_1 + T(\ps\U_2)\U_1 + T\U_2(\ps\U_1)$, so $\snorm{\ps\U} \le \rho + \kappa_2 + \kappa_1$, where $\rho = \snorm{\ps T}$.
    Differentiating $A = \U_2P_{N,1}\U_2^\dagger$ gives
    $\snorm{\ps A} = O(\kappa_2+\Delta_1^{-1}\snorm{\ps\wt H_1})$, while the Riesz formula gives $\snorm{\ps B} = O(\Delta_1^{-1}\snorm{\ps\wt H_2})$ after absorbing the fixed constant loss in the perturbed band gap.
    Hence, \cref{lem:parameterised-unitary-rotation} gives the displayed bound on $\kappa$ and yields the closure property of $T$.
\end{proof}

%% file: appendices/section-proofs/sw-transformation-proofs.tex
\section{Proofs of Main Results in Section~\ref{sec:sw-transformation}}\label{app:sw-transformation-proofs}

\LInverse*

\begin{proof}
    Since $H_0 = \Pi_+ \in \D_\tn{d}$, \cref{cor:L-preserves-block-structure} gives $\L_{\Delta H_0}(\D_\tn{od}) \subseteq \D_\tn{od}$, and linearity is immediate from bilinearity of the commutator.
    For $X \in \D_\tn{od}$ we have $\Pi_+ X = X_{+-}$ and $X \Pi_+ = X_{-+}$, so
    \begin{equation}\label{eq:L-projector-form}
        \L_{\Delta H_0}(X) = \Delta \comm{\Pi_+}{X} = \Delta \big( X_{+-} - X_{-+} \big).
    \end{equation}
    Applying \cref{eq:L-projector-form} twice, and noting that $\L_{\Delta H_0}(X)$ has blocks $\Delta X_{+-}$ and $-\Delta X_{-+}$,
    \begin{equation}
        \L_{\Delta H_0}^2 (X) = \Delta\big( \Delta X_{+-} + \Delta X_{-+} \big) = \Delta^2 X .
    \end{equation}
    Hence, $\L_{\Delta H_0}$ restricted to $\D_\tn{od}$ is a bijection of $\D_\tn{od}$ onto itself, with two-sided inverse $\Delta^{-2}\L_{\Delta H_0}$, that is $X \mapsto \Delta^{-1}(X_{+-} - X_{-+})$.
    Since $H_0 = \Pi_+$ gives $\Sigma = \Pi_+$, and $\Pi_+ X_{+-} = X_{+-}$, $X_{-+}\Pi_+ = X_{-+}$, this is exactly the right-hand side of \cref{eq:L-inverse}.

    For the identification with the pseudoinverse, $\Delta H_0 = 0\cdot\Pi_- + \Delta\,\Pi_+$ is the spectral decomposition, so $\sigma(\Delta H_0) = \{0,\Delta\}$ with spectral projectors $\Pi_-$ and $\Pi_+$. \Cref{eq:liouville_pseudoinverse_def} then has exactly two terms,
    \begin{equation}
        \L_{\Delta H_0}^{\pattce}(X)
        = \frac{\Pi_- X \Pi_+}{0-\Delta} + \frac{\Pi_+ X \Pi_-}{\Delta - 0}
        = \Delta^{-1}\big( X_{+-} - X_{-+} \big),
    \end{equation}
    which agrees with the inverse just computed.
    Uniqueness is injectivity: if $\L_{\Delta H_0}(S) = \L_{\Delta H_0}(S')$ with $S, S' \in \D_\tn{od}$, then $S - S' \in \ker\big(\L_{\Delta H_0}\big) \cap \D_\tn{od} = \{0\}$.
\end{proof}

\regularityperiodicitycanonicalsw*

\begin{proof}
    Since $V\in{\rm C}^2$, the Hamiltonian family $\widetilde{H}(s)=\Delta H_0+V(s)$ is ${\rm C}^2$.
    The condition $2\snorm{V}<\Delta$ ensures that its low-energy band remains separated from the high-energy band by a uniform positive gap. 
    We may therefore choose a fixed contour $\mathscr{C}$ enclosing the low-energy band and no other eigenvalues for every $s\in[0,1]$ and write the corresponding spectral projector as
    \begin{equation}
        \widetilde{\Pi}_-(s) = \frac{1}{2\pi\i}\oint_{\mathscr{C}} \dd{z} \, (zI-\widetilde{H}(s))^{-1}.
    \end{equation}
    The resolvent $(zI-\widetilde{H}(s))^{-1}$ is ${\rm C}^2$ in $s$ uniformly for $z\in\mathscr{C}$, and hence $\widetilde{\Pi}_-(s)$ is a ${\rm C}^2$ family of projectors.
    The canonical direct rotation from the fixed projector $\Pi_-$ to $\widetilde{\Pi}_-(s)$ is a smooth and analytic function of these two projectors whenever $\norm{\Pi_- - \widetilde{\Pi}_-(s)}<1$.
    It follows that the unitary $\e^{-S(s)}$ is ${\rm C}^2$. 
    Since the canonical generator is the unique anti-Hermitian and block-off-diagonal logarithm satisfying $\norm{S(s)}<\pi/2$~\cite{bravyi2011schrieffer}, the family $S(s)$ is also ${\rm C}^2$.
    Moreover, $V(0)=V(1)$ implies $\widetilde{H}(0)=\widetilde{H}(1)$ and therefore $\widetilde{\Pi}_-(0) = \widetilde{\Pi}_-(1)$.

    The operators $S(0)$ and $S(1)$ are consequently the canonical generators associated with the same ordered pair of projectors $\bigl(\Pi_-,\widetilde{\Pi}_-(0)\bigr)$. 
    Uniqueness of the canonical direct rotation gives $S(0)=S(1)$ and hence $\e^{-S(0)}=\e^{-S(1)}$.

    Finally, since $\V$ is fixed, $\U(s)=\e^{-S(s)}\V$ is ${\rm C}^2$ and satisfies $\U(0)=\U(1)$.
\end{proof}

\swcoefficients*

\begin{proof}
    We suppress the dependence on the parameter $s$ for notational convenience in this proof.
    Let 
    \begin{equation}
        \mathcal{B}(t) \coloneqq \e^{S(t)} (\Delta H_0 + t V) \e^{-S(t)} = \sum_{k \geq 0} \frac{1}{k!} \L^k_{S(t)}(\Delta H_0 + t V) \eqqcolon \sum_{j \geq 0} t^j B_j.
    \end{equation}
    Since $S(t) = O(t)$, the $k$-th BCH term contributes to the order $t^k$ and higher in the expansion.
    Recalling that $S(t) = t S_1 + t^2 S_2 + t^3 S_3 + O(t^4)$ and therefore collecting powers of $t$ up to third order gives
    \begin{align}
        B_0 &= \Delta H_0, \\
        B_1 &= V 
                + \comm{S_1}{\Delta H_0}, \\
        B_2 &= \comm{S_2}{\Delta H_0} 
                + \frac{1}{2} \comm{S_1}{\comm{S_1}{\Delta H_0}} + \comm{S_1}{V}, \\
        B_3 &= \comm{S_3}{\Delta H_0} 
                + \comm{S_2}{V}
                + \frac{1}{2} \left( \comm{S_2}{\comm{S_1}{\Delta H_0}} + \comm{S_1}{\comm{S_2}{\Delta H_0}} + \comm{S_1}{\comm{S_1}{V}} \right)
                + \frac{1}{6} \comm{S_1}{\comm{S_1}{\comm{S_1}{\Delta H_0}}}.
    \end{align}
    By construction $S(t)$ block-diagonalises $\wt{H}(t)$, so $\mathcal{B}(t) \in \D_\tn{d}$ and hence $(B_j)_\tn{od} = 0$ for every $j \geq 0$.
    Recall also that $S_j \in \D_\tn{od}$ for every $j$.
    We now prove the expressions for the first three Schrieffer--Wolff generator coefficients in turn.

    \medskip
    \noindent\textsl{First Schrieffer--Wolff coefficient $S_1$.}
    The first Schrieffer--Wolff coefficient $S_1$ is determined by setting the block-off-diagonal part of $B_1$ to zero, and thus using \cref{lem:block-commutators} we have
    \begin{equation}
        V_{\tn{od}} + \comm{S_1}{\Delta H_0} = 0.
    \end{equation}
    Solving for $S_1$ gives
    \begin{equation}
        S_1 = \L_{\Delta H_0}^{\pattce}( V_{\tn{od}} ),
    \end{equation}
    yielding \cref{eq:sw1}, since $V_{\tn{od}} \in \D_\tn{od}$ and hence \cref{lem:L-inverse} applies.

    \medskip
    \noindent\textsl{Second Schrieffer--Wolff coefficient $S_2$.}
    The second Schrieffer--Wolff coefficient $S_2$ is determined by setting the block-off-diagonal part of $B_2$ to zero. 
    Using \cref{lem:block-commutators}, we have
    \begin{equation}
        \comm{S_2}{\Delta H_0} + \comm{S_1}{V_{\tn{d}}} = 0.
    \end{equation}
    Solving for $S_2$ gives
    \begin{equation}
        S_2 = \L_{\Delta H_0}^{\pattce}(\comm{S_1}{V_{\tn{d}}}),
    \end{equation}
    yielding \cref{eq:sw2}, since $\comm{S_1}{V_{\tn{d}}} \in \D_\tn{od}$ via \cref{lem:block-commutators} and hence \cref{lem:L-inverse} applies.

    \medskip
    \noindent\textsl{Third Schrieffer--Wolff coefficient $S_3$.}
    The third Schrieffer--Wolff coefficient $S_3$ is determined by setting the block-off-diagonal part of $B_3$ to zero.
    Using \cref{lem:block-commutators}, we have
    \begin{equation}
        \comm{S_3}{\Delta H_0} + \comm{S_2}{V_{\tn{d}}} + \frac{1}{2}\comm{S_1}{\comm{S_1}{V_{\tn{od}}}} + \frac{1}{6}\comm{S_1}{\comm{S_1}{\comm{S_1}{\Delta H_0}}} = 0.
    \end{equation}
    Recalling that $\comm{S_1}{\Delta H_0} = -V_{\tn{od}}$, we can rewrite the third-order equation as
    \begin{equation}
        \comm{S_3}{\Delta H_0} + \comm{S_2}{V_{\tn{d}}} - \frac{1}{3}\comm{S_1}{\comm{S_1}{\comm{S_1}{\Delta H_0}}} = 0.
    \end{equation}
    Solving for $S_3$ gives
    \begin{equation}
        S_3 = \L_{\Delta H_0}^{\pattce}\left(\comm{S_2}{V_{\tn{d}}} - \frac{1}{3}\comm{S_1}{\comm{S_1}{\comm{S_1}{\Delta H_0}}}\right) = \L_{\Delta H_0}^{\pattce}\left(\comm{S_2}{V_{\tn{d}}}\right) - \frac{1}{3} \L_{\Delta H_0}^{\pattce}\left(\comm{S_1}{\comm{S_1}{\comm{S_1}{\Delta H_0}}}\right),
    \end{equation}
    yielding \cref{eq:sw3}, since $\comm{S_2}{V_{\tn{d}}} - \frac{1}{3}\comm{S_1}{\comm{S_1}{\comm{S_1}{\Delta H_0}}} \in \D_\tn{od}$ via \cref{lem:block-commutators} and hence \cref{lem:L-inverse} applies.
    The second equality follows from \cref{lem:liouville-properties}.

    In each case the argument of $\L_{\Delta H_0}^{\pattce}$ lies in $\D_\tn{od}$ by \cref{lem:block-commutators}, so \cref{lem:L-inverse} applies, hence the corresponding equation has a unique solution in $\D_\tn{od}$, given by the stated formula.
\end{proof}

\heffcoefficients*

\begin{proof}
    We suppress the dependence on the parameter $s$ for notational convenience in this proof.
    The proof of \cref{prop:sw-coefficients} provides explicit expressions for the coefficients in the expansion of $\mathcal{B}(t) = \e^{S(t)} \wt{H}(t) \e^{-S(t)}$ in powers of $t$.
    Using this expansion, we now prove the expressions for the first three effective Hamiltonian coefficients in turn.
    Recall that at each order, the Schrieffer-Wolff coefficients $S_j$ are chosen such that the block-off-diagonal part of $B_j$ vanishes; hence, we only need to consider the block-diagonal part when computing the effective Hamiltonian coefficients.
    We first identify the form of the effective Hamiltonian coefficients $K_j$ in terms of the Schrieffer--Wolff coefficients $S_j$, then evaluate them explicitly using the expressions for $S_j$ obtained from \cref{prop:sw-coefficients}.

    \medskip
    \noindent\textsl{First effective Hamiltonian coefficient $K_1$.}
    The block-diagonal part of $B_1$ is simply the block-diagonal part of $V$.
    Hence,
    \begin{equation}
        K_1 = (B_1)_{--} = (V)_{--}.
    \end{equation}

    \medskip
    \noindent\textsl{Second effective Hamiltonian coefficient $K_2$.}
    The block-diagonal part of $B_2$ is given by
    \begin{equation}
        (B_2)_{\tn{d}} = \comm{S_1}{V_\tn{od}} + \frac{1}{2} \comm{S_1}{\comm{S_1}{\Delta H_0}} = \frac{1}{2} \comm{S_1}{V_\tn{od}}.
    \end{equation}
    Hence,
    \begin{equation}\label{eq:K2-raw}
        K_2 = (B_2)_{--} = \frac{1}{2} \big( \comm{S_1}{V_\tn{od}} \big)_{--}.
    \end{equation}

    \medskip
    \noindent\textsl{Third effective Hamiltonian coefficient $K_3$.}
    The block-diagonal part of $B_3$ is given by
    \begin{equation}
        (B_3)_{\tn{d}} = \comm{S_2}{V_\tn{od}} + \frac{1}{2} \big( \comm{S_1}{\comm{S_2}{\Delta H_0}} + \comm{S_2}{\comm{S_1}{\Delta H_0}} + \comm{S_1}{\comm{S_1}{V_\tn{d}}} \big).
    \end{equation}
    Recalling that $\comm{S_1}{\Delta H_0} = -V_\tn{od}$ and $\comm{S_2}{\Delta H_0} = - \comm{S_1}{V_\tn{d}}$, we can simplify the expression for $(B_3)_{\tn{d}}$ and hence,
    \begin{equation}\label{eq:K3-raw}
        K_3 = (B_3)_{--} = \frac{1}{2} \big( \comm{S_2}{V_\tn{od}} \big)_{--}.
    \end{equation}

    \medskip
    \noindent\textsl{Evaluation of coefficients.}
    From \cref{lem:L-inverse} and \cref{eq:sw1}, 
    \begin{align}\label{eq:S1-evaluation}
        S_1 &= \Delta^{-1} ( \Sigma V_{+-} - V_{-+} \Sigma ), &
        (S_1)_{-+} &= - \Delta^{-1} V_{-+} \Sigma, &
        (S_1)_{+-} &= \Delta^{-1} \Sigma V_{+-}.
    \end{align}
    Inserting these into \cref{eq:K2-raw} gives 
    \begin{equation}
        K_2 = \frac{1}{2} \big( (S_1)_{-+} V_{+-} - V_{-+} (S_1)_{+-} \big) = -\Delta^{-1} V_{-+} \Sigma V_{+-},
    \end{equation}
    which yields \cref{eq:K2}.
    Next, \cref{eq:sw2} and \cref{eq:S1-evaluation} show that 
    \begin{align}
        \big( \comm{S_1}{V_\tn{d}} \big)_{-+} &= - \Delta^{-1} ( V_{-+} \Sigma V_{++} - V_{--} V_{-+} \Sigma ), &
        \big( \comm{S_1}{V_\tn{d}} \big)_{+-} &= \Delta^{-1} ( \Sigma V_{+-} V_{--} - V_{++} \Sigma V_{+-} ).
    \end{align}
    By applying \cref{lem:L-inverse} again, 
    \begin{align}\label{eq:S2-evaluation}
        (S_2)_{-+} &= - \Delta^{-2} ( V_{-+} \Sigma V_{++} \Sigma - V_{--} V_{-+} \Sigma^2 ), &
        (S_2)_{+-} &= \Delta^{-2} ( \Sigma^2 V_{+-} V_{--} - \Sigma V_{++} \Sigma V_{+-} ).
    \end{align}
    Substituting \cref{eq:S2-evaluation} into \cref{eq:K3-raw} gives 
    \begin{equation}
        K_3 = \frac{1}{2} \big( (S_2)_{-+} V_{+-} - V_{-+} (S_2)_{+-} \big) = \Delta^{-2} ( V_{-+} \Sigma V_{++} \Sigma V_{+-} - \frac{1}{2} \acomm{V_{--}}{V_{-+} \Sigma^2 V_{+-}} ),
    \end{equation}
    which yields \cref{eq:K3}.
\end{proof}

\uniformcontrolexacttruncated*

\begin{proof}
    The undifferentiated coefficient estimates follow from the standard Schrieffer--Wolff bounds in Refs.~\cite{bravyi2011schrieffer,bravyi2016complexity}.
    We show that these estimates hold uniformly in $s$ and that the corresponding series may be differentiated term by term.
    We suppress the explicit dependence on $s$ when it is clear from context.

    Since $\norm{V(s)}\leq v_0 < \Delta/32$, the low- and high-energy bands of $\Delta H_0+V(s)$ remain uniformly separated.
    Since $V$ is closed and ${\rm C}^2$, \cref{prop:regularity-periodicity-canonical-sw} shows that the exact canonical generator $S(s)$ and the unitary $\e^{-S(s)}$ are closed and ${\rm C}^2$.

    For each fixed $s$, \cite[Lemma~3.4]{bravyi2011schrieffer} guarantees that the series for $S(t,s)$ and $H_{\tn{eff}}(t,s)$ converge absolutely whenever
    \begin{equation}
        \abs{t}<\frac{\Delta}{16\norm{V(s)}}.
    \end{equation}
    Suppose that $v_0>0$.
    Choose a fixed constant $16 < \alpha < 32$ and set $R\coloneqq \Delta/(\alpha v_0)$.
    Since $32 v_0 < \Delta$ and $\alpha < 32$, we have $R >1$.
    Moreover, since $\alpha >16$ and $\norm{V(s)}\leq v_0$, we have
    \begin{equation}
        R < \frac{\Delta}{16 v_0} \leq \frac{\Delta}{16\norm{V(s)}}
    \end{equation}
    whenever $V(s) \neq 0$.
    Thus, the circle $\abs{t}=R$ lies inside the convergence domain of both series for every $s\in[0,1]$.

    Applying the norm bound arguments of Refs.~\cite{bravyi2011schrieffer,bravyi2016complexity}, for a constant order $p$, on this common circle gives uniform bounds on the undifferentiated coefficients.
    Recall that $S_j$ and $K_j$ are (homogeneous) operator polynomials of degree $j$ in $V(s)$.
    Differentiating their expressions with respect to $s$ replaces one occurrence of $V$ by $\ps V$ and produces at most $j$ corresponding terms.
    Applying the same norm bound arguments to these differentiated expressions gives a constant $C > 0$, depending only on the fixed choice of $\alpha$, such that, uniformly in $s$ and for every $j\geq1$,
    \begin{align}
        \snorm{S_j} &\leq C\alpha^j\Delta^{-j}v_0^j, & \snorm{\ps S_j} &\leq Cj\alpha^j\Delta^{-j}v_1v_0^{j-1}, \label{eq:coefficient-estimates-S}\\
        \snorm{K_j} &\leq C\alpha^j\Delta^{1-j}v_0^j, & \snorm{\ps K_j} &\leq Cj\alpha^j\Delta^{1-j}v_1v_0^{j-1}. \label{eq:coefficient-estimates-K}
    \end{align}

    We now verify that these estimates give uniform convergence in $s$.
    Set $q \coloneqq \alpha v_0/\Delta = R^{-1}$.
    Since $32 v_0 < \Delta$ and $\alpha < 32$, we have $q <1$.
    The estimates in \cref{eq:coefficient-estimates-S,eq:coefficient-estimates-K} may therefore be rewritten as
    \begin{align}
            \snorm{S_j} &\leq Cq^j, & \snorm{\ps S_j} &\leq C\alpha\Delta^{-1}v_1jq^{j-1}, \\
            \snorm{K_j} &\leq C\Delta q^j, & \snorm{\ps K_j} &\leq C\alpha v_1jq^{j-1}.
    \end{align}
    Since
    \begin{align}
        \sum_{j\geq1}q^j &= \frac{q}{1-q}, &
        \sum_{j\geq1}jq^{j-1} &= \frac{1}{(1-q)^2},
    \end{align}
    the coefficient series and their differentiated series are summable uniformly in $s$.
    The Weierstrass M-test therefore gives uniform convergence of the series for $S$ and $H_{\tn{eff}}$ and of their differentiated series.
    Since every $S_j$ and $K_j$ is ${\rm C}^1$ in $s$, the termwise-differentiation theorem gives
    \begin{align}
        \ps S(s) &= \sum_{j\geq1}\ps S_j(s), &
        \ps H_{\tn{eff}}(s) &= \sum_{j\geq1}\ps K_j(s).
    \end{align}

    Summing the first two coefficient bounds gives
    \begin{align}
        \snorm{S} &\leq C\frac{q}{1-q} = O(\Delta^{-1}v_0), &
        \snorm{\ps S} &\leq C\alpha\Delta^{-1}v_1\frac{1}{(1-q)^2} = O(\Delta^{-1}v_1).
    \end{align}
    For every fixed and constant $p\geq1$, the corresponding scalar tails satisfy
    \begin{align}
        \sum_{j\geq p+1}q^j &= \frac{q^{p+1}}{1-q}, &
        \sum_{j\geq p+1}jq^{j-1} &= \frac{q^p\bigl((p+1)-pq\bigr)}{(1-q)^2}.
    \end{align}
    Since $q < 1$ and $\alpha$ is fixed, these identities give
    \begin{align}
        \snorm{R^S_{p+1}} &= O(\Delta^{-(p+1)}v_0^{p+1}), & \snorm{\ps R^S_{p+1}} &= O(\Delta^{-(p+1)}v_1v_0^p), \\
        \snorm{R^{H_{\tn{eff}}}_{p+1}} &= O(\Delta^{-p}v_0^{p+1}), & \snorm{\ps R^{H_{\tn{eff}}}_{p+1}} &= O(\Delta^{-p}v_1v_0^p).
    \end{align}
    The implicit constants may depend on the fixed order $p$, but are independent of $s$, $\Delta$, $v_0$, and $v_1$.

    Finally, if $\V$ is independent of $s$, then $\U(s)=\e^{-S(s)}\V$ is closed and ${\rm C}^2$.
    The Duhamel formula and the anti-Hermiticity of $S$ give
    \begin{equation}
        \kappa = \snorm{\ps\U} \leq \snorm{\ps S} = O(\Delta^{-1}v_1).
    \end{equation}
\end{proof}

\staticleadingswgenerator*

\begin{proof}
    Since $H_0$ and $V_{\tn{od}}$ are independent of $s$, $\ps S_1 = 0$. 
    Writing $S(s)=S_1+R_2^S(s)$, we therefore have $\ps S(s) = \ps R_2^S(s)$.
    The differentiated remainder estimate in \cref{lemma:uniform-control-exact-truncated}, applied with $p=1$, gives $\snorm{\ps R_2^S} = O(\Delta^{-2}v_0v_1)$.
    Finally, we apply the Duhamel formula to see that $\snorm{\ps \U} \leq \snorm{\ps S} = O(\Delta^{-2}v_0v_1)$.
\end{proof}

\excitationsectordecomposition*

\begin{proof}
    The projectors $\{Q_X\}_{X\subseteq[A]}$ are mutually orthogonal and resolve the identity.
    That is, if $X\neq Y$, then choosing $a$ in the symmetric difference of $X$ and $Y$ shows that $Q_X Q_Y$ contains the factor $p_a q_a = 0$ on the penalty qubit $m_a$, so $Q_XQ_Y = 0$.
    By expanding the resolution of the identity on each penalty qubit, we see
    \begin{equation}
        \sum_{X\subseteq[A]} Q_X = \bigotimes_{a\in[A]} (p_a + q_a) = I.
    \end{equation}

    On the subspace selected by $Q_X$ we have $q_a Q_X = \ind{a\in X}Q_X$, and therefore
    \begin{equation}
        \bs{H}_0 Q_X = \sum_{a\in[A]} q_a Q_X = \abs{X}\, Q_X ,
    \end{equation}
    which is \cref{eq:H0-sector-eigenvalue}.
    Taking adjoints gives $Q_X \bs{H}_0 = \abs{X} Q_X$, so $\bs{H}_0$ commutes with every $Q_X$, i.e., $\comm{\bs{H}_0}{Q_X} = 0$ for all $X$.
    Summing over $X$ yields
    \begin{equation}\label{eq:H0-sector-decomposition}
        \bs{H}_0 = \bs{H}_0\sum_{X\subseteq[A]} Q_X = \sum_{\emptyset\neq X\subseteq[A]} \abs{X}\, Q_X;
    \end{equation}
    note that the term $X = \emptyset$ vanishes.
    Since $\abs{X}\geq 1$ for every nonempty $X$, the only vanishing coefficient in \cref{eq:H0-sector-decomposition} comes from $Q_\emptyset = \bs{\Pi}_-$.
    Hence, the spectral gap of $\bs{H}_0$ above ${\rm Im}(\bs{\Pi}_-)$ is one.

    Now set $\bs{\Sigma} \coloneqq \sum_{\emptyset\neq X\subseteq[A]} \abs{X}^{-1} Q_X$.
    By orthogonality of the $Q_X$ together with \cref{eq:H0-sector-decomposition},
    \begin{equation}
        \bs{\Sigma}\,\bs{H}_0 = \bs{H}_0\,\bs{\Sigma} = \sum_{\emptyset\neq X\subseteq[A]} Q_X = \bs{\Pi}_+,
    \end{equation}
    and $\bs{\Sigma}\bs{\Pi}_- = 0$.
    Thus, $\bs{\Sigma}$ inverts $\bs{H}_0$ on ${\rm Im}(\bs{\Pi}_+)$ and vanishes on $\ker(\bs{H}_0) = {\rm Im}(\bs{\Pi}_-)$, which for a Hermitian operator is precisely the Moore--Penrose pseudoinverse; hence $\bs{\Sigma} = \bs{H}_0^{\pattce}$ and the claimed decomposition holds.

    Finally, for mutually orthogonal projectors we have $\norm{\sum_X c_X Q_X} = \max_{X : Q_X\neq 0}\abs{c_X}$, so that $\norm{\bs{\Sigma}_X} = \abs{X}^{-1}$ and $\norm{\bs{\Sigma}} = \max_{\emptyset\neq X\subseteq[A]}\abs{X}^{-1} = 1$.
\end{proof}

\linverseparallel*

\begin{proof}
    By \cref{clm:excitation-sector-decomposition}, $\bs{H}_0$ is block-diagonal with respect to the global splitting $\bs{\Pi}_- \oplus \bs{\Pi}_+$, with $(\bs{H}_0)_{--} = 0$ and $(\bs{H}_0)_{++} = \bs{H}_0 \succeq \bs{\Pi}_+$ invertible on ${\rm Im}(\bs{\Pi}_+)$ with inverse $\bs{\Sigma}$.
    Thus, $\bs{H}_0$ satisfies \cref{eq:H0-block-structure} for the global splitting, and \cref{lem:L-inverse} applies with $H_0 \mapsto \bs{H}_0$, $\Pi_\pm \mapsto \bs{\Pi}_\pm$ and $\Sigma \mapsto \bs{\Sigma}$.
\end{proof}

\swgeneratorsparallel*

\begin{proof}
    The proof of \cref{prop:sw-coefficients} uses the fact that $\Delta \bs{H}_0 \in D_\tn{d}$, the block-commutator rules of \cref{lem:block-commutators} applied to the global splitting, and the invertibility of $\L_{\Delta\bs{H}_0}$ on $D_\tn{od}$ established in \cref{cor:L-inverse-parallel}. 
\end{proof}

\heffcoefficientsparallel*

\begin{proof}
    The proof of \cref{prop:heff-coefficients} uses only the block structure of $\Delta H_0$, the block-commutator rules of \cref{lem:block-commutators}, and the inversion of $\L_{\Delta H_0}$ on block-off-diagonal operators.
    By \cref{clm:excitation-sector-decomposition} and \cref{cor:L-inverse-parallel} these hold for $\bs{H}_0$ with respect to the global splitting, hence the argument applies under the substitutions $H_0\mapsto\bs{H}_0$, $\Pi_\pm\mapsto\bs{\Pi}_\pm$, $\Sigma\mapsto\bs{\Sigma}$ and $V\mapsto\bs{V}$.
\end{proof}

\ParallelEffectiveHamiltonianDecomposition*

\begin{proof}
    We suppress the dependence on the parameter $s$ for notational convenience in this proof.
    
    For $X, Y \subseteq [A]$ and $\alpha \in \{-,+\}$, define
    \begin{align}
        \bs{V}_{XY} &\coloneqq Q_X \bs{V} Q_Y, &
        \bs{V}_{X\alpha} &\coloneqq Q_X \bs{V} \bs{\Pi}_\alpha, &
        \bs{V}_{\alpha X} &\coloneqq \bs{\Pi}_\alpha \bs{V} Q_X.
    \end{align}
    Since $Q_{\emptyset} = \bs{\Pi}_-$, we have that 
    \begin{align}
        \bs{V}_{+-} = \sum_{\emptyset\neq X\subseteq[A]} \bs{V}_{X-}, &
        \bs{V}_{-+} = \sum_{\emptyset\neq X\subseteq[A]} \bs{V}_{-X}.
    \end{align}

    Taking the decomposition $\bs{V} = \sum_{a\in[A]} V^a$ into account, we see that $\bs{V}_{XY} = \sum_{a\in[A]} (V^a)_{XY}$ and $\bs{V}_{X-} = \sum_{a\in[A]} (V^a)_{X-}$.
    It follows by definition that $Q_X V^a \bs{\Pi}_- = 0$ unless $X \subseteq \{a\}$ and therefore
    \begin{equation}\label{eq:Va-plus-minus}
        (V^a)_{+-} = Q_{\{a\}} V^a \bs{\Pi}_-.
    \end{equation}
    The same argument applies to $\bs{\Pi}_- V^a Q_X$.
    
    Now we prove the decomposition of $\bs{K}_2$ and $\bs{K}_3^{(1)}$ in turn.

    \medskip
    \noindent\textsl{Decomposition of $\bs{K}_2$.}
    By the definition of $\bs{K}_2$ in \cref{eq:K2-parallel}, the decomposition of $\bs{V}$ we see 
    \begin{equation}
        \bs{K}_2 = - \Delta^{-1} \sum_{a,b \in [A]} \sum_{\emptyset \neq X \subseteq [A]} \bs{\Pi}_- V^a Q_X \bs{\Sigma}_X Q_X V^b \bs{\Pi}_-.
    \end{equation}
    From \cref{eq:Va-plus-minus} (and the analogous expression for $(V^a)_{-+}$), the second-order terms labelled by the triple $(a,b,X)$ can be nonzero only if $\emptyset \neq X \subseteq a \cap b$.
    Using \cref{eq:singleton-sigma} allows us to factorise as 
    \begin{equation}
        \bs{K}_2 = - \Delta^{-1} \sum_{a \in [A]} \bs{\Pi}_- V^a \bs{\Sigma}_{\{a\}} V^a \bs{\Pi}_- = - \Delta^{-1} \sum_{a \in [A]} (V^a)_{-+} (H_0^a)^{\pattce} (V^a)_{+-}.
    \end{equation}
    By \cref{eq:Omega-a} we yield \cref{eq:K2-decomposition}.
    
    \medskip
    \noindent\textsl{Decomposition of $\bs{K}_3^{(1)}$.}
    By the definition of $\bs{K}_3^{(1)}$ in \cref{eq:K3-parallel}, the decomposition of $\bs{V}$ we see 
    \begin{equation}
        \bs{K}_3^{(1)} = \Delta^{-2} \sum_{a,b,c \in [A]} \sum_{\emptyset \neq X,Y \subseteq [A]} \bs{\Pi}_- V^a Q_X \bs{\Sigma}_X Q_X V^b Q_Y \bs{\Sigma}_Y Q_Y V^c \bs{\Pi}_-.
    \end{equation}
    By \cref{eq:Va-plus-minus}, the rightmost and leftmost factors can be nonzero only when $Y=\{c\}$ and $X=\{a\}$, respectively.
    Suppose that $b\neq c$.
    Since $V^b$ acts only on the system qubits and mediator $m_b$, the assumption $\bs{\Pi}_-V^b\bs{\Pi}_- = 0$ implies that $Q_{\{c\}} V^b Q_{\{c\}} = 0$.
    Every remaining action of $V^b$ on $Q_{\{c\}}$ creates an additional excitation on $m_b$, producing the sector $Q_{\{b,c\}}$, from which the leftmost single-mediator perturbation cannot return to $\bs{\Pi}_-$.
    If $b = c$, the second reduced inverse removes the component in which $V^c$ returns the state to the ground sector, leaving only the component that remains in $Q_{\{c\}}$.
    Consequently, a contribution can be nonzero only when $a = b = c$ and $X = Y = \{a\}$.
    Using \cref{eq:singleton-sigma} allows us to factorise as 
    \begin{equation}
        \bs{K}_3^{(1)} = \Delta^{-2} \sum_{a \in [A]} \bs{\Pi}_- V^a \bs{\Sigma}_{\{a\}} V^a \bs{\Sigma}_{\{a\}} V^a \bs{\Pi}_- = \Delta^{-2} \sum_{a \in [A]} (V^a)_{-+} (H_0^a)^{\pattce} (V^a)_{++} (H_0^a)^{\pattce} (V^a)_{+-}.
    \end{equation}
    By \cref{eq:Psi-a} we yield \cref{eq:K3-decomposition}.
\end{proof}

%% file: appendices/section-proofs/berry-phase-simulation-proofs.tex
\section{Proofs of Main Results in Section~\ref{sec:berry-phase-simulation}}\label{app:berry-phase-simulation-proofs}
In this appendix, we include some additional supporting proofs for the main results presented in Section~\ref{sec:berry-phase-simulation}.
We first outline these supporting results then proceed with the proofs of the main results.

\subsection{Supporting Results}
The context and relevance of these supporting results is with respect to the content, notational conventions, and assumptions used in Section~\ref{sec:berry-phase-simulation}.

\begin{restatable}{lemma}{resolventbounds}\label{lem:resolvent-bounds}
    For every $s\in[0,1]$, the contour $\mathscr{C}(s)$ encloses $\lambda_0(s)$ and no other eigenvalue of $H(s)$, and encloses $\mu_0(s)$ and no other eigenvalue of $\widehat{H}(s)$.
    Moreover, for all $z\in\mathscr{C}(s)$,
    \begin{align}
        \|R(z,s)\| &\le \frac{2}{\gamma_\star}, &
        \|\widehat{R}(z,s)\| &\le \frac{4}{\gamma_\star}.
    \end{align}
\end{restatable}

\begin{proof}
    Since $H(s)$ is Hermitian, $\|R(z,s)\| = {\rm dist}(z, {\rm spec}(H(s)))^{-1}$, and likewise for $\widehat{R}$.
    For $z \in \mathscr{C}(s)$, $\abs{z - \lambda_0(s)} = \gamma_\star/2$ while $\abs{z - \lambda_j(s)} \ge \gamma(s) - \gamma_\star/2 \ge \gamma_\star/2$ for $j\ge 1$, which proves the bound on $\|R\|$ and the enclosure for $H(s)$.
    By Weyl's inequality and Assumption~\ref{item:a2}, $\abs{\mu_0(s) - \lambda_0(s)} \le \varepsilon$ and $\mu_j(s) - \lambda_0(s) \ge \gamma_\star - \varepsilon$ for $j\ge 1$.
    Therefore,
    \begin{equation}
        \abs{z - \mu_0(s)} \ge \tfrac{\gamma_\star}{2} - \varepsilon \ge \tfrac{\gamma_\star}{4}, \qquad \abs{z - \mu_j(s)} \ge \gamma_\star - \varepsilon - \tfrac{\gamma_\star}{2} \ge \tfrac{\gamma_\star}{4} \quad (j\ge 1),
    \end{equation}
    giving the bound on $\|\widehat{R}\|$.
    Finally $\mu_0(s)$ lies inside $\mathscr{C}(s)$ since $\varepsilon < \gamma_\star/2$, and every $\mu_j(s)$, $j\ge 1$, lies outside since $\mu_j(s) - \lambda_0(s) \ge \gamma_\star - \varepsilon > \gamma_\star/2$.
\end{proof}

\begin{remark}\label{rem:riesz-equivalence}
    Expanding the resolvent in the eigenbasis of $H(s)$,
    \begin{equation}
        \frac{1}{2\pi\i}\oint_{\mathscr{C}(s)} \dd{z}\, R(z,s) = \sum_{j=0}^{2^n - 1} \left( \frac{1}{2\pi\i}\oint_{\mathscr{C}(s)} \frac{\dd{z}}{z - \lambda_j(s)} \right) \ketbra{\phi_j(s)} = \ketbra{\phi_0(s)},
    \end{equation}
    since by \cref{lem:resolvent-bounds} only the pole at $\lambda_0(s)$ is enclosed (residue $1$), the integrands for $j\ge 1$ being holomorphic inside $\mathscr{C}(s)$.
    The same computation gives $\widehat{P}(s) = \ketbra{\psi_0(s)}$.
\end{remark}

\begin{restatable}{lemma}{resolventidentity}\label{lem:resolvent-identity}
    For all $z\in\mathscr{C}(s)$ and $s\in[0,1]$,
    \begin{equation}
        R(z,s)-\widehat R(z,s)=R(z,s)E(s)\widehat R(z,s).
    \end{equation}
\end{restatable}

\begin{proof}
    Dropping dependencies and applying $A^{-1} - B^{-1} = A^{-1}(B - A)B^{-1}$ with $A = z - H$, $B = z - \widehat{H}$,
    $R - \widehat{R} = R[(z - \widehat{H}) - (z - H)]\widehat{R} = R(H - \widehat{H})\widehat{R} = R E \widehat{R}$.
\end{proof}

\begin{restatable}{proposition}{derivativeSpectralProjectorBound}\label{prop:derivative-spectral-projector-bound}
    For all $s\in[0,1]$,
    \begin{align}
        \|\ps  P(s)\| &= O(\gamma_\star^{-1} \norm{\ps H(s)}), &
        \|\ps  \widehat{P}(s)\| &= O(\gamma_\star^{-1}\norm{\ps\widehat H(s)}),
    \end{align}
    where $\norm{\ps\widehat H(s)}\le \norm{\ps H(s)}+\norm{\ps E(s)}$ because $\widehat H=H-E$.
\end{restatable}

\begin{proof}
    Fix $s_0$.
    By continuity of the spectrum the contour $\mathscr{C}(s_0)$ stays admissible on a neighbourhood of $s_0$, so we differentiate \cref{eq:riesz} under the integral sign with locally constant contour.
    Using $\ps  R = R\,(\ps  H)\, R$,
    \begin{equation}\label{eq:projector-derivative}
        \ps  P(s) = \frac{1}{2\pi\i}\oint_{\mathscr{C}(s)} \dd{z}\, R(z,s)\,(\ps  H(s))\, R(z,s),
    \end{equation}
    and analogously for $\widehat{P}$ with $R \to \widehat{R}$, $H \to \widehat{H}$.
    The contour estimate with \cref{lem:resolvent-bounds} gives
    \begin{equation}
        \|\ps  P\| \le \frac{\pi\gamma_\star}{2\pi}\Big(\frac{2}{\gamma_\star}\Big)^2 \norm{\ps H} = \frac{2\norm{\ps H}}{\gamma_\star}, \qquad \|\ps  \widehat{P}\| \le \frac{\pi\gamma_\star}{2\pi}\Big(\frac{4}{\gamma_\star}\Big)^2 \norm{\ps\widehat{H}} = \frac{8\norm{\ps\widehat{H}}}{\gamma_\star}.
    \end{equation}
    Finally, differentiating the identity $\widehat H=H-E$ gives $\norm{\ps\widehat H}\le\norm{\ps H}+\norm{\ps E}$.
\end{proof}

\begin{restatable}{claim}{circularDistance}\label{claim:circular-distance}
    For all $a, b\in\mathbb{R}$, $\ d_{2\pi}(a,b) \le \tfrac{\pi}{2}\abs{\e^{\i a} - \e^{\i b}}$.
\end{restatable}

\begin{proof}
    Let $\theta = d_{2\pi}(a,b) \in [0,\pi]$, so $\abs{\e^{\i a} - \e^{\i b}} = 2\sin(\theta/2)$.
    Concavity of $\sin$ on $[0,\pi/2]$ gives $\sin(\theta/2) \ge (\theta/2)\cdot\tfrac{2}{\pi} = \theta/\pi$, whence $\theta \le \pi\sin(\theta/2) = \tfrac{\pi}{2}\cdot 2\sin(\theta/2)$.
\end{proof}

\begin{restatable}{proposition}{duhamelBound}\label{prop:duhamel}
    Let $A(s), B(s)$ be continuous anti-Hermitian families and $U(s), V(s)$ solve $\ps  U = AU$, $\ps  V = BV$ with $U(0) = V(0) = I$.
    Then $\|U(s) - V(s)\| \le \int_0^s \|A(t) - B(t)\|\,\dd{t}$ for all $s\in[0,1]$.
\end{restatable}

\begin{proof}
    Both $U, V$ are unitary.
    Using $\partial_t U^\dagger = -U^\dagger A$, $\ \partial_t(U^\dagger V) = U^\dagger(B - A)V$, so $U^\dagger(s)V(s) - I = \int_0^s U^\dagger(t)(B(t) - A(t))V(t)\,\dd{t}$.
    By unitary invariance, $\|U(s) - V(s)\| = \|I - U^\dagger(s)V(s)\| \le \int_0^s \|A(t) - B(t)\|\,\dd{t}$.
\end{proof}

\begin{restatable}{proposition}{generatorCommutatorBound}\label{prop:generator-commutator-bound}
    Let $X(s), Y(s)$ be ${\rm C}^1$ and $K = \comm{\ps  X}{X}$, $L = \comm{\ps  Y}{Y}$.
    Then $\|K - L\| \le 2\|\ps  X - \ps  Y\|\,\|X\| + 2\|\ps  Y\|\,\|X - Y\|$.
\end{restatable}

\begin{proof}
    Adding and subtracting $\comm{\ps  Y}{X}$, $\ K - L = \comm{\ps  X - \ps  Y}{X} + \comm{\ps  Y}{X - Y}$; the bound follows from $\|\comm{W}{Z}\| \le 2\|W\|\|Z\|$ and submultiplicativity.
\end{proof}

\begin{restatable}{lemma}{generatorComparison}\label{lem:generator-comparison}
    For all $s\in[0,1]$,
    \begin{equation}
        \|{\rm G}(s) - \widehat{{\rm G}}(s)\| = O\big( \gamma_\star^{-2}\max\{\norm{\ps H}, \norm{\ps\widehat{H}}\}\,\varepsilon + \gamma_\star^{-1} \norm{\ps E}\big).
    \end{equation}
\end{restatable}

\begin{proof}
    By \cref{prop:generator-commutator-bound} with $X = P$, $Y = \widehat{P}$ and $\|P\| = 1$, $\ \|{\rm G} - \widehat{{\rm G}}\| \le 2\|\ps  P - \ps \widehat{P}\| + 2\|\ps \widehat{P}\|\,\|P - \widehat{P}\|$.
    The first term is bounded by \cref{prop:derivative-projector-difference}.
    By \cref{prop:spectral-projector-bound,prop:derivative-spectral-projector-bound}, the second is
    $O(\gamma_\star^{-2}\norm{\ps\widehat H}\varepsilon)$, which is already contained in the first term's maximum.
\end{proof}

\subsection{Main Results}
We now proceed to prove the main results of Section~\ref{sec:berry-phase-simulation}, making use of the supporting results established above.

\spectralprojectorbound*

\begin{proof}
    By \cref{eq:riesz} and the contour estimate, $\|P(s) - \widehat{P}(s)\| \le \frac{\abs{\mathscr{C}(s)}}{2\pi} \sup_{z\in\mathscr{C}(s)} \|R(z,s) - \widehat{R}(z,s)\|$.
    By \cref{lem:resolvent-identity,lem:resolvent-bounds}, $\sup_{z\in\mathscr{C}(s)} \|R - \widehat{R}\| \le \tfrac{2}{\gamma_\star}\cdot\varepsilon\cdot\tfrac{4}{\gamma_\star} = 8\varepsilon\gamma_\star^{-2}$, and since $\abs{\mathscr{C}(s)} = \pi\gamma_\star$ the claim follows.
\end{proof}

\derivativeProjectorDifference*

\begin{proof}
    Dropping dependencies, by \cref{eq:projector-derivative} and the contour estimate, $\|\ps  P - \ps  \widehat{P}\| \le \frac{\abs{\mathscr{C}}}{2\pi} \sup_{z\in\mathscr{C}} \|\Omega(z)\|$ with $\Omega \coloneqq R(\ps  H)R - \widehat{R}(\ps \widehat{H})\widehat{R}$.
    Telescoping,
    \begin{equation}
        \Omega = (R - \widehat{R})(\ps  H)R + \widehat{R}(\ps  H - \ps \widehat{H})R + \widehat{R}(\ps \widehat{H})(R - \widehat{R}),
    \end{equation}
    with $\ps  H - \ps \widehat{H} = \ps  E$.
    By \cref{lem:resolvent-bounds} and $\|R - \widehat{R}\| \le 8\varepsilon\gamma_\star^{-2}$ (proof of \cref{prop:spectral-projector-bound}),
    \begin{equation}
        \|\Omega\| \le \frac{8\varepsilon}{\gamma_\star^2}\norm{\ps H}\frac{2}{\gamma_\star} + \frac{4}{\gamma_\star}\norm{\ps E}\frac{2}{\gamma_\star} + \frac{4}{\gamma_\star}\norm{\ps\widehat{H}}\frac{8\varepsilon}{\gamma_\star^2} = O\big(\gamma_\star^{-3}\max\{\norm{\ps H}, \norm{\ps\widehat{H}}\}\varepsilon + \gamma_\star^{-2}\norm{\ps E}\big).
    \end{equation}
    Multiplying by $\abs{\mathscr{C}}/2\pi = \gamma_\star/2$ yields the claim.
\end{proof}

\katoIntertwining*

\begin{proof}
    Differentiating $P^2 = P$ gives $\ps  P = (\ps  P)P + P(\ps  P)$; multiplying by $P$ on both sides yields $P(\ps  P)P = 0$, whence
    \begin{equation}
        \comm{{\rm G}}{P} = (\ps  P)P + P(\ps  P) - 2P(\ps  P)P = \ps  P.
    \end{equation}
    Setting $Q(s) \coloneqq T^\dagger(s) P(s) T(s)$ and using $\ps  T = {\rm G}T$, $\ps  T^\dagger = -T^\dagger {\rm G}$, we get $\ps  Q = T^\dagger(\ps  P - \comm{{\rm G}}{P})T = 0$, so $Q(s) = Q(0) = P(0)$.
\end{proof}

\berryphaseholonomy*

\begin{proof}
    Let $\ket{\xi(s)} = T(s)\ket{\phi_0(0)}$.
    By \cref{eq:xi-in-image} and rank-one-ness, $\ket{\xi(s)} \propto \ket{\phi_0(s)}$; as both are normalised, $c(s) \coloneqq \braket{\phi_0(s)}{\xi(s)}$ is unimodular and ${\rm C}^1$, so it admits a ${\rm C}^1$ phase lift $\beta$ with $c(s) = \e^{\i\beta(s)}$ and, since $\ket{\xi(0)} = \ket{\phi_0(0)}$, $\beta(0) = 0$.
    Thus, $\ket{\xi(s)} = \e^{\i\beta(s)}\ket{\phi_0(s)}$, and taking the inner product of $\ket{\xi(1)}$ with $\ket{\phi_0(0)}$ and using \cref{eq:holonomy-phase},
    \begin{equation}\label{eq:trace-via-beta}
        \Tr[P(0)\, T(1)] = \braket{\phi_0(0)}{\xi(1)} = \e^{\i\beta(1)}\braket{\phi_0(0)}{\phi_0(1)},
    \end{equation}
    so it suffices to identify $\beta(1)$.
    On one hand $\ps \ket{\xi} = {\rm G}\ket{\xi} = \e^{\i\beta} {\rm G}\ket{\phi_0}$ with $\mel{\phi_0}{{\rm G}}{\phi_0} = 0$ (as in \cref{eq:parallel-transport}, using $P\ket{\phi_0} = \ket{\phi_0}$), so $\mel{\phi_0}{\ps }{\xi} = 0$.
    On the other hand, differentiating $\ket{\xi} = \e^{\i\beta}\ket{\phi_0}$ gives $\mel{\phi_0}{\ps }{\xi} = \e^{\i\beta}(\i\,\ps \beta + \mel{\phi_0}{\ps }{\phi_0})$.
    Comparing and dividing by $\e^{\i\beta} \ne 0$ yields $\ps \beta = \i\mel{\phi_0}{\ps }{\phi_0} = A(s)$, and integrating with $\beta(0) = 0$ gives $\beta(1) = \varTheta$; substituting into \cref{eq:trace-via-beta} proves \cref{eq:berry-holonomy-identity}.
\end{proof}

\berryPhaseComparison*

\begin{proof}
    By Assumption~\ref{item:a3}, both paths are closed, so by \cref{prop:berry-holonomy} in periodic gauges $\Tr[P(0)T(1)] = \e^{\i\vartheta}$ and $\Tr[\widehat{P}(0)\widehat{T}(1)] = \e^{\i\widehat{\vartheta}}$ are unimodular.
    The triangle inequality gives
    \begin{equation}
        \big|\e^{\i\vartheta} - \e^{\i\widehat{\vartheta}}\big| \le \big|\Tr[P(0)(T(1) - \widehat{T}(1))]\big| + \big|\Tr[(P(0) - \widehat{P}(0))\widehat{T}(1)]\big|.
    \end{equation}
    By H\"older's trace inequality, the first term is $\le \|P(0)\|_1\|T(1) - \widehat{T}(1)\| = \|T(1) - \widehat{T}(1)\|$.
    For the second, $P(0) - \widehat{P}(0)$ has rank $\le 2$, so $\|P(0) - \widehat{P}(0)\|_1 \le 2\|P(0) - \widehat{P}(0)\|$ and the term is $\le 2\|P(0) - \widehat{P}(0)\| = O(\gamma_\star^{-1}\varepsilon)$ by \cref{prop:spectral-projector-bound}.
    Combining these estimates with \cref{prop:duhamel}, bounding the generator-difference integral by the uniform estimates of \cref{lem:generator-comparison}, and applying Claim~\ref{claim:circular-distance} yields the result.
\end{proof}

\pullbackVsSimulator*

\begin{proof}
    Differentiating $\ket{\psi_0} = \U^\dagger\ket{\wt{\phi}_0}$ and substituting into $\widehat{A}$,
    \begin{equation}
        \wt{A}(s) - \widehat{A}(s) = \i\bra{\wt{\phi}_0}(I - \U\U^\dagger)\ps \ket{\wt{\phi}_0} - \i\mel{\wt{\phi}_0}{\U(\ps \U^\dagger)}{\wt{\phi}_0}.
    \end{equation}
    Since $\U\U^\dagger$ is the orthogonal projector onto the low-energy band containing $\ket{\wt{\phi}_0}$, the first term vanishes, so $\big|\wt{A}(s)-\widehat A(s)\big|\le\norm{\ps \U}$.
    Integrating gives the same inequality for compatible real lifts of the two phases; minimizing over multiples of $2\pi$ gives the stated circular-distance bound.
\end{proof}

\berryPhaseSimulation*

\begin{proof}
    By \cref{prop:berry-holonomy} applied to $\widehat{\Upsilon}$ in the periodic gauge $\ket{\psi_0(s)}$, the phase $\widehat{\vartheta}$ of \cref{lem:pullback-vs-simulator} agrees modulo $2\pi$ with the holonomy phase $\arg\Tr[\widehat{P}(0)\widehat{T}(1)]$ of \cref{lem:berry-phase-comparison}, so the two results combine.
    The triangle inequality for $d_{2\pi}$ gives
    \begin{equation}
        d_{2\pi}(\vartheta,\wt\vartheta)\le d_{2\pi}(\vartheta,\widehat\vartheta)+d_{2\pi}(\widehat\vartheta,\wt\vartheta),
    \end{equation}
    and the claim follows from \cref{lem:berry-phase-comparison,lem:pullback-vs-simulator}.
\end{proof}

%% file: appendices/section-proofs/perturbative-gadget-reductions-proofs.tex
\section{Proofs of Main Results in Section~\ref{sec:perturbative-gadget-reductions}}\label{app:perturbative-gadget-reductions-proofs}
In this appendix, we include some additional supporting proofs for the main results presented in Section~\ref{sec:perturbative-gadget-reductions}.
We first outline these supporting results then proceed with the proofs of the main results.

\subsection{Supporting Results}

\begin{restatable}{proposition}{SWGenSimpleBounds}\label{prop:SW-gen-simple-bounds}
    Let $S$ be anti-Hermitian with $\norm{S} = o(1)$, and let $X$ be any operator.
    Then
    \begin{equation}
        \norm{\e^{S} X \e^{-S} - X} \le 2\norm{S}\norm{X}\big(1 + O(\norm{S})\big),
        \qquad
        \norm{\e^{\pm S} - I} = O(\norm{S}).
    \end{equation}
\end{restatable}

\begin{proof}
    Expanding in nested commutators, $\e^{S} X \e^{-S} - X = \sum_{p\ge 1}\tfrac{1}{p!}\L_S^p(X)$ with $\norm{\L_S^p(X)} \le (2\norm{S})^p\norm{X}$.
    Hence, $\norm{\e^{S} X \e^{-S} - X} \le (\e^{2\norm{S}} - 1)\norm{X}$, and $\e^{2\norm{S}} - 1 = 2\norm{S}(1 + O(\norm{S}))$ since $\norm{S} = o(1)$.
    The second bound follows identically from $\norm{\e^{\pm S} - I} \le \e^{\norm{S}} - 1 = O(\norm{S})$.
\end{proof}

\subsection{Main Results}

\carriedDrivingSimulationError*

\begin{proof}
    Throughout, $\bs\V$ is the fixed encoding isometry, independent of $s$, satisfying
    $\bs\V^\dagger\bs\V = I$ and $\bs\V\bs\V^\dagger = \bs{\Pi}_-$. Consequently
    \begin{align}\label{eq:pi-minus-v-identity}
        \bs{\Pi}_-\bs\V &= \big(\bs\V\bs\V^\dagger\big)\bs\V = \bs\V\big(\bs\V^\dagger\bs\V\big) = \bs\V, &
        \bs\V^\dagger\bs{\Pi}_- &= \bs\V^\dagger .
    \end{align}
    We write $\bs X_{\alpha\beta} = \bs{\Pi}_\alpha\bs X\bs{\Pi}_\beta$ for $\alpha,\beta\in\{-,+\}$.

    We prove \cref{eq:carried-final-bounds} in several steps outlined below.

    \medskip
    \noindent\textsl{Step 1.}\\
    \indent Since $\bs\U(s) = \e^{-\bs{S}(s)}\bs\V$, and using \cref{eq:pi-minus-v-identity,eq:exact-effective-hamiltonian}, we find that
    \begin{equation}\label{eq:carried-conjugation-identity}
        \bs\U(s)^\dagger\wt{\bs{H}}_\zeta(s)\bs\U(s)
        = \bs\V^\dagger\e^{\bs{S}(s)}\wt{\bs{H}}_\zeta(s)\e^{-\bs{S}(s)}\bs\V
        = \bs\V^\dagger\Big(\bs{\Pi}_-\e^{\bs{S}(s)}\wt{\bs{H}}_\zeta(s)\e^{-\bs{S}(s)}\bs{\Pi}_-\Big)\bs\V
        = \bs\V^\dagger\bs{H}_{\tn{eff}}(s)\bs\V.
    \end{equation}
    Substituting
    $\bs{H}_{\tn{eff}} = \bs{H}^{[p]}_{\tn{eff}} + R^{\bs{H}_{\tn{eff}}}_{p+1}$ and recalling the
    definition \cref{eq:carried-matching-defect} of $\bs{M}_p$ gives
    \cref{eq:carried-error-decomposition},
    \begin{equation}
        \bs{E}_\zeta(s) = \bs{M}_p(s) - \bs\V^\dagger R^{\bs{H}_{\tn{eff}}}_{p+1}(s)\bs\V .
    \end{equation}

    \medskip
    \noindent\textsl{Step 2.}\\
    \indent The operator $\bmw(s)\otimes I_{{\rm med}}$ acts as the identity on the mediator qubits
    and therefore commutes with $\bs{\Pi}_\pm$, so it is block-diagonal. 
    Hence
    \begin{align}\label{eq:V-zeta-blocks}
        \big(\bs{V}_\zeta(s)\big)_{--} &= \zeta\bmw(s)\bs{\Pi}_-, &
        \big(\bs{V}_\zeta(s)\big)_{-+} &= \bs{V}_{-+}, \\
        \big(\bs{V}_\zeta(s)\big)_{++} &= \bs{V}_{++} + \zeta\bmw(s)\bs{\Pi}_+, &
        \big(\bs{V}_\zeta(s)\big)_{+-} &= \bs{V}_{+-} .
    \end{align}
    First note that the block-off-diagonal parts of $\bs{V}_\zeta$ are independent of $s$, and that
    \begin{equation}\label{eq:carried-pullback}
        \bs\V^\dagger\big(\bmw(s)\bs{\Pi}_-\big)\bs\V = \bmw(s)\,\bs\V^\dagger\bs\V = \bmw(s).
    \end{equation}
    Finally, since $\bs{\Pi}_\pm$ do not depend on $s$, we may differentiate blockwise: $\ps\big(\bs{V}_\zeta\big)_{\alpha\beta} = \big(\ps\bs{V}_\zeta\big)_{\alpha\beta}$.

    \medskip
    \noindent\textsl{Step 3.}\\
    \indent Differentiating \cref{eq:carried-matching-defect} and using
    $\ps\bs{H}_\zeta(s) = \zeta\,\ps\bmw(s)$ gives
    \begin{equation}\label{eq:carried-defect-derivative-raw}
        \ps\bs{M}_p(s) = \zeta\,\ps\bmw(s) - \bs\V^\dagger\bigg(\sum_{j=1}^{p}\ps\bs K_j(s)\bigg)\bs\V .
    \end{equation}
    By \cref{eq:K1-parallel} and \cref{eq:V-zeta-blocks}, the first coefficient is
    $\bs K_1 = \big(\bs{V}_\zeta\big)_{--}$, whose only $s$-dependent part is
    $\zeta\bmw(s)\bs{\Pi}_-$; hence, by \cref{eq:carried-pullback},
    \begin{equation}
        \bs\V^\dagger\big(\ps\bs K_1(s)\big)\bs\V = \zeta\,\ps\bmw(s).
    \end{equation}
    By \cref{eq:K2-parallel}, the second coefficient
    $\bs K_2 = -\Delta^{-1}\big(\bs{V}_\zeta\big)_{-+}\bs\Sigma\big(\bs{V}_\zeta\big)_{+-}$ involves
    only block-off-diagonal parts, which are static by \cref{eq:V-zeta-blocks}; hence
    $\ps\bs K_2 = 0$.
    Substituting both into \cref{eq:carried-defect-derivative-raw} we find
    \begin{equation}
        \ps\bs{M}_p(s) = 0,
    \end{equation}
    for $p\in\{1,2\}$.
    For $p = 3$ the same cancellation leaves only the third coefficient, giving the second case of
    \cref{eq:carried-defect-derivative},
    \begin{equation}\label{eq:carried-defect-p3}
        \ps\bs{M}_3(s) = -\,\bs\V^\dagger\big(\ps\bs K_3(s)\big)\bs\V .
    \end{equation}

    \medskip
    \noindent\textsl{Step 4.}\\
    \indent Differentiating \cref{eq:K3-parallel} and using \cref{eq:V-zeta-blocks},
    \begin{equation}\label{eq:carried-K3-leak}
        \ps\bs K_3(s) = \frac{\zeta}{\Delta^{2}}\left(
            \bs{V}_{-+}\bs\Sigma\big[\ps\bmw(s)\bs{\Pi}_+\big]\bs\Sigma\bs{V}_{+-}
            - \tfrac{1}{2}\acomm{\ps\bmw(s)\bs{\Pi}_-}{\bs{V}_{-+}\bs\Sigma^{2}\bs{V}_{+-}}\right).
    \end{equation}
    By \cref{eq:V-zeta-blocks} the off-diagonal blocks appearing here are those of $\bs{V}_\zeta$,
    so $\norm{\bs{V}_{\pm\mp}} = \norm{(\bs{V}_\zeta)_{\pm\mp}} \le \snorm{\bs{V}_\zeta} \le v_0$.
    Using $\norm{\bs\Sigma} = 1$ from \cref{clm:excitation-sector-decomposition},
    $\snorm{\ps\bmw\,\bs{\Pi}_\alpha} \le w_1$ by \cref{eq:carried-driving-scales}, and
    $\norm{\acomm{X}{Y}}\le 2\norm{X}\norm{Y}$, each of the two terms of
    \cref{eq:carried-K3-leak} is bounded by $\zeta\Delta^{-2}v_0^2 w_1$, so that
    \begin{align}\label{eq:carried-K3-bound}
        \snorm{\ps\bs K_3} &\le 2\zeta\Delta^{-2}v_0^{2}w_1, &
        \snorm{\ps\bs{M}_3} &\le \snorm{\ps\bs K_3} \le 2\zeta\Delta^{-2}v_0^{2}w_1,
    \end{align}
    the second inequality because $\bs\V$ is an isometry. This proves
    \cref{eq:carried-defect-derivative}.

    \medskip
    \noindent\textsl{Step 5.}\\
    \indent Since $\bs{V}$ is independent of $s$, we have
    $\ps\bs{V}_\zeta(s) = \zeta\big(\ps\bmw(s)\big)\otimes I_{{\rm med}}$ and hence
    $\snorm{\ps\bs{V}_\zeta}\le\zeta w_1$ by \cref{eq:carried-driving-scales}. Applying
    \cref{lemma:uniform-control-exact-truncated} with $v_1 = \zeta w_1$, gives
    \begin{align}\label{eq:carried-remainder-bounds}
        \snorm{R^{\bs{H}_{\tn{eff}}}_{p+1}} &= O\big(\Delta^{-p}v_0^{\,p+1}\big), &
        \snorm{\ps R^{\bs{H}_{\tn{eff}}}_{p+1}} &= O\big(\Delta^{-p}v_0^{\,p}\,\zeta w_1\big).
    \end{align}
    As $\bs\V$ is an isometry independent of $s$, \cref{eq:carried-error-decomposition} yields
    \begin{align}
        \snorm{\bs{E}_\zeta} &\le \snorm{\bs{M}_p} + \snorm{R^{\bs{H}_{\tn{eff}}}_{p+1}}, &
        \snorm{\ps\bs{E}_\zeta} &\le \snorm{\ps\bs{M}_p} + \snorm{\ps R^{\bs{H}_{\tn{eff}}}_{p+1}}.
    \end{align}
    The first of these together with \cref{eq:carried-remainder-bounds} is the first bound of
    \cref{eq:carried-final-bounds}.
    For $p\in\{1,2\}$ we have $\ps\bs{M}_p = 0$ by Step 3, so
    $\snorm{\ps\bs{E}_\zeta} = O(\Delta^{-p}v_0^{\,p}\,\zeta w_1)$.
    For $p = 3$, combining \cref{eq:carried-K3-bound,eq:carried-remainder-bounds} gives
    \begin{equation}
        \snorm{\ps\bs{E}_\zeta}
        \le 2\zeta\Delta^{-2}v_0^{2}w_1 + O\big(\Delta^{-3}v_0^{3}\,\zeta w_1\big)
        = O\big(\Delta^{-2}v_0^{2}\,\zeta w_1\big),
    \end{equation}
    the first term dominating since $v_0 < \Delta/2$. 
\end{proof}

\AssembledBoundsGeneral*

\begin{proof}
    The bounds on $J_2,J_3,J_4$ are immediate.
    Moreover, $\snorm{\ps\bs{H}_\zeta}\le \zeta w_1\le d_1$ and, since $\widehat{\bs H}_\zeta=\bs H_\zeta-\bs E_\zeta$,
    \begin{equation}
        \snorm{\ps\widehat{\bs H}_\zeta}
        \le \snorm{\ps\bs H_\zeta}+\snorm{\ps\bs E_\zeta}
        =O\big(d_1(1+A^2\Delta^{-q'})\big).
    \end{equation}
    Substitution into $J_1$ proves the remaining claim.
\end{proof}

\PolyControlGeneral*

\begin{proof}
    By \cref{prop:assembled-bounds-general}, there exists a constant $c\geq 1$ such that
    \begin{align}
        J_1 &\leq c\,\gamma_\star^{-2}d_1\bigl(1+A^2\Delta^{-q'}\bigr)\varepsilon, &
        J_2 &\leq c\,\gamma_\star^{-1}A^2d_1\Delta^{-q'}, \\
        J_3 &\leq \gamma_\star^{-1}\varepsilon, &
        \snorm{\ps\bs{U}} &\leq c\,Ad_1\Delta^{-q}.
    \end{align}
    Moreover, recall from \cref{eq:carried-driving-scales} that
    \begin{equation}
        d_1=\max\{1,\zeta w_1\}=O(1+\zeta B) \leq \poly{n}.
    \end{equation}

    We may assume without loss of generality that $\delta = 1/\poly{n}$.
    Choose
    \begin{align}\label{eq:eps-kappa-choice-general}
        \varepsilon &\coloneqq \min\bigg\{\frac{\gamma_\star\delta}{8},\frac{\gamma_\star^2\delta}{8c\,d_1}\bigg\}, &
        \kappa &\coloneqq \frac{\delta}{4}.
    \end{align}
    These are both inverse-polynomially small and are independent of $\Delta$.
    Moreover, $\varepsilon\leq\gamma_\star/8<\gamma_\star/4$, so Assumption~\ref{item:a2} is satisfied.

    Having fixed these targets, choose
    \begin{equation}\label{eq:delta-choice-general}
        \Delta \coloneqq 
        \max\bigg\{
            \Delta_\ast(\varepsilon,\eta,\kappa),
            A^{2/q'},
            \left(\frac{4c\,\gamma_\star^{-1}A^2d_1}{\delta}\right)^{1/q'},
            \left(\frac{4c\,Ad_1}{\delta}\right)^{1/q}
        \bigg\}.
    \end{equation}
    The first term ensures that the reduction is a valid $(\eta,\varepsilon,\kappa)$-regular simulation.
    The remaining terms imply
    \begin{align}
        A^2 \Delta^{-q'} &\leq 1, &
        c \gamma_\star^{-1} A^2 d_1 \Delta^{-q'} &\leq \frac{\delta}{4}, &
        c A d_1 \Delta^{-q} &\leq \frac{\delta}{4}=\kappa.
    \end{align}
    It follows that
    \begin{align}
        J_1 &\leq 2c\,\gamma_\star^{-2}d_1\varepsilon \leq\frac{\delta}{4}, &
        J_2 &\leq c\,\gamma_\star^{-1}A^2d_1\Delta^{-q'} \leq\frac{\delta}{4}, &
        J_3 &\leq\gamma_\star^{-1}\varepsilon \leq\frac{\delta}{8}, &
        J_4 &\leq\kappa =\frac{\delta}{4}.
    \end{align}
    Consequently,
    \begin{equation}
        J_1+J_2+J_3+J_4\leq\frac{7\delta}{8}\leq\delta.
    \end{equation}

    Finally, $\varepsilon^{-1},\eta^{-1},\kappa^{-1}=\poly{n}$, and hence
    \begin{equation}
        \Delta_\ast(\varepsilon,\eta,\kappa)=\poly{n}.
    \end{equation}
    Since $q$ and $q'$ are fixed positive constants and $A$, $d_1$, $\gamma_\star^{-1}$ and $\delta^{-1}$ are polynomially bounded, every other term in \cref{eq:delta-choice-general} is also polynomially bounded.
    Therefore $\Delta=\poly{n}$, completing the proof.
\end{proof}

\FirstOrderReduction*

\begin{proof}
    For each fixed $s\in[0,1]$, the hypotheses are those of the first-order reduction of \cite[Lemma~4]{bravyi2016complexity}.
    Moreover, $\snorm{V} \leq \norm{V_1} + \snorm{V_2} \leq 2L_{0}$.
    Since both the matching condition and the perturbation bound are uniform in $s$, the proof of \cite[Lemma~4]{bravyi2016complexity} applies with the same value of $\Delta$ at every $s$. 
    It establishes \cref{item:c1,item:c2,item:c3} whenever $\Delta = \Omega(\varepsilon^{-1}L_{0}^2 + \eta^{-1}L_{0})$.

    It remains to verify \cref{item:c4,item:c5}. 
    Let $S(s)$ be the canonical exact Schrieffer--Wolff generator and set $\U(s)=\e^{-S(s)}\V$ to define the family $\Phi = \{\U(s)\}_{s\in[0,1]}$.
    By \cref{lemma:uniform-control-exact-truncated}, $S$ and $\U$ are ${\rm C}^2$.

    Since $V(1)=V(0)$ and $H_0$ is independent of $s$, we have $\wt{H}(1)=\wt{H}(0)$.
    Therefore, the exact low-energy spectral projectors of $\wt H(0)$ and $\wt H(1)$ coincide. 
    By uniqueness of the canonical block-off-diagonal Schrieffer--Wolff generator satisfying $\snorm{S}<\pi/2$~\cite{bravyi2011schrieffer}, it follows that $S(1)=S(0)$.
    Hence, $\U(1)=\U(0)$, proving \cref{item:c4}.

    Because $V_1$ is independent of $s$ and $V_2(s)$ is block diagonal, $\partial_s \big( V(s) \big)_{\tn{od}} = 0$.
    Moreover, $\snorm{V} = O(L_{0})$ and $\snorm{\ps V} = L_{1}$.
    Applying \cref{cor:static-leading-sw-generator} gives $\snorm{\ps\U} = O(\Delta^{-2}L_{0}L_{1})$.
    Thus, $\Delta = \Omega(\kappa^{-1/2}(L_{0}L_{1})^{1/2})$ ensures $\snorm{\ps\U}\leq\kappa$, proving \cref{item:c5}.
\end{proof}

\ParallelFirstOrderSimulationBerryPhase*

\begin{proof}
    The validity threshold of \cref{eq:parallel-threshold-first-order} is polynomial in $\varepsilon^{-1}$, $\eta^{-1}$, $\kappa^{-1}$, and $A$, and the bounds of \cref{eq:parallel-first-order-bounds-pt1,eq:parallel-first-order-bounds-pt2} satisfy the hypotheses of \cref{prop:assembled-bounds-general} with $q=2$ and $q'=1$.
    Let $C$ be the constant implicit in \cref{thm:berry-phase-simulation}, and apply \cref{lem:poly-control-general} with target accuracy $\delta/C = 1/\poly{n}$.
    This yields simulation targets $\varepsilon,\kappa = 1/\poly{n}$ and a penalty $\Delta = \poly{n}$ for which $\wt{\Upsilon}$ is a regular $(\eta,\varepsilon,\kappa)$-simulation of $\Upsilon$ with $J_1+\cdots+J_4 \le \delta/C$.
    \Cref{thm:berry-phase-simulation} then gives $d_{2\pi}(\vartheta,\wt{\vartheta}) \le C\,(J_1+\cdots+J_4) \le \delta$.
\end{proof}

\SecondOrderReduction*

\begin{proof}
    The proof follows the same structure as that of \cref{lem:first-order-reduction} and the static result of \cite[Lemma~6]{bravyi2016complexity}, after verifying the added block-diagonal hypothesis on $V_2$; thus, \cref{item:c1,item:c2,item:c3,item:c4} hold analogously.
    We verify \cref{item:c5} for the third-order reduction.
    Since $\snorm{V} = O(\Delta^{2/3}L_0)$ and $\snorm{\ps V} = L_1$, \cref{cor:static-leading-sw-generator} gives $\snorm{\ps\U} = O(\Delta^{-4/3}L_0L_1)$.
    Thus, $\Delta = \Omega(\kappa^{-3/4}(L_0L_1)^{3/4})$ ensures $\snorm{\ps\U}\leq\kappa$, completing the verification of \cref{item:c5}.
\end{proof}

\ThirdOrderReduction*

\begin{proof}
    The proof follows the same structure as that of \cref{lem:first-order-reduction} and the static result of \cite[Lemma~5]{bravyi2016complexity}; thus, \cref{item:c1,item:c2,item:c3,item:c4} hold analogously.
    We verify \cref{item:c5} for the third-order reduction.
    Noting that $\snorm{V} = O(\Delta^{2/3} L_0)$ and $\snorm{\ps V} = L_1$, we apply \cref{cor:static-leading-sw-generator} to obtain $\snorm{\ps\U} = O(\Delta^{-3/4}L_{0}L_{1})$.
    Thus, $\Delta = \Omega(\kappa^{-4/3}(L_{0}L_{1})^{4/3})$ ensures $\snorm{\ps\U}\leq\kappa$, completing the verification of \cref{item:c5}.
\end{proof}

\SecondOrderParameterisedReduction*

\begin{proof}
    First note that $\snorm{V} = O(\Delta^{1/2}L_0)$ and $\snorm{\ps V} = O(\Delta^{1/2}L_1)$.
    Let $S(s)$ be the canonical Schrieffer--Wolff generator and set $\U(s) = \e^{-S(s)} \V$ where $\V$ is the fixed encoding isometry.
    Following the proofs of \cref{lem:first-order-reduction} and \cite[Lemma~4]{bravyi2016complexity} we establish that \cref{item:c1,item:c4} are satisfied.
    Next, recall that the second-order effective Hamiltonian $H_{\tn{eff}}^{[2]}(s)$ is given by the left-hand side of \cref{eq:second-order-parameterised-uniform-condition}, which exactly coincides with $\overline{H}(s)$ by assumption.
    By \cref{lemma:uniform-control-exact-truncated} we observe that $\snorm{\overline{H} - H_{\tn{eff}}} = \snorm{R_3^{H_{\tn{eff}}}} = O(\Delta^{-1/2}L_0^3) \leq \varepsilon$ and therefore, $\sup_s \norm{H(s) - \U(s)^\dagger \wt{H}(s) \U(s)} \leq \varepsilon$.
    Consequently, $\Delta = \Omega(\varepsilon^{-2}L_0^6)$ proves \cref{item:c2}.
    We see that \cref{eq:uniform-control-s} gives $\snorm{S} = O(\Delta^{-1/2}L_0)$, and hence $\snorm{\U - \V} = O(\Delta^{-1/2}L_0)$. 
    Therefore, $\Delta = \Omega(\eta^{-2}L_0^2)$ proves \cref{item:c3}.
    Finally, \cref{eq:uniform-control-s} gives $\snorm{\ps S} = O(\Delta^{-1/2}L_1)$, and hence $\snorm{\ps\U} = O(\Delta^{-1/2}L_1)$. 
    Hence, $\Delta = \Omega(\kappa^{-2}L_1^2)$ proves \cref{item:c5}.
\end{proof}

\begin{corollary}\label{cor:gadgetised-driving-simulation-error}
    Let $\Upsilon$ and $\wt\Upsilon$ be as in Definitions~\ref{def:parameter-dependent-second-order-layer} and~\ref{def:general-driving-simulator}.
    Let $v_0 > 0$ satisfy $\snorm{V} \leq v_0$ and $v_1 > 0$ satisfy $\snorm{\ps V} \leq v_1$.
    Take $\bs\V$ to be the fixed encoding isometry with image ${\rm Im}(\bs\Pi_-)$, put $\bs\U(s) = \e^{-\bs{S}(s)}\bs\V$, and set
    \begin{equation}
        \bs{E}(s) \coloneqq \bs{H}(s) - \bs\U(s)^\dagger \wt{\bs{H}}(s) \bs\U(s).
    \end{equation}
    Then,
    \begin{align}
        \snorm{\bs{E}} = O(\Delta^{-2} v_0^3), &
        \snorm{\ps \bs{E}} = O(\Delta^{-2} v_0^2 v_1), &
    \end{align}
\end{corollary}

\begin{proof}
    The proof follows that of Lemma~\ref{lem:carried-driving-simulation-error} noting that the second-order effective Hamiltonian is exactly the encoded target Hamiltonian, hence $\bs{M}(s) = 0$.
\end{proof}

\begin{proposition}\label{prop:parallel-second-order-exact-assembled-bounds}
    Set $e_1\coloneqq\max\left\{1,A,\snorm{\bs H},\snorm{\ps\bs H}\right\}$.
    Suppose an exact second-order parallel gadget reduction with penalty $\Delta$ satisfies, for some $q \in (0,1]$,
    \begin{align}
        \kappa = O(\Delta^{-q} A e_1), &
        \snorm{\ps \bs{E}} = O(\Delta^{-q} A^2 e_1).
    \end{align}
    Then, 
    \begin{align}
        J_1 &= O(\gamma_\star^{-2} e_1(1 + A^2 \Delta^{-q})\varepsilon), & J_2 &= O(\gamma_\star^{-1}e_1A^2 \Delta^{-q})\\
        J_3 &= \gamma_\star^{-1}\varepsilon, & J_4 &= O(e_1 A \Delta^{-q})
    \end{align}
\end{proposition}

\begin{proof}
    The proof follows an analogous argument to that of Proposition~\ref{prop:assembled-bounds-general}.
\end{proof}

%% file: appendices/section-proofs/extended-hardness-results-proofs.tex
\section{Proofs of Main Results in Section~\ref{sec:extended-hardness-results}}\label{app:extended-hardness-results-proofs}

\gsbpesimulatorhardness*

\begin{proof}
    By consequence of Assumption~\ref{item:a2} and Assumption~\ref{item:a3}, $\wt{\Upsilon}$ is closed, gapped, and non-degenerate, so $\wt{x}$ is syntactically valid.

    \medskip
    \noindent\textsl{Promise.} 
    By \cref{thm:berry-phase-simulation}, $d_{2\pi}(\vartheta, \wt{\vartheta}) \le \upsilon$.
    If $x\in\sc{yes}$ then $\vartheta \in [a+\g, b-\g]_{2\pi}$, so $\wt{\vartheta} \in [a+\g-\upsilon, b-\g+\upsilon]_{2\pi} = [a+\g', b-\g']_{2\pi} = \wt{I}_{\sc{yes}}$; the \sc{no} case is identical.
    Since $2\upsilon < \g$, the arcs stay disjoint with separation $2\g' \ge 1/\poly{n}$.

    \medskip
    \noindent\textsl{Overlap.} 
    Set $\ket{a} = \ket{\wt{\phi}_0(0)}$, $\ket{b} = \V\ket{\phi_0(0)}$, $\ket{c} = \V\ket{\xi} = \ket{\wt{\xi}}$.
    By \cref{lem:ground-state-simulation}, $\norm{\ket{a} - \ket{b}} \le \mu$, and since $\V^\dagger\V = I$, $\abs{\braket{b}{c}}^2 = \abs{\braket{\phi_0(0)}{\xi}}^2 \ge \delta_0$.
    The norm-tracking lemma gives $\abs{\braket{\wt{\phi}_0(0)}{\wt\xi}}^2\ge(\sqrt{\delta_0}-\mu)^2=\delta_0'$.
    The stronger hypothesis $\mu \le \sqrt{\delta_0}/2$ implies $\delta_0' \ge \delta_0/4$.
    Finally, $C_{\wt\xi}$ is polynomially sized since it is the concatenation of $C_\xi$ and a polynomial-size description of $\V$.
    From Ref.~\cite{waite2026physically} we know this description satisfies Items~\cref{item:efficient-compilation,item:uniform-preparation}.

    \medskip
    \noindent\textsl{Computability.} 
    The presentation of $\wt{\Upsilon}$ and the description of $\V$ are computable from $x$ in polynomial time.
    The simulator contains polynomially many qubits and local terms, and $C_{\wt\xi}$ is obtained from $C_\xi$ by composing polynomially many explicitly described local isometries.
    Finally, $\g/2$ and $\delta_0/4$ have polynomial-bit descriptions and are computable by elementary rational arithmetic.
    Therefore, the map $x\mapsto\wt{x}$ is a polynomial-time many-one reduction.
\end{proof}

\spatiallysparsefivelocal*

\begin{proof}
    The proof follows by applying the spatially-sparse reduction of Ref.~\cite{oliveira2008complexity} to the \cl{BQP}-hardness construction of Ref.~\cite{hayakawa2025computational}.
    The resulting local Hamiltonians are of identical form to those in Ref.~\cite{hayakawa2025computational}, but with a spatially-sparse interaction graph (cf. Ref.~\cite{waite2026physically}).
    We conclude that the guiding state's overlap and classical description are preserved under the spatially-sparse reduction by the results of Ref.~\cite{waite2026physically}, thus \cref{prop:gsbpe-simulator-hardness} applies.
    Containment in \cl{BQP} follows trivially from the existing algorithm in Ref.~\cite{hayakawa2025computational}, which applies to every constant-locality family satisfying the promises in \cref{sec:problem-statement}.
\end{proof}

\maintheorem*

\begin{proof}
    Containment in \cl{BQP} is the guided-state Berry-phase algorithm of Ref.~\cite{hayakawa2025computational}, which applies to every constant-locality family satisfying the promises in \cref{sec:problem-statement}.
    Hardness follows by applying the geometrical reduction chain of Ref.~\cite{oliveira2008complexity} to the spatially-sparse $5$-local Hamiltonians of \cref{theorem:spatially-sparse-5-local}.
    The procedure occurs in the following way: 
    \begin{enumerate}
        \item Reduce the locality of the bulk terms from $5$-local to $3$-local using the \textsl{subdivision} gadget, requiring only second-order perturbation gadget constructions in parallel;
        \item Reduce the locality of the bulk terms from $3$-local to $2$-local using the \textsl{3-to-2-local} gadget, requiring a third-order perturbation gadget construction in parallel;
        \item Reduce the Pauli-degree of each qubit to $3$ using repeated applications of the \textsl{fork} gadget, requiring only a constant number of second-order perturbation gadget constructions;
        \item Remove any non-planar crossings in the interaction graph using the \textsl{cross} gadget, requiring only a constant number of second-order perturbation gadget constructions;
        \item Embed the resulting planar graph into a square or triangular lattice using the \textsl{subdivision} gadget to reroute edges, requiring only a constant number of second-order perturbation gadget constructions.
    \end{enumerate}
    See \cref{fig:geometric-embedding} for an illustration of the geometric embedding procedure, adpated from Ref.~\cite{oliveira2008complexity}.
    Note that this illustration holds for both the square and triangular lattice embeddings since the square lattice is a special case of the triangular lattice.
    Following the simulation parameter choices outlined in \cref{sec:uniform-simulation} for a composition of a constant number of gadget reductions and the error analysis of Section~\ref{sec:perturbative-gadget-reductions}, we conclude that at each stage, there exists a sufficiently large, but polynomially-large penalty $\Delta$ such that the resulting Berry phase error is sufficiently small.
    We conclude that the guiding state's overlap and classical description are preserved under the geometrical reductions by the results of Ref.~\cite{waite2026physically}, thus \cref{prop:gsbpe-simulator-hardness} applies.
\end{proof}

%% file: appendices/section-proofs/berry-phase-estimation-parameterised-1-local-proofs.tex
\section{Proofs of Main Results in Section~\ref{sec:berry-phase-estimation-parameterised-1-local}}\label{app:berry-phase-estimation-parameterised-1-local-proofs}

\oneLocalDiscreteBerryPhase*

\begin{proof}
    For each $j$, choose a normalised parallel-transported ground state $\ket{u_j(s)}$ satisfying $\braket{u_j(s)}{\partial_su_j(s)}=0$.
    Let $Q_j(s)=I-P_j(s)$ and let $R_j(s)$ be the reduced resolvent of $g_j(s)$ above its ground state, so $\snorm{R_j}\leq\gamma_\star^{-1}$.
    Write $\lambda_{0,j}(s)$ for the ground-state eigenvalue of $g_j(s)$.
    Differentiating the eigenvalue equation once and projecting onto $Q_j(s)$ gives
    \begin{equation}\label{eq:one-local-first-eigenvector-identity}
        \partial_s\ket{u_j(s)}
        =-R_j(s)Q_j(s)\big(\partial_sg_j(s)\big)\ket{u_j(s)},
    \end{equation}
    where we used the parallel-transport condition to remove the component in the image of $P_j(s)$.
    Consequently,
    \begin{equation}\label{eq:one-local-first-eigenvector-derivative}
        \norm{\partial_su_j(s)}\leq\frac{L_1}{\gamma_\star}.
    \end{equation}
    Differentiating the eigenvalue equation a second time gives
    \begin{equation}\label{eq:one-local-second-eigenvector-identity}
            Q_j(s)\partial_s^2\ket{u_j(s)} = - R_j(s) Q_j(s) \big[\big(\ps^2 g_j(s)\big)\ket{u_j(s)} + 2\big(\ps g_j(s)-(\ps \lambda_{0,j}(s))I\big)\ps\ket{u_j(s)} \big].
    \end{equation}
    The Hellmann--Feynman identity gives $\abs{\partial_s\lambda_{0,j}(s)}\leq L_1$.
    Applying \cref{eq:one-local-first-eigenvector-derivative} to \cref{eq:one-local-second-eigenvector-identity} therefore gives
    \begin{equation}\label{eq:one-local-second-eigenvector-derivative-orthogonal}
        \norm{Q_j(s)\partial_s^2u_j(s)}
        \leq\frac{L_2}{\gamma_\star}+4\frac{L_1^2}{\gamma_\star^2}.
    \end{equation}
    Differentiating the parallel-transport condition gives
    \begin{equation}\label{eq:one-local-parallel-second-derivative}
        \braket{u_j(s)}{\partial_s^2u_j(s)}
        =-\norm{\partial_su_j(s)}^2.
    \end{equation}
    Combining \cref{eq:one-local-first-eigenvector-derivative,eq:one-local-second-eigenvector-derivative-orthogonal,eq:one-local-parallel-second-derivative} yields
    \begin{equation}\label{eq:one-local-second-eigenvector-derivative}
        \sup_{s\in[0,1]}\norm{\partial_s^2u_j(s)}\leq M_2.
    \end{equation}

    Put $h=1/N$ and $a_{j,\ell}=\braket{u_j(s_{\ell+1})}{u_j(s_\ell)}$, where $s_N=1$.
    The parallel-transport condition and the fundamental theorem of calculus imply
    \begin{equation}\label{eq:one-local-neighbouring-overlap-bound}
        \abs{a_{j,\ell}-1}
        \leq\int_{s_\ell}^{s_{\ell+1}}\dd{t}\,\int_{s_\ell}^{t}\dd{r}\,\norm{\partial_r^2u_j(r)}
        \leq\frac{M_2}{2N^2}.
    \end{equation}
    Since $N\geq2n\overline M_2$, every $a_{j,\ell}$ lies in the disk of radius $1/2$ centred at $1$ and hence $\abs{\arg a_{j,\ell}}\leq2\abs{a_{j,\ell}-1}\leq M_2/N^2$.

    Let $\vartheta_j$ be the local Berry phase.
    Parallel transport gives $\ket{u_j(1)}=\e^{\i\vartheta_j}\ket{u_j(0)}$, and therefore
    \begin{equation}\label{eq:one-local-discrete-continuous-comparison}
        \varDelta_{j,N}=\e^{\i\vartheta_j}\prod_{\ell=0}^{N-1}a_{j,\ell}.
    \end{equation}
    Summing the arguments of the $nN$ neighbouring overlaps and using \cref{eq:one-local-berry-phase-sum} proves the first bound in \cref{eq:one-local-discrete-bp-bound}.
    Moreover, \cref{eq:one-local-neighbouring-overlap-bound} and the inequality $\prod_k(1-x_k)\geq1-\sum_kx_k$ give
    \begin{equation}
        \abs{\varDelta_N}
        \geq\left(1-\frac{M_2}{2N^2}\right)^{nN}
        \geq1-\frac{nM_2}{2N}
        \geq\frac{3}{4},
    \end{equation}
    which completes the proof.
\end{proof}

\oneLocalBerryPhaseInP*

\begin{proof}
    Fix a target phase accuracy $\epsilon_{\rm alg}\in(0,1)$ and choose
    \begin{equation}\label{eq:one-local-mesh-size}
        N=\left\lceil\frac{4n\overline M_2}{\epsilon_{\rm alg}}\right\rceil.
    \end{equation}
    By \cref{lem:one-local-discrete-bp}, the exact discrete phase satisfies $d_{2\pi}(\vartheta_N,\vartheta)\leq\epsilon_{\rm alg}/4$ and $\abs{\varDelta_N}\geq3/4$.

    The presentation supplies $L_1$, $L_2$, and $\gamma_\star^{-1}=\poly{n}$, so $N=\poly{n,\epsilon_{\rm alg}^{-1}}$.
    Every mesh point $s_\ell=\ell/N$ has a binary description of length $O(\log N)$.
    At each mesh point, the local terms can be grouped by qubit and their scalar parts removed in classical polynomial time.
    For a traceless $2\times2$ Hermitian matrix $g_j(s)$, the ground-state projector is
    \begin{equation}\label{eq:one-local-explicit-projector}
        P_j(s)=\frac{1}{2}\left(I-\frac{g_j(s)}{\norm{g_j(s)}}\right),
        \qquad
        \norm{g_j(s)}=\frac{\gamma_j(s)}{2}\geq\frac{\gamma_\star}{2}.
    \end{equation}
    Thus, the projector can be evaluated easily and the inverse-polynomial gap makes the normalisation in \cref{eq:one-local-explicit-projector} polynomially well-conditioned.

    Query the presentation and perform the constant-dimensional matrix operations using $b_{\rm prec}$ bits of precision, where
    \begin{equation}\label{eq:one-local-working-precision}
        b_{\rm prec}=\left\lceil c_{\rm prec}\log_2\!\left(\frac{mnN}{\epsilon_{\rm alg}\gamma_\star}\right)\right\rceil
    \end{equation}
    for a sufficiently large absolute constant $c_{\rm prec}$.
    This choice ensures that each grouped approximation to $g_j(s_\ell)$ is within $\gamma_\star/4$ of the exact matrix and therefore has norm at least $\gamma_\star/4$.
    The map $g\mapsto\frac{1}{2}(I-g/\norm{g})$ is $O(\gamma_\star^{-1})$-Lipschitz on the promised domain.
    Grouping the $m$ presented terms introduces at most an additional factor of $m$ in the evaluation error.
    A telescoping expansion of the products in \cref{eq:one-local-bargmann-invariant,eq:one-local-discrete-phase} therefore shows that the accumulated error from the projector evaluations and fixed-precision matrix operations is $O(mnN2^{-b_{\rm prec}}/\gamma_\star)$.
    Choosing $c_{\rm prec}$ sufficiently large makes the resulting approximation $\widehat \varDelta_N$ satisfy $\abs{\widehat \varDelta_N-\varDelta_N}\leq\epsilon_{\rm alg}/16$.
    Since $\abs{\varDelta_N}\geq3/4$ and $\epsilon_{\rm alg}<1$, the approximation $\widehat \varDelta_N$ is nonzero and
    \begin{equation}\label{eq:one-local-argument-stability}
        d_{2\pi}(\arg\widehat \varDelta_N,\arg \varDelta_N)
        \leq\frac{2\abs{\widehat \varDelta_N-\varDelta_N}}{\abs{\varDelta_N}}
        \leq\frac{\epsilon_{\rm alg}}{6}.
    \end{equation}
    Compute a $b_{\rm prec}$-bit representative $\widehat\vartheta\in[0,2\pi)$ such that
    \begin{equation}\label{eq:one-local-output-argument-precision}
        d_{2\pi}(\widehat\vartheta,\arg\widehat \varDelta_N)
        \leq\frac{\epsilon_{\rm alg}}{12}.
    \end{equation}
    Combining this output-precision bound with the numerical and discretisation errors gives
    \begin{equation}\label{eq:one-local-final-estimation-error}
        d_{2\pi}(\widehat\vartheta,\vartheta)
        \leq\frac{\epsilon_{\rm alg}}{12}+\frac{\epsilon_{\rm alg}}{6}+\frac{\epsilon_{\rm alg}}{4}
        \leq\epsilon_{\rm alg}.
    \end{equation}

    The algorithm uses $O(mN+nN)$ arithmetic operations on $b_{\rm prec}=O(\log n+\log\epsilon_{\rm alg}^{-1})$-bit numbers under the supplied polynomial bounds.
    Standard fixed-precision algorithms for arithmetic, square roots, and the argument of a nonzero complex number run in time polynomial in $b_{\rm prec}$.
    Hence the total bit complexity is polynomial in $n$ and $\epsilon_{\rm alg}^{-1}$.

    For an instance of \sc{GSBPE}, choose $\epsilon_{\rm alg}=\g/4$.
    The promise arcs are separated by $2\g$, so the estimate in \cref{eq:one-local-final-estimation-error} determines which promised arc contains $\vartheta$ in classical polynomial time.
    The construction never uses the guiding state, completing the proof.
\end{proof}

%% file: appendices/perturbative-gadget-extensions.tex
\section{Perturbative Gadget Extensions}
\label{app:perturbative-gadget-extensions}
In this appendix we analyse explicit examples of parameter-dependent gadgets, extending the discussions of both Ref.~\cite{oliveira2008complexity} and Ref.~\cite{schuch2009computational}.
Ref.~\cite{oliveira2008complexity} introduces the gadgets called: \textsl{subdivision}, \textsl{fork} and \textsl{cross}; we call this ``family 1''.
Ref.~\cite{schuch2009computational} introduces the gadgets called: \textsl{parameter-fixing}, \textsl{Pauli-to-Ising}, \textsl{Ising-to-XX} and \textsl{XX-to-Heisenberg}; we call this ``family 2''.

We start by setting up some general notation and results that will be used throughout this appendix.
Note that in the analysis of the gadget families from each reference, we frequently encounter the same form in the components defined.
For the sake of completeness, we repeat some of these calculations in the respective subsections below.

Note that the norm bounds on the relevant components can be directly obtained from \cref{lemma:uniform-control-exact-truncated}.
The calculations below are somewhat extensive, but we include them to verify explicitly that the static constructions admit the required parameterised extensions, that the penalty Hamiltonians and encoding isometries remain independent of $s$, and that all unwanted second-order terms are cancelled.
They may therefore be regarded as proofs that these gadgets can be used within the extended hardness reductions developed in this work.

\subsection{Setup}
We remark on some basic properties of functions that are ${\rm C}^2$ on the interval $[0,1]$ and periodic with $f(0) = f(1)$.
\begin{fact}\label{fact:c2-periodic-functions}
    Let $f, g : [0,1]\rightarrow\mathbb{R}$ be ${\rm C}^2$ functions with $f(0) = f(1)$ and $g(0) = g(1)$.
    Then the functions: 
    \begin{inparaenum}[(i)]
        \item $f(s)g(s)$,
        \item $f(s) + g(s)$,
        \item $f(s) - g(s)$,
        \item $f^p(s) + g^q(s)$ for any $p, q \in \mathbb{Z}^+$,
    \end{inparaenum}
    are all ${\rm C}^2$ and satisfy the same boundary conditions.
\end{fact}

In what follows, the functions $f$ and $g$ are assumed to be ${\rm C}^2$ and periodic with $f(0) = f(1)$ and $g(0) = g(1)$.
Additionally, the target operators are assumed to be single-qubit non-identity Pauli operators.
If we further assume that $\snorm{\ps^r f} = O(c_{r,f})$ and $\snorm{\ps^r g} = O(c_{r,g})$, it follows that $\snorm{\ps^r f^p} = O(c_{r,f}^p)$ and $\snorm{\ps^r g^q} = O(c_{r,g}^q)$ for any $p, q \in \mathbb{Z}^+$ where $p$ and $q$ scale as a constant with respect to the relevant variable.
The norms of the functions constructed from $f$ and $g$ as in Fact~\ref{fact:c2-periodic-functions} are also $O(\max\{c_{r,f}^p, c_{r,g}^q\})$ for any $p, q \in \mathbb{Z}^+$ and any derivative order $r$.
We shall assume $c_{r,f}, c_{r,g} \leq \poly{n}$ for all relevant derivative orders $r$ and where $n$ will be the system size of the Hamiltonians we consider.

For each target system we consider, we include a remainder term that is not directly simulated by the gadget.
The operator 
\begin{equation}
    \bs{\Gamma}(s) = \sum_{j} \Gamma_j(s)
\end{equation}
contains all remaining interactions in the target Hamiltonian.
Assume that the total number of terms $M$ in the remainder Hamiltonian is a polynomial in the system size and that the degree of each system qubit is a constant.
We assume that each $\Gamma_j(s)$ is at most $5$-local and may be written as $\Gamma_j(s) = g_j(s) P_j$, where $P_j$ is a tensor product of at most $5$ non-trivial Pauli operators and $g_j$ is ${\rm C}^2$ with $g_j(0) = g_j(1)$ for all $j$.
For simplicity, we assume that $c_{r,g_j} = O(1)$ for all relevant derivative orders $r$ and all $j$.

\subsection{Family 1 Gadgets}
In each gadget analysis below, we consider an unperturbed Hamiltonian $H_0 = \ketbra{1}_m$ that acts on a mediator qubit $m$.
The ground- and excited-space projectors are therefore $\Pi_- = \ketbra{0}_m$ and $\Pi_+ = \ketbra{1}_m$.
We also define $H_0^{\pattce} = \Pi_+$ and the fixed isometry to be $\V = I_{\rm sys} \otimes \ket{0}_m$ (therefore $\V\V^\dagger = \Pi_-$).

\subsubsection{Subdivision gadget}
\label{app:parameterised-subdivision-gadget}
This gadget is visualised in \cref{fig:subdivision}.
Let the parameterised target Hamiltonian be 
\begin{equation}
    H(s) = \bs{\Gamma}(s) + f(s)A_aB_b.
\end{equation}

Define
\begin{align}
    \mathdcal{Q}(s) &= A_a - f(s)B_b, &
    \mathdcal{C}(s) &= \frac{1}{2}(1 + f^2(s)),
\end{align}
and construct the simulator Hamiltonian $\wt{H}(s) = \Delta H_0 + V(s)$ with $V(s) = \sqrt{\Delta} V_1(s) + V_2(s)$ such that 
\begin{align}
    V_1(s) &= \frac{1}{\sqrt{2}} \mathdcal{Q}(s)X_m, &
    V_2(s) &= \mathdcal{C}(s) + \bs{\Gamma}(s).
\end{align}
A short calculation shows that 
\begin{align}
    \big( V(s) \big)_{--} &= \big(\mathdcal{C}(s) + \bs{\Gamma}(s)\big)\Pi_-, &
    \big( V(s) \big)_{++} &= \big(\mathdcal{C}(s) + \bs{\Gamma}(s)\big)\Pi_+, \\
    \big( V(s) \big)_{-+} &= \sqrt{\frac{\Delta}{2}} \mathdcal{Q}(s)\ketbra{0}{1}_m, &
    \big( V(s) \big)_{+-} &= \sqrt{\frac{\Delta}{2}} \mathdcal{Q}(s)\ketbra{1}{0}_m,
\end{align}
and therefore $\big( V(s) \big)_{\tn{od}} = \sqrt{\Delta/2}\,\mathdcal{Q}(s)X_m$ and $\big( V(s) \big)_{\tn{d}} = \mathdcal{C}(s) + \bs{\Gamma}(s)$.

The second-order effective Hamiltonian is given by
\begin{equation}\label{eq:subdivision-effective-hamiltonian::app}
    H_\tn{eff}^{[2]}(s) = \big( V(s) \big)_{--} - \frac{1}{\Delta} \big( V(s) \big)_{-+} \Pi_+ \big( V(s) \big)_{+-} = \left( \mathdcal{C}(s) + \bs{\Gamma}(s) - \frac{1}{2}\mathdcal{Q}(s)^2\right) \Pi_-.
\end{equation}
Noting that $A_a^2 = B_b^2 = I_{\tn{sys}}$ and $\comm{A_a}{B_b} = 0$, we have $\mathdcal{Q}(s)^2 = (1+f^2(s)) - 2 f(s) A_a B_b$.
Therefore, 
\begin{equation}
    H_\tn{eff}^{[2]}(s) = \left( \bs{\Gamma}(s) + f(s)A_aB_b\right)\Pi_- = H(s)\Pi_- = \V H(s) \V^\dagger.
\end{equation}

By Fact~\ref{fact:c2-periodic-functions}, both $\wt{H}(s)$ and $S^{[2]}(s)$ are closed ${\rm C}^2$ functions of $s$.
Applying \cref{lem:second-order-reduction} gives an $(\eta, \varepsilon, \kappa)$-regular simulation of the effective Hamiltonian up to second order, and \cref{cor:parallel-second-order-simulation-berry-phase} gives criteria for controlling the Berry phase to the desired accuracy.

\subsubsection{Fork gadget}
\label{app:parameterised-fork-gadget}
This gadget is visualised in \cref{fig:fork}.
Let the parameterised target Hamiltonian be
\begin{equation}
    H(s) = \bs{\Gamma}(s) + f(s)A_aB_b + g(s)A_aC_c.
\end{equation}
Define 
\begin{align}
    \mathdcal{Q}(s) &= A_a - f(s) B_b - g(s) C_c, &
    \mathdcal{C}(s) &= \frac{1}{2} \left( 1 + f^2(s) + g^2(s) \right) + f(s)g(s) B_b C_c.
\end{align}
Note that in this case we can decompose both $\mathdcal{Q}$ and $\mathdcal{C}$ into a sum over components, but for brevity we will not do so here.
We construct the simulator Hamiltonian $\wt{H}(s) = \Delta H_0 + V(s)$ with $V(s) = \sqrt{\Delta} V_1(s) + V_2(s)$ such that 
\begin{align}
    V_1(s) &= \frac{1}{\sqrt{2}} \mathdcal{Q}(s) X_m, & 
    V_2(s) &= \mathdcal{C}(s) + \bs\Gamma(s).
\end{align}
A short calculation shows that
\begin{align}
    \big( V(s) \big)_{--} &= \big( \mathdcal{C}(s) + \bs\Gamma(s) \big) \Pi_-, & 
    \big( V(s) \big)_{++} &= \big( \mathdcal{C}(s) + \bs\Gamma(s) \big) \Pi_+, \\
    \big( V(s) \big)_{-+} &= \sqrt{\frac{\Delta}{2}} \mathdcal{Q}(s) \ketbra{0}{1}_m, &
    \big( V(s) \big)_{+-} &= \sqrt{\frac{\Delta}{2}} \mathdcal{Q}(s) \ketbra{1}{0}_m,
\end{align}
and therefore $\big( V(s) \big)_\tn{od} = \sqrt{\Delta/2} \mathdcal{Q}(s) X_m$ and $\big( V(s) \big)_\tn{d} = \mathdcal{C}(s) + \bs\Gamma(s)$.

The second-order effective Hamiltonian is given by the same expression as \cref{eq:subdivision-effective-hamiltonian::app}.
Noting that $A_a^2 = B_b^2 = C_c^2 = I_{\rm sys}$ and $\comm{A_a}{B_b} = \comm{A_a}{C_c} = \comm{B_b}{C_c} = 0$, we have $\mathdcal{Q}(s)^2 = (1 + f^2(s) + g^2(s)) + 2f(s) g(s) B_b C_c - 2f(s) A_a B_b - 2g(s) A_a C_c$.
Therefore, 
\begin{equation}
    H_\tn{eff}^{[2]}(s) = \big( \bs\Gamma(s) + f(s)A_aB_b + g(s)A_aC_c \big) \Pi_- = H(s)\Pi_- = \V H(s) \V^\dagger.
\end{equation}

By Fact~\ref{fact:c2-periodic-functions}, both $\wt{H}(s)$ and $S^{[2]}(s)$ are closed ${\rm C}^2$ functions of $s$.
Applying \cref{lem:second-order-reduction} gives an $(\eta, \varepsilon, \kappa)$-regular simulation of the effective Hamiltonian up to second order, and \cref{cor:parallel-second-order-simulation-berry-phase} gives criteria for controlling the Berry phase to the desired accuracy.

\subsubsection{Cross gadget}
\label{app:parameterised-cross-gadget}
This gadget is visualised in \cref{fig:cross}.
Let the parameterised target Hamiltonian be 
\begin{equation}
    H(s) = \bs{\Gamma}(s) + f(s) A_a B_b + g(s) C_cD_d.
\end{equation}
Define 
\begin{align}\label{eq:cross-gadget-QC-components::app}
    \mathdcal{Q}(s) &= \sum_{\ell = 1}^{4} \mathdcal{Q}^{\ell}(s), & 
    \mathdcal{C}(s) &= \sum_{\ell = 1}^{5} \mathdcal{C}^{\ell}(s),
\end{align}
with the components for $\mathdcal{Q}(s)$ given by
\begin{align}
    \mathdcal{Q}^{1}(s) &= A_a, & 
    \mathdcal{Q}^{2}(s) &= - f(s) B_b, & 
    \mathdcal{Q}^{3}(s) &= C_c, & 
    \mathdcal{Q}^{4}(s) &= - g(s) D_d,
\end{align}
and the components for $\mathdcal{C}(s)$ given by
\begin{align}
    \mathdcal{C}^{1}(s) &= \frac{1}{2}(2+f^2(s)+g^2(s)), &
    \mathdcal{C}^{2}(s) &= A_aC_c, &
    \mathdcal{C}^{3}(s) &= -f(s)B_bC_c, \\
    \mathdcal{C}^{4}(s) &= -g(s)A_aD_d, &
    \mathdcal{C}^{5}(s) &= f(s)g(s)B_bD_d.
\end{align}

We construct the simulator Hamiltonian $\wt{H}(s) = \Delta H_0 + V(s)$ with $V(s) = \sqrt{\Delta} V_1(s) + V_2(s)$ such that 
\begin{align}
    V_1(s) &= \frac{1}{\sqrt{2}} \left( \sum_{\ell = 1}^{4} \mathdcal{Q}^{\ell}(s) \right)X_m, & 
    V_2(s) &= \sum_{\ell = 1}^{5} \mathdcal{C}^{\ell}(s) + \bs\Gamma(s).
\end{align}
A short calculation shows that
\begin{align}
    \big( V(s) \big)_{--} &= \left( \sum_{\ell = 1}^{5} \mathdcal{C}^{\ell}(s) + \bs\Gamma(s) \right) \Pi_-, & 
    \big( V(s) \big)_{++} &= \left( \sum_{\ell = 1}^{5} \mathdcal{C}^{\ell}(s) + \bs\Gamma(s) \right) \Pi_+, \\
    \big( V(s) \big)_{-+} &= \sqrt{\frac{\Delta}{2}}\left( \sum_{\ell = 1}^{4} \mathdcal{Q}^{\ell}(s) \right) \ketbra{0}{1}_m, &
    \big( V(s) \big)_{+-} &= \sqrt{\frac{\Delta}{2}}\left( \sum_{\ell = 1}^{4} \mathdcal{Q}^{\ell}(s) \right) \ketbra{1}{0}_m,
\end{align}
and therefore $\big( V(s) \big)_\tn{od} = \sqrt{\Delta/2} \mathdcal{Q}(s) X_m$ and $\big( V(s) \big)_\tn{d} = \mathdcal{C}(s) + \bs\Gamma(s)$, using \cref{eq:cross-gadget-QC-components::app}.

The second-order effective Hamiltonian is given by the same expression as \cref{eq:subdivision-effective-hamiltonian::app}.
Noting that $A_a^2 = B_b^2 = C_c^2 = D_d^2 = I_{\rm sys}$ and $\comm{A_a}{B_b} = \comm{A_a}{C_c} = \comm{B_b}{C_c} = \comm{A_a}{D_d} = \comm{B_b}{D_d} = \comm{C_c}{D_d} = 0$, we have $\mathdcal{Q}(s)^2 = (2 + f^2(s) + g^2(s)) + 2 ( A_a C_c - f(s) B_b C_c - g(s) A_a D_d + f(s) g(s) B_b D_d ) - 2 ( f(s) A_a B_b + g(s) C_c D_d )$.
Therefore, 
\begin{equation}
    H_\tn{eff}^{[2]}(s) = \big( \bs\Gamma(s) + f(s)A_aB_b + g(s)C_c D_d \big) \Pi_- = H(s)\Pi_- = \V H(s) \V^\dagger.
\end{equation}

By Fact~\ref{fact:c2-periodic-functions}, both $\wt{H}(s)$ and $S^{[2]}(s)$ are closed ${\rm C}^2$ functions of $s$.
Applying \cref{lem:second-order-reduction} gives an $(\eta, \varepsilon, \kappa)$-regular simulation of the effective Hamiltonian up to second order, and \cref{cor:parallel-second-order-simulation-berry-phase} gives criteria for controlling the Berry phase to the desired accuracy.

\subsection{Family 2 Gadgets}
We present the family 2 gadgets in the same mediator-qubit form used by Schuch and Verstraete~\cite{schuch2009computational}, which builds on the second-order perturbative-gadget framework of Oliveira and Terhal~\cite{oliveira2008complexity}.
As with the family 1 gadgets, one could introduce a uniform decomposition of the perturbation into terms of the form
\begin{equation}
    \sum_{a\in[A]}\sum_{\sigma\in\{X,Y,Z\}}\sum_{\ell=1}^{{\sf d}_{a,\mathcal{Q}}(\sigma)}\mathcal{Q}^{a,\ell}(s)\otimes\sigma_{m_a}.
\end{equation}
Such a decomposition is straightforward but obscures the simple structure of the individual constructions, so we instead treat each gadget directly.

Each construction introduces a mediator qubit with a gadget-dependent penalty Hamiltonian $H_0$ and fixed encoding isometry $\mathcal{V}$.
In each case, we construct a simulator of the same form used throughout this work, namely
\begin{equation}
    \wt{H}(s)=\Delta H_0 + \sqrt{\Delta}V_1(s) + V_2(s).
\end{equation}
The mediator operators in $V_1(s)$ are centred with respect to the chosen ground state so that $\bigl(V_1(s)\bigr)_{--}=0$, while $V_2(s)$ contains the carried interaction $\bs{\Gamma}(s)$ together with the scalar and one-local counterterms required to cancel the unwanted second-order contributions.

The gadgets are used in a chain which successively reproduce general Pauli interactions using Ising interactions, Ising interactions using XX interactions, and XX interactions using Heisenberg interactions.
In the corresponding static construction the \textsl{parameter-fixing} gadget uses the mediator qubit to fix the value of a particular coupling, ensuring that the effective interaction has the desired strength.
For a parameterised coupling $f(s)$, this would naturally introduce factors of $\sqrt{f(s)}$, which need not be ${\rm C}^2$ even when $f(s)$ is ${\rm C}^2$.
We avoid this by placing the full factor $f(s)$ on one mediator edge and assigning a fixed weight to the other; consequently, the fixed isometry remains independent of $s$.
The resulting weights of the perturbation are therefore typically unequal.

\subsubsection{Parameter-fixing gadget}
We note that the \textsl{parameter-fixing} gadet name is somewhat misleading in the context with which it is used here.
However, for the sake of consistency with the existing literature, we retain this terminology.

Let the parameterised target Hamiltonian be
\begin{equation}
    H(s) = \bs{\Gamma}(s) + f(s) A_a B_b.
\end{equation}
Define
\begin{align}
    \ket{g_{\pi/4}} &= \frac{1}{\sqrt{2}} (\ket{0} + \e^{\i \pi/4} \ket{1}), &
    \ket{e_{\pi/4}} &= \frac{1}{\sqrt{2}} (\ket{0} - \e^{\i \pi/4} \ket{1}) \\
    \Pi_{-,\pi/4} &= \ket{g_{\pi/4}}\bra{g_{\pi/4}}, &
    \Pi_{+,\pi/4} &= \ket{e_{\pi/4}}\bra{e_{\pi/4}}.
\end{align}
Let the fixed encoding isometry be $\V_{\pi/4} = I_{\rm sys} \otimes \ket{g_{\pi/4}}_m$ for the mediator qubit $m$.

Construct the simulator Hamiltonian $\wt{H}(s) = \Delta \Pi_{+,\pi/4} + V(s)$ with $V(s) = \sqrt{\Delta} V_1(s) + V_2(s)$ such that 
\begin{align}
    V_1(s) &= f(s) A_a \hat{X}_m + B_b \hat{Y}_m, &
    V_2(s) &= \bs{\Gamma}(s) + \frac{1}{2}(1 + f^2(s)),
\end{align}
where we have defined the operators $\hat{\sigma} = \sigma - \frac{1}{\sqrt{2}}I$ for $\sigma \in \{X_m, Y_m\}$.
Define $O_{\alpha\beta} = \Pi_{\alpha,\pi/4} O \Pi_{\beta,\pi/4}$ for $\alpha, \beta \in \{+,-\}$.
A short calculation shows that 
\begin{align}
    \big(V(s)\big)_{--} &= \left(\bs{\Gamma}(s)+\frac{1}{2}\big(1+f^2(s)\big)I_{\rm sys}\right)\Pi_{-,\pi/4}, \\
    \big(V(s)\big)_{++} &= \left(\bs{\Gamma}(s)+\frac{1}{2}\big(1+f^2(s)\big)I_{\rm sys}-\sqrt{2\Delta}\big(f(s)A_a+B_b\big)\right)\Pi_{+,\pi/4}, \\
    \big(V(s)\big)_{-+} &= \i\sqrt{\frac{\Delta}{2}}\big(-f(s)A_a+B_b\big)\ketbra{g_{\pi/4}}{e_{\pi/4}}, \\
    \big(V(s)\big)_{+-} &= -\i\sqrt{\frac{\Delta}{2}}\big(-f(s)A_a+B_b\big)\ketbra{e_{\pi/4}}{g_{\pi/4}}.
\end{align}
and therefore 
\begin{align}
    \big( V(s) \big)_\tn{od} &= \i \sqrt{\frac{\Delta}{2}}\big( -f(s) A_a + B_b \big) \left( \ketbra{g_{\pi/4}}{e_{\pi/4}} - \ketbra{e_{\pi/4}}{g_{\pi/4}} \right), \\
    \big(V(s)\big)_{\tn{d}} &=\left(\bs{\Gamma}(s) + \frac{1}{2}\big( 1 + f^2(s) \big) I_{\rm sys}\right)I_m - \sqrt{2\Delta} \big( f(s)A_a + B_b\big)\Pi_{+,\pi/4}.
\end{align}

The second-order effective Hamiltonian is given by
\begin{equation}\label{eq:subdivision-effective-hamiltonian-v2::app}
    H_\tn{eff}^{[2]}(s) = \big( V(s) \big)_{--} - \frac{1}{\Delta} \big( V(s) \big)_{-+} \Pi_{+,\pi/4} \big( V(s) \big)_{+-}.
\end{equation}
Noting that $A_a^2 = B_b^2 = I$ and $\comm{A_a}{B_b} = 0$, we have 
\begin{equation}
    H_\tn{eff}^{[2]}(s) = \big( \bs{\Gamma}(s) + f(s) A_a B_b \big) \Pi_{-,\pi/4} = H(s) \Pi_{-,\pi/4} = \V_{\pi/4} H(s) \V_{\pi/4}^\dagger.
\end{equation}

By Fact~\ref{fact:c2-periodic-functions}, both $\wt{H}(s)$ and $S^{[2]}(s)$ are closed ${\rm C}^2$ functions of $s$.
Applying \cref{lem:second-order-reduction} gives an $(\eta, \varepsilon, \kappa)$-regular simulation of the effective Hamiltonian up to second order, and \cref{cor:parallel-second-order-simulation-berry-phase} gives criteria for controlling the Berry phase to the desired accuracy.
\subsubsection{Pauli-to-Ising gadget}

Let the parameterised target Hamiltonian be
\begin{equation}
    H(s) = \bs{\Gamma}(s) + f(s)X_aY_b.
\end{equation}
Let the fixed encoding isometry be $\V_{\pi/4} = I_{\rm sys}\otimes\ket{g_{\pi/4}}_m$ for the mediator qubit $m$.

Construct the simulator Hamiltonian $\wt{H}(s) = \Delta\Pi_{+,\pi/4} + V(s)$ with $V(s) = \sqrt{\Delta}V_1(s) + V_2(s)$ such that
\begin{align}
    V_1(s) &= f(s)X_a\hat{X}_m + Y_b\hat{Y}_m, & 
    V_2(s) &= \bs{\Gamma}(s) + \frac{1}{2}\big(1+f^2(s)\big)I_{\rm sys}.
\end{align}
The two-qubit interactions appearing in $V_1(s)$ are of the Ising forms $X_aX_m$ and $Y_bY_m$.
Define $O_{\alpha\beta}=\Pi_{\alpha,\pi/4}O\Pi_{\beta,\pi/4}$ for $\alpha,\beta\in\{+,-\}$.
A short calculation shows that
\begin{align}
    \big( V(s) \big)_{--} &= \left(\bs{\Gamma}(s)+\frac{1}{2}\big(1+f^2(s)\big)I_{\rm sys}\right)\Pi_{-,\pi/4}, \\
    \big( V(s) \big)_{++} &= \left(\bs{\Gamma}(s)+\frac{1}{2}\big(1+f^2(s)\big)I_{\rm sys}-\sqrt{2\Delta}\big(f(s)X_a+Y_b\big)\right)\Pi_{+,\pi/4}, \\
    \big( V(s) \big)_{-+} &= \i\sqrt{\frac{\Delta}{2}}\big(-f(s)X_a+Y_b\big)\ketbra{g_{\pi/4}}{e_{\pi/4}}, \\
    \big( V(s) \big)_{+-} &= -\i\sqrt{\frac{\Delta}{2}}\big(-f(s)X_a+Y_b\big)\ketbra{e_{\pi/4}}{g_{\pi/4}}.
\end{align}
It follows that
\begin{align}
    \big( V(s) \big)_{\tn{od}} &= \i\sqrt{\frac{\Delta}{2}} \big(-f(s)X_a + Y_b\big) (\ketbra{g_{\pi/4}}{e_{\pi/4}} - \ketbra{e_{\pi/4}}{g_{\pi/4}}), \\
    \big( V(s) \big)_{\tn{d}} &= \big(\bs{\Gamma}(s) + \frac{1}{2}\big( 1 + f^2(s) \big)I_{\rm sys} \big)I_m - \sqrt{2\Delta}\big( f(s)X_a + Y_b \big) \Pi_{+,\pi/4}.
\end{align}

The second-order contribution is
\begin{equation}
    \frac{1}{\Delta}\big( V(s) \big)_{-+}\Pi_{+,\pi/4}\big( V(s) \big)_{+-} = \frac{1}{2}\big( -f(s)X_a + Y_b \big)^2 \Pi_{-,\pi/4}.
\end{equation}
Noting that $X_a^2 = Y_b^2 = I$ and $\comm{X_a}{Y_b} = 0$, we have
\begin{equation}
    \frac{1}{2}\big( -f(s)X_a + Y_b \big)^2 = \frac{1}{2}\big( 1 + f^2(s) \big)I_{\rm sys} - f(s)X_aY_b.
\end{equation}
Therefore,
\begin{equation}
    H_{\tn{eff}}^{[2]}(s) = \big( \bs{\Gamma}(s) + f(s)X_aY_b \big)\Pi_{-,\pi/4} = H(s)\Pi_{-,\pi/4} = \V_{\pi/4}H(s)\V_{\pi/4}^{\dagger}.
\end{equation}

By Fact~\ref{fact:c2-periodic-functions}, both $\wt{H}(s)$ and $S^{[2]}(s)$ are closed ${\rm C}^2$ functions of $s$.
Applying \cref{lem:second-order-reduction} gives an $(\eta, \varepsilon, \kappa)$-regular simulation of the effective Hamiltonian up to second order, and \cref{cor:parallel-second-order-simulation-berry-phase} gives criteria for controlling the Berry phase to the desired accuracy.

\subsubsection{Ising-to-XX gadget}
Let the parameterised target Hamiltonian be
\begin{equation}
    H(s) = \bs{\Gamma}(s) + f(s)X_aX_b.
\end{equation}
Define
\begin{align}
    \ket{-i} &= \frac{1}{\sqrt{2}}(\ket{0} - \i\ket{1}), & \ket{i} &= \frac{1}{\sqrt{2}}(\ket{0} + \i\ket{1}), \\
    \Pi_{-,i} &= \ketbra{-i}{-i}, & \Pi_{+,i} &= \ketbra{i}{i}.
\end{align}
Let the fixed encoding isometry be $\V_i = I_{\rm sys}\otimes\ket{-i}_m$ for the mediator qubit $m$.

For $k\in\{a,b\}$, define the shifted XX interaction
\begin{equation}
    \hat{\Psi}_{km}^{XX} = X_k X_m + Y_k (Y_m + I_m).
\end{equation}
The two-qubit part of $\hat{\Psi}_{km}^{XX}$ is the XX interaction $X_kX_m + Y_kY_m$, while the one-local term ensures that its ground-space block vanishes.

Construct the simulator Hamiltonian $\wt{H}(s) = \Delta\Pi_{+,i} + V(s)$ with $V(s) = \sqrt{\Delta}V_1(s) + V_2(s)$ such that
\begin{align}
    V_1(s) &= \frac{1}{\sqrt{2}}\left(f(s)\hat{\Psi}_{am}^{XX} - \hat{\Psi}_{bm}^{XX}\right), & 
    V_2(s) &= \bs{\Gamma}(s) + \frac{1}{2}\big( 1 + f^2(s) \big)I_{\rm sys}.
\end{align}
A short calculation therefore gives
\begin{align}
    \big( V(s) \big)_{--} &= \left(\bs{\Gamma}(s) + \frac{1}{2}\big(1 + f^2(s)\big) I_{\rm sys}\right)\Pi_{-,i}, \\
    \big( V(s) \big)_{++} &= \left(\bs{\Gamma}(s) + \frac{1}{2}\big(1 + f^2(s)\big) I_{\rm sys} + \sqrt{2\Delta}\big(f(s)Y_a - Y_b\big)\right)\Pi_{+,i}, \\
    \big( V(s) \big)_{-+} &= \i\sqrt{\frac{\Delta}{2}}\big(f(s)X_a - X_b\big)\ketbra{-i}{i}, \\
    \big( V(s) \big)_{+-} &= -\i\sqrt{\frac{\Delta}{2}}\big(f(s)X_a - X_b\big)\ketbra{i}{-i}.
\end{align}
It follows that
\begin{align}
    \big(V(s)\big)_{\tn{od}} &= \i\sqrt{\frac{\Delta}{2}}\big(f(s)X_a - X_b\big)\left(\ketbra{-i}{i} - \ketbra{i}{-i}\right), \\
    \big(V(s)\big)_{\tn{d}} &= \big(\bs{\Gamma}(s) + \frac{1}{2}\big(1 + f^2(s)\big)I_{\rm sys}\big)I_m + \sqrt{2\Delta}\big(f(s)Y_a - Y_b\big)\Pi_{+,i}.
\end{align}

The second-order contribution is
\begin{equation}
    \frac{1}{\Delta}\big(V(s)\big)_{-+}\Pi_{+,i}\big(V(s)\big)_{+-} = \frac{1}{2}\big(f(s)X_a - X_b\big)^2\Pi_{-,i}.
\end{equation}
Using $X_a^2 = X_b^2 = I$ and $\comm{X_a}{X_b} = 0$, we obtain
\begin{equation}
    \frac{1}{2}\big(f(s)X_a - X_b\big)^2 = \frac{1}{2}\big(1 + f^2(s)\big)I_{\rm sys} - f(s)X_aX_b.
\end{equation}
Therefore,
\begin{equation}
    H_{\tn{eff}}^{[2]}(s) = \big(\bs{\Gamma}(s) + f(s)X_aX_b\big)\Pi_{-,i} = H(s)\Pi_{-,i} = \V_iH(s)\V_i^{\dagger}.
\end{equation}

The corresponding constructions for $Y_aY_b$ and $Z_aZ_b$ are obtained by applying the cyclic permutation $X\mapsto Y\mapsto Z\mapsto X$ to the target, simulator, and penalty Pauli operators.
By Fact~\ref{fact:c2-periodic-functions}, both $\wt{H}(s)$ and $S^{[2]}(s)$ are closed ${\rm C}^2$ functions of $s$.
Applying \cref{lem:second-order-reduction} gives an $(\eta, \varepsilon, \kappa)$-regular simulation of the effective Hamiltonian up to second order, and \cref{cor:parallel-second-order-simulation-berry-phase} gives criteria for controlling the Berry phase to the desired accuracy.

\subsubsection{XX-to-Heisenberg gadget}

Let the parameterised target Hamiltonian be
\begin{equation}
    H(s) = \bs{\Gamma}(s) + f(s)(X_aX_b+Y_aY_b).
\end{equation}
Let the fixed encoding isometry be $\V = I_{\rm sys}\otimes\ket{0}_m$ for the mediator qubit $m$.

For $k\in\{a,b\}$, define the shifted Heisenberg interaction
\begin{equation}
    \hat{\Psi}_{km}^{\mathrm{Heis}} = X_kX_m + Y_kY_m + Z_k(Z_m - I_m).
\end{equation}
The two-qubit part of $\hat{\Psi}_{km}^{\mathrm{Heis}}$ is the Heisenberg interaction $X_kX_m+Y_kY_m+Z_kZ_m$, while the one-local term ensures that its ground-space block vanishes.

Construct the simulator Hamiltonian $\wt{H}(s)=\Delta\Pi_++V(s)$ with $V(s)=\sqrt{\Delta}V_1(s)+V_2(s)$ such that
\begin{align}
    V_1(s) &= \frac{1}{\sqrt{2}}\big(f(s)\hat{\Psi}_{am}^{\mathrm{Heis}} - \hat{\Psi}_{bm}^{\mathrm{Heis}}\big), & 
    V_2(s) &= \bs{\Gamma}(s)+\big(1 + f^2(s)\big)I_{\rm sys} - f^2(s)Z_a - Z_b.
\end{align}
A short calculation therefore gives
\begin{align}
    \big(V(s)\big)_{--} &= \left(\bs{\Gamma}(s)+\big(1+f^2(s)\big)I_{\rm sys}-f^2(s)Z_a-Z_b\right)\Pi_-, \\
    \big(V(s)\big)_{++} &= \left(\bs{\Gamma}(s)+\big(1+f^2(s)\big)I_{\rm sys}-f^2(s)Z_a-Z_b+\sqrt{2\Delta}\big(-f(s)Z_a+Z_b\big)\right)\Pi_+, \\
    \big(V(s)\big)_{-+} &= \sqrt{\frac{\Delta}{2}}\left(f(s)(X_a-\i Y_a)-(X_b-\i Y_b)\right)\ketbra{0}{1}, \\
    \big(V(s)\big)_{+-} &= \sqrt{\frac{\Delta}{2}}\left(f(s)(X_a+\i Y_a)-(X_b+\i Y_b)\right)\ketbra{1}{0}.
\end{align}
It follows that
\begin{align}
    \big(V(s)\big)_{\tn{od}} &= \sqrt{\frac{\Delta}{2}}\big( f(s)(X_a - \i Y_a) - (X_b - \i Y_b) \big)\ketbra{0}{1} + \sqrt{\frac{\Delta}{2}}\big( f(s)(X_a + \i Y_a) - (X_b + \i Y_b) \big)\ketbra{1}{0}, \\
    \big(V(s)\big)_{\tn{d}} &= \big( \bs{\Gamma}(s) + \big( 1 + f^2(s) \big)I_{\rm sys} - f^2(s)Z_a - Z_b \big)I_m + \sqrt{2\Delta}\big( -f(s)Z_a + Z_b \big)\Pi_+.
\end{align}

To evaluate the second-order contribution, first note that $(X_k-\i Y_k)(X_k+\i Y_k)=2(I-Z_k)$.
The cross terms satisfy $(X_a-\i Y_a)(X_b+\i Y_b)+(X_b-\i Y_b)(X_a+\i Y_a)=2(X_aX_b+Y_aY_b)$.
It follows that
\begin{equation}
    \frac{1}{\Delta}\big(V(s)\big)_{-+}\Pi_+\big(V(s)\big)_{+-}=\left(\big(1+f^2(s)\big)I_{\rm sys}-f^2(s)Z_a-Z_b-f(s)(X_aX_b+Y_aY_b)\right)\Pi_-.
\end{equation}

The second-order effective Hamiltonian is given by the same expression as \cref{eq:subdivision-effective-hamiltonian::app}.
Substituting the preceding expressions gives
\begin{equation}
    H_{\tn{eff}}^{[2]}(s) = \big(\bs{\Gamma}(s) + f(s)(X_aX_b + Y_aY_b)\big)\Pi_- = H(s)\Pi_- = \V H(s)\V^\dagger.
\end{equation}

The analogous constructions for the target interactions $Y_aY_b+Z_aZ_b$ and $Z_aZ_b+X_aX_b$ follow by cyclically permuting the Pauli operators and choosing the penalty projectors in the $X$ and $Y$ eigenbases, respectively.
By Fact~\ref{fact:c2-periodic-functions}, both $\wt{H}(s)$ and $S^{[2]}(s)$ are closed ${\rm C}^2$ functions of $s$.
Applying \cref{lem:second-order-reduction} gives an $(\eta, \varepsilon, \kappa)$-regular simulation of the effective Hamiltonian up to second order, and \cref{cor:parallel-second-order-simulation-berry-phase} gives criteria for controlling the Berry phase to the desired accuracy.